\documentclass[acmsmall,screen,nonacm]{acmart}

\renewcommand\footnotetextcopyrightpermission[1]{}
\AtBeginDocument{%
  
}

\usepackage{mathtools}
\usepackage{mathrsfs}
\usepackage{braket}
\newif\ifMintedLiveHighlight
\ifdefined\pdfshellescape
  \ifnum\pdfshellescape=1 \MintedLiveHighlighttrue \fi
\fi
\ifMintedLiveHighlight
  \usepackage[newfloat]{minted}
\else
  \usepackage[newfloat,frozencache]{minted}
\fi
\setmintedinline{breaklines}
\newcommand{\edgeBZ}{\mathrm{edgeB}_Z}
\newcommand{\YsubZero}{_{Y_0}}
\newcommand{\translationOfV}{\mathcal{T}(V)}
\usepackage{multirow}
\usepackage{latexml}
\providecommand{\anon}[2][]{#2}
\makeatletter
\@ifundefined{if@ACM@anonymous}{\newif\if@ACM@anonymous}{}
\makeatother
\usepackage{array}
\newcolumntype{L}[1]{>{\raggedright\arraybackslash}p{#1}}
\iflatexml
  \newcommand{\typewriterInlineCode}[3][]{\texttt{#3}}
  \newcommand{\typewriterInlineCodeInHtml}{\let\mintinline\typewriterInlineCode}
\else
  \newcommand{\typewriterInlineCodeInHtml}{}
\fi
\iflatexml
  \makeatletter
  \let\mintinlineWithoutOptions\mintinline
  \renewcommand{\mintinline}[2][]{%
    \ifx\relax#1\relax\expandafter\@firstoftwo\else\expandafter\@secondoftwo\fi
    {\mintinlineWithoutOptions{#2}}{\lstinline[language=#2,#1]}}
  \@ifundefined{@@listings@block}{}{%
    \let\listingsBlockInSurroundingFont\@@listings@block
    \let\listingGroupWithoutTypewriterClass\@listingGroup
    \newcommand{\listingGroupWithTypewriterClass}[1]{%
      \listingGroupWithoutTypewriterClass{#1 ltx_font_typewriter}}%
    \renewcommand{\@@listings@block}[3]{{\ttfamily
      \let\@listingGroup\listingGroupWithTypewriterClass
      \listingsBlockInSurroundingFont{#1}{#2}{#3}}}}
  \AtBeginDocument{}
  \makeatother
\fi
\iflatexml
  \newcolumntype{N}{r}
\else
  \newcolumntype{N}{r@{\enspace}}
\fi
\usepackage{graphicx}
\iflatexml\else
  \usepackage{zref-savepos}
\fi
\usepackage{tikz}
\usetikzlibrary{positioning,arrows.meta,calc}
\DeclareRobustCommand{\strength}[1]{\tikz[baseline=-0.6ex]{\draw[line width=0.4pt] (0,0) circle (1.2ex);
  \fill (0,0) -- (90:1.2ex) arc[start angle=90, delta angle={-#1*360}, radius=1.2ex] -- cycle;}}
\newcounter{modularityrow}
\newcommand{\modularityrow}[1]{\refstepcounter{modularityrow}\label{row:#1}\themodularityrow}

\allowdisplaybreaks

\newcounter{figureAnnotationCounter}
\iflatexml
  
  \newcommand{\figureAnnotation}[1]{%
    \quad\normalfont\text{\scriptsize(#1)}%
  }
\else
  
  \newcommand{\figureAnnotation}[1]{%
    \stepcounter{figureAnnotationCounter}%
    \zsavepos{figureAnnotation\the\value{figureAnnotationCounter}}%
    \rlap{\hbox to
      \dimexpr\zposx{figureAnnotationLeft}sp+\linewidth-\zposx{figureAnnotation\the\value{figureAnnotationCounter}}sp\relax
    {\hfill\normalfont\text{\scriptsize #1}}}%
  }
\fi

\AtBeginDocument{%
  \theoremstyle{acmdefinition}%
  \newtheorem{remark}[theorem]{Remark}%
}

\ExplSyntaxOn
\ior_new:N \g__mixin_in_ior
\iow_new:N \g__mixin_out_iow
\str_new:N \l__mixin_line_str
\str_new:N \l__mixin_dir_str
\str_new:N \l__mixin_name_str
\str_new:N \l__mixin_ext_str
\str_new:N \l__mixin_base_str
\str_new:N \l__mixin_mixinname_str
\cs_new_protected:Npn \mixinprovide #1#2 {
  \file_parse_full_name:nNNN {#1} \l__mixin_dir_str \l__mixin_name_str \l__mixin_ext_str
  \str_set:Nx \l__mixin_base_str { \l__mixin_name_str \l__mixin_ext_str }
  \cs_set:Npx \mixinbase { \str_use:N \l__mixin_base_str }
  \file_if_exist:nT {#1} {
    \ior_open:Nn \g__mixin_in_ior {#1}
    \iow_open:Nn \g__mixin_out_iow { \l__mixin_base_str }
    \ior_str_map_inline:Nn \g__mixin_in_ior {
      \str_set:Nn \l__mixin_line_str { ##1 }
      \int_compare:nNnTF {#2} = {1} {
        \regex_replace_once:nnN { \s*\#.* } { } \l__mixin_line_str
        \regex_replace_once:nnN { \s+ \Z } { } \l__mixin_line_str
        \str_if_empty:NF \l__mixin_line_str
          { \iow_now:Ne \g__mixin_out_iow { \str_use:N \l__mixin_line_str } }
      }{
        \iow_now:Ne \g__mixin_out_iow { \str_use:N \l__mixin_line_str }
      }
    }
    \ior_close:N \g__mixin_in_ior
    \iow_close:N \g__mixin_out_iow
  }
}
\ExplSyntaxOff

\ExplSyntaxOn
\cs_new_protected:Npn \__mixin_halve_indent:N #1 {
  \regex_replace_all:nnN { (\ \ )\ \  } { \1 } #1
}
\cs_new_protected:Npn \mixinstripprovide #1 {
  \file_parse_full_name:nNNN {#1} \l__mixin_dir_str \l__mixin_name_str \l__mixin_ext_str
  \str_set:Nx \l__mixin_base_str { \l__mixin_name_str \l__mixin_ext_str }
  \cs_set:Npx \mixinbase { \str_use:N \l__mixin_base_str }
  \str_set:Nx \l__mixin_mixinname_str { \l__mixin_name_str }
  \regex_replace_once:nnN { \.mixin \Z } { } \l__mixin_mixinname_str
  \file_if_exist:nT {#1} {
    \ior_open:Nn \g__mixin_in_ior {#1}
    \iow_open:Nn \g__mixin_out_iow { \l__mixin_base_str }
    \iow_now:Ne \g__mixin_out_iow { -~ \str_use:N \l__mixin_mixinname_str : }
    \ior_str_map_inline:Nn \g__mixin_in_ior {
      \str_set:Nn \l__mixin_line_str { ##1 }
      \regex_replace_once:nnN { \s*\#.* } { } \l__mixin_line_str
      \regex_replace_all:nnN { Builtin,\s* } { } \l__mixin_line_str
      \regex_replace_once:nnN { \s+ \Z } { } \l__mixin_line_str
      \str_if_empty:NF \l__mixin_line_str
        {
          \__mixin_halve_indent:N \l__mixin_line_str
          \iow_now:Ne \g__mixin_out_iow { \prg_replicate:nn { 2 } { ~ } \str_use:N \l__mixin_line_str }
        }
    }
    \ior_close:N \g__mixin_in_ior
    \iow_close:N \g__mixin_out_iow
  }
}
\ExplSyntaxOff
\iflatexml
  \newcommand{\inputmintedmixinbase}[2][]{\expandafter\inputmintedfile\expandafter{\mixinbase}{#1}{#2}}
  \newcommand{\inputmintedfile}[3]{\inputminted[#2]{#3}{#1}}
\else
  \newcommand{\inputmintedmixinbase}[2][]{\inputminted[#1]{#2}{\mixinbase}}
\fi
\newcommand{\mixinstrip}[2][]{\mixinstripprovide{#2}\inputmintedmixinbase[#1]{yaml}}

\ExplSyntaxOn
\cs_new_protected:Npn \mixincheatsheetprovide #1#2 {
  \file_parse_full_name:nNNN {#1} \l__mixin_dir_str \l__mixin_name_str \l__mixin_ext_str
  \str_set:Nx \l__mixin_base_str { \l__mixin_name_str \l__mixin_ext_str }
  \cs_set:Npx \mixinbase { \str_use:N \l__mixin_base_str }
  \file_if_exist:nT {#1} {
    \ior_open:Nn \g__mixin_in_ior {#1}
    \iow_open:Nn \g__mixin_out_iow { \l__mixin_base_str }
    \iow_now:Ne \g__mixin_out_iow { -~#2: }
    \ior_str_map_inline:Nn \g__mixin_in_ior {
      \str_set:Nx \l__mixin_line_str { \prg_replicate:nn { 2 } { ~ } ##1 }
      \iow_now:Ne \g__mixin_out_iow { \str_use:N \l__mixin_line_str }
    }
    \ior_close:N \g__mixin_in_ior
    \iow_close:N \g__mixin_out_iow
  }
}
\ExplSyntaxOff
\newcommand{\mixincheatsheet}[3][]{\mixincheatsheetprovide{#2}{#3}\inputmintedmixinbase[#1]{yaml}}

\newlength{\namedmixinindent}

\newcommand{\Path}{\mathrm{Path}}
\newcommand{\Reference}{\mathrm{Reference}}
\newcommand{\Tag}{\mathrm{Tag}}
\newcommand{\String}{\mathrm{String}}
\newcommand{\Label}{\mathrm{Label}}
\DeclareMathOperator{\dom}{dom}
\DeclareMathOperator{\overrides}{\mathit{overrides}}
\DeclareMathOperator{\supers}{\mathit{supers}}
\DeclareMathOperator{\inherits}{\mathit{inherits}}
\DeclareMathOperator{\has}{\mathit{has}}

\newcommand{\doc}{\mathrm{doc}}
\DeclareMathOperator{\bases}{\mathit{bases}}
\DeclareMathOperator{\resolve}{\mathit{resolve}}
\DeclareMathOperator{\this}{\mathit{this}}
\newcommand{\snoc}{\mathbin{\triangleright}}
\newcommand{\cat}{\mathbin{\frown}}
\DeclareMathOperator{\init}{\mathit{init}}
\DeclareMathOperator{\last}{\mathit{last}}
\DeclareMathOperator{\at}{\mathit{at}}
\DeclareMathOperator{\items}{\mathit{items}}
\DeclareMathOperator{\lbl}{\mathit{label}}
\DeclareMathOperator{\lbls}{\mathit{labels}}
\DeclareMathOperator{\isnull}{\mathit{null}}
\DeclareMathOperator{\kv}{\mathit{kv}}
\DeclareMathOperator{\lexical}{\mathit{lexical}}
\DeclareMathOperator{\qualified}{\mathit{qualified}}
\DeclareMathOperator{\indexed}{\mathit{indexed}}
\DeclareMathOperator{\eval}{\mathit{eval}}
\DeclareMathOperator{\lfp}{lfp}

\newcommand{\olbl}[1]{\text{\mintinline{yaml}{#1}}}
\iflatexml
  \newcommand{\tuple}[1]{\mathord{(#1)}}
\else
  \newcommand{\tuple}[1]{(#1)}
\fi
\iflatexml
  \newenvironment{breakablepathcommas}{}{}
\else
  \newenvironment{breakablepathcommas}{%
    \mathchardef\unbreakablecomma=\mathcode`\,%
    \mathcode`\,="8000
    \begingroup\lccode`\~=`\,\lowercase{\endgroup\def~}{\unbreakablecomma\allowbreak}%
  }{}
\fi
\iflatexml
  \renewcommand{\hphantom}[1]{}
\fi

\newcommand{\YAMLNode}{\mathrm{Node}}
\newcommand{\YAMLScalar}{\mathrm{Scalar}}
\newcommand{\YAMLSequence}{\mathrm{Sequence}}
\newcommand{\YAMLMapping}{\mathrm{Mapping}}
\newcommand{\seqtag}{\text{\ttfamily !!seq}}
\newcommand{\maptag}{\text{\ttfamily !!map}}
\newcommand{\strtag}{\text{\ttfamily !!str}}
\newcommand{\nulltag}{\text{\ttfamily !!null}}

\makeatletter
\@ifundefined{ifPaperChinese}{\newif\ifPaperChinese\PaperChinesefalse}{}
\makeatother

\newcommand{\bilingual}[2]{#1}
\ifPaperChinese
  \usepackage{CJKutf8}
  \makeatletter
  \DeclareFontFamily{C70}{song}{\hyphenchar\font\m@ne}
  \DeclareFontShape{C70}{song}{m}{n}{<-> CJK * gbsnu}{\CJKnormal}
  \DeclareFontShape{C70}{song}{bx}{n}{<-> CJKb * gbsnu}{\CJKbold}
  \makeatother
  \renewcommand{\bilingual}[2]{\texorpdfstring{#2}{#1}}
  \makeatletter
  \let\InheritanceParagraphFont\@parfont
  \def\@parfont{\normalfont\InheritanceParagraphFont}
  \makeatother
\fi

\usepackage[ruled,vlined,linesnumbered]{algorithm2e}
\SetKwProg{Fn}{function}{:}{}
\SetKwFunction{TABLED}{Tabled}
\SetKwFunction{RESOLVE}{Resolve}

\begin{document}
\ifPaperChinese
  \begin{CJK*}{UTF8}{gbsn}
  \renewcommand{\proofname}{证明}
  \makeatletter
  \newcommand{\inheritanceThmName}[1]{\@ifundefined{inheritance@thm@#1}{#1}{\csname inheritance@thm@#1\endcsname}}
  \renewcommand{\thmname}[1]{\inheritanceThmName{#1}}
  \@namedef{inheritance@thm@Theorem}{定理}
  \@namedef{inheritance@thm@Lemma}{引理}
  \@namedef{inheritance@thm@Corollary}{推论}
  \@namedef{inheritance@thm@Proposition}{命题}
  \@namedef{inheritance@thm@Conjecture}{猜想}
  \@namedef{inheritance@thm@Definition}{定义}
  \@namedef{inheritance@thm@Example}{例}
  \@namedef{inheritance@thm@Remark}{注}
  \@namedef{inheritance@thm@Notation}{记法}
  \makeatother
\fi

\title[A Calculus of Inheritance]{\bilingual{A Calculus of Inheritance}{一种继承的演算}}
\author{Bo Yang}
\orcid{0000-0003-2757-9115}
\affiliation{%
  \institution{Figure AI Inc.}
  \city{San Jose}
  \state{California}
  \country{USA}
}
\email{yang-bo@yang-bo.com}
\makeatletter
\if@ACM@anonymous\else
  \thanks{This work was conducted independently prior to the author's employment at Figure AI}
\fi
\makeatother

\begin{abstract}
\ifPaperChinese
  正如 $\lambda$-演算是一门由三个原语(抽象、应用、变量)构成的纯函数式编程语言,我们提出继承演算;它是一门由三个原语(mixin、定义、引用)构成的纯声明式编程语言。我们引入 multi-target $\mathtt{this}$ 来解决继承与模块化相冲突的问题,从而统一面向对象的构造与配置的构造:模块、类、对象、方法、字段和局部变量成为单一的 mixin 抽象,而继承简单地成为集合并,即配置语言的深度合并。因此,组合是可交换、幂等且可结合的,从结构上消除了多重继承的线性化问题;继承关系具有可替换性;每个程序的每条路径都可以被继承和扩展;每棵子树都可以拆成若干部分,把它们一起继承,行为与整体相同;并且程序没有副作用,覆盖每个输入。我们把这些性质称为由构造保证的模块化。该演算没有函数,却宏表达了惰性 $\lambda$-演算,并且按 Felleisen 的意义比它更具宏表达力。其语义是一阶的、指称的,并由 tabling 求值,即使对循环继承层级也是如此。

  继承演算提炼自 MIXINv2,一个实现;在该实现中,同一份代码可充当不同的函数颜色,而普通算术符合逻辑编程的关系语义。
\else
  Just as the $\lambda$-calculus is a pure functional programming language built from three primitives (abstraction, application, variable), we present inheritance-calculus, a pure declarative programming language built from three primitives (mixin, definition, reference). We introduce multi-target $\mathtt{this}$ to solve the problem that inheritance conflicts with modularity, and thereby unify object-oriented and configuration constructs: modules, classes, objects, methods, fields, and locals become a single mixin abstraction, and inheritance becomes simply set union, the deep merge of configuration languages. Consequently, composition is commutative, idempotent, and associative, structurally eliminating the multiple-inheritance linearization problem; the inheritance relation has substitutability; every path of every program can be inherited and extended; every subtree can be split into parts that, inherited together, behave as the whole; and programs have no side effects and cover every input. We call these properties modularity by construction. The calculus has no functions, yet it macro-expresses the lazy $\lambda$-calculus and is more macro-expressive than it in Felleisen's sense. Its semantics is first-order, denotational, and evaluated by tabling, even for cyclic inheritance hierarchies.

  Inheritance-calculus is distilled from MIXINv2, an implementation in which the same code acts as different function colors and ordinary arithmetic agrees with the relational semantics of logic programming.
\fi
\end{abstract}



\maketitle

\section{\bilingual{Introduction}{引言}}
\label{sec:introduction}
\ifPaperChinese
本文提出一个问题:
\begin{quote}
  一门语言能否既是通用的,又由构造保证模块化?
\end{quote}
\else
This paper asks one question:
\begin{quote}
  Can a language be both general-purpose and modular by construction?
\end{quote}
\fi

\ifPaperChinese
Meyer 列出五条模块化标准:可分解性、可组合性、可理解性、连续性与保护性~\cite[第~40--46~页]{meyer1997-object-oriented-software-construction}。我们说一门语言具有\emph{由构造保证的模块化}(modularity by construction),当它对每个程序保证全部五条,而不是把它们留给程序员靠纪律和设计模式逐个类去遵守的原则。面向对象语言、配置语言与逻辑编程各自追求过这些标准,分别靠继承、靠合并、靠添加子句,而每一条路线都由构造拿到其中一部分(第~\ref{sec:prior-modularity}~节与表~\ref{tab:modularity-means})。
\else
Meyer names five criteria of modularity: decomposability, composability, understandability, continuity,
and protection~\cite[pp.~40--46]{meyer1997-object-oriented-software-construction}.
We say that a language has \emph{modularity by construction} when it guarantees all five for every
program, rather than leave them to principles that programmers uphold class by class through discipline and design
patterns. Object-oriented languages, configuration languages, and logic programming have each pursued
these criteria, by inheritance, by merge, and by adding clauses, and each obtains some of them by
construction (Section~\ref{sec:prior-modularity} and Table~\ref{tab:modularity-means}).
\fi

\ifPaperChinese\SetupFloatingEnvironment{listing}{name=清单}\fi
\begin{listing}[tbp]
  \caption{\ifPaperChinese
  初看 inheritance-calculus:一个复数类及其加法方法,用 inheritance-calculus 记号表示。每行从左到右分三栏:inheritance-calculus 源,其记法在第~\ref{sec:syntax}~节精确定义;\texttt{\#} 之后是该行扮演的角色;\texttt{|} 之后是它所对应的 Python 代码。每行的角色(模块、类、字段、方法、参数、局部、应用、实例、import)只是标注:每一行都是同一种 \mintinline{yaml}{- name:} 定义或方括号引用继承,而非一种独立的语言特性。
  \else
  An inheritance-calculus example: a complex-number class with an addition
  method, expressed in inheritance-calculus notation.
  Each line reads left to right in three columns: the inheritance-calculus
  source, whose notation Section~\ref{sec:syntax} defines precisely; after
  \texttt{\#}, the role that line plays; and after \texttt{|},
  the corresponding Python code.
  The role on each line (module, class, field, method, parameter, local,
  application, instance, and import) is only an annotation: every line is
  the same \mintinline{yaml}{- name:} definition or bracketed-reference
  inheritance, not a separate language feature.
  \fi}
  \label{lst:cheatsheet}
  \mixincheatsheet{../packages/mixinv2-examples/src/mixinv2_examples/ComplexModule.mixin.yaml}{ComplexModule}%
\end{listing}

\ifPaperChinese
面向对象语言伸向过配置语言的自由组合,配置语言伸向过面向对象语言的计算;两者停在同一处:组合深度嵌套的模块时,要么以错误失败,要么按次序只保留被组合的模块中的一个(第~\ref{sec:prior-modularity}~节)。逻辑编程不会走到那个地方:定子程序以添加子句组合,能做一般计算,但既没有嵌套的作用域,也没有晚绑定的 self。
\else
Object-oriented languages have reached toward the free combination of configuration languages, and
configuration languages toward the computation of object-oriented languages; both stop at the same
place, where a composition of deeply nested modules either fails with an error or keeps only one of
the modules it composes, chosen by their order (Section~\ref{sec:prior-modularity}).
Logic programming does not reach that place: definite programs compose by adding clauses and compute
in general, but have neither nested scopes nor a late-bound self.
\fi

\label{sec:introduction-approach}
\ifPaperChinese
本文提出的继承演算越过了那个地方:它的覆写就是深度合并,即配置语言的组合方式:在同一名字下相遇的两个定义,无论在多深的层次,都被保留,而不是按优先级保留一个,或以拒绝一个不留。继承演算的程序只由成员定义及其组合构成,即面向对象语言与配置语言共有的构造;面向对象语言把这种组合称为继承,配置语言把它称为合并。我们称这样的声明式语言是\emph{纯}(pure)的,如同纯 $\lambda$-演算只保留函数。模块、类、对象、方法、字段与局部绑定都是同一种构造,即由继承组合的定义 mixin(清单~\ref{lst:cheatsheet}),因此面向对象的构造与配置的构造成为一种。
\else
Inheritance-calculus, the language this paper presents, goes past that place: its overriding is deep
merge, the composition of configuration languages: two definitions that
meet under one name, at any depth, are both kept, rather than one by priority or neither by rejection.
A program of inheritance-calculus consists only of member definitions and their composition, the
constructs that object-oriented and configuration languages share; object-oriented languages call this
composition inheritance, and configuration languages call it merge.
We call such a declarative language \emph{pure}, as the pure $\lambda$-calculus keeps functions alone.
A module, a class, an object, a method, a field, and a local binding are all the same construct, a
mixin of definitions composed by inheritance (Listing~\ref{lst:cheatsheet}), so
object-oriented constructs and configuration constructs become one.
\fi

\ifPaperChinese
计算也由同样的 mixin 完成:晚绑定的引用使 mixin 能够计算,而惰性 $\lambda$-演算可以嵌入其中(第~\ref{sec:translation}~节)。于是程序中做计算的部分,正是深度合并在任意深度上组合的部分,关于组合的保证就是关于每个程序的保证:在一门通用语言中,模块化的各项要求由构造成立(第~\ref{sec:modularity}~节)。所以对开篇问题的回答是肯定的。
\else
Computation is made of the same mixins: late-bound references let a mixin compute, and the lazy
$\lambda$-calculus embeds in it (Section~\ref{sec:translation}).
The parts of a program that compute are thus the parts that deep merge composes, at any depth, and a
guarantee about composition is a guarantee about every program: the requirements of modularity hold by
construction in a general-purpose language (Section~\ref{sec:modularity}).
The answer to the opening question is therefore yes.
\fi

\ifPaperChinese
我们的贡献有三项。
\begin{description}
  \item[multi-target $\this$] 一种 $\this$ 的新语义,它在求值中沿到达一个晚绑定引用的作用域的每条继承路线解析该引用,使面向对象语言的继承与配置语言的合并能够成为同一种组合,即深度合并。
  \item[纯声明式编程] 继承演算,一门以深度合并为组合的纯声明式语言。它的语义是一阶的最小不动点,由 tabling 求值,循环继承也包括在内;惰性 $\lambda$-演算可以嵌入其中,并对 Lévy--Longo 树充分;按 Felleisen 的意义,它比惰性 $\lambda$-演算更具宏表达力。
  \item[由构造保证的模块化] 把 Meyer 的五条标准各读作形式定义的组合性质;继承演算由构造具有这些性质的证明;以及对替代设计的考察,其中没有一个由构造满足全部五条标准。
\end{description}
我们还观察到:普通算术符合逻辑编程的关系语义,并有 Datalog 编码;同一份代码可充当不同的函数颜色~\cite{nystrom2015-function-color};值递归使循环规则收敛(第~\ref{sec:emergent-phenomena}~节)。
附录~\ref{app:well-definedness} 的结果、第~\ref{sec:modularity}~节关于半格律与分解的定理,以及该节可替换性定理的语义部分,已在 Lean~4 中机械化\anon[~(补充材料)]{~\cite{yang2026-mixinv2}};其余结果为纸面证明。
\else
We make three contributions.
\begin{description}
  \item[Multi-target $\this$] A new semantics of $\this$, which
    resolves a late-bound reference during evaluation along every inheritance route that reaches its
    scope, so that the inheritance of object-oriented languages and the merge of configuration languages
    can become one composition, deep merge.
  \item[Pure declarative programming] Inheritance-calculus, a pure declarative language whose
    composition is deep merge. Its semantics is a first-order least fixed point evaluated by tabling,
    cyclic inheritance included; the lazy $\lambda$-calculus embeds in it adequately for L\'evy--Longo
    trees; and it is more macro-expressive than the lazy $\lambda$-calculus in Felleisen's sense.
  \item[Modularity by construction] A reading of each of Meyer's five criteria as formally defined
    properties of composition; proofs that inheritance-calculus has these properties by construction; and
    a survey of alternative designs, none of which meets all five criteria by construction.
\end{description}
We also observe that ordinary arithmetic agrees with the relational semantics of logic programming, with
a Datalog encoding; that the same code acts as different function
colors~\cite{nystrom2015-function-color}; and that value recursion makes cyclic rules converge
(Section~\ref{sec:emergent-phenomena}).
The results of Appendix~\ref{app:well-definedness}, the theorems of Section~\ref{sec:modularity} on the
semilattice laws and on decomposition, and the semantic part of its substitutability theorem are
mechanized in Lean~4\anon[~(supplementary material)]{~\cite{yang2026-mixinv2}}; the remaining results
are proved on paper.
\fi

\ifPaperChinese
本文其余部分沿开篇问题的两半展开。第~\ref{sec:prior-modularity}~节回顾先前设计怎样追求每条标准,以及这些标准在哪里相冲突。第~\ref{sec:syntax}~节与第~\ref{sec:mixin-trees}~节定义继承演算的语法与语义。对问题中要求一门通用语言的那一半,我们以惰性 $\lambda$-演算作为可编程性的代理,第~\ref{sec:translation}~节证明第二项贡献中的嵌入与表达力结果;第~\ref{sec:modularity}~节以第三项贡献中的定理回答另一半。第~\ref{sec:case-study}~节与第~\ref{sec:examples}~节给出示例,第~\ref{sec:emergent-phenomena}~节报告在 MIXINv2\footnote{继承演算与 MIXINv2\anon[~(补充材料)]{~\cite{yang2026-mixinv2}} 在罕见情形下于语法与语义上皆有差别:MIXINv2 照顾人体工程学,而继承演算更简单一致。本文每个例子在两者中都合法。} 中编程时遇到的现象。第~\ref{sec:semantic-variants}~节考察替代设计,第~\ref{sec:related-work}~节讨论进一步的相关工作,第~\ref{sec:conclusion}~节总结。
\else
The rest of the paper follows the two halves of the opening question.
Section~\ref{sec:prior-modularity} recalls how prior designs pursue each criterion and where the
criteria conflict.
Sections~\ref{sec:syntax} and~\ref{sec:mixin-trees} define the syntax and the semantics of
inheritance-calculus. For the half that asks for a general-purpose language, we take the lazy
$\lambda$-calculus as a proxy for programmability, and Section~\ref{sec:translation} proves the
embedding and expressiveness results of the second contribution; Section~\ref{sec:modularity} answers
the other half with the theorems of the third. Sections~\ref{sec:case-study} and~\ref{sec:examples} give examples, and
Section~\ref{sec:emergent-phenomena} reports phenomena met while programming in MIXINv2\footnote{Inheritance-calculus and MIXINv2\anon[~(supplementary material)]{~\cite{yang2026-mixinv2}} differ in rare cases, in both syntax and semantics: MIXINv2 caters to ergonomics, while inheritance-calculus is simpler and more uniform. Every example in this paper is valid in both.}.
Section~\ref{sec:semantic-variants} examines alternative designs, Section~\ref{sec:related-work}
discusses further related work, and Section~\ref{sec:conclusion} concludes.
\fi

\section{\bilingual{Modularity in Prior Designs}{先前设计中的模块化}}
\label{sec:prior-modularity}
\begin{table}[tbp]
  \centering
  \small
  \renewcommand{\arraystretch}{0.98}
  \setlength{\tabcolsep}{3pt}
  \caption{\bilingual{Ways in which prior work obtains Meyer's criteria: composability (Comp.), protection (Prot.), understandability (Und.), continuity (Cont.), and decomposability (Dec.); the filled sector of a mark is its strength. Comp.: a third per law among commutativity, idempotence, and associativity. Prot.: an abnormal condition \strength{1/3} contained where it occurs, or encoded as pure data that is \strength{2/3} free of side effects or covers every input, \strength{1} both. Und.: inherited members keep \strength{1/3} their protocol, \strength{2/3} their types, \strength{1} their behavior; $^{<:}$ marks one stated through subtyping. Cont.: a third per dimension of extensibility, \strength{1} for every dimension. Dec.: the finest level down to which what a mechanism builds splits into parts that recombine into the whole's behavior, \strength{1/3} file-level definitions, \strength{2/3} sub-file definitions, \strength{1} AST leaves. A row sits in the group of its weakest guarantee; a blank means that its presumed intent does not require the criterion.}{先前工作获得 Meyer 各条标准的方式:可组合性(组合)、保护性(保护)、可理解性(理解)、连续性(连续)与可分解性(分解);标记的填充扇形是其强度。组合:交换律、幂等律与结合律中每成立一条占三分之一。保护:异常条件 \strength{1/3} 被限定在它发生的地方,或被编码为纯数据,\strength{2/3} 无副作用或覆盖每个输入,\strength{1} 两者兼具。理解:所继承的成员保留 \strength{1/3} 其协议,\strength{2/3} 其类型,\strength{1} 其行为;$^{<:}$ 标出借子类型表述的一档。连续:每有一个可扩展的维度占三分之一,\strength{1} 为每个维度。分解:机制所构造的东西能拆成重新组合后具有整体行为的各部分的最细一级,\strength{1/3} 文件级定义,\strength{2/3} 文件内定义,\strength{1} AST 叶。每行归入其最弱保证所在的组;留空表示其推定意图不需要这条标准。}}
  \label{tab:modularity-means}
  \begin{tabular}{@{}Nl*{5}{c}@{}}
    \toprule
    & & \bilingual{\textbf{Comp.}}{\textbf{组合}} & \bilingual{\textbf{Prot.}}{\textbf{保护}} & \bilingual{\textbf{Und.} (locality)}{\textbf{理解}(局部性)} & \bilingual{\textbf{Cont.}}{\textbf{连续}} & \bilingual{\textbf{Dec.}}{\textbf{分解}} \\
    \midrule
    \multicolumn{7}{@{}l}{\bilingual{\textit{By a principle applied case by case}}{\textit{靠逐案应用的原则}}} \\
    \modularityrow{convergent-operations} & \bilingual{convergent operations}{收敛操作} & \strength{2/3} &  &  &  &  \\
    \modularityrow{disciplined-exceptions} & \bilingual{disciplined exceptions}{有纪律的异常} &  & \strength{1/3} &  &  &  \\
    \modularityrow{liskov-substitution} & \bilingual{Liskov substitution}{Liskov 替换} &  &  & \strength{1} &  &  \\
    \modularityrow{subtype-proofs} & \bilingual{subtype proofs}{子类型证明} &  &  & \strength{1}\rlap{$^{<:}$} &  &  \\
    \modularityrow{open-closed} & \bilingual{open--closed}{开闭} &  &  &  & \strength{1/3} &  \\
    \modularityrow{add-dont-replace} & \bilingual{add, don't replace}{添加而不替换} &  &  &  & \strength{1/3} &  \\
    \modularityrow{multi-dimensional-separation} & \bilingual{multi-dimensional separation}{多维分离} &  &  &  & \strength{1} &  \\
    \modularityrow{drop-in-files} &\bilingual{drop-in files}{drop-in 文件} &  &  &  & \strength{1/3} & \strength{1/3} \\
    \addlinespace
    \multicolumn{7}{@{}l}{\bilingual{\textit{By convention}}{\textit{靠约定}}} \\
    \modularityrow{replicated-data-types} & \bilingual{replicated data types}{复制数据类型} & \strength{1} &  &  &  &  \\
    \modularityrow{nixos-modules} & \bilingual{NixOS modules}{NixOS 模块} & \strength{1} &  & \strength{2/3}\rlap{$^{<:}$} & \strength{2/3} & \strength{2/3} \\
    \modularityrow{functions-as-objects} & \bilingual{functions as objects}{对象模拟函数} &  &  &  &  & \strength{1} \\
    \modularityrow{exceptions-as-values} & \bilingual{exceptions as values}{异常作为值} &  & \strength{2/3} &  &  &  \\
    \modularityrow{assertion-redeclaration} & \bilingual{assertion redeclaration}{断言重声明} &  &  & \strength{1} &  &  \\
    \modularityrow{runtime-contracts} & \bilingual{runtime contracts}{运行时契约} &  &  & \strength{1} &  &  \\
    \modularityrow{visitor} & \bilingual{visitor}{访问者} &  &  &  & \strength{1/3} &  \\
    \modularityrow{interpreters-over-a-signature} & \bilingual{interpreters over a signature}{签名的解释器} &  &  &  & \strength{2/3} &  \\
    \modularityrow{coproducts-of-signatures} & \bilingual{coproducts of signatures}{签名的余积} &  &  &  & \strength{2/3} &  \\
    \modularityrow{class-families} & \bilingual{class families refined by inheritance}{靠继承精化的类族} &  &  &  & \strength{2/3} &  \\
    \addlinespace
    \multicolumn{7}{@{}l}{\bilingual{\textit{By construction}}{\textit{由构造}}} \\
    \modularityrow{linearization} & \bilingual{multiple inheritance with linearization}{带有线性化的多重继承} & \strength{1/3} &  &  &  &  \\
    \modularityrow{sharing-by-name} & \bilingual{sharing by name}{按名共享} & \strength{1/3} &  &  &  &  \\
    \modularityrow{trait-sum} & \bilingual{trait sum}{trait 之和} & \strength{1} &  &  &  & \strength{2/3} \\
    \modularityrow{symmetric-merge} & \bilingual{symmetric concatenation}{对称拼接} & \strength{2/3} &  &  &  & \strength{2/3} \\
    \modularityrow{override} & \bilingual{module override}{模块覆写} & \strength{2/3} &  &  &  &  \\
    \modularityrow{superimposition} & \bilingual{superimposition}{叠加} & \strength{2/3} &  &  &  & \strength{2/3} \\
    \modularityrow{unification} & \bilingual{unification}{合一} & \strength{1} & \strength{2/3} &  &  & \strength{2/3} \\
    \modularityrow{choice} & \bilingual{choice}{选择} & \strength{1} &  &  &  & \strength{2/3} \\
    \modularityrow{clause-union} & \bilingual{clause union}{子句的并} & \strength{1} &  &  & \strength{2/3} & \strength{2/3} \\
    \modularityrow{effect-systems} & \bilingual{effect systems}{效果系统} &  & \strength{1/3} &  &  &  \\
    \modularityrow{monads} & \bilingual{monads}{monad} &  & \strength{2/3} &  &  &  \\
    \modularityrow{total-functional-programming} & \bilingual{total functional programming}{全函数式编程} &  & \strength{1} &  &  &  \\
    \modularityrow{subclassing} & \bilingual{subclassing}{子类化} &  &  & \strength{2/3}\rlap{$^{<:}$} & \strength{1/3} &  \\
    \modularityrow{matching} & \bilingual{matching}{匹配} &  &  & \strength{1/3}\rlap{$^{<:}$} &  &  \\
    \modularityrow{extension-points} & \bilingual{extension points}{扩展点} &  &  &  & \strength{1/3} &  \\
    \modularityrow{open-methods} & \bilingual{open methods}{开放方法} &  &  &  & \strength{2/3} &  \\
    \modularityrow{open-data-types} & \bilingual{open data types}{open data types} &  &  &  & \strength{2/3} & \strength{2/3} \\
    \modularityrow{inheritance-calculus} & \bilingual{inheritance-calculus (ours)}{继承演算(本文)} & \strength{1} & \strength{1} & \strength{1} & \strength{1} & \strength{1} \\
    \bottomrule
  \end{tabular}
\end{table}

\ifPaperChinese
本节回顾面向对象语言、配置语言与逻辑编程怎样追求 Meyer 的标准,分别通过继承、合并与添加子句;这一意义上的配置语言,是指其程序为由合并组合的嵌套记录的语言,例如 Jsonnet~\cite{cunningham2014-jsonnet}、CUE~\cite{van-lohuizen2019-cue} 与 NixOS 模块系统~\cite{nixos-contributors-nixos-modules}。表~\ref{tab:modularity-means} 概括了这些方式,并按每一行借以得到其标准的方式把各行分组:靠逐案应用的原则、靠约定,或由构造。
\else
This section recalls how object-oriented languages, configuration languages, and logic programming
pursue Meyer's criteria, by inheritance, by merge, and by adding clauses; configuration languages, in
this sense, are languages whose programs are nested records composed by merge, such as
Jsonnet~\cite{cunningham2014-jsonnet}, CUE~\cite{van-lohuizen2019-cue}, and the NixOS module
system~\cite{nixos-contributors-nixos-modules}.
Table~\ref{tab:modularity-means} summarizes these ways, grouping its rows by how each obtains its
criteria: by a principle applied case by case, by convention, or by construction.
\fi

\subsection{\bilingual{Composability}{可组合性}}
\label{sec:prior-composability}
\ifPaperChinese
可组合性要求软件元素``可以彼此自由组合以产生新系统'',``可能处在一个截然不同的环境中'',而不是它们被开发时所在的环境~\cite[第~42~页]{meyer1997-object-oriented-software-construction}。它可以理解为半格律,即交换律、幂等律与结合律;模块代数为模块的组合陈述了这三条律~\cite{bergstra1990-module-algebra},它们使一次组合既不依赖运算数的次序,也不依赖其重复与分组,从而模块能够自由组合。
\else
Composability asks for software elements that ``may then be freely combined with each other to produce
new systems'', ``possibly in an environment quite different'' from the one they were developed
in~\cite[p.~42]{meyer1997-object-oriented-software-construction}.
It can be understood as the semilattice laws, commutativity, idempotence, and associativity, which
module algebra states for the combination of modules~\cite{bergstra1990-module-algebra} and which make a
combination depend neither on the order of its operands, nor on their repetition, nor on their grouping,
so that modules combine freely.
\fi

\paragraph{\bilingual{Object-oriented languages}{面向对象语言}}
\ifPaperChinese
当一次调用选中多个继承来的方法时,它们``必须按某种特定的顺序被调用'',它们的值``以某种方式组合''~\cite{moon1986-flavors};而继承在其指称语义中是不对称的:一个对象是一个生成器的不动点,生成器是一个``其形式参数表示自引用''的函数~\cite[第~11~页]{cook1989-denotational-inheritance},而继承对它施加一个包装器,它``以自引用的方式修改一个自引用结构''~\cite[第~15~页]{cook1989-denotational-inheritance}。线性化(表~\ref{tab:modularity-means} 第~\ref{row:linearization}~行)因此给一个类的超类排序~\cite{moon1986-flavors, barrett1996-c3-linearization, simionato2003-python-mro, odersky2005-scalable-components};它只保留一次重复的超类,但按设计不可交换,而对某些层级它``约束过度,因而不存在''~\cite[第~70~页]{barrett1996-c3-linearization}。其他组合运算各得到其中一些定律,前提如显式解决冲突、标签不相交或不含 \texttt{super}(第~\ref{row:trait-sum}、\ref{row:symmetric-merge}、\ref{row:sharing-by-name}、\ref{row:override}~行与第~\ref{row:superimposition}~行)~\cite{scharli2003-traits, ducasse2006-traits-fine-grained-reuse, harper1991-symmetric-record-concatenation, meyer1997-object-oriented-software-construction, cunningham2014-jsonnet, bracha1992-jigsaw-mixins, apel2008-feature-algebra}。
\else
When several inherited methods are selected for one call, they ``must be called in some particular
order'' and their values ``combined somehow''~\cite{moon1986-flavors}, and inheritance in its
denotational semantics is asymmetric: an object is the fixed point of a generator, a function ``whose
formal parameter represents self-reference''~\cite[p.~11]{cook1989-denotational-inheritance}, and
inheritance applies to it a wrapper, which modifies ``a self-referential structure in a self-referential
way''~\cite[p.~15]{cook1989-denotational-inheritance}.
Linearization (Table~\ref{tab:modularity-means}, row~\ref{row:linearization}) therefore orders the
superclasses of a class~\cite{moon1986-flavors, barrett1996-c3-linearization,
simionato2003-python-mro, odersky2005-scalable-components}; it keeps a repeated superclass once but is
not commutative by design, and for some hierarchies it is ``over-constrained and thus does not
exist''~\cite[p.~70]{barrett1996-c3-linearization}.
Other composition operators obtain some of the laws each, under conditions such as explicit conflict
resolution, disjoint labels, or the absence of \texttt{super} (rows~\ref{row:trait-sum}, \ref{row:symmetric-merge}, \ref{row:sharing-by-name},
\ref{row:override}, and~\ref{row:superimposition})~\cite{scharli2003-traits,
ducasse2006-traits-fine-grained-reuse, harper1991-symmetric-record-concatenation,
meyer1997-object-oriented-software-construction, cunningham2014-jsonnet, bracha1992-jigsaw-mixins,
apel2008-feature-algebra}.
\fi

\paragraph{\bilingual{Configuration languages}{配置语言}}
\ifPaperChinese
配置语言按设计得到这些定律:作为管理员各自独立贡献的脚本的原则,靠一个作为并半格的状态,靠一个值的格,或者靠按优先级而不是按模块的次序放置定义(第~\ref{row:convergent-operations}、\ref{row:replicated-data-types}、\ref{row:choice}、\ref{row:unification}~行与第~\ref{row:nixos-modules}~行)~\cite{couch2001-maelstrom, shambaugh2016-rehearsal, shapiro2011-conflict-free-replicated-data-types, dezani1996-filter-models-conjunctive-disjunctive, cue-authors-cue-specification, nixos-contributors-nixos-modules}。它们靠把计算挡在所合并的东西之外来保住这些定律:``不禁止递归会使 CUE 图灵完备''~\cite{cue-authors-code-generation-use-case},嵌套交集提前算出嵌套类~\cite{nystrom2006-nested-intersection},而为了计算,配置语言转而依靠合并够不着其函数体的一等函数,或者用对象模拟函数,写成模块的软件包就是这样(第~\ref{row:functions-as-objects}~行)~\cite{cunningham2014-jsonnet, dolstra2006-purely-functional-deployment, dream2nix-contributors-package-func}。
\else
Configuration languages obtain the laws by design: as principles for scripts that administrators
contribute independently, by a state that is a join semilattice, by a lattice of values, or by placing
definitions by priority rather than by the order of modules
(rows~\ref{row:convergent-operations}, \ref{row:replicated-data-types}, \ref{row:choice},
\ref{row:unification}, and~\ref{row:nixos-modules})~\cite{couch2001-maelstrom, shambaugh2016-rehearsal, shapiro2011-conflict-free-replicated-data-types, dezani1996-filter-models-conjunctive-disjunctive, cue-authors-cue-specification, nixos-contributors-nixos-modules}.
They keep the laws by keeping computation out of what they merge: ``Not disallowing recursion would
make CUE Turing complete''~\cite{cue-authors-code-generation-use-case}, nested intersection computes
nested classes ahead of time~\cite{nystrom2006-nested-intersection}, and to compute, configuration
languages fall back on first-class functions, whose bodies merge does not reach, or emulate a function
with an object, as packages written as modules do
(row~\ref{row:functions-as-objects})~\cite{cunningham2014-jsonnet, dolstra2006-purely-functional-deployment, dream2nix-contributors-package-func}.
\fi

\paragraph{\bilingual{Logic programming}{逻辑编程}}
\ifPaperChinese
子句的并(第~\ref{row:clause-union}~行)连同一般计算得到半格律。定子程序的并是子句集合的并,所以满足它们。与没有函数符号的 Datalog 不同,后者``的表达力在 ptime 之内'',Prolog 中的子句的并所组合的逻辑程序``能表达全部递归可枚举(r.e.)谓词''~\cite[第~299--300~页]{abiteboul1995-foundations-of-databases},但这些程序是全局谓词上扁平的子句集合,没有嵌套的作用域,也没有晚绑定的 self。
\else
Clause union (row~\ref{row:clause-union}) obtains the semilattice laws together with general computation.
A union of definite programs is a union of clause sets, so it obeys them.
Unlike Datalog, which has no function symbols and whose ``expressive power \ldots\ lies within ptime'',
clause union in Prolog composes logic programs, which ``can express all recursively enumerable (r.e.)
predicates''~\cite[pp.~299--300]{abiteboul1995-foundations-of-databases}, but these programs are flat sets
of clauses over global predicates, with no nested scopes and no late-bound self.
\fi

\subsection{\bilingual{Protection}{保护性}}
\label{sec:prior-protection}
\ifPaperChinese
保护性要求运行时异常条件的影响``局限于该模块'',或者至多传到``少数相邻模块''~\cite[第~45~页]{meyer1997-object-oriented-software-construction};它针对的不是``错误的避免或纠正'',而是``它们的传播''~\cite[第~46~页]{meyer1997-object-oriented-software-construction}。保护性可以理解为把异常条件或副作用限定在它发生的地方,也可以理解为把异常条件编码为纯数据结构,即一个普通的值,使程序没有副作用并覆盖每个输入。
\else
Protection asks that the effect of an abnormal condition at run time ``remain confined to that module''
or reach at worst ``a few neighboring
modules''~\cite[p.~45]{meyer1997-object-oriented-software-construction}; it addresses not ``the
avoidance or correction of errors'' but ``their
propagation''~\cite[p.~46]{meyer1997-object-oriented-software-construction}.
Protection can be understood as containing an abnormal condition or a side effect where it occurs, or as
encoding an abnormal condition as a pure data structure, an ordinary value, so that a program has no
side effects and covers every input.
\fi

\paragraph{\bilingual{Containment}{控制}}
\ifPaperChinese
控制限制一个异常条件或一个副作用传到多远。有纪律的异常只允许两种回应,即重试该例程或向调用者失败,而普通的异常把控制转到一个可能很远的处理者。效果系统``用类型记录一个程序所执行的副作用''(第~\ref{row:disciplined-exceptions}~行与第~\ref{row:effect-systems}~行)~\cite[第~46、417、427~页]{meyer1997-object-oriented-software-construction}\cite[\S~7.1]{peytonjones1993-imperative-functional-programming}。
\else
Containment limits how far an abnormal condition or a side effect travels.
Disciplined exceptions admit only two responses, retrying the routine or failing to its caller, where
ordinary exceptions transfer control to a possibly remote handler.
Effect systems ``use types to record the side-effects performed by a program''
(rows~\ref{row:disciplined-exceptions} and~\ref{row:effect-systems})~\cite[pp.~46, 417, 427]{meyer1997-object-oriented-software-construction}\cite[\S~7.1]{peytonjones1993-imperative-functional-programming}.
\fi

\paragraph{\bilingual{Encoding as data}{编码为数据}}
\ifPaperChinese
编码为数据使一个异常条件成为一个普通的值。通过 monad 进行的输入输出属于``排除副作用''的纯函数式解法,而 CUE 这样的配置语言``是纯的,或者说封闭的,或者说没有副作用的'',但它们都不覆盖每个输入。把错误作为值返回是一种约定。语言自身检测到的失败,例如模式匹配失败,留在值之外。在 CUE 中合一相互冲突的值得到一个最小的值,该值``通常表示一个错误''(第~\ref{row:monads}、\ref{row:exceptions-as-values}~行与第~\ref{row:unification}~行)~\cite[\S\S~1, 2]{peytonjones1993-imperative-functional-programming}\cite[\S\S~2.1, 3.5]{peytonjones1999-imprecise-exceptions}\cite{cue-authors-configuration-use-case, cue-authors-cue-specification}。全函数式编程由构造同时得到无副作用与覆盖每个输入。它要求``所有的情形分析都必须完全'',并要求递归作用于语法子成分,于是``我们的编程语言不再是图灵完备的!''(第~\ref{row:total-functional-programming}~行)~\cite[第~755--759~页]{turner2004-total-functional-programming}。
\else
Encoding as data makes an abnormal condition an ordinary value.
Input and output through monads, one of the purely functional solutions ``which rule out side effects'',
and configuration languages such as CUE, which is ``pure, or hermetic, or side-effect free'', do not
cover every input.
Returning errors as values is a convention.
Failures that the language itself detects, such as a failed pattern match, stay outside the values.
Unifying conflicting values in CUE gives a least value that ``usually indicates an error''
(rows~\ref{row:monads}, \ref{row:exceptions-as-values}, and~\ref{row:unification})~\cite[\S\S~1, 2]{peytonjones1993-imperative-functional-programming}\cite[\S\S~2.1, 3.5]{peytonjones1999-imprecise-exceptions}\cite{cue-authors-configuration-use-case, cue-authors-cue-specification}.
Total functional programming obtains both freedom from side effects and coverage of every input by
construction.
It asks that ``All case analysis must be complete'' and that recursion be on syntactic subcomponents, so
that ``Our programming language is no longer Turing complete!''
(row~\ref{row:total-functional-programming})~\cite[pp.~755--759]{turner2004-total-functional-programming}.
\fi

\subsection{\bilingual{Understandability}{可理解性}}
\label{sec:prior-understandability}
\ifPaperChinese
可理解性要求读者``无需了解其他模块''就能理解每个模块,或者至多察看``少数其他模块''~\cite[第~43~页]{meyer1997-object-oriented-software-construction}。它可以理解为信息隐藏~\cite{parnas1972-criteria-decomposing-modules},也可以理解为局部性:读者从一个模块得出的结论在包含它的每个程序中都成立,正如可替换性与行为子类型所要求的那样。
\else
Understandability asks that a reader can understand each module ``without having to know the others''
or, at worst, by examining ``only a few of the
others''~\cite[p.~43]{meyer1997-object-oriented-software-construction}.
It can be understood as information hiding~\cite{parnas1972-criteria-decomposing-modules}, or as
locality: what a reader concludes from a module holds in every program that contains it, as
substitutability and behavioral subtyping ask.
\fi

\paragraph{\bilingual{Information hiding}{信息隐藏}}
\ifPaperChinese
信息隐藏这样分解一个系统:每个模块在其接口之后对其他模块隐藏一个设计决定~\cite{parnas1972-criteria-decomposing-modules}。抽象类型在语言中强制实施它:用 Reynolds 的话说,类型结构是``一种强制实施抽象层次的语法纪律''~\cite[第~513~页]{reynolds1983-parametric-polymorphism}。
\else
Information hiding decomposes a system so that each module hides a design decision from the others
behind its interface~\cite{parnas1972-criteria-decomposing-modules}.
Abstract types enforce it in the language: type structure is, in Reynolds's words, ``a syntactic
discipline for enforcing levels of abstraction''~\cite[p.~513]{reynolds1983-parametric-polymorphism}.
\fi

\paragraph{\bilingual{Locality}{局部性}}
\ifPaperChinese
在面向对象语言中,继承来的代码可以向下调用被子类覆写的方法,所以程序员通常不能``不研读其超类的源代码就写出一个子类''~\cite[第~208~页]{ruby2000-safely-creating-subclasses}。面向对象语言以替换这一原则对抗这一点,即 Liskov 替换原则(第~\ref{row:liskov-substitution}~行):用 Liskov 的话说,子类型的对象``提供另一个类型(超类型)的对象的全部行为,外加一些额外的东西''~\cite{liskov1987-data-abstraction-hierarchy}。这一原则在两种关系上被精确化,而这两种关系不同,因为``继承不是子类型''~\cite{cook1990-inheritance-not-subtyping}。
\begin{description}
  \item[行为子类型] 在子类型关系上,它关联一对带规格的类型~\cite{liskov1994-behavioral-subtyping},它对 supertype abstraction 的有效性``既是必要的也是充分的'',而 supertype abstraction 是一种模块化推理~\cite{leavens2015-behavioral-subtyping-modular-reasoning}。
  \item[可替换性] 在继承关系上,扩展应当对其基类可替换,然而一个精化其基类的扩展仍精化该基类的一个精化,这一性质``一般不成立'',Mikhajlov 与 Sekerinski 认为它``构成脆弱基类问题的实质''~\cite{mikhajlov1998-fragile-base-class}。
\end{description}
然而,``满足替换原则的真正的 behavior compatibility 比通常设想的更受限''~\cite[\S~4]{wegner1988-inheritance-incremental-modification}。若不逐对证明行为子类型(第~\ref{row:subtype-proofs}~行),表~\ref{tab:modularity-means} 中的类型系统与契约保住的就更少:子类化、NixOS 模块的类型检查与匹配保住所继承的东西的类型或协议,断言重声明与运行时契约保住程序员所写的规格(第~\ref{row:subclassing}、\ref{row:nixos-modules}、\ref{row:matching}、\ref{row:assertion-redeclaration}~行与第~\ref{row:runtime-contracts}~行)~\cite{gosling2015-java-se8-specification, nixos-contributors-nixos-modules, abadi1996-subtyping-and-matching, meyer1997-object-oriented-software-construction, findler2001-behavioral-contracts}。
\else
In object-oriented languages, inherited code can call down to methods that a subclass overrides, so
that programmers usually cannot ``write a subclass without studying the source code of its
superclasses''~\cite[p.~208]{ruby2000-safely-creating-subclasses}.
Against this, object-oriented languages hold substitution as a principle, the Liskov substitution
principle (row~\ref{row:liskov-substitution}): in Liskov's words, the objects of a subtype ``provide all
the behavior of objects of another type (the supertype) plus something
extra''~\cite{liskov1987-data-abstraction-hierarchy}.
The principle has been made precise on two relations, which differ, since ``Inheritance Is Not
Subtyping''~\cite{cook1990-inheritance-not-subtyping}.
\begin{description}
  \item[Behavioral subtyping] On subtyping, it relates a pair of types with
    specifications~\cite{liskov1994-behavioral-subtyping} and ``is both necessary and sufficient for the
    validity of supertype abstraction'', a form of modular
    reasoning~\cite{leavens2015-behavioral-subtyping-modular-reasoning}.
  \item[Substitutability] On inheritance, an extension is to be substitutable for its base, yet that an
    extension refining its base class still refines a refinement of it ``does not hold in general'',
    which, Mikhajlov and Sekerinski find, ``constitutes the essence of the fragile base class
    problem''~\cite{mikhajlov1998-fragile-base-class}.
\end{description}
However, ``true behavioral compatibility which satisfies the principle of substitution is more
restrictive than generally supposed''~\cite[\S~4]{wegner1988-inheritance-incremental-modification}.
Short of proving behavioral subtyping pair by pair (row~\ref{row:subtype-proofs}), the type systems and
contracts of Table~\ref{tab:modularity-means} keep less: subclassing, the
type checking of NixOS modules, and matching keep the types or the protocol of what is inherited, and
assertion redeclaration and run-time contracts keep the specifications that programmers write
(rows~\ref{row:subclassing}, \ref{row:nixos-modules}, \ref{row:matching},
\ref{row:assertion-redeclaration}, and~\ref{row:runtime-contracts})~\cite{gosling2015-java-se8-specification,
nixos-contributors-nixos-modules, abadi1996-subtyping-and-matching,
meyer1997-object-oriented-software-construction, findler2001-behavioral-contracts}.
\fi

\subsection{\bilingual{Continuity}{连续性}}
\label{sec:prior-continuity}
\ifPaperChinese
连续性要求``问题规格中的一个小改动''只改动``一个模块或少数几个模块'',这条标准``与 extendibility 这一总体目标直接相关''~\cite[第~44~页]{meyer1997-object-oriented-software-construction};在这些语言中,它可以理解为可扩展性,即改动可以通过添加来完成。Szyperski 把一个系统中每个``可以单独扩展''的特征称为``(独立)可扩展性的一个维度''~\cite[第~102~页]{szyperski2002-component-software}。表~\ref{tab:modularity-means} 中由构造得到可扩展性的先例,沿一个或两个维度得到它。
\else
Continuity asks that ``a small change in a problem specification'' change ``just one module, or a small
number of modules'', a criterion ``directly connected to the general goal of
extendibility''~\cite[p.~44]{meyer1997-object-oriented-software-construction}; in these languages it
can be understood as extensibility, that a change can be made by adding.
Szyperski calls each feature of a system ``that is open for separate extension'' a ``dimension of
(independent) extensibility''~\cite[p.~102]{szyperski2002-component-software}.
The precedents of Table~\ref{tab:modularity-means} that obtain extensibility by construction do so
along one or two dimensions.
\fi

\paragraph{\bilingual{Extension by adding}{以添加扩展}}
\ifPaperChinese
Meyer 的开闭原则要求``模块应当既开放又封闭''~\cite[第~57~页]{meyer1997-object-oriented-software-construction},继承是它的体现~\cite[第~497、608~页]{meyer1997-object-oriented-software-construction};但``既然封闭不可能完全,它就必须是策略性的''~\cite{martin1996-open-closed-principle},即由设计者作出的选择(第~\ref{row:open-closed}~行)。插件架构把以添加扩展写成规则,``添加,而不替换'',而它们的扩展点强制这一规则:``没有办法替换已有的功能''~\cite[第~24、67~页]{gamma2004-contributing-to-eclipse}(第~\ref{row:add-dont-replace}~行与第~\ref{row:extension-points}~行)。配置靠添加的文件与模块增长:drop-in 目录按每个应用自己的规则合并这些文件,而在 NixOS 中它们既声明选项、又定义其他模块的选项,沿两个维度扩展一份配置(第~\ref{row:drop-in-files}~行与第~\ref{row:nixos-modules}~行)~\cite{systemd-contributors-sysctl-d, systemd-contributors-tmpfiles-d, sudo-contributors-sudoers, nixos-contributors-nixos-modules}。
\else
Meyer's open--closed principle asks that ``Modules should be both open and
closed''~\cite[p.~57]{meyer1997-object-oriented-software-construction}, and inheritance is its
embodiment~\cite[pp.~497, 608]{meyer1997-object-oriented-software-construction}; but ``since closure
cannot be complete, it must be strategic''~\cite{martin1996-open-closed-principle}, a choice the
designer makes (row~\ref{row:open-closed}).
Plug-in architectures make extension by adding a rule, ``Add, don't replace'', and their extension points
enforce it: ``There isn't a way to replace existing
functionality''~\cite[pp.~24, 67]{gamma2004-contributing-to-eclipse}
(rows~\ref{row:add-dont-replace} and~\ref{row:extension-points}).
Configurations grow by added files and modules, which drop-in directories merge by a rule of each
application, and which in NixOS both declare options and define the options of others, extending a
configuration along two dimensions
(rows~\ref{row:drop-in-files} and~\ref{row:nixos-modules})~\cite{systemd-contributors-sysctl-d, systemd-contributors-tmpfiles-d, sudo-contributors-sudoers, nixos-contributors-nixos-modules}.
\fi

\paragraph{\bilingual{One dimension}{一个维度}}
\ifPaperChinese
Strachey 区分了两种多态,参数多态与特设多态~\cite{strachey2000-fundamental-concepts},Cardelli 与 Wegner 又加入第三种,即包含多态,``用以刻画面向对象编程''~\cite[第~475~页]{cardelli1985-understanding-types}:子类的对象可以出现在任何要求其超类对象的地方,因此加入一个新的子类(第~\ref{row:subclassing}~行)就扩展了程序,而客户照旧不变~\cite[第~478~页]{cardelli1985-understanding-types}\cite[第~499~页]{meyer1997-object-oriented-software-construction}。分成子类与分成数据类型上的函数,这两种分解中的任一种只开放一个维度:如 Cook 关于数据抽象所说,按构造子还是按观测者分解,``决定了加入一个新构造子或观测者有多容易''~\cite{cook1991-object-oriented-versus-abstract-data-types}\cite[第~131~页]{clifton2000-multijava},而用 Krishnamurthi 等人的话说,每种风格``轻易容纳一种扩展,而使另一种几乎不可能''~\cite[第~92~页]{krishnamurthi1998-synthesizing-object-oriented-functional},而访问者模式(第~\ref{row:visitor}~行)把程序从一侧移到另一侧~\cite[第~335--336~页]{gamma1994-design-patterns}。
\else
Strachey distinguishes two kinds of polymorphism, parametric and ad
hoc~\cite{strachey2000-fundamental-concepts}, and Cardelli and Wegner add a third, inclusion
polymorphism, ``to model object-oriented programming''~\cite[p.~475]{cardelli1985-understanding-types}:
an object of a subclass can appear wherever an object of one of its superclasses is required, so a new
subclass (row~\ref{row:subclassing}) extends a program whose clients stay
unchanged~\cite[p.~478]{cardelli1985-understanding-types}\cite[p.~499]{meyer1997-object-oriented-software-construction}.
Either decomposition, into subclasses or into functions over a datatype, opens one dimension: as Cook says of
data abstraction, decomposing by constructor or by observer ``determines how easy it is to add a new
constructor or observer''~\cite{cook1991-object-oriented-versus-abstract-data-types}\cite[p.~131]{clifton2000-multijava},
and in those of Krishnamurthi et al., each style ``accommodates one form of extension easily and renders
the other nearly impossible''~\cite[p.~92]{krishnamurthi1998-synthesizing-object-oriented-functional},
and the visitor pattern (row~\ref{row:visitor}) moves a program from one side to the
other~\cite[pp.~335--336]{gamma1994-design-patterns}.
\fi

\paragraph{\bilingual{Multiple dimensions}{多个维度}}
\ifPaperChinese
Krishnamurthi 等人要求新变体与新工具都能在不改动已有代码的情况下加入~\cite[第~91~页]{krishnamurthi1998-synthesizing-object-oriented-functional},Wadler 的表达式问题同样要求这两个维度~\cite{wadler1998-expression-problem},而多维关注点分离(第~\ref{row:multi-dimensional-separation}~行)要求``添加性的而非侵入性的改动'',``沿所需要的任意多个维度''~\cite[第~109、112~页]{tarr1999-n-degrees-of-separation}。有了开放方法(第~\ref{row:open-methods}~行),它如类型类与 CLOS 的泛型函数那样使特设多态可扩展,面向对象语言``能够支持既加入新的子类、又给已有的类加入新方法''~\cite[第~130~页]{clifton2000-multijava}\cite[第~64、68~页]{wadler1989-type-classes}\cite{bobrow1988-clos-specification},尽管 MultiJava 拒绝其可用方法中没有唯一最具体者的泛型函数,并且为了使这一检查保持模块化,外部方法既不能是抽象的,也不能加入接口~\cite[第~134--135~页]{clifton2000-multijava}。签名的解释器、签名的余积与靠继承精化的类族,作为程序员遵循的设计模式得到两个维度(第~\ref{row:interpreters-over-a-signature}、\ref{row:coproducts-of-signatures}~行与第~\ref{row:class-families}~行)~\cite{carette2009-finally-tagless, oliveira2012-object-algebras, swierstra2008-data-types-a-la-carte, odersky2005-scalable-components, wang2016-expression-problem-trivially},而 open data types 与 open functions 把两个维度都做进语言(第~\ref{row:open-data-types}~行)~\cite{loh2006-open-data-types}。面向特性的软件开发沿任意多个维度排列特性,于是``一个程序由 n 个维度方程规定''~\cite[第~48、52~页]{batory2003-refinements-multi-dimensional}。添加子句(第~\ref{row:clause-union}~行)沿两个维度扩展一个逻辑程序,即对已有项的一个新谓词,以及已有谓词的一个新情形,而不改动任何写下的子句。
\else
Krishnamurthi et al.\ ask that a new variant and a new tool can both be added without changing existing
code~\cite[p.~91]{krishnamurthi1998-synthesizing-object-oriented-functional}, Wadler's expression
problem asks for both dimensions as well~\cite{wadler1998-expression-problem}, and multi-dimensional
separation of concerns (row~\ref{row:multi-dimensional-separation}) asks for ``additive, rather than
invasive, change'' ``along as many dimensions as are
needed''~\cite[pp.~109, 112]{tarr1999-n-degrees-of-separation}.
With open methods (row~\ref{row:open-methods}), which make ad hoc polymorphism extensible as type classes
and the generic functions of CLOS do, object-oriented languages ``can support the addition of both new
subclasses and new methods to existing classes''~\cite[p.~130]{clifton2000-multijava}\cite[pp.~64, 68]{wadler1989-type-classes}\cite{bobrow1988-clos-specification},
though MultiJava rejects a generic function whose applicable methods have no unique most specific one,
and, to keep that checking modular, lets external methods be neither abstract nor added to
interfaces~\cite[pp.~134--135]{clifton2000-multijava}.
Interpreters over a signature, coproducts of signatures, and families of classes refined by inheritance
obtain both dimensions as design patterns that programmers follow
(rows~\ref{row:interpreters-over-a-signature}, \ref{row:coproducts-of-signatures},
and~\ref{row:class-families})~\cite{carette2009-finally-tagless, oliveira2012-object-algebras, swierstra2008-data-types-a-la-carte, odersky2005-scalable-components, wang2016-expression-problem-trivially},
and open data types and open functions build both into the language
(row~\ref{row:open-data-types})~\cite{loh2006-open-data-types}.
Feature-oriented software development arranges features along any number of dimensions, so that ``a
program is specified by n dimensional
equations''~\cite[pp.~48, 52]{batory2003-refinements-multi-dimensional}.
Adding clauses (row~\ref{row:clause-union}) extends a logic program along both, by a new predicate over
existing terms and by a new case of an existing predicate, without changing a clause that is written.
\fi

\subsection{\bilingual{Decomposability}{可分解性}}
\label{sec:prior-decomposability}
\ifPaperChinese
连续性向程序添加部分,而可分解性要求把一个问题拆成``少数几个较不复杂的子问题'',每个``足够独立,使得可以分别继续工作'',这一过程``常常会自我重复,因为每个子问题可能仍然复杂到需要进一步分解''~\cite[第~40~页]{meyer1997-object-oriented-software-construction};这是 Parnas 期望模块化程序设计带来的管理上的好处~\cite[第~1054~页]{parnas1972-criteria-decomposing-modules}。它可以理解为分解,即组合的逆,在广度上与深度上。
\else
Where continuity adds parts to a program, decomposability asks that a problem be split into ``a small
number of less complex subproblems'', each ``independent enough to allow further work to proceed
separately'', a process that ``will often be
self-repeating since each subproblem may still be complex enough to require further
decomposition''~\cite[p.~40]{meyer1997-object-oriented-software-construction}; this is the managerial
benefit that Parnas expected of modular programming~\cite[p.~1054]{parnas1972-criteria-decomposing-modules}.
It can be understood as decomposition, the converse of composition, in breadth and in depth.
\fi

\paragraph{\bilingual{File-level definitions}{文件级定义}}
\ifPaperChinese
``软件包应当把它们的配置文件安装在 /usr/lib/tmpfiles.d 中。/etc/tmpfiles.d 中的文件留给本地管理员''~\cite{systemd-contributors-tmpfiles-d}:一个 drop-in 目录(第~\ref{row:drop-in-files}~行)把一份配置分给由不同维护者保管的文件,而拆不到比文件更细。
\else
``Packages should install their configuration files in /usr/lib/tmpfiles.d. Files in /etc/tmpfiles.d are
reserved for the local administrator''~\cite{systemd-contributors-tmpfiles-d}: a drop-in directory
(row~\ref{row:drop-in-files}) splits a configuration among files kept by different maintainers, and
nothing finer than a file.
\fi

\paragraph{\bilingual{Sub-file definitions}{文件内定义}}
\ifPaperChinese
组合运算拆得更细,一直拆到函数、属性这样的定义。一个逻辑程序``可以被视为其全部子句之并''~\cite[第~447~页]{bugliesi1994-modularity-logic-programming},合一两个 struct 得到``一个新的 struct c,它有 a 与 b 两者的全部字段''~\cite{cue-authors-cue-specification},而一个特性是其树中``所有路径/原子 introduction 的叠加''~\cite[第~41~页]{apel2008-feature-algebra};选择、trait 之和、对称拼接、open data types 与 NixOS 模块也这样拆(第~\ref{row:clause-union}、\ref{row:unification}、\ref{row:superimposition}、\ref{row:choice}、\ref{row:trait-sum}、\ref{row:symmetric-merge}、\ref{row:open-data-types}~行与第~\ref{row:nixos-modules}~行)~\cite{dezani1996-filter-models-conjunctive-disjunctive, scharli2003-traits, harper1991-symmetric-record-concatenation, nixos-contributors-nixos-modules}。这样拆开时,``一个程序应当表现得如同所有 open data types 与 open functions 都是封闭的、定义在一处''~\cite[\S~2.4]{loh2006-open-data-types}。在这样的定义内部,结构``被视为不透明''~\cite[第~39~页]{apel2008-feature-algebra};Tarr 等人``把原子单元当作不可分'',Tarr 的模型``与它们打交道,但从不往里看''~\cite[第~111~页]{tarr1999-n-degrees-of-separation}。``许多 SPL 实现支持粒度粗的特性'',然而``在从遗留应用中抽取特性时,细粒度的扩展是必不可少的''~\cite{kastner2008-granularity-software-product-lines}。
\else
Composition operators split further, down to definitions such as functions and properties.
A logic program ``can be viewed as the union of all its
clauses''~\cite[p.~447]{bugliesi1994-modularity-logic-programming}, unifying two structs gives ``a new
struct c which has all fields of both a and b''~\cite{cue-authors-cue-specification}, and a feature is
``the superimposition of all paths / atomic introductions'' in its
tree~\cite[p.~41]{apel2008-feature-algebra}; choice, trait sum, symmetric concatenation, open data types,
and NixOS modules split likewise (rows~\ref{row:clause-union}, \ref{row:unification},
\ref{row:superimposition}, \ref{row:choice}, \ref{row:trait-sum}, \ref{row:symmetric-merge},
\ref{row:open-data-types}, and~\ref{row:nixos-modules})~\cite{dezani1996-filter-models-conjunctive-disjunctive, scharli2003-traits, harper1991-symmetric-record-concatenation, nixos-contributors-nixos-modules}.
Split so, ``a program should behave as if all open data types and open functions were closed, defined in
one place''~\cite[\S~2.4]{loh2006-open-data-types}.
Inside such a definition, the structure ``is regarded as opaque''~\cite[p.~39]{apel2008-feature-algebra};
Tarr et al.\ ``treat primitive units as indivisible'', and Tarr's model ``works with them, but never
looks inside them''~\cite[p.~111]{tarr1999-n-degrees-of-separation}.
``Many SPL implementations support features with coarse granularity'', yet ``fine-grained extensions are
essential when extracting features from legacy
applications''~\cite{kastner2008-granularity-software-product-lines}.
\fi

\paragraph{\bilingual{AST leaves}{AST 叶}}
\ifPaperChinese
Hughes 论证说,``人们能够怎样划分原问题,直接取决于人们能够怎样把各解粘合在一起'',并追求``更小、更简单、更通用的模块''~\cite[\S~2]{hughes1989-why-functional-programming-matters};他把一个求和的计算``通过把一个一般的递归模式与方框中的部分粘合在一起''来模块化~\cite[\S~3]{hughes1989-why-functional-programming-matters}。在一个定义内部的其他拆分都只是逐案成立:重构``只要满足其前置条件''就保持行为~\cite[第~iii~页]{opdyke1992-refactoring-frameworks},``两个静态切片之并不一定是一个有效的切片''~\cite{delucia2003-unions-of-slices},而注解``把特性直接赋给 AST 的子树'',各部分却仍留在同一份源码中~\cite[\S~4.2]{kastner2008-granularity-software-product-lines};用对象模拟的函数按约定在其函数体内部拆开(第~\ref{row:functions-as-objects}~行)。
\else
``The ways in which one can divide up the original problem depend directly on the ways in which one can
glue solutions together'', Hughes argues, seeking ``smaller and simpler and more general
modules''~\cite[\S~2]{hughes1989-why-functional-programming-matters}; he modularises the computation of
a sum ``by glueing together a general recursive pattern and the boxed
parts''~\cite[\S~3]{hughes1989-why-functional-programming-matters}.
Other splits inside a definition hold case by case: refactorings are behavior preserving ``provided that
their preconditions are met''~\cite[p.~iii]{opdyke1992-refactoring-frameworks}, ``the union of two
static slices is not necessarily a valid slice''~\cite{delucia2003-unions-of-slices}, and annotations
``assign features directly to subtrees of the AST'' while the parts stay in one
source~\cite[\S~4.2]{kastner2008-granularity-software-product-lines}; a function emulated by an object
splits inside its body by convention (row~\ref{row:functions-as-objects}).
\fi

\subsection{\bilingual{Conflicts Among the Criteria}{条件之间的矛盾}}
\label{sec:prior-conflicts}
\ifPaperChinese
这些标准在两处相冲突。第一处在可分解性,它与可组合性相冲突,也与读作信息隐藏的可理解性相冲突。Meyer 认为可组合性``独立于可分解性'',而这两条标准``常常相互冲突''~\cite[第~43~页]{meyer1997-object-oriented-software-construction},而他自己的记法保留了这一冲突的一边:尽管他的可分解性在深度上重复,他却让类与特性保持平坦,并指出块结构``带来许多违反 Small Interfaces 规则的机会'',而``嵌套类的能力与继承所提供的某些设施是冗余的''~\cite[第~49--50~页]{meyer1997-object-oriented-software-construction}。他论证说,复用依赖于``把一个类从其上下文中取出''的能力,而嵌套在另一个类中的类失去了这种能力:``我们将不能单独复用它''~\cite[第~209~页]{meyer1997-object-oriented-software-construction},而可组合性所要求的正是这种取出;只有簇可以嵌套,而簇是``一种组织工具而非语言构造''~\cite[第~50~页]{meyer1997-object-oriented-software-construction}。与之相反,Bracha 把层次(hierarchy)算作一条独立的标准,``模块可以由更小的模块构成'',它``不同于组合,组合是把同一层次上的模块结合起来'',并认为全局名字空间违反了它~\cite[第~13、16~页]{bracha1992-jigsaw-mixins}。在 gbeta 中保留嵌套的 Ernst 把嵌套理解为上下文性(contextuality),即一种``相对于一个已有的知识体,即上下文''给出的含义,并得出结论:``一般的块结构不支持把各个可单独复用的实体分离开''~\cite[第~104、117~页]{ernst1999-gbeta}。可分解性还与信息隐藏相冲突:一个模块所隐藏的东西没有一个组合能借以够到它的名字,所以它不能被拆出,而上述分解停在它们视为不透明的单元上。上述语言于是以限制继承、以放弃一般计算、或以只组合扁平的子句集合来保住这些要求,剩下的则留给程序员的纪律。
\else
The criteria conflict in two places.
The first is at decomposability, which conflicts with composability and with understandability read as
information hiding.
Meyer holds composability ``independent of decomposability'' and the two criteria ``often at
odds''~\cite[p.~43]{meyer1997-object-oriented-software-construction}, and his own notation keeps one side
of this conflict: although his decomposability repeats in depth, he keeps classes and features flat,
observing that block structure ``introduces many opportunities to violate the Small Interfaces rule''
and that ``the ability to nest classes is redundant with some of the facilities provided by
inheritance''~\cite[pp.~49--50]{meyer1997-object-oriented-software-construction}.
Reuse, he argues, relies on the ability ``to extract a class from its context'', which a class nested in
another loses: ``we will not be able to reuse it by
itself''~\cite[p.~209]{meyer1997-object-oriented-software-construction}, though composability asks for
just such an extraction; only clusters nest, ``an organizational tool rather than a
linguistic construct''~\cite[p.~50]{meyer1997-object-oriented-software-construction}.
Bracha, by contrast, counts hierarchy as a criterion of its own, ``Modules can be built out of smaller
modules'', ``distinct from composition, where modules are combined at the same level'', and holds that a global
name space violates it~\cite[pp.~13, 16]{bracha1992-jigsaw-mixins}.
Ernst, who keeps nesting in gbeta, reads it as contextuality, a meaning given ``relative to an existing
body of knowledge, the context'', and concludes that ``general block structure does not support the
separation of individually reusable entities''~\cite[pp.~104, 117]{ernst1999-gbeta}.
Decomposability also conflicts with information hiding: what a module hides has no name by which
composition could reach it, so it cannot be split off, and the decompositions above stop at units they
treat as opaque.
The languages above thus keep these requirements by restricting inheritance, by giving up general
computation, or by composing flat sets of clauses, and the rest is left to the discipline of programmers.
\fi

\ifPaperChinese
第二处使可组合性与面向对象语言借以得到连续性的晚绑定相对立:面向对象语言与配置语言停在既能计算又能合并的成员上。面向对象语言保留组合在同一名字下带来的定义中的一个;半格律与可替换性会要求它们保留其中的每一个,但一次调用所选中的各个方法要按顺序调用,而线性化为每一次组合重新计算这个顺序,所以它``可能改变继承层级中类之间的父子关系''~\cite{scharli2003-traits}:一个类拆成若干超类再组合回去,可能在一个部分与它的父类之间出现另一个类。为了计算,配置语言转而依靠上文所说的一等函数,而它们在合并够不着的地方。Abadi 与 Cardelli 的对象不靠函数也能计算,因为它们编码了 $\lambda$-演算~\cite[\S~2.2]{abadi1994-primitive-objects},但一次覆写只替换一个方法~\cite{abadi1996-theory-of-objects}。用对象模拟的函数与写成模块的软件包趋近这样的成员:嵌套作用域,其指向外层对象的引用是晚绑定的。于是深度合并沿多条继承路线带入同一个嵌套作用域。经由两个外层对象被继承的嵌套 trait 在 Scala~3 中被以错误拒绝,而 Scala~3 的线性化能叠加方法的覆写~\cite[第~5~章 ``Traits'']{scala-contributors-scala3-specification}(附录~\ref{app:scala-multi-target})。Scala~3 所拒绝的嵌套 trait 与 NixOS 模块系统所拒绝的嵌套模块,各自会使其体中的每个引用有两个目标,每条继承路线一个。本文其余部分给出继承演算,尽管有上述困难,它仍由构造同时得到全部五条标准(第~\ref{row:inheritance-calculus}~行)。
\else
The second sets composability against the late binding by which object-oriented languages
obtain continuity: object-oriented and configuration languages stop at members that compute and still
merge.
Object-oriented languages keep one of the definitions that a composition brings under one name; the
semilattice laws and substitutability would have them keep every one, but the methods selected for one
call are called in an order, which linearization computes anew for each composition, so that it ``may
change the parent-child relationships among classes in the inheritance
hierarchy''~\cite{scharli2003-traits}: a class split into superclasses and combined back can find
another class between a part and its parent.
To compute, configuration languages fall back on the first-class functions above, out of the reach of merge.
The objects of Abadi and Cardelli compute without functions, since they encode the
$\lambda$-calculus~\cite[\S~2.2]{abadi1994-primitive-objects}, but an override replaces one
method~\cite{abadi1996-theory-of-objects}.
Functions emulated by objects and packages written as modules approach such members: nested scopes whose
references to their enclosing objects are late bound.
Deep merge then brings one nested scope in along several inheritance routes.
A nested trait inherited through two enclosing objects is rejected with an error in Scala~3, whose
linearization stacks method overrides~\cite[ch.~5, ``Traits'']{scala-contributors-scala3-specification}
(Appendix~\ref{app:scala-multi-target}).
The nested trait that Scala~3 rejects and the nested module that the NixOS module system rejects
would each give every reference in its body two targets, one for each inheritance route.
The rest of this paper presents inheritance-calculus, which, despite these difficulties, obtains all five
criteria by construction (row~\ref{row:inheritance-calculus}).
\fi

\section{\bilingual{Syntax}{语法}}
\label{sec:syntax}

\subsection{\bilingual{Abstract Syntax Trees}{抽象语法树}}
\label{sec:ast}

\ifPaperChinese
一个程序是单个 YAML 文档,我们把它的根记作 $\doc$;$\doc$ 的具体类型是下文给出的 YAML AST。
\else
A program is a single YAML document, whose root we write $\doc$; the type
of $\doc$ is the YAML AST given below.
\fi

\paragraph{\bilingual{The AST}{AST}}
\ifPaperChinese
AST 就是 YAML 的 AST。我们回顾 YAML 1.2.2 的\emph{表示模型}(Representation Model)~\cite{ben-kiki2021-yaml}——YAML 解析器产生的节点图:每个节点都带有一个\emph{标签}(tag),且只属于以下三种之一:
\else
The AST is just YAML's AST. We recall the YAML 1.2.2 \emph{Representation
Model}~\cite{ben-kiki2021-yaml}, the node graph a YAML parser produces: every node carries
a \emph{tag} and is one of three kinds:
\fi

{\small
\[
  \begin{array}{@{}l@{~~}c@{~~}l@{\quad}l@{}}
    \YAMLNode & = & \YAMLScalar(\Tag \times \String)
      & \text{\bilingual{a tag and a content string}{标签与内容串}} \\
      & \uplus & \YAMLSequence(\Tag \times \YAMLNode^{*})
      & \text{\bilingual{a tag and a sequence of nodes}{标签与节点序列}} \\
      & \uplus & \YAMLMapping(\Tag \times (\YAMLNode \times \YAMLNode)^{*})
      & \text{\bilingual{a tag and a sequence of node pairs}{标签与节点对的序列}} \\
  \end{array}
\]}

\ifPaperChinese
我们还回顾同一规范中的 \emph{Core Schema}~\cite{ben-kiki2021-yaml},即它解析标签的方案,它给每个节点指派一个标签。继承演算只读其中四个标签:$\seqtag$、$\maptag$、$\strtag$、$\nulltag$(分别是 \mintinline{yaml}{tag:yaml.org,2002:seq} 等的简写)。
\else
We also recall the YAML 1.2.2 \emph{Core Schema}~\cite{ben-kiki2021-yaml}, the
tag-resolution scheme of the same specification, which assigns a tag to every node.
Inheritance-calculus reads only four of these tags: $\seqtag$, $\maptag$, $\strtag$, and
$\nulltag$ (shorthand for \mintinline{yaml}{tag:yaml.org,2002:seq} and so on).
\fi

\ifPaperChinese
全局根 $\doc$ 是一个 $\seqtag$ 序列,即一个 \emph{mixin},其每个成员要么是一条定义,要么是一个引用,正如清单~\ref{lst:cheatsheet} 所示。
\begin{description}
  \item[定义]
    一条\emph{定义}是一个单键的 $\maptag$ 映射,即一个标签后跟一个冒号,它把该标签绑定到一个 mixin,从而使该标签下的路径存在,例如 \mintinline{yaml}{Complex:} 或 \mintinline{yaml}{re:};这里作为键的 $\strtag$ 标量就是一个\emph{标签}。这是本演算唯一的定义形式。

  \item[引用]
    一个\emph{引用}是一个由若干 $\strtag$ 标量构成的 $\seqtag$ 序列,解析到一个 mixin 来继承。从左到右读它的诸标签,就是投影,逗号扮演着别的语言里 \mintinline{java}{.} 的角色。一个\emph{词法引用},例如 \mintinline{yaml}{[a, Plus]},读作 \mintinline{java}{a.Plus}:它向上爬到包围携带它的 mixin、且定义了其首标签 \mintinline{yaml}{a} 的最内层作用域,再从那里向下投影 \mintinline{yaml}{a} 与 \mintinline{yaml}{Plus}。第二位上可选的一个 \mintinline{yaml}{~}(即一个标签为 $\nulltag$ 的标量)用作 \mintinline{java}{this} 关键字,而一个作用域名后跟这样一个 \mintinline{yaml}{~} 便构成一个 qualified this,即 \mintinline{java}{ComplexModule.this},亦即 Java 内部类的 qualified \mintinline{java}{this}~\cite{gosling2000-java-language-specification}。于是\emph{带 qualified this 的引用} \mintinline{yaml}{[ComplexModule, ~, Float]} 读作 \mintinline{java}{ComplexModule.this.Float}:它向上爬到自身标签为 \mintinline{yaml}{ComplexModule} 的外围作用域,再从那里向下投影 \mintinline{yaml}{Float}。两种形式的唯一区别在于首标签如何找到它的作用域。词法引用取最近\emph{定义}了该标签的作用域,即把它直接写下、而非继承来的作用域。而 qualified this 取最近\emph{自身即}该标签的作用域,因而即使有更近的作用域定义了同名标签,它仍能按名指定一个选定的外围作用域。

  \item[Mixins]
    一个 mixin 的诸成员构成一个 $\seqtag$ 序列,而 YAML 既可以把序列写成块序列(每个成员占一行,以 \mintinline{yaml}{-} 起头),也可以写成方括号里的流序列;两者是同一份文档。于是一个单成员的体既可行内写作 \mintinline{yaml}{[[a, Plus]]},也可等价地跨行写成只含单个成员 \mintinline{yaml}{- [a, Plus]} 的块序列:外层方括号是成员列表,内层方括号是它所含的那一个引用。这个体本身就是一个 mixin,它继承 \mintinline{yaml}{[a, Plus]} 解析到的那个 mixin。此外,清单~\ref{lst:cheatsheet} 的最后一行 \mintinline{yaml}{total: [[_totalCall, result]]} 读作 \mintinline{scala}{class total extends _totalCall.result},而不读作绑定 \mintinline{javascript}{let total = _totalCall.result}。
\end{description}
\else
The global root $\doc$ is a $\seqtag$ sequence, a \emph{mixin} whose members are each a definition or a reference, as Listing~\ref{lst:cheatsheet} shows.
\begin{description}
  \item[Definitions]
    A \emph{definition} is a single-key $\maptag$ mapping, a label followed by a colon, that binds the label to a mixin, making the path under that label exist, such as \mintinline{yaml}{Complex:} or \mintinline{yaml}{re:}; the $\strtag$ scalar serving as the key is a \emph{label}. This is the calculus's only form of definition.

  \item[References]
    A \emph{reference} is a $\seqtag$ sequence of $\strtag$ scalars that resolves to a mixin to inherit from. Reading its labels from left to right is projection, with the comma in the role of the \mintinline{java}{.} of other languages. A \emph{lexical reference} such as \mintinline{yaml}{[a, Plus]} reads as \mintinline{java}{a.Plus}: it climbs to the innermost scope that encloses its carrying mixin and defines its leading label \mintinline{yaml}{a}, then projects \mintinline{yaml}{a} and \mintinline{yaml}{Plus} downward from there. An optional \mintinline{yaml}{~} in second position, a scalar whose tag is $\nulltag$, serves as the \mintinline{java}{this} keyword, and a scope label followed by such a \mintinline{yaml}{~} forms a qualified this, \mintinline{java}{ComplexModule.this}, the qualified \mintinline{java}{this} of Java's inner classes~\cite{gosling2000-java-language-specification}. So the \emph{qualified-this reference} \mintinline{yaml}{[ComplexModule, ~, Float]} reads as \mintinline{java}{ComplexModule.this.Float}: it climbs to the enclosing scope whose own label is \mintinline{yaml}{ComplexModule} and projects \mintinline{yaml}{Float} from there. The two forms differ only in how the leading label finds its scope. A lexical reference takes the nearest scope that \emph{defines} that label, one that writes it rather than inheriting it. A qualified this takes the nearest scope that \emph{is} that label, so it can name a chosen enclosing scope even when a nearer one defines the same label.

  \item[Mixins]
    The members of a mixin form a $\seqtag$ sequence, which YAML writes either as a block sequence with one member per \mintinline{yaml}{-} line or as a flow sequence in brackets; the two are the same document. So a one-member body appears inline as \mintinline{yaml}{[[a, Plus]]} or, equivalently, across lines as the block sequence holding the single member \mintinline{yaml}{- [a, Plus]}: the outer brackets are the member list, and the inner brackets are the one reference it holds. This body is itself a mixin that inherits the mixin \mintinline{yaml}{[a, Plus]} resolves to. Also, the last line of Listing~\ref{lst:cheatsheet}, \mintinline{yaml}{total: [[_totalCall, result]]}, reads as \mintinline{scala}{class total extends _totalCall.result}, not as the binding \mintinline{javascript}{let total = _totalCall.result}.
\end{description}
\fi


\ifPaperChinese
语义到底需要读出程序的多少?一个查询就够。在树中的每个位置,它询问该处继承了什么:在那个位置提供了哪些引用。把一个位置写作一条\emph{路径},并沿用上文表层语法中的引用,这个查询就是
\begin{equation}\label{eq:primitives}
  \inherits_{\doc} : \Path \rightharpoonup \mathcal{P}(\Reference).
\end{equation}
下标 $\doc$ 标明该查询由文档 $\doc$(即上文那个全局根)决定。第~\ref{sec:syntactic-helpers}~节给出路径与引用的精确定义,并由 AST 算出 $\inherits$。
\else
How much of a program must the semantics read? A single query
suffices. At each position in the tree it asks what the program
inherits there: which references are supplied at that position.
Writing a position as a \emph{path} and taking references as in the
surface syntax above, the query is
\begin{equation}\label{eq:primitives}
  \inherits_{\doc} : \Path \rightharpoonup \mathcal{P}(\Reference).
\end{equation}
The subscript $\doc$ marks that the query is determined by the document $\doc$, the
global root above. Section~\ref{sec:syntactic-helpers} makes paths and references precise and
computes $\inherits$ from the AST.
\fi

\subsection{\bilingual{Syntactic Helpers}{语法辅助函数}}
\label{sec:syntactic-helpers}

\paragraph{\bilingual{Path}{路径}}
\ifPaperChinese
一个\emph{标签}是一个 $\strtag$ 标量的内容串,故 $\Label = \String$。一个\emph{路径}是一个标签的有限序列
$(\ell_1, \ell_2, \ldots, \ell_n)$,即 $\Path = \Label^{*}$ 的一个元素,它标识 AST 中的一个位置。
我们用 $p$ 表示一条路径。\emph{根路径}是空序列 $()$。
给定一条路径 $p$ 与一个标签 $\ell$,$p \snoc \ell$ 是序列 $p$ 用 $\ell$ 扩展后的结果;给定一个标签序列 $w$,$p \cat w$ 是把 $w$ 接到 $p$ 末尾的拼接。
对于非根路径,外围路径 $\init(p)$ 是 $p$ 去掉其最后一个元素后的结果,
而末尾标签 $\last(p)$ 是 $p$ 的最后一个元素。
\else
A \emph{label} is the content string of a $\strtag$ scalar, so $\Label = \String$.
A \emph{path} is a finite sequence of labels
$(\ell_1, \ell_2, \ldots, \ell_n)$, an element of $\Path = \Label^{*}$, that
identifies a position in the AST.
We write $p$ for a path. The \emph{root path} is the empty sequence $()$.
Given a path $p$ and a label $\ell$, $p \snoc \ell$ is the sequence $p$
extended with $\ell$; given a sequence of labels $w$, $p \cat w$ is the
concatenation appending $w$ to $p$. For a nonroot path, the enclosing path $\init(p)$ is $p$
with its last element removed, and the final label $\last(p)$ is the last
element of $p$.
\fi

\paragraph{\bilingual{Reference}{引用}}
\ifPaperChinese
一个\emph{引用} $(n, w) \in \Reference = \mathbb{N} \times \Label^{*}$ 是 de~Bruijn 索引化的:一个索引 $n$,数出从携带该引用的 mixin 的外围作用域出发,向外要爬过多少层作用域才到达它解析所在的作用域;以及一个投影列表 $w \in \Label^{*}$,从那里向下依次跟随的标签。第~\ref{sec:ast}~节已给出引用的两种具体 YAML 写法,即词法引用与限定 $\this$;\eqref{eq:helpers} 中的辅助函数 $\indexed$ 以分别对应这两种写法的两条方程,把写在路径 $p$ 处的引用 $r$ 索引化为 $\indexed(p, r)$。
\else
A \emph{reference} $(n, w) \in \Reference = \mathbb{N} \times \Label^{*}$ is de-Bruijn-indexed: an index $n$ counting how many scopes to climb outward, from the enclosing scope of its carrying mixin, to the scope where it resolves, together with a projection list $w \in \Label^{*}$ of labels to follow downward from there. Section~\ref{sec:ast} gives the two concrete YAML forms of references, lexical and qualified this; the helper $\indexed$ in~\eqref{eq:helpers}, with an equation for each, indexes the reference $r$ written at path $p$ into $\indexed(p, r)$.
\fi

\ifPaperChinese
我们由第~\ref{sec:ast}~节的 YAML AST,经~\eqref{eq:helpers} 中的语法辅助函数,算出~\eqref{eq:primitives} 中的查询 $\inherits$。\eqref{eq:helpers} 及本文其余部分中的每个函数,都和 $\inherits$ 一样由文档 $\doc$ 决定,故都应读作带一个 $\doc$ 下标;为简洁起见,一律略去该下标。
\else
We obtain the query $\inherits$ of~\eqref{eq:primitives} from the YAML
AST of Section~\ref{sec:ast} by the syntactic helpers of~\eqref{eq:helpers}.
Every function in~\eqref{eq:helpers} and the rest of the paper, like $\inherits$ itself, is
determined by the document $\doc$ and so should be read as carrying a $\doc$ subscript, which we
omit throughout for conciseness.
\fi
\begin{equation}\label{eq:helpers}
\begin{aligned}
  \items(\YAMLSequence(\seqtag, S)) &= S, \\
  \lbl(\YAMLScalar(\strtag, \ell)) &= \ell, \\
  \lbls(s_1, \ldots, s_k) &= (\lbl(s_1), \ldots, \lbl(s_k)), \\
  \isnull(s) &\iff \YAMLScalar(\nulltag, {\_}) = s, \\
  \kv(\YAMLMapping(\maptag, (\, \tuple{s, m} \,))) &= (\lbl(s), m), \\
  \at(()) &= \doc, \\
  \{\, \at(p \snoc \ell) \,\} &= \Set{ m | \mathit{node} \in \items(\at(p)),\ \text{s.t.}\ (\ell, m) = \kv(\mathit{node}) }, \\
  \has(p, \ell) &\iff \exists\, \mathit{node} \in \items(\at(p)).\; (\ell, \_) = \kv(\mathit{node}), \\
  \lexical(p, \ell) &=
    \begin{cases}
      0 & \text{if } \has(p, \ell), \\
      1 + \lexical(\init(p), \ell) & \text{otherwise},
    \end{cases} \\
  \qualified(p, \ell) &=
    \begin{cases}
      0 & \text{if } \last(p) = \ell, \\
      1 + \qualified(\init(p), \ell) & \text{otherwise},
    \end{cases} \\
  \indexed(p,\; (s_1, s_2,\ldots,s_k)) &=\bigl(\qualified(\init(p),\ \lbl(s_1)),\ \lbls(s_3,\ldots,s_k)\bigr) \text{\quad if }\isnull(s_2), \\
  \indexed(p,\; (s_1,\ldots,s_k)) &= \bigl(\lexical(\init(p),\ \lbl(s_1)),\ \lbls(s_1,\ldots,s_k)\bigr), \\
  \dom(\inherits) &= \Set{ p | p = () \vee \has(\init(p), \last(p)) }, \\
  \inherits(p) &= \Set{ \indexed(p, \items(r)) | r \in \items(\at(p)) }.
\end{aligned}
\end{equation}

\ifPaperChinese
\eqref{eq:helpers} 把 $\inherits$ 定义为定义域是 $\dom(\inherits)$ 的偏函数;该定义域与取值 $\inherits(p)$%
\footnote{作为点查询,$\inherits$ 从不枚举一个节点所定义的标签,而是按名查找每个标签。因此其他 $\inherits$ 实现可以支撑有限 YAML 文档之外的 AST:把路径映射到无法穷举的来源(例如 URL),或非严格地解析 $\doc$、让每个节点的源文本保持未解析,直到某查询沿相应路径强制其子树。}
都预设一个有限、无 YAML anchor 的\emph{良构}程序;良构性是写下文档上一个可判定的局部性质,在附录~\ref{app:well-definedness} 中确立。本文的集合推导式遵循偏函数的标准约定:当内部某个偏函数对某一项无定义时,该项被静默略去(详见附录~\ref{app:well-definedness});例如下文 $\bases$ 中 $\inherits$ 对不在其定义域的 override 无定义,该项即被略去。
\else
Equation~\eqref{eq:helpers} defines $\inherits$ as the partial function with domain $\dom(\inherits)$; both this domain and the value $\inherits(p)$%
\footnote{Being a point query, $\inherits$ never enumerates the labels that a node defines; it looks each label up by name. Other $\inherits$ implementations may therefore back an AST other than a finite YAML document: one whose paths map to sources that cannot be listed exhaustively, such as URLs, or a $\doc$ parsed non-strictly, each node's source text left unparsed until a query forces the subtree along the relevant path.}
presuppose a finite program without YAML anchors that is \emph{well-formed}, a decidable, local property of the written document established in Appendix~\ref{app:well-definedness}.
The set comprehensions in this paper follow the standard convention for partial functions: a term is silently omitted when an inner partial function is undefined on it (see Appendix~\ref{app:well-definedness}); for instance, in $\bases$ below, $\inherits$ is undefined on an override outside its domain and that term is omitted.
\fi

\section{\bilingual{Semantics}{语义}}
\label{sec:mixin-trees}

\subsection{\bilingual{Mixin Trees}{Mixin 树}}

\ifPaperChinese
继承演算只有一个抽象,即
\emph{深可合并 mixin}(deep-mergeable mixin),其成员本身也是深可合并 mixin。%
\footnote{
  在不引起歧义之处,我们简写为``mixin''。正如清单~\ref{lst:cheatsheet} 所示,
  它推广了它在先前工作中的同名物:Bracha--Cook 的
  mixin~\cite{bracha1990-mixin-inheritance}。
}
它可以同时充当模块、类、对象、方法、字段以及
局部绑定,正如清单~\ref{lst:cheatsheet} 所展示的那样,因为
没有不透明的方法:每个值都是
一个透明的 mixin,而继承总是对子树执行一次深度
合并:当一个标签在继承链上被定义不止一次时,合并以交换、幂等、结合的方式(第~\ref{sec:modularity}~节)把这些定义合在一起,而非由后者遮蔽前者。
我们把所得到的结构称为一棵\emph{mixin 树}。

查询 $\inherits$(第~\ref{sec:syntactic-helpers}~节)只读出 AST 在每条路径处写下的内容。继承把这份 AST 扩展成一棵 mixin 树,它与 AST 有两处可观测的不同。其一,路径更多:一个标签在某节点上存在,不仅因为该节点写下了它,也因为该节点的某个 super mixin 定义了它,于是继承合成出没人在那里写过的路径,例如清单~\ref{lst:cheatsheet} 的 \mintinline{yaml}{[a, Plus]},它之所以存在,正因 \mintinline{yaml}{a} 继承了 \mintinline{yaml}{Complex}。其二,每条路径上可观测的更多:除了节点直接写下的引用,mixin 树还暴露该节点传递继承到的 super mixin。这两点都由单个函数捕获,即这棵树的观察式语义,
\begin{equation}\label{eq:supers-signature}
  \supers : \Path \to \mathcal{P}(\Path \times \Path),
\end{equation}
它把每条路径映射为它所暴露的继承;一条不存在的路径无所继承,故 $\supers$ 在该处返回空集。第~\ref{sec:definitions}~节通过与另外四个函数的相互递归来定义 $\supers$。
\else
Inheritance-calculus has a single abstraction, the
\emph{deep-mergeable mixin}, whose members are themselves deep-mergeable mixins.%
\footnote{
  Where no ambiguity arises we write simply ``mixin.'' As
  Listing~\ref{lst:cheatsheet} shows, it generalizes its namesake in prior
  work, the Bracha--Cook mixin~\cite{bracha1990-mixin-inheritance}.
}
It can simultaneously serve as module, class, object, method, field, and
local binding, as Listing~\ref{lst:cheatsheet} illustrates, because there
are no opaque methods: every value is
a transparent mixin, and inheritance always performs a deep
merge of subtrees: where a label is defined more than once along the
inheritance chain, the merge combines those definitions commutatively,
idempotently, and associatively (Section~\ref{sec:modularity}) rather
than letting a later one shadow an earlier one.
We call the resulting structure a \emph{mixin tree}.

The query $\inherits$ (Section~\ref{sec:syntactic-helpers}) reads off only
what the AST writes at each path. Inheritance expands that AST into a mixin
tree, which differs observably from the AST in two ways. First, it has more paths: a
label exists at a node not only when the node writes it but also when one of the
node's super mixins does, so inheritance synthesizes paths that were never written
there, such as Listing~\ref{lst:cheatsheet}'s \mintinline{yaml}{[a, Plus]},
which exists because \mintinline{yaml}{a} inherits \mintinline{yaml}{Complex}.
Second, more is observable at each path: beyond the references a node writes
directly, the mixin tree exposes the super mixins it reaches transitively.
A single function captures both, the observational semantics of the mixin
tree,
\begin{equation}\label{eq:supers-signature}
  \supers : \Path \to \mathcal{P}(\Path \times \Path),
\end{equation}
mapping each path to the inheritance it exposes; a path that does not exist
inherits from nothing, so $\supers$ returns the empty set there.
Section~\ref{sec:definitions} defines $\supers$ by mutual recursion with four
supporting functions.
\fi

\subsection{\bilingual{Semantic Functions}{语义函数}}
\label{sec:definitions}

\paragraph{Supers}
\ifPaperChinese
直观上,$\supers(p)$ 收集 $p$ 所继承的每一条路径:
$p$ 自身的 $\overrides$,
加上 $p$ 的每个
直接基 $\bases$ 的 $\overrides$,如此传递下去。
每个结果与到达它所经由的继承点语境
$\init(p_{\mathrm{base}})$ 配成对。
我们下文定义的 $\this$ 解析需要这一来源,
以便把一个定义点 mixin 映射回
那些纳入它的继承点路径:
\else
Intuitively, $\supers(p)$ collects every path that $p$ inherits from:
the $\overrides$ of $p$ itself,
plus the $\overrides$ of each
direct base $\bases$ of $p$, and so on transitively.
Each result is paired with the inheritance-site context
$\init(p_{\mathrm{base}})$ through which it is reached.
This provenance is needed by $\this$ resolution,
which we define below, to map a definition-site mixin back to
the inheritance-site paths that incorporate it:
\fi
\begin{equation}\label{eq:supers}
\begin{aligned}
  \supers(()) &= \{((),\; ())\}, \\
  \supers(p_{\mathrm{init}} \snoc \ell) &=
  \Set{ (\init(p_{\mathrm{base}}),\; p_{\mathrm{override}}) |
    \begin{aligned}
      &p_{\mathrm{base}} \in \bases^*(p_{\mathrm{init}} \snoc \ell), \\
      &p_{\mathrm{override}} \in \overrides(p_{\mathrm{base}})
    \end{aligned}
  }
\end{aligned}
\end{equation}
\ifPaperChinese
\begin{sloppypar}
这里 $\bases^*$ 表示 $\bases$ 的自反传递闭包。根没有外围继承点,其唯一的 super 即自身,故 $\supers(()) = \{((),\, ())\}$。于是我们把一个标签 $\ell$ 在 $p$ 处\emph{存在}定义为 $\supers(p \snoc \ell) \neq \varnothing$;附录~\ref{app:well-definedness} 的引理~\ref{lem:path-existence} 表明这恰当 $p$ 的某个 super mixin 写下 $\ell$,无论 $p$ 是直接写下它还是经继承到达。
$\supers$ 公式依赖两个函数:
$\overrides$ 与单跳引用目标 $\bases$。
我们先定义 $\overrides$。
\end{sloppypar}
\else
Here $\bases^*$ denotes the reflexive-transitive closure of $\bases$.
The root has no enclosing inheritance site and is its own only super, so $\supers(()) = \{((),\, ())\}$. We therefore take a label $\ell$ to \emph{exist} at $p$ exactly when $\supers(p \snoc \ell) \neq \varnothing$; Lemma~\ref{lem:path-existence} in Appendix~\ref{app:well-definedness} shows this holds exactly when some super mixin of $p$ writes $\ell$, whether $p$ writes it directly or reaches it through inheritance.
The $\supers$ formula depends on two functions:
$\overrides$ and the one-hop reference targets $\bases$.
We define $\overrides$ first.
\fi

\paragraph{Overrides}
\ifPaperChinese
$\overrides(p)$ 收集 $p$ 外围作用域的 super mixin 所贡献的多处同名定义;如同面向对象的方法覆写从各超类收集一个方法的同名定义,深度合并把它们合为一体。当该路径存在时,它自身也是它的一个 override,与继承来的同名定义并列;根为基例:
\else
$\overrides(p)$ collects the several same-name definitions contributed by
the super mixins of $p$'s enclosing scope; like object-oriented method
override, which collects a method's same-name definitions from the
superclasses, deep merge combines them into one. When the path exists, it is
itself one of its overrides, alongside the same-name definitions it inherits;
the root is the base case:
\fi
\begin{equation}\label{eq:overrides}
\begin{aligned}
  \overrides(()) &= \{()\}, \\
  \overrides(p_{\mathrm{init}} \snoc \ell) &= \Set{ p_{\mathrm{override}} |
    \begin{aligned}
      &(\_,\; p_{\mathrm{branch}}) \in \supers(p_{\mathrm{init}}), \\
      &\text{s.t.}\; p_{\mathrm{branch}} \snoc \ell \in \dom(\inherits), \\
      &p_{\mathrm{override}} \in \left\{ p_{\mathrm{init}} \snoc \ell,\; p_{\mathrm{branch}} \snoc \ell \right\}
    \end{aligned}
  }.
\end{aligned}
\end{equation}
\ifPaperChinese
还剩下要定义 $\bases$,即 $\supers$ 的另一个依赖项。
\else
It remains to define $\bases$, the other dependency of $\supers$.
\fi

\paragraph{Bases}
\ifPaperChinese
直观上,$\bases(p)$ 是 $p$ 通过引用直接
继承的那些路径,类比于面向对象语言中的直接
基类。
具体地说,$\bases$ 把 $p$ 的 $\overrides$ 中的每个引用
解析一步(根没有外围继承点,故 $\bases(()) = \varnothing$):
\else
Intuitively, $\bases(p)$ are the paths that $p$ directly
inherits from via references, analogous to the direct base
classes in object-oriented languages.
Concretely, $\bases$ resolves every reference in
$p$'s $\overrides$ one step (the root has no enclosing inheritance site, so $\bases(()) = \varnothing$):
\fi
\begin{equation}\label{eq:bases}
\begin{aligned}
  \bases(()) &= \varnothing, \\
  \bases(p_{\mathrm{init}} \snoc \ell) &=
  \Set{ p_{\mathrm{target}} |
    \begin{aligned}
      &p_{\mathrm{override}} \in \overrides(p_{\mathrm{init}} \snoc \ell), \\
      &(n,\; w) \in \inherits(p_{\mathrm{override}}), \\
      &p_{\mathrm{target}} \in \resolve(p_{\mathrm{init}},\; p_{\mathrm{override}},\; n,\; w)
    \end{aligned}
  }.
\end{aligned}
\end{equation}
\ifPaperChinese
$\bases$ 公式调用引用解析函数 $\resolve$,我们接下来定义它。
\else
The $\bases$ formula calls the reference resolution function $\resolve$, which we define next.
\fi

\paragraph{\bilingual{Reference resolution}{引用解析}}
\ifPaperChinese
直观上,$\resolve$ 把一个语法引用转化为它在
完全继承后的树中所指向的路径集合。
它取一条继承点路径 $p_{\mathrm{site}}$、
一条定义点路径 $p_{\mathrm{def}}$、
一个 de~Bruijn 索引 $n$,
以及投影列表 $w$。
解析分两个阶段进行:
$\this$ 从外围作用域 $\init(p_{\mathrm{def}})$ 出发
执行 $n$ 次向上步进,
把定义点路径映射为继承点路径;
然后把投影列表 $w$ 向下拼接:
\else
Intuitively, $\resolve$ turns a syntactic reference into the
set of paths it points to in the fully inherited tree.
It takes an inheritance-site path $p_{\mathrm{site}}$,
a definition-site path $p_{\mathrm{def}}$,
a de~Bruijn index $n$,
and the projection list $w$.
Resolution proceeds in two phases:
$\this$ performs $n$ upward steps
starting from the enclosing scope $\init(p_{\mathrm{def}})$,
mapping the definition-site path to inheritance-site paths;
then the projection list $w$ is concatenated downward:
\fi
\begin{equation}\label{eq:resolve}
  \resolve(p_{\mathrm{site}},\; p_{\mathrm{def}},\; n,\; w)
  = \Set{
    p_{\mathrm{current}} \cat w
    |
    p_{\mathrm{current}} \in
    \this(\{p_{\mathrm{site}}\},\; \init(p_{\mathrm{def}}),\; n)
  }
\end{equation}
\ifPaperChinese
$\resolve$ 返回一个集合而非单条路径,因为
$\this$ 可能解析到多个目标。
我们现在定义 $\this$。
\else
$\resolve$ returns a set rather than a single path because
$\this$ may resolve to multiple targets.
We now define $\this$.
\fi

\paragraph{\bilingual{$\this$ resolution}{$\this$ 解析}}
\ifPaperChinese
$\this$ 回答这个问题:
``在完全继承后的树中,定义点作用域
$p_{\mathrm{def}}$ 究竟住在哪里?''
每一步在前沿 $S$ 中每条路径的 supers 当中,
找出那些其 override 分量匹配
$p_{\mathrm{def}}$ 的,并把对应的
继承点路径收集为新的前沿。
经过 $n$ 步后,前沿就包含答案:
\else
$\this$ answers the question:
``in the fully inherited tree, where does
the definition-site scope $p_{\mathrm{def}}$ actually live?''
Each step finds, among the supers of every path in the
frontier $S$, those whose override component matches
$p_{\mathrm{def}}$, and collects the corresponding
inheritance-site paths as the new frontier.
After $n$ steps the frontier contains the answer:
\fi
{\small
  \begin{equation}\label{eq:this}
    \this(S,\; p_{\mathrm{def}},\; n) =
    \begin{cases}
      S & \text{if } n = 0 \\[6pt]
      \this\!\left(
        \Set{ p_{\mathrm{site}} |
          \begin{aligned}
            &p_{\mathrm{current}} \in S, \\
            &(p_{\mathrm{site}},\; p_{\mathrm{override}})
            \in \supers(p_{\mathrm{current}}), \\
            &\text{s.t.}\; p_{\mathrm{override}} = p_{\mathrm{def}}
          \end{aligned}
        },\;
        \init(p_{\mathrm{def}}),\;
        n - 1
      \right)
      & \text{if } n > 0
    \end{cases}
\end{equation}}%
\ifPaperChinese
每一步 $\init$ 把 $p_{\mathrm{def}}$ 缩短一个标签,
而 $n$ 减少一。%
\footnote{
  若在某一步不存在匹配对
  $(p_{\mathrm{site}},\; p_{\mathrm{override}})$,
  则前沿变为空集,而
  $\resolve$ 返回 $\varnothing$:该引用没有
  目标;不继承任何路径。
  MIXINv2\anon[~(补充材料)]{~\cite{yang2026-mixinv2}} 更为严格:它的解析器拒绝
  前沿变为空的引用,在编译期报告一个
  解析错误。
}

这就完成了 $\supers$(方程~\ref{eq:supers})及其依赖的四个函数的整套定义。

附录~\ref{app:evaluation-trace} 把这五个方程逐步
应用到一个突出 $\resolve$ 集合值本质的例子上;
在那里前沿 $S$ 携带多个目标,而
$\this$ 同时解析到它们,而非如 Scala~3 与 NixOS 模块
这类面向对象语言中所见的单个目标
(附录~\ref{app:scala-multi-target})。
\else
At each step $\init$ shortens $p_{\mathrm{def}}$ by one label
and $n$ decreases by one.%
\footnote{
  If no matching pair
  $(p_{\mathrm{site}},\; p_{\mathrm{override}})$ exists
  at some step, the frontier becomes the empty set, and
  $\resolve$ returns $\varnothing$: the reference has no
  target; no paths are inherited.
  MIXINv2\anon[~(supplementary material)]{~\cite{yang2026-mixinv2}} is stricter: its resolver rejects
  references whose frontier becomes empty, reporting a
  resolution error at compile time.
}

This completes the definition of $\supers$ (equation~\ref{eq:supers}) and the four functions it depends on.

Appendix~\ref{app:evaluation-trace} applies the five equations step by
step to an example highlighting the set-valued nature of $\resolve$,
where the frontier $S$ carries multiple targets that
$\this$ resolves to simultaneously rather than the single target
seen in object-oriented languages such as Scala~3 and NixOS modules
(Appendix~\ref{app:scala-multi-target}).
\fi

\subsection{\bilingual{Recursive Evaluation}{递归求值}}
\ifPaperChinese
方程~(\ref{eq:supers})--(\ref{eq:this}) 的这五个相互递归的函数\emph{就是}解释器:观察者
通过选择检查哪条路径来驱动计算,
而每个方程按需展开。
由于指称是最小不动点而非任何特定策略的轨迹(定理~\ref{thm:well-defined}),
(\ref{eq:supers}) 中的自反传递闭包 $\bases^*$
可以按广度优先或
深度优先来探索而不影响结果;
路径可以被 intern,从而结构相等
退化为指针相等;
而这五个函数构成一个以路径为键的可记忆化
动态规划递推式。

在操作上,求值使用\emph{tabled
求值}~\cite{tamaki1986-tabled-resolution}:
每个查询在首次遇到时被记忆化,而
重入调用,即对应于依赖图中的环的调用,
返回当前的累加器而非发散。
在指称上,五个方程定义了一个单调的
直接后承
算子~$T_P$~\cite{vanemden1976-predicate-logic-semantics};
它的右端只使用 $\in$、相等以及正向集合概括,没有否定、集合
差或补。
由 Knaster--Tarski
定理~\cite{tarski1955-lattice-fixpoint-theorem},$T_P$ 有一个
最小不动点 $\lfp(T_P)$;它充当
指称语义。
这一指称刻画把上文引入的观察式
语义形式化:观察者驱动的查询
由同一组不动点方程计算。取最小不动点作为指称,
是因为它恰好容纳这些方程所强制的继承事实而别无其他,
正对应把 $\supers(p)$ 读作 $p$ 可证继承到的内容。
求值是否终止取决于程序。
回忆 $\supers$(方程~\ref{eq:supers})取
$\bases$(方程~\ref{eq:bases})的自反传递
闭包;所得的图就是
\emph{继承层级}(inheritance hierarchy)。
继承层级有限且无环的程序在一趟内终止:
这是对有限多可达查询的普通记忆化递归。
当层级包含环时,
若可达查询只有有限多个且每个查询只能累积有限多个答案,求值
终止:对所得有限格的迭代收敛~\cite{bancilhon1986-recursive-query-strategies},
如同在 Datalog~\cite{vanemden1976-predicate-logic-semantics} 中那样。
当层级无限时,求值可能
不终止;这是图灵完备性所固有的。
附录~\ref{app:well-definedness} 证明 tabled
求值绝不返回错误答案,并重述
这些终止条件。
\else
The five mutually recursive functions of equations~(\ref{eq:supers})--(\ref{eq:this}) \emph{are} the interpreter: the observer
drives the computation by choosing which path to
inspect, and each equation unfolds on demand.
Because the denotation is a least fixed point rather than the trace
of any particular strategy (Theorem~\ref{thm:well-defined}),
the reflexive-transitive closure $\bases^*$
in~(\ref{eq:supers}) may be explored breadth-first or
depth-first without affecting the result;
paths may be interned so that structural equality
reduces to pointer equality;
and the five functions form a memoizable
dynamic-programming recurrence keyed by paths.

Operationally, evaluation uses \emph{tabled
evaluation}~\cite{tamaki1986-tabled-resolution}:
each query is memoised on first encounter, and
re-entrant calls, which correspond to cycles in the dependency graph,
return the current accumulator rather than diverging.
Denotationally, the five equations define a monotone
immediate consequence
operator~$T_P$~\cite{vanemden1976-predicate-logic-semantics}
whose right-hand sides use only $\in$, equality, and
positive set comprehension, with no negation, set
difference, or complement.
By the Knaster--Tarski
theorem~\cite{tarski1955-lattice-fixpoint-theorem}, $T_P$ has a
least fixed point $\lfp(T_P)$, which serves as
the denotational semantics.
This denotational characterization formalizes the observational
semantics introduced above: the observer-driven queries are
computed by the same fixed-point equations. The least fixed
point is the intended denotation because it admits exactly the
inheritance facts the equations force and no others, matching
the reading of $\supers(p)$ as what $p$ provably inherits.
Whether evaluation terminates depends on the program.
Recall that $\supers$ (equation~\ref{eq:supers}) takes the
reflexive-transitive closure of $\bases$
(equation~\ref{eq:bases}); the resulting graph is the
\emph{inheritance hierarchy}.
Programs with a finite acyclic inheritance hierarchy terminate in
one pass, by ordinary memoised recursion over the finitely many
reachable queries.
When the hierarchy contains cycles,
evaluation terminates when finitely many queries are
reachable and each can accumulate only finitely many answers, since
iteration over the resulting finite
lattice converges~\cite{bancilhon1986-recursive-query-strategies},
as in Datalog~\cite{vanemden1976-predicate-logic-semantics}.
When the hierarchy is infinite, evaluation may
not terminate; this is inherent to Turing completeness.
Appendix~\ref{app:well-definedness} proves that tabled
evaluation never returns a wrong answer and restates
these termination conditions.
\fi

\section{\bilingual{Embedding the \texorpdfstring{$\lambda$}{λ}-Calculus}{嵌入 \texorpdfstring{$\lambda$}{λ}-演算}}
\label{sec:translation}
\label{sec:forward-translation}

%
\ifPaperChinese
下面的翻译将纯 $\lambda$-演算映射到继承演算。$\lambda$-演算包含三种项:
\[
  \begin{array}{r@{\;::=\;}l}
    M & x \mid \lambda x.\, M \mid M_1\; M_2
  \end{array}
\]
翻译函数 $\mathcal{T}$ 是AST的同态:每个 $\lambda$-抽象在翻译后的mixin树中引入一个作用域层级,继承演算的求值机制随后在这棵树上计算弱头范式。关键设计选择是使用变量名作作用域标签,并依赖继承演算的shadowing机制来处理嵌套的同名变量。
\else
The following translation maps pure $\lambda$-calculus to
inheritance-calculus. The $\lambda$-calculus comprises three
kinds of terms:
\[
  \begin{array}{r@{\;::=\;}l}
    M & x \mid \lambda x.\, M \mid M_1\; M_2
  \end{array}
\]
The translation function $\mathcal{T}$ is a homomorphism on AST:
each $\lambda$-abstraction introduces one scope level in the
translated mixin tree, on which the inheritance-calculus evaluation
mechanism then computes the weak head normal form. A key design choice is to use variable
names as scope labels, relying on inheritance-calculus shadowing
to handle nested variables with the same name.
\fi

\newlength{\transcode}%
\settowidth{\transcode}{\mintinline{yaml}{- __whnf: [[__call, __result, __whnf]]}}%
\begin{center}
\begin{tabular}{@{}r@{\;}l@{}}
$\mathcal{T}(x) ={}$ & \mintinline[escapeinside=||]{yaml}{[[|$x$|]]} \\[1ex]
$\mathcal{T}(\lambda x.\, M) ={}$ & \begin{minipage}[t]{\transcode}
\begin{minted}[escapeinside=||]{yaml}
- __whnf:
  - |$x$|: [[__parameter]]
  - __parameter: []
  - __result: |$\mathcal{T}(M)$|
\end{minted}
\end{minipage} \\[1ex]
$\mathcal{T}(M_1\; M_2) ={}$ & \begin{minipage}[t]{\transcode}
\begin{minted}[escapeinside=||]{yaml}
- __callee: |$\mathcal{T}(M_1)$|
- __call:
  - [__callee, __whnf]
  - __parameter: |$\mathcal{T}(M_2)$|
- __whnf: [[__call, __result, __whnf]]
\end{minted}
\end{minipage} \\
  \end{tabular}
\end{center}

\ifPaperChinese
一个变量 $x$ 翻译为单成员 mixin \mintinline[escapeinside=||]{yaml}{[[|$x$|]]},其唯一成员是词法引用 \mintinline[escapeinside=||]{yaml}{[|$x$|]},在翻译后的mixin树中解析为绑定 $x$ 的抽象。这里外层方括号是成员列表,即 mixin 体本身,而内层 \mintinline[escapeinside=||]{yaml}{[|$x$|]} 是那个引用;翻译函数 $\mathcal{T}$ 的陪域一律是 mixin 体,故每条规则的右端都是这种成员列表。嵌套的同名变量通过名字遮蔽处理:内层的 $x$ 遮蔽外层的 $x$,继承演算的 name resolution 机制自动选择最近的定义。%
\footnote{\label{fn:shadowing}依赖名字遮蔽便不需要显式的 de~Bruijn 索引。我们假设翻译引入的合成标签,即 $\olbl{__parameter}$、$\olbl{__result}$、$\olbl{__call}$、$\olbl{__callee}$、$\olbl{__whnf}$,与诸如 $x$ 的用户定义变量名来自不相交的字母表,因此不会发生冲突。}

一个抽象 $\lambda x.\, M$ 本身已是弱头范式,故我们把它的形状放在单一自有定义 $\olbl{__whnf}$ 之下:一个抽象的弱头范式就是它自己,而一个应用通过同一个 $\olbl{__whnf}$ 投影抵达其被调用者的范式。$\olbl{__whnf}$ 之下的形状有三个 label:变量名 $x$ 作为指向 $\olbl{__parameter}$ 的别名;固定的 $\olbl{__parameter}$ 扮演形参声明的角色,接收实际参数;以及 $\olbl{__result}$,持有体的翻译,一旦提供了参数,它便成为缩约项。我们称一个拥有 $\olbl{__parameter}$ 的节点为\emph{抽象形状}。别名 \mintinline[escapeinside=||]{yaml}{|$x$|: [[__parameter]]} 是一个关键的设计:它使应用规则保持可组合,因为应用只需知道固定的 $\olbl{__parameter}$ 标签而无需知道绑定变量的名字,同时让变量引用 $\mathcal{T}(x) ={}$ \mintinline[escapeinside=||]{yaml}{[[|$x$|]]} 按名字解析到正确的 $\olbl{__parameter}$。当应用向 $\olbl{__parameter}$ 写入值时,别名会自动中继它,使得体可以通过引用 \mintinline[escapeinside=||]{yaml}{[|$x$|]} 访问参数。

一个应用 $M_1\; M_2$ 先用 $\olbl{__callee}$ 标签命名被调用者的翻译 $\mathcal{T}(M_1)$。$\olbl{__call}$ 节点继承引用 \mintinline{yaml}{[__callee, __whnf]},它在施用之前先把被调用者规约到其抽象形状,并在 $\olbl{__parameter}$ 中填入参数的翻译 $\mathcal{T}(M_2)$,而不予求值。继承抽象形状的同时覆写 $\olbl{__parameter}$,使 $\olbl{__call}$ 的 $\olbl{__result}$ 成为代入了参数的体,即单步 $\beta$-规约的缩约项。投影 \mintinline{yaml}{[[__call, __result, __whnf]]} 随后取该缩约项的弱头范式,于是 $\olbl{__whnf}$ 一路追踪规约直到末端。

这三条规则构成了从纯 $\lambda$-演算到继承演算的完整翻译:每个 $\lambda$-项都有一个像,且构造是可组合的。
下一小节将精确刻画翻译关于惰性 $\lambda$-演算的充分性。
\else
A variable $x$ becomes the one-member mixin \mintinline[escapeinside=||]{yaml}{[[|$x$|]]}, whose sole
member is the lexical reference \mintinline[escapeinside=||]{yaml}{[|$x$|]} that resolves in the
translated mixin tree to the abstraction that binds $x$. The outer
brackets are the member list, that is, the mixin body itself,
while the inner \mintinline[escapeinside=||]{yaml}{[|$x$|]} is the reference; the codomain of $\mathcal{T}$
is always a mixin body, so the right-hand side of every rule is
such a member list.
Nested variables with the same name are handled by name
shadowing: an inner $x$ shadows an outer $x$, and the
inheritance-calculus name resolution mechanism automatically
selects the nearest definition.%
\footnote{\label{fn:shadowing}Relying on name shadowing
  eliminates the need for explicit de~Bruijn indices.
  We assume that the labels the translation introduces, namely
  $\olbl{__parameter}$, $\olbl{__result}$, $\olbl{__call}$,
  $\olbl{__callee}$, and $\olbl{__whnf}$, and the
  user-defined variable names such as~$x$ are drawn from disjoint
  alphabets, so no collision can arise.
}

An abstraction $\lambda x.\, M$ is already a weak head normal form,
so we place its shape under a single own definition $\olbl{__whnf}$:
the weak head normal form of an abstraction is itself, and an
application reaches its operator's normal form by the same
$\olbl{__whnf}$ projection. The shape under $\olbl{__whnf}$ has three
labels: the variable name $x$ as an alias for $\olbl{__parameter}$;
the fixed $\olbl{__parameter}$, which plays the role of a parameter
declaration and receives the actual argument;
and the $\olbl{__result}$ that holds the translation of the body,
which becomes the contractum once an argument is supplied. We call a
node that owns the $\olbl{__parameter}$ an \emph{abstraction shape}.
The alias \mintinline[escapeinside=||]{yaml}{|$x$|: [[__parameter]]} is a crucial design choice:
it keeps the application rule compositional, since an application
needs only the fixed $\olbl{__parameter}$ label and not the binder's
variable name, while letting the variable
reference $\mathcal{T}(x) ={}$ \mintinline[escapeinside=||]{yaml}{[[|$x$|]]} resolve by name to the correct
$\olbl{__parameter}$. When an application writes a value to $\olbl{__parameter}$, the
alias automatically relays it, so the body can access the argument
via the reference \mintinline[escapeinside=||]{yaml}{[|$x$|]}.

An application $M_1\; M_2$ first names the operator translation
$\mathcal{T}(M_1)$ with the label $\olbl{__callee}$. The $\olbl{__call}$ node inherits the reference
\mintinline{yaml}{[__callee, __whnf]}, which reduces the operator to its
abstraction shape before applying it, and fills the
$\olbl{__parameter}$ with the argument translation $\mathcal{T}(M_2)$, left
unevaluated. Inheriting the abstraction shape while overriding
$\olbl{__parameter}$ makes the $\olbl{__result}$ of the $\olbl{__call}$
the body with the argument substituted, the contractum of a single
$\beta$-step. The projection \mintinline{yaml}{[[__call, __result, __whnf]]}
then takes the weak head normal form of that contractum, so
$\olbl{__whnf}$ follows the reduction to the end.

These three rules form a complete translation from pure
$\lambda$-calculus to inheritance-calculus: every
$\lambda$-term has an image, and the construction is
compositional.
The next subsection
pins down the adequacy of the translation for the lazy
$\lambda$-calculus.
\fi

\subsection{\bilingual{Adequacy for the Lazy \texorpdfstring{$\lambda$}{λ}-Calculus}{惰性 $\lambda$-演算的充分性}}
\label{sec:levy-longo-tree}

\ifPaperChinese
翻译 $\mathcal{T}$ 将 $\lambda$-演算嵌入继承演算的一个子语言中。mixin 树以附加的语义结构增广了 AST;对于翻译后的 $\lambda$-项,这一附加结构的观测恰好读回 Lévy--Longo 树~\cite{levy1978-reductions-lambda-calcul, longo1983-set-theoretical-models}(定理~\ref{thm:llt-correspondence})。我们借 mixin 树的收敛陈述该翻译关于惰性 $\lambda$-演算的充分性。
\else
The translation $\mathcal{T}$ embeds the $\lambda$-calculus
into a sublanguage of inheritance-calculus.
A mixin tree augments the AST with
additional semantic structure; for translated
$\lambda$-terms, the observations of this additional structure read
back to exactly the L\'evy--Longo
tree~\cite{levy1978-reductions-lambda-calcul, longo1983-set-theoretical-models}
(Theorem~\ref{thm:llt-correspondence}).
We state the adequacy of the translation for the lazy $\lambda$-calculus through convergence.
\fi

\begin{definition}[\bilingual{Inheritance-convergence}{继承收敛}]\label{def:convergence}
  \ifPaperChinese
  若
  \else
  A mixin tree \emph{converges} if
  \fi
  \[
    \supers\!\bigl((\text{``\_\_whnf''}, \text{``\_\_parameter''})\bigr) \;\neq\; \varnothing
  \]
  \ifPaperChinese
  \sloppy
  在继承演算语义 $\lfp(T_P)$ 下成立,则称 mixin 树\emph{收敛}。等价地,存在一对 $(\_,\, q) \in \supers((\text{``\_\_whnf''}))$ 使 $q \snoc \text{``\_\_parameter''} \in \dom(\inherits)$:即根的 $\olbl{__whnf}$ 沿继承($\bases^*$)追踪规约链,最终到达某个其下 $\olbl{__parameter}$ 有定义的基,亦即一个抽象形状。此等价即路径存在性引理(引理~\ref{lem:path-existence})在 $(\text{``\_\_whnf''})$ 处的实例;该收敛判据与 $M$ 是否有弱头范式的对应关系由定理~\ref{thm:adequacy} 证明,而非在本定义中假定。
  \else
  holds under the inheritance-calculus semantics $\lfp(T_P)$.
  Equivalently, there is a pair $(\_,\, q) \in \supers((\text{``\_\_whnf''}))$
  with $q \snoc \text{``\_\_parameter''} \in \dom(\inherits)$: the root's $\olbl{__whnf}$
  follows the reduction chain by inheritance ($\bases^*$) and reaches a base under
  which $\olbl{__parameter}$ is defined, an abstraction shape. This equivalence is
  the instance of the path-existence lemma (Lemma~\ref{lem:path-existence}) at
  $(\text{``\_\_whnf''})$; the correspondence between this convergence criterion and
  whether $M$ has a weak head normal form is established by
  Theorem~\ref{thm:adequacy}, not assumed in this definition.
  \fi
\end{definition}

\begin{theorem}[\bilingual{Adequacy for convergence}{收敛的充分性}]\label{thm:adequacy}
  \ifPaperChinese
  闭合 $\lambda$-项 $M$ 有弱头范式，当且仅当 $\mathcal{T}(M)$ 收敛。
  \else
  A closed $\lambda$-term $M$ has a weak head normal form if and
  only if $\mathcal{T}(M)$ converges.
  \fi
\end{theorem}

\ifPaperChinese
证明见附录~\ref{app:levy-longo-tree-proofs}。由于 $\lambda$-演算是图灵完备的,且该嵌入保持并反映收敛性,继承演算也是图灵完备的。下一小节考察继承演算是否也比惰性 $\lambda$-演算更具宏表达力。
\else
The proofs are in Appendix~\ref{app:levy-longo-tree-proofs}.
Since the $\lambda$-calculus is Turing complete and the embedding preserves and reflects convergence,
inheritance-calculus is Turing complete as well.
The next subsection examines whether inheritance-calculus is
also more macro-expressive than the lazy $\lambda$-calculus.
\fi

\subsection{\bilingual{Expressive Asymmetry}{表达力不对称}}
\label{sec:asymmetry}
\label{sec:formal-separation}
\ifPaperChinese
我们从 Felleisen~\cite{felleisen1991-expressive-power} 回顾两个概念。若对一个同时包含语言 $L_0$ 与 $L_1$ 的语言全集,每当 $L_0$ 包含或能宏表达该全集的一组构造时,$L_1$ 也包含或能宏表达它们,则称 $L_1$ 相对于该全集\emph{至少与 $L_0$ 一样具有宏表达力}(他的定义~3.17);并且,当 $L_0$ 的一个保守扩展分离了 $L_0$ 的操作等价所等同的两个短语时,$L_0$ 不能宏表达所添加的构造(他的定理~3.14(i))。两个在惰性 $\lambda$-演算中可观测等价的 $\lambda$-项,因此在 $\lambda$-演算中无法区分,可以被一个混合了 $\lambda$-演算与继承演算构造子的上下文分离;该上下文通过开放扩展在已有定义上继承新的可观测投影。例如,在丘奇布尔值及其相等判断 $\mathrm{eq}$ 之下,$\mathrm{eq}\;\mathrm{false}\;\mathrm{false}$ 与 $\mathrm{eq}\;\mathrm{true}\;\mathrm{true}$ 都归约为 $\mathrm{true}$,因此没有 $\lambda$-上下文能分离它们;一个继承任一项、并添加一个投影读取该项的路径 $(\olbl{__callee}, \olbl{__call}, \olbl{__parameter})$,即 $\mathrm{eq}$ 所收到的第一个操作数的上下文,能区分二者,因为那个操作数在一者中是 $\mathrm{false}$,在另一者中是 $\mathrm{true}$。因此,由 Felleisen 的定理~3.14(i),惰性 $\lambda$-演算不能宏表达继承。
\else
We recall two notions from Felleisen~\cite{felleisen1991-expressive-power}.
A language $L_1$ is \emph{at least as macro-expressive as} a language $L_0$ \emph{with respect to} a
language universe that contains both if $L_1$ contains or can macro-express a set of constructs of the
universe whenever $L_0$ contains or can macro-express them (his Definition~3.17); and when a
conservative extension of $L_0$ separates two phrases that the operational equivalence of $L_0$
identifies, $L_0$ cannot macro-express the added constructs (his Theorem~3.14(i)).
Two $\lambda$-terms observationally equivalent in the lazy
$\lambda$-calculus, and therefore indistinguishable there,
can be
separated by a context that mixes $\lambda$-calculus and
inheritance-calculus constructors, inheriting new observable
projections onto an existing definition via open extension.
For instance, with Church booleans and their equality $\mathrm{eq}$, the terms
$\mathrm{eq}\;\mathrm{false}\;\mathrm{false}$ and $\mathrm{eq}\;\mathrm{true}\;\mathrm{true}$ both reduce to
$\mathrm{true}$, so no $\lambda$-context separates them; a context that inherits either term and adds a
projection of the term's path $(\olbl{__callee}, \olbl{__call}, \olbl{__parameter})$, the first operand
that $\mathrm{eq}$ received, tells them apart, since that operand is $\mathrm{false}$ in one and
$\mathrm{true}$ in the other.
By Theorem~3.14(i) of Felleisen, the lazy $\lambda$-calculus therefore cannot macro-express
inheritance.
\fi

\begin{theorem}[\bilingual{Expressive asymmetry}{表达力的不对称}]
  \label{thm:expressive-asymmetry}
  \ifPaperChinese
  相对于语言全集 $\mathscr{U}$(定义~\ref{def:language-universe}),继承演算至少与惰性 $\lambda$-演算一样具有宏表达力,而惰性 $\lambda$-演算并非至少与继承演算一样具有宏表达力。
  \else
  With respect to the language universe $\mathscr{U}$
  (Definition~\ref{def:language-universe}),
  inheritance-calculus is at least as macro-expressive as
  the lazy $\lambda$-calculus, and the lazy
  $\lambda$-calculus is not at least as macro-expressive
  as inheritance-calculus.
  \fi
\end{theorem}

\ifPaperChinese
$\lambda$-演算可以直接通过 $\mathcal{T}$ 宏表达于继承演算中(定理~\ref{thm:forward-macro}),但惰性 $\lambda$-演算无法宏表达继承(定理~\ref{thm:cartesian-nonexpressibility});语言全集 $\mathscr{U}$、各构造与证明见附录~\ref{app:expressiveness}。
\else
The $\lambda$-calculus is macro-expressed directly by $\mathcal{T}$ (Theorem~\ref{thm:forward-macro}),
but the lazy $\lambda$-calculus cannot macro-express inheritance
(Theorem~\ref{thm:cartesian-nonexpressibility}); the language universe $\mathscr{U}$, the constructions,
and the proofs are in Appendix~\ref{app:expressiveness}.
\fi

\section{\bilingual{Modularity by Construction}{由构造保证的模块化}}
\label{sec:modularity}

\ifPaperChinese
我们现在转向开篇问题的另一半。本节为每条标准选定一种读法,把它定义为组合的性质,并陈述继承演算借以具有这些性质的定理,证明在附录~\ref{app:well-definedness} 与附录~\ref{app:merge-algebra}。
\else
We now turn to the other half of the opening question. This section chooses one reading for
each criterion, defines it as properties of composition, and states the theorems by which
inheritance-calculus has them, with the proofs in Appendices~\ref{app:well-definedness} and~\ref{app:merge-algebra}.
\fi

\ifPaperChinese
这些性质以观测比较程序。程序 $P$ 与一组新增顶层定义 $Q$ 的\emph{合成体}是同时携带二者的文档,$P$ 的文本原样保留;良构性要求于合成体而非单独的 $Q$:$Q$ 中书写的引用可以只经 $P$ 解析。一个程序在路径 $p$ 处的\emph{观测}是使 $p \cat w$ 存在的后缀 $w$ 之集(第~\ref{sec:definitions}~节);更细的、$p$ 处的\emph{书写观测}是 $\supers(p)$ 中各对的、属于 $\dom(\inherits)$ 的 override 分量之集。
\else
The properties compare programs by their observations. The \emph{composite} of a program~$P$ and
additional top-level definitions~$Q$ is the document carrying both, with $P$'s text kept as written;
well-formedness is required of the composite, not of $Q$ alone, since a reference written in~$Q$ may
resolve only through~$P$. The \emph{observation} of a program at a path~$p$ is the set of suffixes $w$
such that $p \cat w$ exists (Section~\ref{sec:definitions}); the finer \emph{written observation} at $p$
is the set of those override components of the pairs in $\supers(p)$ that lie in $\dom(\inherits)$.
\fi

\paragraph{\bilingual{Composability}{可组合性}}
\ifPaperChinese
我们把可组合性读作半格律(第~\ref{sec:prior-modularity}~节)。
\else
We read composability as the semilattice laws (Section~\ref{sec:prior-modularity}).
\fi

\begin{definition}[\bilingual{Deep merge}{深度合并}]\label{def:merge}
  \ifPaperChinese
  设路径 $p_1, \ldots, p_k$ 各指向同一外围作用域中命名的 mixin。这些 mixin 的\emph{深度合并}
  就是 mixin 值 \mintinline[escapeinside=||]{yaml}{[[|$p_1$|], |$\ldots$|, [|$p_k$|]]}(竖线内为元语言拼接):
  一个恰好携带指向诸基的继承引用、不携带任何自有定义的 mixin。
  这就是本文所说的\emph{组合}:该 mixin
  经由方程~(\ref{eq:overrides})的深度合并同时继承全部诸基。
  下文称这类 mixin 为\emph{包装}。
  \else
  Let the paths $p_1, \ldots, p_k$ each point to a mixin named in a
  common enclosing scope. The \emph{deep merge} of those mixins is
  the mixin value
  \mintinline[escapeinside=||]{yaml}{[[|$p_1$|], |$\ldots$|, [|$p_k$|]]}
  (the bars splice metalanguage into the object syntax):
  a mixin carrying exactly the inheritance references to the bases
  and no own definitions.
  This is \emph{composition} as this paper uses the word: that mixin inherits all the bases
  at once, through the deep merge of
  equation~(\ref{eq:overrides}). We call such a mixin a
  \emph{wrapper} below.
  \fi
\end{definition}

\begin{theorem}[\bilingual{Associativity}{结合性}]\label{thm:merge-associative}
  \ifPaperChinese
  设路径 $p_1$、$p_2$、$p_3$ 各指向文档中写下的 mixin,其首标签均定义于文档顶层。取一个保留标签
  \mintinline{yaml}{__merge}:规定除下示定义外,它不书写于文档任何他处,
  既不作定义键,也不出现于任何引用。在顶层将五个包装命名于其下,
  五个包装标签取与 $p_1$、$p_2$、$p_3$ 的首标签互异者:
  \begin{minted}[escapeinside=||]{yaml}
__merge:
- LeftPair:
  - [|$p_1$|]
  - [|$p_2$|]
- RightPair:
  - [|$p_2$|]
  - [|$p_3$|]
- mergedLeft:
  - [LeftPair]
  - [|$p_3$|]
- mergedRight:
  - [|$p_1$|]
  - [RightPair]
- mergedFlat:
  - [|$p_1$|]
  - [|$p_2$|]
  - [|$p_3$|]
  \end{minted}
  则对每条路径后缀 $w$,三条路径
  $(\text{``\_\_merge''}, \text{``mergedLeft''}) \cat w$、
  $(\text{``\_\_merge''}, \text{``mergedRight''}) \cat w$、
  $(\text{``\_\_merge''}, \text{``mergedFlat''}) \cat w$ 处
  $\supers$ 的 override 分量中落在保留字子树\emph{之外}
  (即不以 $(\text{``\_\_merge''})$ 为前缀)的路径的集合相同;特别地,
  三条路径的存在性(即其 $\supers$ 非空,第~\ref{sec:definitions}~节)一致。
  \else
  Let the paths $p_1$, $p_2$, $p_3$ each point to a mixin written in
  the document, each with its first label defined
  at the top level. Take a reserved label \mintinline{yaml}{__merge}:
  we stipulate that outside the definition displayed below it is
  written nowhere in the document, neither as a definition key nor
  inside any reference. Name the five wrappers under it at the
  top level, choosing the five wrapper labels distinct from the
  first labels of $p_1$, $p_2$, $p_3$:
  \begin{minted}[escapeinside=||]{yaml}
__merge:
- LeftPair:
  - [|$p_1$|]
  - [|$p_2$|]
- RightPair:
  - [|$p_2$|]
  - [|$p_3$|]
- mergedLeft:
  - [LeftPair]
  - [|$p_3$|]
- mergedRight:
  - [|$p_1$|]
  - [RightPair]
- mergedFlat:
  - [|$p_1$|]
  - [|$p_2$|]
  - [|$p_3$|]
  \end{minted}
  Then for every path suffix $w$, the sets of those override
  components of $\supers$ that lie \emph{outside} the reserved-word
  subtree (that is, do not have $(\text{``\_\_merge''})$ as a
  prefix) agree at the three paths
  $(\text{``\_\_merge''},\allowbreak \text{``mergedLeft''}) \cat w$,
  $(\text{``\_\_merge''},\allowbreak \text{``mergedRight''}) \cat w$, and
  $(\text{``\_\_merge''},\allowbreak \text{``mergedFlat''}) \cat w$; in
  particular the existence (non-empty $\supers$,
  Section~\ref{sec:definitions}) of the three paths agrees.
  \fi
\end{theorem}

\begin{theorem}[\bilingual{Commutativity}{交换性}]\label{thm:merge-commutative}
  \ifPaperChinese
  设路径 $p_1$、$p_2$ 各指向文档中写下的 mixin,其首标签均定义于文档顶层。
  如定理~\ref{thm:merge-associative} 取保留标签 \mintinline{yaml}{__merge},
  在顶层将两个包装命名于其下,两个包装标签取与 $p_1$、$p_2$ 的首标签互异者:
  \begin{minted}[escapeinside=||]{yaml}
__merge:
- mergedForward:
  - [|$p_1$|]
  - [|$p_2$|]
- mergedReversed:
  - [|$p_2$|]
  - [|$p_1$|]
  \end{minted}
  则对每条路径后缀 $w$,两条路径
  $(\text{``\_\_merge''}, \text{``mergedForward''}) \cat w$ 与
  $(\text{``\_\_merge''}, \text{``mergedReversed''}) \cat w$ 处
  $\supers$ 的 override 分量中落在保留字子树之外的路径的集合相同;特别地,
  两条路径的存在性一致。
  \else
  Let the paths $p_1$, $p_2$ each point to a mixin written in the
  document, each with its first label defined at the top level. As in
  Theorem~\ref{thm:merge-associative} take the reserved label
  \mintinline{yaml}{__merge} and name two wrappers under it at the top
  level, choosing the two wrapper labels distinct from the first labels
  of $p_1$, $p_2$:
  \begin{minted}[escapeinside=||]{yaml}
__merge:
- mergedForward:
  - [|$p_1$|]
  - [|$p_2$|]
- mergedReversed:
  - [|$p_2$|]
  - [|$p_1$|]
  \end{minted}
  Then for every path suffix $w$, the sets of those override
  components of $\supers$ that lie outside the reserved-word subtree
  agree at the two paths
  $(\text{``\_\_merge''},\allowbreak \text{``mergedForward''}) \cat w$ and
  $(\text{``\_\_merge''},\allowbreak \text{``mergedReversed''}) \cat w$; in
  particular the existence of the two paths agrees.
  \fi
\end{theorem}

\begin{theorem}[\bilingual{Idempotence}{幂等性}]
  \label{thm:merge-idempotent}
  \ifPaperChinese
  设路径 $p$ 指向文档中写下的 mixin(该路径存在,第~\ref{sec:definitions}~节),
  其首标签定义于文档顶层。另设 $\dom(\inherits)$ 前缀封闭:书写路径的每个前缀
  也是书写路径(语法层的 $\at$ 沿前缀递归,故每份由语法层产出的文档都满足)。
  如定理~\ref{thm:merge-associative} 取保留标签 \mintinline{yaml}{__merge},
  在顶层将 $p$ 与自身的合成命名于其下,包装标签取与 $p$ 的首标签相异者:
  \begin{minted}[escapeinside=||]{yaml}
__merge:
- mergedTwice:
  - [|$p$|]
  - [|$p$|]
  \end{minted}
  记 $\sigma = (\text{``\_\_merge''}, \text{``mergedTwice''})$。则对每条路径
  后缀 $w$,$\supers(\sigma \cat w)$ 与 $\supers(p \cat w)$ 的 override
  分量中\emph{书写}且落在保留字子树之外的路径的集合相同;特别地,两条路径的
  存在性(第~\ref{sec:definitions}~节)一致。
  \else
  Let the path $p$ point to a mixin written in the document (the
  path exists, Section~\ref{sec:definitions}), with its first label
  defined at the top level. Assume moreover that $\dom(\inherits)$
  is prefix-closed: every prefix of a written path is written (the
  syntax layer's $\at$ recurses through prefixes, so every document
  the syntax layer produces satisfies this). As in
  Theorem~\ref{thm:merge-associative} take the reserved label
  \mintinline{yaml}{__merge} and name under it at the top level the
  composition of $p$ with itself, choosing the wrapper label distinct
  from the first label of $p$:
  \begin{minted}[escapeinside=||]{yaml}
__merge:
- mergedTwice:
  - [|$p$|]
  - [|$p$|]
  \end{minted}
  Write $\sigma = (\text{``\_\_merge''}, \text{``mergedTwice''})$. Then
  for every path suffix $w$, the sets of those override components
  of $\supers(\sigma \cat w)$ and of $\supers(p \cat w)$ that are
  \emph{written} and lie outside the reserved-word subtree agree; in
  particular the existence (Section~\ref{sec:definitions}) of the
  two paths agrees.
  \fi
\end{theorem}

\paragraph{\bilingual{Continuity}{连续性}}
\ifPaperChinese
我们把连续性读作可扩展性(第~\ref{sec:prior-modularity}~节),并要求它沿每个维度成立:一门语言具有\emph{多维度可扩展性},当一个程序既能以数据的一个新情形、也能以对数据的一个新操作来扩展,而不改动已有代码,正如第~\ref{sec:prior-modularity}~节中 Krishnamurthi 等人所要求的。首标签定义于良构程序 $P$ 顶层的非空路径 $p$ 是\emph{可扩展的},若存在 $P$ 与 $Q$ 的良构合成体,$Q$ 中某 mixin 经引用继承该首标签,且合成体在\emph{对应路径}(该 mixin 的路径接 $p$ 的其余部分)处的观测至少含有一个 $P$ 在 $p$ 处的观测所没有的后缀。首标签未定义于 $P$ 的路径与根路径仅凭并置即可扩展:$Q$ 径直写下缺失的定义,合成体在 $p$ 自身处观测到新路径。合成体在对应路径处还保留 $P$ 的观测,这是可替换性,而不是可扩展性(定理~\ref{thm:inheritance-substitutability})。一种语言具有\emph{由构造保证的可扩展性},若其每个良构程序的每条路径都是可扩展的。数据的情形与对其的操作都是路径,所以由构造保证的可扩展性给出多维度可扩展性;第~\ref{sec:case-study}~节的案例研究在两个维度上扩展同一个程序。
\else
We read continuity as extensibility (Section~\ref{sec:prior-modularity}) and ask for it along every dimension: a language has \emph{multi-dimensional extensibility} when a
program can be extended both with a new case of its data and with a new operation over it, without
changing existing code, as Krishnamurthi et al.\ ask in Section~\ref{sec:prior-modularity}.
A nonempty path~$p$ whose first label is defined at the top level of a well-formed program~$P$ is
\emph{extensible} if some well-formed composite of $P$ and~$Q$, where a mixin of~$Q$ inherits that
first label through a reference, has an observation at the \emph{corresponding path} (that mixin's path
followed by the remainder of~$p$) with at least one suffix that $P$'s observation at~$p$ lacks. A path
whose first label is not defined in~$P$, and likewise the root, is extensible by juxtaposition alone:
$Q$ writes the missing definitions outright, and the composite observes new paths at~$p$ itself. That
the composite also keeps $P$'s observation at the corresponding path is substitutability, not
extensibility (Theorem~\ref{thm:inheritance-substitutability}). A language has \emph{extensibility by
construction} if every path of every well-formed program is extensible. The cases of the data and the
operations over them are both paths, so extensibility by construction gives multi-dimensional
extensibility; the case study of
Section~\ref{sec:case-study} extends one program in both dimensions.
\fi

\begin{proposition}[\bilingual{Extensibility by
    construction}{由构造保证的可扩展性}]
  \label{prop:extensibility-by-construction}
  \ifPaperChinese
  设良构文档在顶层以标签 $\ell_P$ 命名一个 mixin,并设
  $p = (\ell_1, \ldots, \ell_k)$ 为任意标签序列($k \ge 0$)。取保留标签
  \mintinline{yaml}{__extend} 与 \mintinline{yaml}{__new}:除下示定义外,
  二者不书写于文档任何他处。在顶层命名如下包装(以 $k = 2$ 为例):
  \begin{minted}[escapeinside=||]{yaml}
__extend:
- [|$\ell_P$|]
- |$\ell_1$|:
  - |$\ell_2$|:
    - __new: []
  \end{minted}
  \begin{sloppypar}
  即 \mintinline{yaml}{__extend} 经引用继承 $\ell_P$,并沿 $p$ 写下嵌套
  定义链,最内层只含 \mintinline{yaml}{__new}。则:
  (i) 路径 $(\text{``\_\_extend''}) \cat p \cat (\text{``\_\_new''})$
  存在,且后缀 $(\text{``\_\_new''})$ 不在 $P$ 于 $(\ell_P) \cat p$ 处的观测中;
  (ii) $\ell_P$ 之下的文档原样未动,其各路径的 $\supers$ 事实不变。
  于是定义中的路径 $(\ell_P) \cat p$ 可扩展,其余部即本命题的 $p$,
  对应路径为 $(\text{``\_\_extend''}) \cat p$;连同定义中的并置情形
  (首标签未定义于 $P$,或根路径),每个良构程序的每条路径都是可扩展的
  (上文的定义):
  继承演算具有由构造保证的可扩展性。
  \end{sloppypar}
  \else
  Let a well-formed document name a mixin at the top level under a
  label $\ell_P$, and let $p = (\ell_1, \ldots, \ell_k)$ be any
  sequence of labels ($k \ge 0$). Take two reserved labels
  \mintinline{yaml}{__extend} and \mintinline{yaml}{__new}, written
  nowhere in the document outside the definition displayed below,
  and name the following wrapper at the top level (shown for
  $k = 2$):
  \begin{minted}[escapeinside=||]{yaml}
__extend:
- [|$\ell_P$|]
- |$\ell_1$|:
  - |$\ell_2$|:
    - __new: []
  \end{minted}
  that is, \mintinline{yaml}{__extend} inherits $\ell_P$ through a
  reference and writes the nested definition chain along $p$, the
  innermost body containing only \mintinline{yaml}{__new}. Then:
  (i) the path
  $(\text{``\_\_extend''}) \cat p \cat (\text{``\_\_new''})$
  exists, and the suffix $(\text{``\_\_new''})$ is not in the
  observation of $P$ at $(\ell_P) \cat p$; (ii) the document under
  $\ell_P$ is untouched and the $\supers$ facts of its paths are
  unchanged. The definition's path $(\ell_P) \cat p$ is therefore
  extensible, its remainder being this proposition's $p$ and its
  corresponding path $(\text{``\_\_extend''}) \cat p$; together with
  the definition's juxtaposition cases (a first label undefined
  in~$P$, or the root), every path of every well-formed program is
  extensible (the definition above): inheritance-calculus has
  extensibility by construction.
  \fi
\end{proposition}

\paragraph{\bilingual{Understandability}{可理解性}}
\ifPaperChinese
我们把可理解性读作局部性,继承以可替换性得到它,即添加的东西不破坏任何已有的东西,依据是第~\ref{sec:prior-modularity}~节回顾的 Leavens 与 Naumann 的结果。每个 mixin 的每个成员都在组合够得着的地方,\footnote{这种够得着来自继承演算程序的形式,它类似 A-正规形式~\cite{flanagan1993-essence-compiling-continuations}:每个中间结果都有一个标签,如同清单~\ref{lst:cheatsheet} 中的调用及其结果,以及第~\ref{sec:translation}~节的翻译中的匿名位置那样。}所以继承演算的一个模块不能像信息隐藏所要求的那样,对其他模块藏起一个设计决定。这就是第~\ref{sec:prior-conflicts}~节所说的可分解性与信息隐藏之间的冲突,而继承演算站在可分解性一边:同样的够得着,让 mixin 能一直拆到叶。脆弱基类也是连续性的失败,因为基类的一次修订迫使其扩展随之改动;我们仍把可替换性归于可理解性,因为在 Leavens 与 Naumann 那里,它正是使推理成为模块化的东西。我们把可理解性表述为可替换性,它是继承所关联的 mixin 之间的关系,依照 Mikhajlov 与 Sekerinski 所说的扩展对其基类的可替换性~\cite{mikhajlov1998-fragile-base-class}。\footnote{我们认为,可理解性既可以用行为子类型解释,行为子类型是带规格的类型之间的关系~\cite{liskov1994-behavioral-subtyping, leavens2015-behavioral-subtyping-modular-reasoning},也可以用可替换性解释。关于可替换性的定理直接陈述在求值所观测到的东西上,而关于子类型的定理只有以类型系统为引理才谈及值,该类型系统的可靠性是一条关于值的定理;所以可替换性对有类型和没有类型的语言同样适用。}我们说一个 mixin $S$ 对一个 mixin $T$ \emph{可替换},当且仅当:(a) $S$ 是否良构只由写在 $S$ 处及其外围作用域处的文本决定,因而不由任何语义函数决定,也不由 $S$ 所继承的任何体决定,从而继承 $T$ 不引发任何错误;(b) $T$ 的观测包含于 $S$ 的观测。这是 Liskov 的替换性质~\cite{liskov1987-data-abstraction-hierarchy} 的包含形式,与第~\ref{sec:prior-modularity}~节引用的她的直观陈述一致:$S$ 可以比 $T$ 观测到更多。当 $T \in \bases^*(S)$,即 $S$ 直接或传递地继承 $T$ 时,一对 $(S, T)$ 在一个文档中\emph{由继承关联};一个语言的继承关系具有可替换性,当且仅当在每个良构程序中每个这样的对都是可替换的。由组合引起的错误会打破它。
\else
We read understandability as locality, which inheritance obtains as substitutability, that what is
added breaks nothing that exists, after the
result of Leavens and Naumann recalled in Section~\ref{sec:prior-modularity}. Every
member of every mixin is within reach of composition,\footnote{The reach comes from the form of
inheritance-calculus programs, akin to A-normal
form~\cite{flanagan1993-essence-compiling-continuations}: every intermediate result has a label, as the
calls and their results do in Listing~\ref{lst:cheatsheet} and the anonymous positions do in the
translation of Section~\ref{sec:translation}.} so a module of inheritance-calculus cannot keep a design
decision from the others, as information hiding asks.
This is the conflict between decomposability and information hiding of
Section~\ref{sec:prior-conflicts}, and inheritance-calculus takes the side of decomposability: the same
reach lets a mixin split down to its leaves. A fragile
base class fails continuity as well, since a revision of the base forces changes to its extensions; we
still place substitutability under understandability, since for Leavens and Naumann it is what makes
reasoning modular. We state understandability as substitutability, a relation between the mixins that
inheritance relates, after Mikhajlov and Sekerinski's substitutability of an extension for its
base~\cite{mikhajlov1998-fragile-base-class}.\footnote{We think that understandability can be explained
either by behavioral subtyping, a relation between types with
specifications~\cite{liskov1994-behavioral-subtyping, leavens2015-behavioral-subtyping-modular-reasoning},
or by substitutability. A theorem of substitutability is stated directly on what evaluation observes,
whereas a theorem of subtyping speaks of values only by using the type system, whose soundness is a
theorem about values, as a lemma; so substitutability applies to languages with types and without
alike.}
We say that a mixin $S$ is \emph{substitutable} for a mixin $T$ when (a) whether $S$ is well
formed is determined by the text written at $S$ and at its enclosing scopes alone, and so by no semantic
function and by no body that $S$ inherits, so that inheriting $T$ raises no error; and (b) the observation
of $T$ is contained in the observation of $S$. This is the containment form of Liskov's substitution
property~\cite{liskov1987-data-abstraction-hierarchy}, in line with her intuitive statement quoted in
Section~\ref{sec:prior-modularity}: $S$ may observe more than $T$. A pair $(S, T)$ is
\emph{related by inheritance} in a document when $T \in \bases^*(S)$, that is, when $S$ inherits $T$
directly or transitively; the inheritance relation of a language is substitutable when every such pair,
in every well-formed program, is substitutable. An error caused by a composition breaks it.
\fi

\begin{theorem}[\bilingual{Substitutability of inheritance}{继承的可替换性}]\label{thm:inheritance-substitutability}
\ifPaperChinese
设 $P$ 是良构程序,$Q$ 是一列顶层定义,使 $P$ 与 $Q$ 的组合体 $C$ 良构。则:
(a) $P$ 中书写的每个位置在 $C$ 中良构,而 $Q$ 中书写的每个位置的良构性只读取 $Q$ 的文本以及其外围作用域写下的标签,在根处即 $P$ 与 $Q$ 的顶层标签;因此继承 $P$ 的一个 mixin 不会引发任何错误;
(b1) 对每条路径 $q$,$\supers_P(q)$ 的每个对都是 $\supers_C(q)$ 的对,从而 $P$ 在每条路径处的观测包含于 $C$ 在该路径处的观测;
(b2) 对 $Q$ 中书写的一个 mixin 的路径 $p_S$、$P$ 中书写的一条路径 $p_T$,若 $p_T \in \bases^*_C(p_S)$,则对每个后缀 $w$,$P$ 在 $p_T \cat w$ 处的书写观测包含于 $C$ 在 $p_S \cat w$ 处的书写观测,特别地 $P$ 在 $p_T$ 处的观测包含于 $C$ 在 $p_S$ 处的观测。
更一般地,对一个良构文档中满足 $p_T \in \bases^*(p_S)$ 的任意一对路径,(b2) 在该文档内成立:继承演算的继承关系具有可替换性。
\else
Let $P$ be a well-formed program and $Q$ a list of top-level definitions such that the composite $C$ of
$P$ and $Q$ is well formed. Then:
(a) every position written in $P$ is well formed in $C$, and the well-formedness of every position
written in $Q$ reads only the text of $Q$ and the labels its enclosing scopes write, which at the root
are the top-level labels of $P$ and $Q$; so inheriting a mixin of $P$ raises no error;
(b1) for every path $q$, every pair of $\supers_P(q)$ is a pair of $\supers_C(q)$, so the observation of
$P$ at every path is contained in the observation of $C$ at that path;
(b2) for the path $p_S$ of a mixin written in $Q$ and a path $p_T$ written in $P$ with
$p_T \in \bases^*_C(p_S)$, and for every suffix $w$, the written observation of $P$ at $p_T \cat w$ is
contained in the written observation of $C$ at $p_S \cat w$; in particular the observation of $P$ at
$p_T$ is contained in the observation of $C$ at $p_S$.
More generally, (b2) holds within a well-formed document for every pair of paths with
$p_T \in \bases^*(p_S)$: the inheritance relation of inheritance-calculus is substitutable.
\fi
\end{theorem}

\paragraph{\bilingual{Decomposability}{可分解性}}
\ifPaperChinese
我们把可分解性读作第~\ref{sec:prior-decomposability}~节的分解:一个 mixin 的各部分,无论是否重叠,都各自是一个 mixin,Meyer 要求连接子问题的简单结构~\cite[第~40~页]{meyer1997-object-oriented-software-construction}在这里只是继承,而一个 mixin 可以这样在每个深度上一直拆到它的叶。可分解性与可组合性在这里是一致的,而第~\ref{sec:prior-conflicts}~节从 Meyer 与 Ernst 那里回顾了二者之间的冲突:在下面的定理中,每个部分,无论是否嵌套,都被放到它所来自的 mixin 之外的一条新路径上,而程序观测到的仍是它先前观测到的。可分解性还要求在各部分上分别继续工作,这靠的是上面两条性质:加到一个部分上的东西经可扩展性到达每个继承者,经可替换性不破坏任何已有的东西。
\else
We read decomposability as the decomposition of Section~\ref{sec:prior-decomposability}: the parts of a mixin, overlapping or not, are each a mixin of
their own, the simple structure that Meyer asks to connect
subproblems~\cite[p.~40]{meyer1997-object-oriented-software-construction} is inheritance alone, and a
mixin splits this way at every depth, down to its leaves.
Decomposability and composability agree here, where Section~\ref{sec:prior-conflicts} recalls a conflict
between them from Meyer and Ernst: in the theorem below, each part, nested or not, is placed at a new
path outside the mixin it came from, and the program observes what it observed before. Decomposability also asks that work proceed separately on each part, which the two properties
above provide: what is added to a part reaches every inheriter by extensibility and breaks nothing by
substitutability.
\fi

\begin{theorem}[\bilingual{Decomposition}{分解}]\label{thm:decomposition}
  \ifPaperChinese
  设 $D$ 是一个文档,其 $\dom(\inherits)$ 前缀封闭;设 $p$ 是一条写下的非根路径,且 $\overrides_D(p) = \{p\}$,即写在 $p$ 处的 mixin $M$ 是其路径上唯一的定义。$M$ 的一个\emph{部分} $i$($i$ 取自一个下标集 $I$)由两样东西组成:一个前缀封闭的后缀集合 $W_i$,满足 $() \in W_i$ 且每个 $p \cat v$($v \in W_i$)都写下;以及一个由对 $(p \cat v, \rho)$ 组成的集合 $E_i$,其中 $\rho \in \inherits_D(p \cat v)$ 且 $v \in W_i$。设各部分覆盖 $M$:每个满足 $\rho \in \inherits_D(p \cat v)$ 的对 $(p \cat v, \rho)$ 都属于某个 $E_i$,且 $p$ 的子树中每条写下的路径都是某个 $p \cat v$,$v$ 属于某个 $W_i$。取一个保留标签 $R$,它不出现在 $D$ 的任何路径中,也不出现在其引用的任何投影列表中,并对 $i \in I$ 取互异且不同于 $R$ 的标签 $\ell_i$;令 $q_i = \init(p) \snoc R \snoc \ell_i$。对写在 $p \cat v$ 的引用 $(n, w)$,当 $n \ge |v|$,即其爬升离开 $M$ 时,令 $\mathit{shift}_v(n, w) = (n + 1, w)$,否则令它为 $(n, w)$。设 $o$ 是 $p$ 的子树之外一条写下的路径,它携带绝对引用 $r = (|\init(o)|, p)$。设 $D'$ 如下扩展 $D$:路径 $\init(p) \snoc R$ 与 $q_i \cat v$($v \in W_i$)是写下的;$\inherits_{D'}(\init(p) \snoc R) = \varnothing$;$\inherits_{D'}(q_i \cat v) = \{\mathit{shift}_v(\rho) \mid (p \cat v, \rho) \in E_i\}$;$\inherits_{D'}(o) = (\inherits_D(o) \setminus \{r\}) \cup \{(|\init(o)|, q_i) \mid i \in I\}$;其余各处 $\inherits_{D'}$ 与 $\inherits_D$ 一致。记 $\varphi$ 为把 $q_i \cat v$ 送到 $p \cat v$、把 $\init(p) \snoc R$ 送到 $\init(p)$、并固定其余每条路径的映射。则 (a) $\dom(\inherits_{D'})$ 前缀封闭,且 $q_i$ 处部分的引用恰为 $E_i$ 中引用的平移;(b) 对每条不出现 $R$ 的路径 $a$,$\supers_D(a)$ 的书写 override 分量恰为 $\supers_{D'}(a)$ 的书写 override 分量在 $\varphi$ 下的像;特别地,每条这样的路径在 $D$ 中存在当且仅当它在 $D'$ 中存在。
  \else
  Let $D$ be a document whose $\dom(\inherits)$ is prefix-closed, and let $p$ be a nonroot written path
  with $\overrides_D(p) = \{p\}$, so that the mixin $M$ written at $p$ is the only definition of its
  path. A \emph{part} $i$ of $M$, for $i$ in a set $I$ of indices, consists of a prefix-closed set
  $W_i$ of suffixes, with $() \in W_i$ and $p \cat v$ written for every $v \in W_i$, and a set $E_i$
  of pairs $(p \cat v, \rho)$ with $\rho \in \inherits_D(p \cat v)$ and $v \in W_i$. Let the parts
  cover $M$: every pair $(p \cat v, \rho)$ with $\rho \in \inherits_D(p \cat v)$ lies in some $E_i$,
  and every written path in the subtree of $p$ is $p \cat v$ for some $v$ in some $W_i$. Take a
  reserved label $R$, occurring in no path of $D$ and in no projection list of its references, and
  distinct labels $\ell_i$ other than $R$ for $i \in I$; let
  $q_i = \init(p) \snoc R \snoc \ell_i$. For a reference $(n, w)$ written at $p \cat v$, let
  $\mathit{shift}_v(n, w)$ be $(n + 1, w)$ when $n \ge |v|$, that is, when its climb leaves $M$, and
  $(n, w)$ otherwise. Let $o$ be a written path outside the subtree of $p$ that carries the absolute
  reference $r = (|\init(o)|, p)$. Let $D'$ extend $D$ as follows: the paths $\init(p) \snoc R$ and
  $q_i \cat v$ for $v \in W_i$ are written; $\inherits_{D'}(\init(p) \snoc R) = \varnothing$;
  $\inherits_{D'}(q_i \cat v) = \{\mathit{shift}_v(\rho) \mid (p \cat v, \rho) \in E_i\}$;
  $\inherits_{D'}(o) = (\inherits_D(o) \setminus \{r\}) \cup \{(|\init(o)|, q_i) \mid i \in I\}$;
  and $\inherits_{D'}$ agrees with $\inherits_D$ elsewhere. Write $\varphi$ for the map that sends
  $q_i \cat v$ to $p \cat v$ and $\init(p) \snoc R$ to $\init(p)$, and fixes every other path. Then
  (a) $\dom(\inherits_{D'})$ is prefix-closed, and the references of the part at $q_i$ are exactly the
  shifts of the references in $E_i$; and (b) for every path $a$ in which $R$ does not occur, the
  written override components of $\supers_D(a)$ are the images under $\varphi$ of those of
  $\supers_{D'}(a)$; in particular, every such path exists in $D$ exactly when it exists in $D'$.
  \fi
\end{theorem}

\ifPaperChinese
按叶分解是划分的情形。$M$ 的一个\emph{叶}是一个对 $(p \cat v, \rho)$,其中 $\rho \in \inherits_D(p \cat v)$;或是一条写下的路径 $p \cat v$,其中 $\inherits_D(p \cat v) = \varnothing$,且没有写下的延伸 $p \cat v \snoc \ell$。当 $p$ 的子树有限时,把 $M$ 的叶划分成 $k \ge 2$ 个非空集合,就给出 $k$ 个部分:$E_i$ 是第 $i$ 个集合中的对,$W_i$ 包含使 $p \cat v$ 为该集合中某叶之路径的前缀的那些 $v$;有限子树中每条写下的路径都是某叶之路径的前缀,所以这些部分覆盖 $M$。
在表层,$q_i$ 处的部分这样写出:复制 $M$ 中由第 $i$ 个集合张成的子树,并为其中每个引用带上该引用的词法查找所找到的定义,使它在部分中找到同一个定义,语法辅助函数计算出的索引也恰为 $\mathit{shift}$ 给出的索引;一个爬升停在 $p$ 的限定 $\this$ 引用改为指名 $\ell_i$ 而非 $\last(p)$,这是同一个引用 $(n, w)$。其余文本不变。
\else
Decomposing along the leaves is the case of a partition. A \emph{leaf} of $M$ is a pair
$(p \cat v, \rho)$ with $\rho \in \inherits_D(p \cat v)$, or a written path $p \cat v$ with
$\inherits_D(p \cat v) = \varnothing$ and no written extension $p \cat v \snoc \ell$. When the subtree
of $p$ is finite, a partition of the leaves of $M$ into $k \ge 2$ nonempty sets gives $k$ parts: $E_i$
holds the pairs in the $i$-th set, and $W_i$ contains those $v$ for which $p \cat v$ is a prefix of the
path of a leaf in it; every written path of a finite subtree is a prefix of the path of a leaf, so these
parts cover $M$.
At the surface, the part at $q_i$ is written by copying the subtree of $M$ spanned by the $i$-th set together
with, for each of its references, the definition that the reference's lexical lookup finds, so that it
finds the same definition in the part and the indices that the syntactic helpers compute are those of
$\mathit{shift}$; a qualified-$\this$ reference whose climb stops
at $p$ names $\ell_i$ instead of $\last(p)$, which is the same reference $(n, w)$. No other text changes.
\fi

\paragraph{\bilingual{Protection}{保护性}}
\ifPaperChinese
在第~\ref{sec:prior-modularity}~节保护性的两种读法中,我们取在那里从 Peyton Jones 等人与 Turner 处回顾的编码为数据,而不取控制,因为控制取决于每个处理者怎样回应,而编码为数据可以由构造成立。由构造的保护性属于整门语言,而不属于任何单个构造:一门语言只保留保持它的构造,才得到它,如同 Turner 的全函数式编程只保留全函数,继承演算只保留 mixin、定义与引用。继承演算没有副作用:其语义是文档所决定的最小不动点,只是文档的函数(定理~\ref{thm:well-defined})。\footnote{MIXINv2 以依赖注入提供输入与输出:执行它们的 Python 模块是程序所继承的 mixin。继承演算程序既不调用也不执行这些宿主运算;它们的效果与失败都留在宿主一侧,不在本文所主张的保护性之内。}它也覆盖每个输入,这关乎语义函数在何处有定义,而非求值是否终止:输入经由组合进入程序,如同清单~\ref{lst:cheatsheet} 中的参数,而良构的复合与其他程序无异,所以命题~\ref{prop:coverage} 适用于它。语义函数借以读取程序的语法辅助函数以两种方式是部分的(附录~\ref{app:well-definedness}):\emph{缺席},即某个标签在一条路径之下无处写下;以及\emph{畸形},即一个形状不对的写下节点,良构性排除了它:
\else
Of the two readings of protection in Section~\ref{sec:prior-modularity}, we take the encoding as data
that we recalled there from Peyton Jones et al.\ and Turner rather than containment, since containment
rests on how each handler responds while an encoding as data can hold by construction.
Protection by construction belongs to a whole language rather than to any one construct: a language
obtains it by keeping only constructs that preserve it, as Turner's total functional programming keeps
only total functions and inheritance-calculus keeps only mixins, definitions, and references. Inheritance-calculus
has no side effects: its semantics is the
least fixed point that a document determines, a function of the document alone
(Theorem~\ref{thm:well-defined}).\footnote{MIXINv2 supplies input and output by dependency
injection: the Python modules that perform them are mixins that a program inherits. An
inheritance-calculus program neither calls nor performs these host operations; their effects and
failures stay on the host side and are outside the protection claimed here.} It also covers every input, which concerns where the semantic functions are defined, not whether
evaluation terminates: an input enters a program by composition, as an argument does in Listing~\ref{lst:cheatsheet}, and a well-formed composite is a program like any
other, so Proposition~\ref{prop:coverage} applies to it.
The syntactic helpers
through which the semantic functions read a program are partial in two ways
(Appendix~\ref{app:well-definedness}): \emph{absence}, a label written nowhere under a path, and
\emph{malformedness}, a written node of the wrong shape, which well-formedness excludes:
\fi

\begin{proposition}[\bilingual{Coverage}{覆盖}]\label{prop:coverage}
\ifPaperChinese
在良构程序上,定理~\ref{thm:well-defined} 的语义函数所遇到的一切无定义都是缺席,由 $\dom(\inherits)$ 作为布尔值报告;没有其他辅助函数失败,因而其指称 $\lfp(T_P)$ 的到达不经任何失败。
\else
On a well-formed program, every undefinedness that the semantic functions of
Theorem~\ref{thm:well-defined} meet is absence, which $\dom(\inherits)$ reports as a Boolean; no
other helper fails, so their denotation $\lfp(T_P)$ is reached without any failure.
\fi
\end{proposition}

\section{\bilingual{Case Study: Nat Arithmetic}{案例研究：Nat(自然数)算术}}
\label{sec:case-study}

\ifPaperChinese
我们在继承演算中实现了布尔逻辑与自然数算术,包括一元与二进制两种形式。
本节详细追踪一元 Nat 的情形:其数据表示、加法与相等判断。
该实现遵循普通的 Gang-of-Four 设计模式~\cite{gamma1994-design-patterns},采用声明式面向对象风格,将每项关注点沿两个轴分解为 mixin:操作轴与情形轴。由于演算本身没有方法,所有名称均为名词或形容词:\mintinline{yaml}{UpperCamelCase} 名称扮演模块、类型、数据构造子与方法的角色,\mintinline{yaml}{lowerCamelCase} 名称扮演字段与参数的角色。
下文展示的源文件都是 MIXINv2 库的一部分\anon[~(作为补充材料附上)]{~\cite{yang2026-mixinv2}},逐字原样呈现且可执行;本节的行为陈述是测试套件执行这些文件所记录的结果,机器生成的答案表见附录~\ref{app:evaluation-trace}。
\else
We implemented boolean logic and natural number arithmetic, both unary
and binary, in inheritance-calculus.
This section traces the unary Nat case in detail: its data representation,
addition, and equality.
The implementation follows ordinary Gang-of-Four design
patterns~\cite{gamma1994-design-patterns} in a declarative object-oriented
style, with each concern split into mixins along two axes: operations and
cases. Because the calculus has no methods of its own, every name is a noun or
adjective: \mintinline{yaml}{UpperCamelCase} names play the roles of modules, types, data
constructors, and methods, and \mintinline{yaml}{lowerCamelCase} names the roles of fields and
parameters.
The source files shown below are part of the MIXINv2
library\anon[~(included as supplementary
material)]{~\cite{yang2026-mixinv2}}, reproduced verbatim and
executable; the behavioral statements of this section are outcomes
recorded by the test suite executing these files, with the
machine-generated answer table in
Appendix~\ref{app:evaluation-trace}.
\fi

\ifPaperChinese
  mixin \mintinline{yaml}{NatData} 是自然数接口及其工厂。\mintinline{yaml}{NatFactory} 是一个工厂,其抽象产品类型被别名为 \mintinline{yaml}{Nat}:
  \else
  The mixin \mintinline{yaml}{NatData} is the natural-number interface and its factory.
  \mintinline{yaml}{NatFactory} is a factory whose abstract product type is aliased as \mintinline{yaml}{Nat}:
  \fi
\mixinstrip{../packages/mixinv2-library/src/mixinv2_library/Builtin/NatData.mixin.yaml}

\ifPaperChinese
  每个数据构造子都是一个独立的实现 mixin,继承 \mintinline{yaml}{NatData}。一个 \mintinline{yaml}{Zero} 值继承 \mintinline{yaml}{Product}:
  \else
  Each data constructor is a separate implementation mixin that inherits
  \mintinline{yaml}{NatData}.
  A \mintinline{yaml}{Zero} value inherits \mintinline{yaml}{Product}:
  \fi
\mixinstrip{../packages/mixinv2-library/src/mixinv2_library/Builtin/NatDataZero.mixin.yaml}

\ifPaperChinese
  一个 \mintinline{yaml}{Successor} 值继承 \mintinline{yaml}{Product},并暴露一个类型为 \mintinline{yaml}{Product} 的 \mintinline{yaml}{predecessor}:
  \else
  A \mintinline{yaml}{Successor} value inherits \mintinline{yaml}{Product} and
  exposes a \mintinline{yaml}{predecessor} whose type is \mintinline{yaml}{Product}:
  \fi
\mixinstrip{../packages/mixinv2-library/src/mixinv2_library/Builtin/NatDataSuccessor.mixin.yaml}
\ifPaperChinese
与 \mintinline{yaml}{NatDataZero} 合在一起,这以开放的方式完成了自然数的 Peano 编码,即函子 $F(X) = 1 + X$ 的初始代数:两个构造子作为独立的 mixin 加入,而非固定在单个封闭的数据类型定义中。
\else
Together with \mintinline{yaml}{NatDataZero}, this completes the Peano encoding of
the natural numbers, the initial algebra of the functor $F(X) = 1 + X$,
in an open way: the two constructors are added as independent mixins
rather than fixed in a single closed datatype definition.
\fi

\ifPaperChinese
加法操作沿情形轴分解为一个\emph{接口} mixin \mintinline{yaml}{NatPlus},该接口 mixin 保留关注点的名称并在抽象 \mintinline{yaml}{Product} 上声明命令,以及每个情形对应一个\emph{实现} mixin,实现 mixin 在名称后附加构造子名称(即下文的 \mintinline{yaml}{NatPlusZero} 与 \mintinline{yaml}{NatPlusSuccessor})。\mintinline{yaml}{NatPlus} 继承 \mintinline{yaml}{NatData}。面向对象方法会用动词命名的地方,命令模式~\cite{gamma1994-design-patterns} 将其具化为以名词命名的命令,因此加法方法被命名为 \mintinline{yaml}{Plus} 而非 \mintinline{yaml}{Add};它声明参数 \mintinline{yaml}{addend},并通过对 \mintinline{yaml}{sum} 的投影暴露其结果,二者在此已指定类型,但尚未计算:
\else
The addition operation is split along the case axis into an \emph{interface}
mixin \mintinline{yaml}{NatPlus}, which keeps the concern's name and declares the command
on the abstract \mintinline{yaml}{Product}, and one \emph{implementation} mixin per case,
which appends the constructor name (\mintinline{yaml}{NatPlusZero} and
\mintinline{yaml}{NatPlusSuccessor} below). \mintinline{yaml}{NatPlus} inherits
\mintinline{yaml}{NatData}. Where an
object-oriented method would be named with a verb, the command
pattern~\cite{gamma1994-design-patterns} reifies it as a noun-named command, so
the addition method is named \mintinline{yaml}{Plus} rather than \mintinline{yaml}{Add}; it declares
the parameter \mintinline{yaml}{addend} and exposes its result through projection
of \mintinline{yaml}{sum}, both typed here but not yet computed:
\fi
\mixinstrip{../packages/mixinv2-library/src/mixinv2_library/Builtin/NatPlus.mixin.yaml}

\ifPaperChinese
  以 $0$、$\mathrm{S}$、$+$ 分别记 \mintinline{yaml}{Zero}、\mintinline{yaml}{Successor}、\mintinline{yaml}{Plus},基础情形为 $0 + m = m$。这里 \mintinline{yaml}{NatPlusZero} 充当一个导入其他模块的模块:它继承 \mintinline{yaml}{Zero} 数据与 \mintinline{yaml}{Plus} 接口,并直接返回被加数:
  \else
  Writing $0$, $\mathrm{S}$, and $+$ for \mintinline{yaml}{Zero}, \mintinline{yaml}{Successor}, and \mintinline{yaml}{Plus}, the base case is $0 + m = m$. Here \mintinline{yaml}{NatPlusZero} acts as a module that
  imports other modules: it inherits the \mintinline{yaml}{Zero} data and the \mintinline{yaml}{Plus}
  interface, and returns the addend directly:
  \fi
\mixinstrip{../packages/mixinv2-library/src/mixinv2_library/Builtin/NatPlusZero.mixin.yaml}

\ifPaperChinese
  递归情形为 $\mathrm{S}(n) + m = n + \mathrm{S}(m)$,方法是以递增后的被加数委托给 $n$ 的 \mintinline{yaml}{Plus}。每个文件为其引用所经过的每条路径都以继承一个类型的方式加以声明,正如此文件对 \mintinline{yaml}{predecessor} 及其 \mintinline{yaml}{Plus} 的再声明;声明是普通的定义,深度合并把它与计算出该值的定义合并在一起:
  \else
  The recursive case is $\mathrm{S}(n) + m = n + \mathrm{S}(m)$,
  implemented by delegating to $n$'s \mintinline{yaml}{Plus} with an incremented addend.
  Each file declares, by inheriting a type, every path its references
  traverse, as this one re-declares \mintinline{yaml}{predecessor} and its
  \mintinline{yaml}{Plus}; a declaration is an ordinary definition, and deep
  merge combines it with the definitions that compute the value:
  \fi
\mixinstrip{../packages/mixinv2-library/src/mixinv2_library/Builtin/NatPlusSuccessor.mixin.yaml}

\ifPaperChinese
相等判断需要对数据构造子进行情形分析;我们使用访问者模式~\cite{gamma1994-design-patterns}。\mintinline{yaml}{NatVisitor} 向 \mintinline{yaml}{Product} 添加一个 \mintinline{yaml}{Acceptance} 命令,即访问者模式的 \mintinline{java}{accept} 方法的名词形式。mixin \mintinline{yaml}{NatVisitor} 将其声明为空,留下 \mintinline{yaml}{VisitorMap} 与 \mintinline{yaml}{Accepted} mixin 由各构造子的实现来填充:
\else
Equality requires case analysis on the data constructor; we use the
visitor pattern~\cite{gamma1994-design-patterns}.
\mintinline{yaml}{NatVisitor} adds an \mintinline{yaml}{Acceptance} command to \mintinline{yaml}{Product},
the noun form of the visitor pattern's \mintinline{java}{accept} method.
The mixin \mintinline{yaml}{NatVisitor} declares it empty, leaving a \mintinline{yaml}{VisitorMap} and an
\mintinline{yaml}{Accepted} mixin for the per-constructor implementations to fill in:
\fi
\mixinstrip{../packages/mixinv2-library/src/mixinv2_library/Builtin/NatVisitor.mixin.yaml}

\ifPaperChinese
  \mintinline{yaml}{Zero} 的实现选择 \mintinline{yaml}{ZeroVisitor} 分支。这里的访问者模式与~\cite{gamma1994-design-patterns} 中的表述不同,那里数据构造子由 \mintinline{java}{accept} 接口固定;此处通过组合允许添加新的构造子,因为每个构造子作为独立的 mixin 贡献其自身的分支:
  \else
  The \mintinline{yaml}{Zero} implementation selects the \mintinline{yaml}{ZeroVisitor} branch.
  The visitor pattern here, unlike its formulation in
  \cite{gamma1994-design-patterns} where the data constructors are fixed by the
  \mintinline{java}{accept} interface, admits new constructors by composition, since each
  constructor contributes its own branch as a separate mixin:
  \fi
\mixinstrip{../packages/mixinv2-library/src/mixinv2_library/Builtin/NatVisitorZero.mixin.yaml}

\ifPaperChinese
  \mintinline{yaml}{Successor} 的实现选择 \mintinline{yaml}{SuccessorVisitor} 分支:
  \else
  The \mintinline{yaml}{Successor} implementation selects the \mintinline{yaml}{SuccessorVisitor} branch:
  \fi
\mixinstrip{../packages/mixinv2-library/src/mixinv2_library/Builtin/NatVisitorSuccessor.mixin.yaml}

\ifPaperChinese
  \mintinline{yaml}{BooleanData} 定义输出的域:
  \else
  \mintinline{yaml}{BooleanData} defines the output sort:
  \fi
\mixinstrip{../packages/mixinv2-library/src/mixinv2_library/Builtin/BooleanData.mixin.yaml}
\ifPaperChinese
  \begin{sloppypar}
  为了验证 $3 + 4 = 7$,我们需要对 Nat 值进行相等判断。mixin \mintinline{yaml}{NatEquality} 导入 \mintinline{yaml}{NatVisitor} 与 \mintinline{yaml}{BooleanData},并在 \mintinline{yaml}{Product} 上声明一个 \mintinline{yaml}{Equal} 命令,由各构造子的实现提供。由于 \mintinline{yaml}{NatEquality} 继承 \mintinline{yaml}{BooleanData},限定 $\this$ \mintinline{yaml}{[NatEquality, ~]} 可以到达 \mintinline{yaml}{Boolean},因此 \mintinline{yaml}{Equal} 以类型 \mintinline{yaml}{[NatEquality, ~, Boolean]} 声明其结果 \mintinline{yaml}{equal}:
  \end{sloppypar}
  \else
  \begin{sloppypar}
  To test that $3 + 4 = 7$, we needed equality on Nat values.
  The mixin \mintinline{yaml}{NatEquality} imports \mintinline{yaml}{NatVisitor} and \mintinline{yaml}{BooleanData},
  and declares an \mintinline{yaml}{Equal} command on \mintinline{yaml}{Product} that the per-constructor implementations provide.
  Because \mintinline{yaml}{NatEquality} inherits \mintinline{yaml}{BooleanData},
  the qualified this \mintinline{yaml}{[NatEquality, ~]} can reach \mintinline{yaml}{Boolean},
  so \mintinline{yaml}{Equal} declares its result
  \mintinline{yaml}{equal} with the type \mintinline{yaml}{[NatEquality, ~, Boolean]}:
  \end{sloppypar}
  \fi
\mixinstrip{../packages/mixinv2-library/src/mixinv2_library/Builtin/NatEquality.mixin.yaml}

\ifPaperChinese
\mintinline{yaml}{Zero} 只与另一个 \mintinline{yaml}{Zero} 相等。这里的继承执行情形分析:继承 \mintinline{yaml}{[other, Acceptance]} 对 \mintinline{yaml}{other} 是由 \mintinline{yaml}{Zero} 还是 \mintinline{yaml}{Successor} 构建进行分派,在 \mintinline{yaml}{ZeroVisitor} 分支返回 \mintinline{yaml}{True},否则返回 \mintinline{yaml}{False}。\mintinline{yaml}{True} 与 \mintinline{yaml}{False} 均为 \mintinline{yaml}{BooleanFactory} 的数据构造子;该工厂来自 \mintinline{yaml}{BooleanData},通过继承的 \mintinline{yaml}{NatEquality} 间接导入,并以 \mintinline{yaml}{[NatEqualityZero, ~, BooleanFactory]} 到达。同一继承机制也提供了结果类型:在 \mintinline{yaml}{Accepted} 中,\mintinline{yaml}{equal} 继承 \mintinline{yaml}{[NatEqualityZero, ~, Boolean]},从而声明其类型为 \mintinline{yaml}{Boolean}:
\else
\mintinline{yaml}{Zero} is equal only to another \mintinline{yaml}{Zero}.
Inheriting here performs case analysis:
inheriting \mintinline{yaml}{[other, Acceptance]} dispatches on whether \mintinline{yaml}{other}
was built by \mintinline{yaml}{Zero} or \mintinline{yaml}{Successor}, returning \mintinline{yaml}{True}
on the \mintinline{yaml}{ZeroVisitor} branch and \mintinline{yaml}{False} otherwise.
Both \mintinline{yaml}{True} and \mintinline{yaml}{False} are data constructors of
\mintinline{yaml}{BooleanFactory}, which comes from \mintinline{yaml}{BooleanData}, imported
indirectly through the inherited \mintinline{yaml}{NatEquality}, and is reached as
\mintinline{yaml}{[NatEqualityZero, ~, BooleanFactory]}.
The same inheritance mechanism also supplies the result type: in
\mintinline{yaml}{Accepted}, \mintinline{yaml}{equal} inherits
\mintinline{yaml}{[NatEqualityZero, ~, Boolean]}, declaring its type to
be \mintinline{yaml}{Boolean}:
\fi
\mixinstrip{../packages/mixinv2-library/src/mixinv2_library/Builtin/NatEqualityZero.mixin.yaml}

\ifPaperChinese
  \mintinline{yaml}{Successor} 只与另一个具有相等前驱的 \mintinline{yaml}{Successor} 相等,因此其实现通过 \mintinline{yaml}{[Successor, ~, predecessor, Equal]} 进行递归:
  \else
  \mintinline{yaml}{Successor} is equal only to another \mintinline{yaml}{Successor} with an equal
  predecessor, so its implementation recurses through
  \mintinline{yaml}{[Successor, ~, predecessor, Equal]}:
  \fi
\mixinstrip{../packages/mixinv2-library/src/mixinv2_library/Builtin/NatEqualitySuccessor.mixin.yaml}
\ifPaperChinese
通过继承来组合 \mintinline{yaml}{Plus} 与 \mintinline{yaml}{Equal} 的实现,无需修改任何 mixin,即可使所得的 Nat 值自动同时具备 \mintinline{yaml}{Plus} 与 \mintinline{yaml}{Equal}。
\else
Composing the \mintinline{yaml}{Plus} and \mintinline{yaml}{Equal} implementations by
inheriting them, with no modifications to any mixin, gives the
resulting Nat values both \mintinline{yaml}{Plus} and \mintinline{yaml}{Equal}
automatically.
\fi
\mixinstrip{../packages/mixinv2-library/tests/NatConstants.mixin.yaml}
\ifPaperChinese
  算术测试继承 \mintinline{yaml}{NatConstants} 以及它所需的操作:
  \else
  The arithmetic test inherits \mintinline{yaml}{NatConstants} together with the operations it needs:
  \fi
\mixinstrip{../packages/mixinv2-library/tests/NatArithmeticTest.mixin.yaml}
\ifPaperChinese
引用 \mintinline{yaml}{[NatArithmeticTest, ~, threePlusFourEqualsSeven, equal]} 所解析到的 mixin 继承 \mintinline{yaml}{True},验证 $(3+4) = 7$;同一文件中的其他用例覆盖了直接相等以及更多加法结果。%
\footnote{
  读取一个结果,需要借助一个打印它的 Python 模块,或者借助元语言观察以在该路径上调用 $\supers$。测试套件使用的入口 mixin \mintinline{yaml}{ArithmeticTest} 继承了此处展示的 \mintinline{yaml}{NatArithmeticTest} 核心,并额外组合了 Python 模块 ToPython,从而以 Python 原生值读回结果;本文展示的是不含 Python 模块的核心。
}
没有任何现有 mixin 被修改。\mintinline{yaml}{NatEquality}(一项新操作)与 \mintinline{yaml}{Boolean}(一种新数据类型)各自只需要一个新 mixin:沿操作轴与情形轴均可扩展。
\else
The mixin that the reference \mintinline{yaml}{[NatArithmeticTest, ~, threePlusFourEqualsSeven, equal]} resolves to
inherits \mintinline{yaml}{True}, verifying $(3+4) = 7$; other cases in the same file
cover direct equality and further addition results.%
\footnote{
  Reading a result requires either a Python module that prints it or a
  metalanguage observation that invokes $\supers$ on the path.
  The entry-point mixin \mintinline{yaml}{ArithmeticTest} used by the test suite
  inherits the \mintinline{yaml}{NatArithmeticTest} core shown here and additionally
  composes the Python module ToPython to read results back as Python natives; the
  paper shows the core without Python modules.
}
No existing mixin was modified.
\mintinline{yaml}{NatEquality}, a new operation, and \mintinline{yaml}{Boolean},
a new data type, each required only a new mixin: extensibility along both
the operation axis and the case axis.
\fi

\paragraph{\bilingual{Tree-level inheritance and return types}{树层次的继承与返回类型}}
\label{sec:discussion-tree-level-inheritance}
\ifPaperChinese
每个操作都是一个独立的 mixin,扩展工厂的子树。继承递归地合并共享同一标签的子树,因此工厂中的每个标签都从所有被继承的 mixin 获得全部操作。关键在于,这同样适用于返回类型:\mintinline{yaml}{NatPlus} 从 \mintinline{yaml}{NatFactory} 返回值,而 \mintinline{yaml}{NatEquality} 独立地向同一工厂添加 \mintinline{yaml}{Equal}。同时继承两者,使得 \mintinline{yaml}{Plus} 返回的值自动携带 \mintinline{yaml}{Equal},无需修改任何一个 mixin。
\else
Each operation is an independent mixin that extends a factory's
subtree.
Inheritance recursively merges subtrees that share a label,
so every label in the factory acquires all operations from all
inherited mixins.
Crucially, this applies to return types: \mintinline{yaml}{NatPlus} returns
values from \mintinline{yaml}{NatFactory}, and \mintinline{yaml}{NatEquality}
independently adds \mintinline{yaml}{Equal} to the same factory.
Inheriting both causes the values returned by \mintinline{yaml}{Plus} to
carry \mintinline{yaml}{Equal} automatically, without modifying either mixin.
\fi

\paragraph{\bilingual{Peano-encoded Nats are tries}{Peano 编码的 Nat 是字典树}}
\label{sec:case-study-cartesian}
\ifPaperChinese
Peano 编码的 Nat 是一个 mixin,其结构映射了构建它所用的数据构造子路径:\mintinline{yaml}{Zero} 是一个平坦 mixin;$\mathrm{S}(0)$ 是一个带有指向一个 \mintinline{yaml}{Zero} mixin 的 \mintinline{yaml}{predecessor} 的 mixin;依此类推。这恰好是字典树~\cite{fredkin1960-trie-memory}的形状。%
\footnote{
  一元 Nat 的字典树可以被视为一个稀疏链表,其深度等于它所代表的数值:在字典树并中,链表中的某些 \mintinline{yaml}{Successor} 节点指向 \mintinline{yaml}{Zero}(叶节点),因此链表中并非每个位置都被占用。\anon[{} 补充材料]{{} MIXINv2 源码~\cite{yang2026-mixinv2}}{} 中包含一个字典树深度为对数级的二进制自然数(BinNat)算术。
}
一个同时继承两个 Nat 值的 mixin 是一个字典树并。下面的测试将 \mintinline{yaml}{Plus} 应用于 $\{1,2\}$ 的字典树并,被加数为 $\{3,4\}$ 的字典树并,然后用 \mintinline{yaml}{Equal} 将结果与自身比较:
\else
A Peano-encoded Nat is a mixin whose structure mirrors the
data-constructor path used to build it: \mintinline{yaml}{Zero} is a flat mixin;
$\mathrm{S}(0)$ is a mixin with a \mintinline{yaml}{predecessor}
pointing to a \mintinline{yaml}{Zero} mixin; and so on.
This is exactly the shape of a trie~\cite{fredkin1960-trie-memory}.%
\footnote{
  The unary Nat trie can be viewed as a sparse linked list
  whose depth equals the number it represents:
  in a trie union, some \mintinline{yaml}{Successor} nodes along
  the list point to \mintinline{yaml}{Zero} (leaf nodes),
  so not every position in the list is occupied.
  The\anon[{} supplementary material]{{} MIXINv2 source~\cite{yang2026-mixinv2}}{} includes a binary natural number
  (BinNat) arithmetic whose trie depth is logarithmic.
}
A mixin that simultaneously inherits two Nat values is a trie union.
The following test applies \mintinline{yaml}{Plus} to a trie union of $\{1,2\}$
with addend a trie union of $\{3,4\}$, then compares the result to itself
with \mintinline{yaml}{Equal}:
\fi
\mixinstrip{../packages/mixinv2-library/tests/NatCartesianProductTest.mixin.yaml}
\ifPaperChinese
\mintinline{yaml}{[NatOneOrTwoPlusThreeOrFour, sum]} 求值为字典树 $\{4,5,6\}$,即所有两两之和的集合。\mintinline{yaml}{Equal} 随后对从 $\{4,5,6\} \times \{4,5,6\}$ 中取出的每一对进行应用;引用 \mintinline{yaml}{[NatCartesianProductTest, ~, NatResultEqualitySelf, equal]} 所解析到的 mixin 同时继承 \mintinline{yaml}{True} 与 \mintinline{yaml}{False}。笛卡尔积直接来自语义方程。\mintinline{yaml}{NatOneOrTwo} 同时继承 \mintinline{yaml}{One} 与 \mintinline{yaml}{Two},因此 $\bases(\text{``NatOneOrTwo''})$ 包含两条路径。当 \mintinline{yaml}{NatOneOrTwoPlusThreeOrFour} 继承 \mintinline{yaml}{[NatOneOrTwo, Plus]} 时,其 $\bases$ 只是一条路径 $(\text{``NatOneOrTwo''}, \text{``Plus''})$,它继承为 \mintinline{yaml}{Successor} 只写了一次的 \mintinline{yaml}{Plus},经由 \mintinline{yaml}{One} 与 \mintinline{yaml}{Two} 两者到达。这个 \mintinline{yaml}{Plus} 通过 $\this$ 解析其 \mintinline{yaml}{addend},到达 \mintinline{yaml}{NatThreeOrFour};该值本身是 \mintinline{yaml}{Three} 与 \mintinline{yaml}{Four} 的字典树并,因此 \mintinline{yaml}{addend} 路径的 $\supers$ 包含两个值。由于 \mintinline{yaml}{[Successor, Plus]} 以递增后的被加数递归委托给 \mintinline{yaml}{[predecessor, Plus]},而 \mintinline{yaml}{predecessor} 也通过字典树并解析,递归在每个字典树节点处分支。因此结果 \mintinline{yaml}{sum} 收集了通过两个操作数的任意组合所能到达的全部路径,即笛卡尔积。这一分支正是多目标 $\this$(方程~\ref{eq:this})被真正触发之处:在字典树并处,\mintinline{yaml}{_increasedAddend} 内的 \mintinline{yaml}{[addend]} 引用一次解析到 \mintinline{yaml}{[NatOneOrTwoPlusThreeOrFour, addend]} 与 \mintinline{yaml}{[NatOneOrTwoPlusThreeOrFour, _recursiveAddition, addend]} 两条路径,故前沿 $S$ 携带多个目标。

同样操作单个值的 \mintinline{yaml}{Plus} 与 \mintinline{yaml}{Equal} 文件,无需修改,即符合逻辑编程的关系语义:$\overrides$ 与 $\supers$ 收集所有继承路径,每个操作都在其上分布。
\else
\begin{sloppypar}
\mintinline{yaml}{[NatOneOrTwoPlusThreeOrFour, sum]} evaluates to the trie $\{4,5,6\}$, the set of all pairwise sums.
\mintinline{yaml}{Equal} then applies to every pair drawn from $\{4,5,6\} \times \{4,5,6\}$;
the mixin that the reference \mintinline{yaml}{[NatCartesianProductTest, ~, NatResultEqualitySelf, equal]} resolves to
inherits both \mintinline{yaml}{True} and \mintinline{yaml}{False}.
The Cartesian product arises directly from the semantic equations.
\mintinline{yaml}{NatOneOrTwo} inherits both \mintinline{yaml}{One} and \mintinline{yaml}{Two},
so $\bases(\text{``NatOneOrTwo''})$ contains two paths.
When \mintinline{yaml}{NatOneOrTwoPlusThreeOrFour} inherits
\mintinline{yaml}{[NatOneOrTwo, Plus]}, its $\bases$ is the single path
$(\text{``NatOneOrTwo''}, \text{``Plus''})$, which inherits the \mintinline{yaml}{Plus} written once
for \mintinline{yaml}{Successor}, reached through both \mintinline{yaml}{One} and
\mintinline{yaml}{Two}.
This \mintinline{yaml}{Plus} resolves its \mintinline{yaml}{addend} via
$\this$, reaching \mintinline{yaml}{NatThreeOrFour}, which is itself a
trie union of \mintinline{yaml}{Three} and \mintinline{yaml}{Four}:
the $\supers$ of the \mintinline{yaml}{addend} path contain both values.
Since \mintinline{yaml}{[Successor, Plus]} recursively delegates to
\mintinline{yaml}{[predecessor, Plus]} with an incremented addend,
and \mintinline{yaml}{predecessor} also resolves through the trie union,
the recursion branches at every trie node.
The result \mintinline{yaml}{sum} therefore collects all paths reachable
through any combination of the two operands, which is the
Cartesian product.
This branching is where multi-target $\this$
(equation~\ref{eq:this}) is exercised: at the trie union the
\mintinline{yaml}{[addend]} reference inside
\mintinline{yaml}{_increasedAddend} resolves at once to the two paths
\mintinline{yaml}{[NatOneOrTwoPlusThreeOrFour, addend]} and
\mintinline{yaml}{[NatOneOrTwoPlusThreeOrFour, _recursiveAddition, addend]},
so the frontier $S$ carries more than one target.
\end{sloppypar}

The same \mintinline{yaml}{Plus} and \mintinline{yaml}{Equal} files that operate on
single values thus agree, without modification, with the relational
semantics of logic programming:
$\overrides$ and $\supers$ collect all inheritance paths, and
every operation distributes over them.
\fi

\section{\bilingual{Other Examples}{其他示例}}
\label{sec:examples}

\ifPaperChinese
在第~\ref{sec:case-study}~节的案例研究之外,诸附录与补充材料还带有更多可执行的例子;下一节报告它们所显现的现象。
\begin{description}
  \item[多外层解析]
    Scala~3 与 NixOS 模块系统都以静态错误拒绝、而在继承演算中良定义的程序(附录~\ref{app:scala-multi-target}),并配有经五个函数逐步推演的机器生成轨迹(附录~\ref{app:evaluation-trace});它演练第~\ref{sec:definitions}~节的多目标 $\this$ 解析。
  \item[Datalog 编码]
    附录~\ref{app:datalog-encoding}~将 Datalog 编码进继承演算,仅用第~\ref{sec:syntax}~节的三个基本构造即可表达变量连接、常量匹配、多体规则和递归规则,全部示例可在 MIXINv2 中执行\anon[]{~\cite{yang2026-mixinv2}}。
  \item[Nat 结果记录]
    案例研究的机器生成答案表(附录~\ref{app:nat-answers}),其中记录了字典树并的结果:诸笛卡尔和与集合值的自相等判断。
  \item[BinNat 算术]
    补充材料\anon[]{~\cite{yang2026-mixinv2}}中的二进制自然数,编码为对数深度的位字典树,并为二进制表示重新实现了 \mintinline{yaml}{Plus} 与 \mintinline{yaml}{Equal}。
  \item[数据库驱动的 web 服务器]
    补充材料\anon[]{~\cite{yang2026-mixinv2}}中的一个用户查询服务,以多种组合方式装配,既有声明式的也有嵌入 Python 的,其单一业务逻辑 mixin 不加修改地同时运行于同步与异步运行时。
\end{description}
\else
Beyond the case study of Section~\ref{sec:case-study}, the appendices and
the supplementary material carry further executable examples; the next
section reports the phenomena they surfaced.
\begin{description}
  \item[Multi-outer resolution]
    The program that Scala~3 and the NixOS module system both reject with a
    static error, well-defined in inheritance-calculus
    (Appendix~\ref{app:scala-multi-target}), together with its
    machine-generated trace through the five functions
    (Appendix~\ref{app:evaluation-trace}); it exercises the multi-target
    $\this$ resolution of Section~\ref{sec:definitions}.
  \item[Datalog encoding]
    Appendix~\ref{app:datalog-encoding} encodes Datalog into
    inheritance-calculus, expressing variable joins,
    constant matches, multi-body rules, and recursive rules with
    the three primitives of Section~\ref{sec:syntax} alone, every example
    executable in MIXINv2\anon[]{~\cite{yang2026-mixinv2}}.
  \item[Recorded Nat outcomes]
    The machine-generated answer table of the case study
    (Appendix~\ref{app:nat-answers}), recording among others the trie-union
    outcomes: the Cartesian sums and the set-valued self-equality.
  \item[BinNat arithmetic]
    Binary naturals in the supplementary
    material\anon[]{~\cite{yang2026-mixinv2}}, encoded as bit tries of
    logarithmic depth, with \mintinline{yaml}{Plus} and
    \mintinline{yaml}{Equal} implemented anew for the binary representation.
  \item[Database-backed web server]
    A user-lookup service in the supplementary
    material\anon[]{~\cite{yang2026-mixinv2}}, assembled in several
    compositions, declarative and embedded in Python, whose single
    business-logic mixin runs unchanged under both a synchronous and an
    asynchronous runtime.
\end{description}
\fi

\section{\bilingual{Results}{结果}}
\label{sec:emergent-phenomena}

\ifPaperChinese
在 MIXINv2 中开发真实程序的过程中，我们发现了三种现象，每一种都是对某种已有设计模式的非预期用法，超出了该模式的原始能力。

\begin{description}
  \item[关系语义]
    第~\ref{sec:case-study}~节的笛卡尔积语义不限于 $\olbl{Nat}$：任何 Datalog 元组都可编码为链表，而深度合并在这些元组上计算集合操作;在我们测试套件中的每个编码程序上,所得答案都与 Datalog 的关系语义相符。
    编码的递归规则（附录~\ref{app:datalog-encoding}）依赖值递归：最小不动点 $\lfp(T_P)$ 将一个引用解析为满足 $S = F(S)$ 的集合 $S$，沿循环累积答案。
    这种值递归使编码后的递归规则收敛。
    这种逻辑编程语义是对编码案例研究算术的工厂模式的非预期用法中涌现出来的。

  \item[多维度可扩展性]\label{item:multi-dimensional-extensibility}
    \begin{sloppypar}
    由构造保证的可扩展性(第~\ref{sec:modularity}~节)使不可扩展的那一半无从写出。
    案例研究具体演练了两个轴:\mintinline{yaml}{NatPlus} 作为新操作加入既有构造子之上,而 \mintinline{yaml}{Successor} 情形作为三个独立 mixin(\mintinline{yaml}{NatDataSuccessor}、\mintinline{yaml}{NatPlusSuccessor}、\mintinline{yaml}{NatEqualitySuccessor})加入并扩展既有操作,不修改任何已展示的文件。
    在其他地方恢复第二个扩展维度需要专门的机制~\cite{liang1995-monad-transformers, lammel2003-scrap-your-boilerplate, loh2006-open-data-types, swierstra2008-data-types-a-la-carte, carette2009-finally-tagless, oliveira2012-object-algebras, kiselyov2013-extensible-effects, wu2014-effect-handlers-in-scope, kiselyov2015-freer-monads, wang2016-expression-problem-trivially, leijen2017-row-typed-algebraic-effects, poulsen2023-hefty-algebras}。
    这第二个可扩展性维度是对访问者模式的非预期用法中涌现出来的。
    \end{sloppypar}

  \item[函数颜色盲]\label{sec:practical-function-color}
    一个函数的颜色,即同步或异步~\cite{nystrom2015-function-color},之所以难以去除,是因为一个异步函数就是一个续体传递风格(continuation-passing style)的函数,它把结果交给回调而非返回,于是它的每个调用者都必须采用同一风格,颜色便沿调用链向上传播。这一现象依赖于输入与输出,而本文的继承演算示例不含输入与输出。在第~\ref{sec:examples}~节列出的、位于补充材料中的 web 服务器中,我们用继承演算描述业务逻辑,并把同步与异步操作实现为注入其中的 Python 模块;我们没有发现任一组模块与继承演算代码之间有任何特殊的兼容性问题。我们把该问题的缺席归因于业务逻辑是一阶的:它既不固定调用链,也不固定求值顺序,颜色无从传播,于是每个操作的颜色都留在它的 Python 模块内部。这一颜色盲是对经由继承的依赖注入的非预期用法中涌现出来的。

\end{description}
\else
Developing real programs in MIXINv2 surfaced three phenomena, each an
off-label use of an adopted design pattern beyond its original abilities.

\begin{description}
  \item[Relational semantics]
    The Cartesian-product semantics of Section~\ref{sec:case-study}
    generalizes beyond $\olbl{Nat}$: any Datalog tuple encodes
    as a linked list, and deep merge computes set operations over such
    tuples; on every encoded program in our test suite, the answers agree with
    the relational semantics of Datalog.
    The recursive rules of the encoding
    (Appendix~\ref{app:datalog-encoding}) rely on value recursion: the least fixed point
    $\lfp(T_P)$ resolves a reference to a set $S$ with
    $S = F(S)$, accumulating answers across the cycle.
    This value recursion makes the encoded recursive rules converge.
    This logic-programming semantics emerged from an off-label use of
    the factory pattern that encodes the case study's arithmetic.

  \item[Multi-dimensional extensibility]\label{item:multi-dimensional-extensibility}
    Extensibility by construction (Section~\ref{sec:modularity}) leaves no
    nonextensible half to write.
    The case study exercises both axes concretely:
    \mintinline{yaml}{NatPlus} joins as a new operation over the
    existing constructors, and the \mintinline{yaml}{Successor} case
    joins as three independent mixins
    (\mintinline{yaml}{NatDataSuccessor},
    \mintinline{yaml}{NatPlusSuccessor},
    \mintinline{yaml}{NatEqualitySuccessor}) extending the existing
    operations, with no previously shown file modified.
    Recovering that second dimension elsewhere takes dedicated
    mechanisms~\cite{liang1995-monad-transformers, lammel2003-scrap-your-boilerplate, loh2006-open-data-types, swierstra2008-data-types-a-la-carte, carette2009-finally-tagless, oliveira2012-object-algebras, kiselyov2013-extensible-effects, wu2014-effect-handlers-in-scope, kiselyov2015-freer-monads, wang2016-expression-problem-trivially, leijen2017-row-typed-algebraic-effects, poulsen2023-hefty-algebras}.
    This second extensibility dimension emerged from an off-label use
    of the visitor pattern.

  \item[Function color blindness]\label{sec:practical-function-color}
    A function's color, synchronous or
    asynchronous~\cite{nystrom2015-function-color}, is hard to remove
    because an asynchronous function is a function in
    continuation-passing style, delivering its result to a callback rather
    than returning it, so each of its callers must adopt the same style and
    the color propagates up the call chain.
    This phenomenon relies on input and output, which the inheritance-calculus
    examples of this paper leave out. In the web server of
    Section~\ref{sec:examples}, in the supplementary material, we described the business logic in
    inheritance-calculus and implemented the synchronous and the asynchronous
    operations as Python modules injected into it, and we found no special
    compatibility issue between either set of modules and the inheritance-calculus code.
    We attribute the absence of the problem to the business logic being
    first-order: it fixes neither a call chain nor an evaluation order
    through which a color could propagate, so the color of each operation
    stays inside its Python module.
    This color blindness emerged from an off-label use of dependency
    injection via inheritance.

\end{description}
\fi

\section{\bilingual{Discussion}{讨论}}
\label{sec:semantic-variants}


\ifPaperChinese
第~\ref{sec:modularity}~节把 Meyer 的五条标准读作组合的若干性质,并为继承演算证明了它们。我们把由构造满足全部五条作为设计目标,并通过调查我们已知的所有候选方案,询问是否有替代设计能满足它。

\paragraph{\bilingual{Degrees of freedom}{自由度}}
以下分析假设 AST 是一棵\emph{命名树}：每条边都携带一个标签，因此每个节点都由从根节点出发的路径唯一标识。对于配置语言，这自然成立：JSON、YAML 或 Nix 属性集中的每个键都是一个标签，而嵌套产生路径。%
\footnote{
  Nix 是一门具有一等函数的函数式语言；仅凭其属性集并不构成命名树 DSL。然而，NixOS 模块系统~\cite{nixos-contributors-nixos-modules}将用户界面限制为通过递归合并组合的嵌套属性集，而正是这种 DSL 层面的结构自然地构成了命名树。
}
对于主流编程语言，例如 Java~\cite{gosling2015-java-se8-specification}、Haskell~\cite{marlow2010-haskell-report} 与 Scala~\cite{odersky2004-scala-overview}，AST 并不直接是命名树，因为匿名表达式位置和隐式作用域缺乏标签。%
\footnote{
  例如，$\mathbf{let}$ 的主体。
}
然而，这些位置仍可被分配合成的新鲜标签。%
\footnote{
  正如编译器为 Java 匿名类合成名字。
}
第~\ref{sec:forward-translation}~节的前向翻译正是这么做的：它为匿名位置赋予诸如 $\olbl{__result}$ 和 $\olbl{__call}$ 的合成标签，使每个节点都获得唯一路径，AST 因而成为一棵命名树。

在此前提下，每种程序是带引用的命名树的语言，都可以由两个有限集来刻画：一个由\emph{节点类型}构成的集合 $\mathcal{N}$，每种节点类型决定一个节点如何组织其子树；以及一个由\emph{引用类型}构成的集合 $\mathcal{R}$，每种引用类型决定树中一个位置如何引用另一个位置。语法由 $\mathcal{N}$ 和 $\mathcal{R}$ 的选择固定；语义则由每种节点类型如何构造组合以及每种引用类型如何解析其目标来决定。

{\small\typewriterInlineCodeInHtml
  \begin{center}
    \begin{tabular}{L{0.2\columnwidth}L{0.36\columnwidth}L{0.32\columnwidth}}
      \bilingual{\textbf{Language}}{\textbf{语言}}
      & \bilingual{\textbf{Node types $\mathcal{N}$}}{\textbf{节点类型 $\mathcal{N}$}}
      & \bilingual{\textbf{Reference types $\mathcal{R}$}}{\textbf{引用类型 $\mathcal{R}$}} \\
      \hline
      Java (OOP)
      & \bilingual{class, interface, method,
      block, conditional, loop,
      \mintinline{java}{throw}, lambda, \ldots}{类、接口、方法、块、条件、循环、\mintinline{java}{throw}、lambda、\ldots}
      & \bilingual{variable, field access, method call,
      \mintinline{java}{extends}/\mintinline{java}{implements},
      \mintinline{java}{import}, \ldots}{变量、字段访问、方法调用、\mintinline{java}{extends}/\allowbreak\mintinline{java}{implements}、\mintinline{java}{import}、\ldots} \\[3pt]
      Haskell (FP)
      & \bilingual{$\lambda$, application, \mintinline{haskell}{let},
      \mintinline{haskell}{case}, \mintinline{haskell}{data},
      type class, instance, \ldots}{$\lambda$、应用、\mintinline{haskell}{let}、\mintinline{haskell}{case}、\mintinline{haskell}{data}、类型类、实例、\ldots}
      & \bilingual{variable, constructor,
      type class dispatch,
      \mintinline{haskell}{import}, \ldots}{变量、构造子、类型类分派、\mintinline{haskell}{import}、\ldots} \\[3pt]
      Scala\footnote{
        \bilingual{Scala is a multi-paradigm language.}{Scala 是一门多范式语言。}
      }
      & \bilingual{class, trait, object, method,
      block, conditional, \mintinline{scala}{match},
      \mintinline{scala}{while}, \mintinline{scala}{throw}, \ldots}{类、trait、对象、方法、块、条件、\mintinline{scala}{match}、\mintinline{scala}{while}、\mintinline{scala}{throw}、\ldots}
      & \bilingual{variable, member access,
      \mintinline{scala}{extends}, \mintinline{scala}{with},
      implicits, \mintinline{scala}{import}, \ldots}{变量、成员访问、\mintinline{scala}{extends}、\mintinline{scala}{with}、隐式参数、\mintinline{scala}{import}、\ldots} \\
      \hline
      $\lambda$-\bilingual{calculus}{演算}
      & \bilingual{abstraction, application}{抽象、应用}
      & \bilingual{variable}{变量} \\
      \bilingual{Inheritance-calculus}{继承演算}
      & \bilingual{mixin, definition}{mixin、定义}
      & \bilingual{reference}{引用} \\
    \end{tabular}
  \end{center}
}

\paragraph{\bilingual{Substitution in a context}{上下文中的替换}}
可替换性关乎继承:继承者保留它所继承的东西的每个观测(定理~\ref{thm:inheritance-substitutability})。在一个上下文中以 $T$ 的一个可替换者 $S$ 代替 $T$ 是另一种操作,而晚绑定把两者分开。在文档 \mintinline{yaml}{m: [t: [x: []], s: [[m, t]], u: [[t]]]} 与 \mintinline{yaml}{k: [[m], t: [z: []]]} 中,\mintinline{yaml}{s} 继承 \mintinline{yaml}{t},因而对它可替换;然而把 \mintinline{yaml}{u} 的引用 \mintinline{yaml}{[t]} 换成 \mintinline{yaml}{[s]} 会失去路径 $(\olbl{k}, \olbl{u}, \olbl{z})$:\mintinline{yaml}{u} 经一个词法引用到达 \mintinline{yaml}{t},晚绑定把它解析到 \mintinline{yaml}{k} 对 \mintinline{yaml}{t} 的覆写,而 \mintinline{yaml}{s} 以绝对路径命名 $(\olbl{m}, \olbl{t})$。当上下文经绝对引用到达 $T$ 时,以 $S$ 替换 $T$ 不丢失任何观测(推论~\ref{cor:behavior-compatibility})。

我们现在对设计空间进行分类，以解释为何我们选择第~\ref{sec:mixin-trees}~节的语义。两个维度保持自由：\emph{组合机制}，决定单一引用类型如何合并定义；以及\emph{引用解析}，决定单一引用类型如何找到其目标。

\subsection*{\bilingual{The Composition Mechanism}{组合机制}}

每种组合机制决定了一个抽象的定义如何被并入另一个抽象。我们把两个抽象 $x$ 与~$y$ 的组合写作 $x \oplus y$。各机制的区别在于\emph{同名相遇}，即 $x$ 与 $y$ 都定义同一个标签之处；我们按此类相遇处有几个贡献得以保留，对相遇的策略进行分类：\emph{深度合并}保留全部贡献，\emph{冲突拒绝}一个不留并报告错误，\emph{优先级合并}保留一个。一门语言逐构造施用策略，更一般地，按相遇定义上的条件施用策略，例如构造的种类、一个注解、一个类型判断、某成员是否抽象，或两个类型是否共享构造子；因此同一门语言可能在不同之处取不同的策略。浅层合并不是第四种策略：它把嵌套结构当作叶子，对其施用优先级合并或冲突拒绝。另有两种机制并入定义时不经过同名相遇：$\beta$-归约与条件并入。

{\small
\begin{center}
  \begin{tabular}{llL{0.17\textwidth}L{0.3\textwidth}}
    \bilingual{\textbf{Mechanism}}{\textbf{机制}}
    & \bilingual{\textbf{Kept at a meeting}}{\textbf{相遇时保留}}
    & \bilingual{\textbf{Extensibility}}{\textbf{可扩展性}}
    & \bilingual{\textbf{Representative}}{\textbf{代表}} \\
    \hline
    \bilingual{Deep merge}{深度合并}
    & \bilingual{all}{全部}
    & \bilingual{universal}{通用} & \bilingual{this paper}{本文} \\
    \bilingual{Conflict rejection}{冲突拒绝}
    & \bilingual{none}{无}
    & \bilingual{none on overlap}{重叠时无扩展}
    & Harper--Pierce \\
    \bilingual{Priority merge}{优先级合并}
    & \bilingual{one}{一个}
    & \bilingual{universal}{通用}
    & JS spread \\
    \hline
    $\beta$-\bilingual{reduction}{归约}
    & ---
    & \bilingual{opt-in}{按需选择} & Haskell \\
    \bilingual{Conditional incorporation}{条件并入}
    & ---
    & \bilingual{causes divergence}{导致发散} & --- \\
  \end{tabular}
\end{center}
}

\begin{description}
  \item[\bilingual{Deep merge}{深度合并}]
    对同名定义的合并递归进入其子树（$\overrides$，公式~\ref{eq:overrides}）。每个标签都是一个扩展点，无论原作者是否预期了扩展。递归需要一棵可以递归进入的子树。在叶子处，即没有子树的构造处，深度合并需要一个组合各贡献的函数，该函数随标签提供：在 MIXINv2 中，一个 \mintinline{python}{@merge} 定义接收其标签的每一个 patch，而非其中最后一个，文档中的示例把这些 patch 收集为一个集合。在继承演算中，每个定义都把一个标签绑定到一个 mixin，而 mixin 本身是一棵树，因此仅靠递归就足够。

  \item[\bilingual{Conflict rejection}{冲突拒绝}]
    同名定义被作为错误拒绝。当拒绝由重叠触发时，只要 $x$ 定义了某个标签，它就拒绝 $x \oplus x$，因此幂等性失效；并且它阻止了对已有标签的扩展~\cite{harper1991-symmetric-record-concatenation}，而这恰恰是可扩展性所要求的可组合性。当拒绝转而由不匹配触发时，例如两个类型的构造子不同，它永远不会在 $x \oplus x$ 上触发。

  \item[\bilingual{Priority merge}{优先级合并}]
    一个来源具有优先权，因此当 $x$ 与 $y$ 对同一标签的定义不同时，$x \oplus y$ 与 $y \oplus x$ 不同，引入了顺序依赖：该策略读取一个贡献来自哪个操作数，交换性因而失效。JS spread~\cite{ecma2018-ecmascript} 也对嵌套对象施用它，把嵌套对象整体替换。

  \item[$\beta$-\bilingual{reduction}{归约}]
    只有作者放置的参数才是扩展点；函数体是封闭的。要使一个新的方面可扩展，原函数必须被重写以接受一个额外参数。任何基于函数应用的组合机制都继承了这一限制，因为 $\beta$-归约使函数体封闭。可扩展性在此类语言中是按需选择的，原因相同。对象代数~\cite{oliveira2012-object-algebras}与 finally tagless 解释器~\cite{carette2009-finally-tagless}在两个维度上绕过了这一限制，即新的数据变体与新的操作，这两个维度由它们的编码预先设计，而在这两个维度上各扩展可以组合；设计维度之外的方面仍是封闭的。

  \item[\bilingual{Conditional incorporation}{条件并入}]
    根据其他定义的存在与否来移除或有条件地包含定义，会产生依赖循环：某定义是否存在，取决于另一定义是否存在，而后者又取决于前者。当添加和移除交替出现时，递归求值（附录~\ref{app:well-definedness}）无法到达不动点，从而导致发散。这与图灵完备性固有的发散不同：图灵完备程序发散是因为计算无界，而条件并入发散是因为直接后果算子 $T_P$ 振荡而非单调递增。
\end{description}

\noindent
因此，一门语言为每个构造或条件确定一种策略，凡取优先级合并、或取由重叠触发的冲突拒绝之处，都是交换性或幂等性可能失效之处。继承演算在每个同名相遇处都取深度合并。在上述机制中，深度合并是我们找到的唯一一个可交换、幂等、可结合（第~\ref{sec:modularity}~节）且使继承关系具有可替换性（定理~\ref{thm:inheritance-substitutability}）的机制。其他机制各自破坏可替换性、使路径不可扩展，或使求值不终止。

\subsection*{\bilingual{The Reference Resolution}{引用解析}}

一旦将 mixin 上的深度合并确定为组合机制，五个语义方程中的三个便由语法决定：$\bases$~(\ref{eq:bases}) 读取其继承来源，$\supers$~(\ref{eq:supers}) 取传递闭包，$\overrides$~(\ref{eq:overrides}) 实现合并。唯一剩余的自由度是引用的解析方式：$\resolve$~(\ref{eq:resolve}) 和 $\this$~(\ref{eq:this}) 方程。我们考察三个已知候选方案。

\begin{description}
  \item[\bilingual{Early binding}{早期绑定}]
    我们回顾 Cook~\cite{cook1989-denotational-inheritance} 的结论：一个让自引用保持不变的派生，仍然指向原结构中本应被修改的部分。早期绑定让引用保持这样：引用在定义时即解析到一个固定路径，与 mixin 后来如何被继承无关。Nix 的递归属性集~\cite{dolstra2010-nixos}是这方面的典型：\mintinline{nix}{rec} 中的引用在属性集定义时解析到同一属性集的一个属性，所以在更新运算符 \mintinline{nix}{//} 覆写该属性之后，该引用读到的仍是旧值，嵌套属性集中的引用也是如此。这使得 mixin 无法充当 $\lambda$-演算的调用栈：$\beta$-归约要求被调用方能观察到调用点提供的参数，而那是一次后续组合。使用早期绑定时，引用在此组合发生前就已被冻结。在方程中，$\resolve$~(\ref{eq:resolve}) 调用 $\this$~(\ref{eq:this})，后者遍历继承的树结构；迟绑定正是使 $\this$ 能够观察到后续组合的 mixin 所贡献路径的原因。若没有图灵完备的 mixin，函数必须被重新引入为组合机制，从而再次暴露上述按需选择的可扩展性问题。

  \item[\bilingual{Dynamic scope}{动态作用域}]
    动态作用域在使用点而非定义点处解析引用。我们回顾 Steele 与 Sussman 的结论~\cite{steele1978-art-of-the-interpreter}：动态作用域使自由变量引用“指向调用链中更上层的其他过程里的变量绑定”，并且它“不适合构造过程抽象”。OmegaConf~\cite{yadan2018-omegaconf} 是 Hydra~\cite{yadan2023-hydra} 的配置库，它的插值采用此方式：一个插值，无论是像 \texttt{\$\{outer\}} 这样的绝对插值，还是像 \texttt{\$\{..outer\}} 这样的相对插值，都相对于其节点被使用时所在的配置解析，所以放到另一个根下的节点读到的是那个根的字段，而非定义该节点的配置的字段。动态作用域破坏了第~\ref{sec:forward-translation}~节的 $\lambda$-演算嵌入。查找引理（引理~\ref{lem:lookup}）依赖于绑定器的分离：变量 $y \neq x$ 的 de~Bruijn 索引从定义点路径向上爬升到绑定 $y$ 的作用域，一个与 $x$ 的绑定器不同的作用域层级，因此在 $x$ 的 redex 处以 \mintinline[escapeinside=||]{yaml}{__parameter: |$\translationOfV$|} 覆写对 $y$ 的解析不产生任何效果。动态作用域下，这种分离不再成立。
    与早期绑定一样，mixin 无法完全替代函数。

  \item[\bilingual{Single-target $\this$}{单目标 $\this$}]
    在 Cook 的语义中，一个晚绑定的引用只有一个目标，因为对象是一个生成器的不动点~\cite{cook1989-denotational-inheritance}。单目标 $\this$ 是一个有前途的候选方案：单路径引理（引理~\ref{lem:single-path}）表明，对于翻译后的 $\lambda$-项，在收敛观测所依赖的每次解析步中，$\this$~(\ref{eq:this}) 中的前沿集 $S$ 至多包含一条路径，因此 $\lambda$-演算嵌入不需要多目标解析。然而，单目标 $\this$ 因另一原因而达不到这一设计目标。当多条继承路径到达同一作用域时，单目标解析必须拒绝该情况或选择其中一条路径。挑选必须读取某个不对称信息：读引用顺序，即呈现为交换性破坏；读分组结构，即呈现为结合性破坏。拒绝则转而拒斥该组合。NixOS 模块系统~\cite{nixos-contributors-nixos-modules}会引发重复声明错误（附录~\ref{app:scala-multi-target}）。这拒绝了两个独立编写的模块定义同一选项的嵌套扩展，而第~\ref{sec:case-study}~节展示了这种组合。

    我们对 MIXINv2 的前三个实现使用了 $\this$ 解析的标准技术：遵循 Cook~\cite{cook1989-denotational-inheritance}的基于闭包的不动点，以及使用 de~Bruijn 指标~\cite{debruijn1972-nameless-dummies}的基于栈的环境查找。当继承引入到同一封闭作用域的多条路径时，三种实现都产生了隐晦的错误。困难在于结构上。闭包捕获单一环境，因此不动点组合子只能为单一 $\this$ 目标求解。基于栈的环境是一条线性链，每个作用域只有一个外围作用域。这两种数据结构都内嵌了一个与多重继承自然产生的多目标情况不相容的单值假设（附录~\ref{app:scala-multi-target}）。公式~(\ref{eq:this})是在放弃这一假设后涌现出来的：它跟踪一个\emph{前沿集} $S$，而非单一当前作用域。
\end{description}

\noindent
下表总结了两个层面上所有替代方案的脆弱性。

{\small\typewriterInlineCodeInHtml
\begin{center}
  \begin{tabular}{llL{0.3\textwidth}L{0.3\textwidth}}
    & \bilingual{\textbf{Candidate}}{\textbf{候选方案}} & \bilingual{\textbf{Vulnerability}}{\textbf{脆弱性}}
    & \bilingual{\textbf{Representative}}{\textbf{代表}} \\
    \hline
    \multirow{4}{*}{\rotatebox[origin=c]{90}{\scriptsize \bilingual{composition}{组合}}}
    & \bilingual{Conflict rejection}{冲突拒绝}
    & \bilingual{self-composition and same-label extensions rejected}{自组合与同名扩展被拒绝}
    & Harper--Pierce \\
    & \bilingual{Priority merge}{优先级合并}
    & \bilingual{ordering dependency; no order may exist}{顺序依赖;可能不存在顺序}
    & Flavors \\
    & $\beta$-\bilingual{reduction}{归约}
    & \bilingual{opt-in extensibility}{按需选择的可扩展性}
    & Haskell \\
    & \bilingual{Conditional incorporation}{条件并入}
    & \bilingual{causes divergence}{导致发散}
    & --- \\
    \hline
    \multirow{3}{*}{\rotatebox[origin=c]{90}{\scriptsize \bilingual{resolution}{解析}}}
    & \bilingual{Early binding}{早期绑定}
    & \bilingual{mixin cannot replace function}{mixin 无法替代函数}
    & Nix~\mintinline{nix}{rec} \\
    & \bilingual{Dynamic scope}{动态作用域}
    & \bilingual{mixin cannot replace function}{mixin 无法替代函数}
    & OmegaConf \\
    & \bilingual{Single-target $\this$}{单目标 $\this$}
    & \bilingual{nested extensions rejected}{嵌套扩展被拒绝}
    & NixOS modules \\
  \end{tabular}
\end{center}
}

\noindent
我们在两个层面考察的每个替代方案，都达不到本节所立的设计目标，失去某条半格律、可替换性，或由构造保证的可扩展性；早期绑定与动态作用域失去最后这一条，因为它们让函数代替 mixin 留在原处，而函数体只在其参数处可扩展。

我们通过尝试本节中的每个替代方案并发现其失败，从而得到了这一语义：优先级合并，无论施于嵌套结构还是叶子，都引入了顺序依赖；早期绑定在组合能够提供引用目标之前就将引用冻结；动态作用域破坏了 $\lambda$-演算嵌入所依赖的替换；单目标 $\this$ 拒绝了可扩展性所要求的多目标组合。深度合并加上多目标迟绑定 $\this$，正是最终剩下的选择。
\else
Section~\ref{sec:modularity} reads Meyer's five criteria as properties of composition and proves them
for inheritance-calculus. We adopt meeting all five by construction as the design goal, and ask whether
an alternative design could meet it by surveying every candidate known to us.

\paragraph{\bilingual{Degrees of freedom}{自由度}}
The analysis below assumes that the AST is a \emph{named tree}:
every edge carries a label, so that every
node is uniquely identified by the path from the root.
For configuration languages this is naturally the case: every
key in a JSON, YAML, or Nix attribute set is a label, and nesting
produces paths.%
\footnote{
  Nix is a functional language with first-class functions;
  its attribute sets alone do not form a named-tree DSL.
  The NixOS module system~\cite{nixos-contributors-nixos-modules}, however,
  restricts the user-facing interface to nested attribute sets
  composed by recursive merging, and it is this DSL-level
  structure that naturally forms a named tree.
}
For mainstream programming languages, such as
Java~\cite{gosling2015-java-se8-specification},
Haskell~\cite{marlow2010-haskell-report}, and
Scala~\cite{odersky2004-scala-overview}, the AST is not directly a
named tree, since anonymous expression positions and implicit
scopes lack labels.%
\footnote{
  For example, the body of a $\mathbf{let}$.
}
Such positions can nonetheless be given synthetic fresh labels.%
\footnote{
  Just as a compiler synthesizes a name for a Java anonymous class.
}
The forward translation of Section~\ref{sec:forward-translation}
does exactly this, giving anonymous positions synthetic labels
such as $\olbl{__result}$ and $\olbl{__call}$, so that every node
acquires a unique path and the AST becomes a named tree.

Under this premise, every language whose programs are such
named trees with references can be characterized by two finite sets:
a set~$\mathcal{N}$ of \emph{node types}, each determining how a
node organizes its subtree, and a set~$\mathcal{R}$ of
\emph{reference types}, each determining how one position in the
tree refers to another.
The syntax is fixed by the choice of $\mathcal{N}$ and
$\mathcal{R}$; the semantics is determined by how each node type
structures composition and how each reference type resolves its
target.

{\small\typewriterInlineCodeInHtml
  \begin{center}
    \begin{tabular}{L{0.2\columnwidth}L{0.36\columnwidth}L{0.32\columnwidth}}
      \textbf{Language}
      & \textbf{Node types $\mathcal{N}$}
      & \textbf{Reference types $\mathcal{R}$} \\
      \hline
      Java (OOP)
      & class, interface, method,
      block, conditional, loop,
      \mintinline{java}{throw}, lambda, \ldots
      & variable, field access, method call,
      \mintinline{java}{extends}/\mintinline{java}{implements},
      \mintinline{java}{import}, \ldots \\[3pt]
      Haskell (FP)
      & $\lambda$, application, \mintinline{haskell}{let},
      \mintinline{haskell}{case}, \mintinline{haskell}{data},
      type class, instance, \ldots
      & variable, constructor,
      type class dispatch,
      \mintinline{haskell}{import}, \ldots \\[3pt]
      Scala\footnote{
        Scala is a multi-paradigm language.
      }
      & class, trait, object, method,
      block, conditional, \mintinline{scala}{match},
      \mintinline{scala}{while}, \mintinline{scala}{throw}, \ldots
      & variable, member access,
      \mintinline{scala}{extends}, \mintinline{scala}{with},
      implicits, \mintinline{scala}{import}, \ldots \\
      \hline
      $\lambda$-calculus
      & abstraction, application
      & variable \\
      Inheritance-calculus
      & mixin, definition
      & reference \\
    \end{tabular}
  \end{center}
}

\paragraph{\bilingual{Substitution in a context}{上下文中的替换}}
Substitutability concerns inheritance: what inherits keeps every observation of what it inherits
(Theorem~\ref{thm:inheritance-substitutability}).
Replacing $T$ by a substitutable $S$ in a context is a different operation, and late binding separates
the two.
In the document \mintinline{yaml}{m: [t: [x: []], s: [[m, t]], u: [[t]]]} with
\mintinline{yaml}{k: [[m], t: [z: []]]}, the mixin \mintinline{yaml}{s} inherits
\mintinline{yaml}{t} and is substitutable for it, yet replacing the reference
\mintinline{yaml}{[t]} of \mintinline{yaml}{u} by \mintinline{yaml}{[s]} loses the path
$(\olbl{k}, \olbl{u}, \olbl{z})$: \mintinline{yaml}{u} reaches \mintinline{yaml}{t} through a
lexical reference, which late binding resolves to the override of \mintinline{yaml}{t} in
\mintinline{yaml}{k}, while \mintinline{yaml}{s} names $(\olbl{m}, \olbl{t})$ absolutely.
Replacing $T$ by $S$ loses no observation where the context reaches $T$ through absolute references
(Corollary~\ref{cor:behavior-compatibility}).

We now classify the design space to explain why we chose the
semantics of Section~\ref{sec:mixin-trees}.
Two axes remain free:
the \emph{composition mechanism}, which determines how the single reference
type incorporates definitions;
and the \emph{reference resolution}, which determines how the single reference
type finds its target.

\subsection*{\bilingual{The Composition Mechanism}{组合机制}}

Every composition mechanism determines how the definitions of one
abstraction are incorporated into another.
We write $x \oplus y$ for the composition of two abstractions $x$
and~$y$.
Mechanisms differ at a \emph{same-label meeting}, where $x$ and $y$
both define one label, and we classify the strategies for such a
meeting by how many of its contributions survive:
\emph{deep merge} keeps all of them,
\emph{conflict rejection} keeps none and reports an error,
and \emph{priority merge} keeps one.
A language applies a strategy per construct, or more generally per
condition on the meeting definitions, such as the kind of construct,
an annotation, a typing judgment, whether a member is abstract, or
whether two types share a constructor; one language may therefore
take different strategies in different places.
A shallow merge is not a fourth strategy: it treats a nested
structure as a leaf and applies priority merge or conflict rejection
to it.
Two further mechanisms incorporate definitions without a same-label
meeting: $\beta$-reduction and conditional incorporation.

{\small
\begin{center}
  \begin{tabular}{llL{0.17\textwidth}L{0.3\textwidth}}
    \bilingual{\textbf{Mechanism}}{\textbf{机制}}
    & \bilingual{\textbf{Kept at a meeting}}{\textbf{相遇时保留}}
    & \bilingual{\textbf{Extensibility}}{\textbf{可扩展性}}
    & \bilingual{\textbf{Representative}}{\textbf{代表}} \\
    \hline
    \bilingual{Deep merge}{深度合并}
    & \bilingual{all}{全部}
    & \bilingual{universal}{通用} & \bilingual{this paper}{本文} \\
    \bilingual{Conflict rejection}{冲突拒绝}
    & \bilingual{none}{无}
    & \bilingual{none on overlap}{重叠时无扩展}
    & Harper--Pierce \\
    \bilingual{Priority merge}{优先级合并}
    & \bilingual{one}{一个}
    & \bilingual{universal}{通用}
    & JS spread \\
    \hline
    $\beta$-\bilingual{reduction}{归约}
    & ---
    & \bilingual{opt-in}{按需选择} & Haskell \\
    \bilingual{Conditional incorporation}{条件并入}
    & ---
    & \bilingual{causes divergence}{导致发散} & --- \\
  \end{tabular}
\end{center}
}

\begin{description}
  \item[Deep merge]
    Merging same-label definitions recurses into their subtrees
    ($\overrides$, equation~\ref{eq:overrides}).
    Every label is an extension point, whether or not the
    original author anticipated extension.
    Recursion needs a subtree to recurse into.
    At a leaf, a construct without one, deep merge needs a function
    that combines the contributions, supplied with the label:
    in MIXINv2, a \mintinline{python}{@merge} definition receives every
    patch to its label rather than the last of them, and the
    documented example collects the patches into a set.
    In inheritance-calculus every definition binds a label to a
    mixin, which is itself a tree, so the recursion alone suffices.

  \item[Conflict rejection]
    Same-label definitions are rejected as errors.
    When the rejection is triggered by an overlap, it rejects
    $x \oplus x$ whenever $x$ defines a label, so idempotence fails;
    and it prevents extending existing
    labels~\cite{harper1991-symmetric-record-concatenation}, precisely
    the composability that extensibility requires.
    When the rejection is triggered by a mismatch instead, such as two
    types with different constructors, it never fires on
    $x \oplus x$.

  \item[Priority merge]
    One source takes precedence,
    so $x \oplus y$ and $y \oplus x$ differ when $x$ and $y$ define
    one label differently, introducing an ordering dependency:
    the strategy reads which operand a contribution comes from, and
    commutativity fails.
    JS spread~\cite{ecma2018-ecmascript} applies it to a nested
    object as well, replacing the nested object whole.

  \item[$\beta$-reduction]
    Only the parameters placed by the author are extension
    points; the function body is closed.
    To make a new aspect extensible, the original function must
    be rewritten to accept an additional parameter.
    Any composition mechanism based on function application
    inherits this limitation, because $\beta$-reduction
    closes the function body.
    Extensibility in such languages is opt-in for the same reason.
    Object algebras~\cite{oliveira2012-object-algebras} and finally
    tagless interpreters~\cite{carette2009-finally-tagless} work around
    this limitation in two dimensions, new data variants and new
    operations, which their encodings design in advance, and in those
    dimensions the extensions can compose; an aspect outside the
    designed dimensions stays closed.

  \item[Conditional incorporation]
    Removing or conditionally including definitions
    based on the presence of other definitions
    creates a dependency cycle:
    whether a definition exists depends on
    whether another definition exists, which in turn
    depends on the first.
    The recursive evaluation
    (Appendix~\ref{app:well-definedness}) cannot
    reach a fixed point when addition and removal
    alternate, causing divergence.
    This differs from the divergence inherent to
    Turing completeness: Turing-complete programs
    diverge because computation is unbounded, whereas
    conditional incorporation diverges because the
    immediate consequence operator $T_P$ oscillates
    rather than ascending monotonically.
\end{description}

\noindent
A language thus fixes a strategy for each construct or condition,
and each place where it takes priority merge, or a conflict
rejection triggered by an overlap, is a place where commutativity
or idempotence can fail.
Inheritance-calculus takes deep merge at every same-label meeting.
Among the mechanisms above, deep merge is the only one we found
that is commutative, idempotent, associative
(Section~\ref{sec:modularity}), and
gives a substitutable inheritance relation
(Theorem~\ref{thm:inheritance-substitutability}).
The other mechanisms each break substitutability, leave paths
nonextensible, or leave evaluation without termination.

\subsection*{\bilingual{The Reference Resolution}{引用解析}}

Once deep merge over mixins is fixed as the composition
mechanism, three of the five semantic equations are determined
by the syntax:
$\bases$~(\ref{eq:bases}) reads its inheritance
sources, $\supers$~(\ref{eq:supers}) takes the transitive
closure, and $\overrides$~(\ref{eq:overrides}) implements
the merge.
The only remaining degree of freedom is how references
resolve: the $\resolve$~(\ref{eq:resolve}) and
$\this$~(\ref{eq:this}) equations.
We examine the three known candidates.

\begin{description}
  \item[Early binding]
    We recall from Cook~\cite{cook1989-denotational-inheritance} that a
    derivation which leaves self-reference unchanged still refers to the
    parts of the original structure that should have been modified.
    Early binding leaves references this way: they resolve at
    definition time to a fixed path, independently of how the mixin is
    later inherited.
    Nix's recursive attribute sets~\cite{dolstra2010-nixos} exemplify
    this: a reference inside \mintinline{nix}{rec} resolves to an
    attribute of the same set when the set is defined, so after the
    update operator \mintinline{nix}{//} overrides that attribute, the
    reference still reads the old value, and so does a reference in a
    nested set.
    This prevents the mixin from serving as a call stack for the
    $\lambda$-calculus: $\beta$-reduction requires the callee to
    observe the argument supplied at the call site, which is a
    later composition.
    With early binding, references are frozen before this
    composition occurs.
    In the equations,
    $\resolve$~(\ref{eq:resolve}) calls
    $\this$~(\ref{eq:this}), which walks the inherited
    tree structure; late binding is what allows
    $\this$ to see paths contributed by later-composed
    mixins.
    Without a Turing-complete mixin, functions must be
    reintroduced as the composition mechanism, re-exposing
    the opt-in extensibility problem above.

  \item[Dynamic scope]
    Dynamic scope resolves references at the point of use rather than
    the definition site.
    We recall from Steele and
    Sussman~\cite{steele1978-art-of-the-interpreter} that dynamic scope
    makes free variable references ``refer to the bindings of variables
    in other procedures higher up in the chain of calls'', and that it
    is ``unsuitable for constructing procedural abstractions''.
    The interpolations of OmegaConf~\cite{yadan2018-omegaconf}, the
    configuration library of Hydra~\cite{yadan2023-hydra}, follow
    this approach: an interpolation, absolute like
    \texttt{\$\{outer\}} or relative like \texttt{\$\{..outer\}},
    resolves against the configuration in which its node is used, so
    a node placed under another root reads the field of that root,
    not the field of the configuration that defined the node.
    Dynamic scope breaks the $\lambda$-calculus embedding
    of Section~\ref{sec:forward-translation}.
    The Lookup Lemma (Lemma~\ref{lem:lookup})
    relies on the separation of binders:
    the de~Bruijn index of a variable $y \neq x$ climbs from the
    definition-site path to the scope binding~$y$, a scope
    level different from~$x$'s binder, so overriding
    \mintinline[escapeinside=||]{yaml}{__parameter: |$\translationOfV$|} at~$x$'s redex
    has no effect on the resolution of~$y$.
    With dynamic scope this separation fails.
    As with early binding, the mixin cannot serve as a complete
    replacement for functions.

  \item[Single-target $\this$]
    \begin{sloppypar}
    In Cook's semantics, a late-bound reference has one target, since an object is the fixed point of
    one generator~\cite{cook1989-denotational-inheritance}.
    Single-target $\this$ is a promising candidate:
    the Single-Path Lemma (Lemma~\ref{lem:single-path}) shows that
    for translated $\lambda$-terms, at every resolution step the
    convergence observation depends on, the frontier set~$S$ in
    $\this$~(\ref{eq:this}) contains at most one path,
    so the $\lambda$-calculus embedding does not require
    multi-target resolution.
    However, single-target $\this$ misses the design goal for a
    different reason.
    When multiple inheritance routes reach the same scope,
    single-target resolution must reject the situation or select
    one route.
    A selection must read some asymmetry: reading reference order
    surfaces as a failure of commutativity, and reading the grouping
    structure surfaces as a failure of associativity.
    A rejection refuses the composition instead.
    The NixOS module system~\cite{nixos-contributors-nixos-modules} raises a
    duplicate-declaration error
    (Appendix~\ref{app:scala-multi-target}).
    This rejects nested extensions where two independently authored
    modules define the same option, a composition that
    Section~\ref{sec:case-study} demonstrates.
    \end{sloppypar}

    Our first three implementations of MIXINv2 used
    standard techniques for $\this$ resolution: closure-based
    fixed points following Cook~\cite{cook1989-denotational-inheritance}, and
    stack-based environment lookup using de~Bruijn
    indices~\cite{debruijn1972-nameless-dummies}.
    All three produced subtle bugs when inheritance introduced
    multiple routes to the same enclosing scope.
    The difficulty was structural.
    A closure captures a single environment, so a fixed-point
    combinator solves for a single $\this$ target.
    A stack-based environment is a linear chain in which each scope
    has a single enclosing scope.
    Both data structures embed a single-valued assumption
    incompatible with the multi-target situation that multiple
    inheritance naturally produces
    (Appendix~\ref{app:scala-multi-target}).
    Equation~(\ref{eq:this}) emerged from abandoning this
    assumption: it tracks a \emph{frontier set}~$S$ rather than a
    single current scope.
\end{description}

\noindent
The following table summarizes the vulnerabilities of all
alternatives at both levels.

{\small\typewriterInlineCodeInHtml
\begin{center}
  \begin{tabular}{llL{0.3\textwidth}L{0.3\textwidth}}
    & \bilingual{\textbf{Candidate}}{\textbf{候选方案}} & \bilingual{\textbf{Vulnerability}}{\textbf{脆弱性}}
    & \bilingual{\textbf{Representative}}{\textbf{代表}} \\
    \hline
    \multirow{4}{*}{\rotatebox[origin=c]{90}{\scriptsize \bilingual{composition}{组合}}}
    & \bilingual{Conflict rejection}{冲突拒绝}
    & \bilingual{self-composition and same-label extensions rejected}{自组合与同名扩展被拒绝}
    & Harper--Pierce \\
    & \bilingual{Priority merge}{优先级合并}
    & \bilingual{ordering dependency; no order may exist}{顺序依赖;可能不存在顺序}
    & Flavors \\
    & $\beta$-\bilingual{reduction}{归约}
    & \bilingual{opt-in extensibility}{按需选择的可扩展性}
    & Haskell \\
    & \bilingual{Conditional incorporation}{条件并入}
    & \bilingual{causes divergence}{导致发散}
    & --- \\
    \hline
    \multirow{3}{*}{\rotatebox[origin=c]{90}{\scriptsize \bilingual{resolution}{解析}}}
    & \bilingual{Early binding}{早期绑定}
    & \bilingual{mixin cannot replace function}{mixin 无法替代函数}
    & Nix~\mintinline{nix}{rec} \\
    & \bilingual{Dynamic scope}{动态作用域}
    & \bilingual{mixin cannot replace function}{mixin 无法替代函数}
    & OmegaConf \\
    & \bilingual{Single-target $\this$}{单目标 $\this$}
    & \bilingual{nested extensions rejected}{嵌套扩展被拒绝}
    & NixOS modules \\
  \end{tabular}
\end{center}
}

\noindent
Each alternative we examined at either level misses the design goal
this section set out, losing a semilattice law, substitutability, or
extensibility by construction, which early binding and dynamic scope
lose by leaving functions, whose bodies are extensible only at their
parameters, in place of mixins.

We arrived at this semantics by trying each alternative in this section
and finding that it failed:
priority merge, whether at nested structures or at leaves,
introduced ordering dependencies;
early binding froze references before composition could
supply them;
dynamic scope broke the substitution that the
$\lambda$-calculus embedding relies on;
single-target $\this$ rejected the multi-target compositions
that extensibility demands.
Deep merge with multi-target late-binding $\this$ is what
remained.
\fi

\section{\bilingual{Further Related Work}{进一步的相关工作}}
\label{sec:related-work}

\ifPaperChinese
第~\ref{sec:prior-modularity}~节逐条标准比较先前设计,第~\ref{sec:semantic-variants}~节比较组合与引用解析的替代设计;本节讨论其余的工作。
\else
Section~\ref{sec:prior-modularity} compares prior designs criterion by criterion, and
Section~\ref{sec:semantic-variants} compares alternative designs of composition and reference
resolution; this section discusses the remaining works.
\fi

\subsection{\bilingual{Computational Models and \texorpdfstring{$\lambda$}{λ}-Calculus Semantics}{计算模型与 \texorpdfstring{$\lambda$}{λ}-演算语义}}

\paragraph{\bilingual{Minimal Turing-complete models}{最小图灵完备模型}}
\ifPaperChinese
Sch\"onfinkel~\cite{schonfinkel1924-combinatory-logic} 与 Curry~\cite{curry1958-combinatory-logic} 把 $\lambda$-演算归约为组合子 $S$、$K$ 与 $I$,消除了约束变量,而应用仍是唯一的运算,Meyer~\cite{meyer1982-model-of-lambda-calculus} 证明了满足五个附加条件的组合代数恰好就是 $\lambda$-代数;Dolan~\cite{dolan2013-mov-turing-complete} 证明了 x86 的 \mintinline{nasm}{mov} 指令单独就已图灵完备,而不提供任何超越 RAM 机的抽象机制。这些模型中的每一个,都在某个已有模型的语义框架内实现了图灵完备性。
\else
Sch\"onfinkel~\cite{schonfinkel1924-combinatory-logic} and
Curry~\cite{curry1958-combinatory-logic} reduced the $\lambda$-calculus to the combinators $S$, $K$,
and $I$, eliminating bound variables while application remains the one operation, and
Meyer~\cite{meyer1982-model-of-lambda-calculus} showed that the combinatory algebras satisfying five
additional conditions are exactly the $\lambda$-algebras;
Dolan~\cite{dolan2013-mov-turing-complete} showed that the x86 \mintinline{nasm}{mov} instruction
alone is Turing complete, providing no abstraction mechanisms beyond those of the RAM Machine.
Each of these models achieves Turing completeness within the semantic framework of an existing model.
\fi

\paragraph{\bilingual{Graph and filter models}{图模型与滤子模型}}
\ifPaperChinese
图模型与滤子模型为 $\lambda$-演算给出了载体为幂集的模型:Berline~\cite{berline2000-webbed-models} 综述了这类由 web 构造的 webbed 模型。在 Barendregt、Coppo 与 Dezani 的滤子模型~\cite{barendregt1983-filter-lambda-model} 中,一个项的解释是可指派给它的交类型之集,对滤子 $d_1$、$d_2$ 与类型 $\sigma$、$\tau$,应用是集合概括 $d_1 \cdot d_2 = \{\tau \mid \exists \sigma \in d_2.\ \sigma \to \tau \in d_1\}$,而 $\lambda$-模型的条件被直接验证,不构造函数空间。Dezani-Ciancaglini、Giovannetti 与 de'Liguoro~\cite{dezani1998-intersection-types-bohm-trees} 将函数空间作为载体的一个 retract 恢复出来,并把滤子模型与 Böhm 树联系起来。这些模型所解释的项保留抽象与应用作为原语。
\else
Graph models and filter models give the $\lambda$-calculus models
whose carrier is a powerset: Berline~\cite{berline2000-webbed-models}
surveys these webbed models, built from webs.
In the filter model of Barendregt, Coppo, and
Dezani~\cite{barendregt1983-filter-lambda-model}, the interpretation
of a term is the set of intersection types assignable to it;
for filters $d_1$, $d_2$ and types $\sigma$, $\tau$, application is the
set comprehension
$d_1 \cdot d_2 = \{\tau \mid \exists \sigma \in d_2.\ \sigma \to \tau \in d_1\}$,
and the conditions of a $\lambda$-model are verified directly,
without constructing a function space.
Dezani-Ciancaglini, Giovannetti, and
de'Liguoro~\cite{dezani1998-intersection-types-bohm-trees} recover the
function space as a retract of the carrier and relate filter models
to B\"ohm trees.
The terms these models interpret keep abstraction and application as
primitives.
\fi

\subsection{\bilingual{Rule Languages with Objects}{带对象的规则语言}}

\paragraph{\bilingual{Objects in deductive languages}{演绎语言中的对象}}
\ifPaperChinese
Kifer、Lausen 与 Wu 的 F-logic~\cite{kifer1995-frame-logic} 把类、属性与方法作为 id-项放在与个体对象相同的论域中,而一个方法的行为由演绎规则定义,一个模块是逻辑之外的元注解。从多个超类继承来的签名相加,一个方法的类型是它们的类型之交;一个显式定义覆写一个继承来的值,而当两个超类提供冲突的值时,结果取决于先进行哪一步继承,且语义接受多个典范模型。Abiteboul 与 Kanellakis~\cite{abiteboul1989-object-identity-query-primitive, abiteboul1998-object-identity-jacm} 在 IQL 中以对象标识扩展了基于规则的查询语言;IQL 是带否定的膨胀式 Datalog,具有集合与元组类型、对象标识符的生成与弱赋值,在复制消除的意义下完备,并把方法、封装与方法继承留在其模型之外;一个对象标识符的值只在其未定义时被赋值一次,因此 IQL 程序靠新增规则而扩展,而非靠在已有对象上新增成员。Datalog 引擎 Souffl\'e 的组件~\cite{souffle-contributors-components} 把关系、规则、事实与嵌套组件组织在一起;一个子组件接收其超组件的每一个元素,因此一个关系的事实与规则累积起来,而 \texttt{.override} 丢弃超组件声明为可覆写的关系的继承来的事实与规则。一个组件在编译时被展平为带前缀的关系的命名空间。在这三种语言中,规则取代了方法体,作为它们所描述的对象或组件之外的一种形式体系,而关系是原语。按第~\ref{sec:semantic-variants}~节的分类,F-logic 对继承的值、Soufflé 对可覆写的关系取优先级合并,IQL 加入否定;Soufflé 在实例化处才解析一个组件未声明的关系,这是动态作用域。
\else
Kifer, Lausen, and Wu's F-logic~\cite{kifer1995-frame-logic} treats
classes, attributes, and methods as id-terms in the same domain as
individual objects, while the behaviour of a method is defined by
deductive rules and a module is a meta-annotation outside the logic.
Signatures inherited from several superclasses add up, and a method's
type is the intersection of their types; an explicit definition
overrides an inherited value, and when two superclasses supply
conflicting values, the result depends on which inheritance step is
taken first, and the semantics accepts several canonical models.
Abiteboul and
Kanellakis~\cite{abiteboul1989-object-identity-query-primitive, abiteboul1998-object-identity-jacm}
extend rule-based query languages with object identity in IQL, an
inflationary Datalog with negation, set and tuple types, object
identifier invention, and weak assignment, which is complete up to
copy elimination and leaves methods, encapsulation, and method
inheritance outside its model; the value of an object identifier is
assigned once, while it is undefined, so an IQL program is extended by
new rules, not by new members on existing objects.
The components of the Datalog engine
Souffl\'e~\cite{souffle-contributors-components} group relations,
rules, facts, and nested components; a sub-component receives every
element of its super-components, so the facts and rules of a relation
accumulate, and \texttt{.override} discards the inherited facts and
rules of a relation that the super-component declared overridable.
A component is flattened at compile time into a namespace of prefixed
relations.
In these three languages, rules take the place of method bodies, as a
formalism apart from the objects or components they describe, and
relations are primitive.
In the taxonomy of Section~\ref{sec:semantic-variants}, F-logic takes priority merge for inherited
values and Souffl\'e for overridable relations, IQL adds negation, and Souffl\'e resolves a relation
that a component does not declare at instantiation, which is dynamic scope.
\fi

\paragraph{\bilingual{Object-oriented attribute grammars}{面向对象属性文法}}
\ifPaperChinese
Hedin 的面向对象属性文法~\cite{hedin1989-object-oriented-attribute-grammars} 在非终结符与产生式的单继承分类之上以属性与方程规约一个语言,其中子类覆写一个方程,并讨论了多重继承会引起的问题。JastAdd~\cite{hedin2003-jastadd} 把这样的方程与普通的 Java 方法结合起来:一个综合属性被翻译为一个抽象 Java 方法,由其方程实现,而一个方程的 self 是 Java 的 \texttt{this}。
\else
Hedin's object-oriented attribute
grammars~\cite{hedin1989-object-oriented-attribute-grammars} specify a
language by attributes and equations over a single-inheritance
classification of nonterminals and productions, in which a subclass
overrides an equation, and discuss the problems that multiple
inheritance would raise.
JastAdd~\cite{hedin2003-jastadd} combines such equations with ordinary
Java methods: a synthesized attribute is translated to an abstract
Java method, which its equations implement, and the self of an
equation is Java's \texttt{this}.
\fi

\paragraph{\bilingual{Feature structures}{特征结构}}
\ifPaperChinese
LOGIN~\cite{ait-kaci1986-login} 把继承内建进逻辑编程:按 Kifer 等人~\cite{kifer1995-frame-logic} 与 Carpenter~\cite{carpenter1992-typed-feature-structures} 的描述,它把定子句的一阶项替换为 $\psi$-项,而 $\psi$-项的合一纳入了结构继承。Carpenter 发展了类型化特征结构,其类型构成一个允许多重继承的继承层次;两个特征结构只在其类型于该层次中相容时才合一,而定子句程序是它们之上的短语结构文法的一个特例。在两者中,这些结构之上的定子句取代了方法。
\else
LOGIN~\cite{ait-kaci1986-login} builds inheritance into logic
programming: as Kifer et al.~\cite{kifer1995-frame-logic} and
Carpenter~\cite{carpenter1992-typed-feature-structures} describe it,
it replaces the first-order terms of definite clauses with
$\psi$-terms, whose unification incorporates structural inheritance.
Carpenter develops typed feature structures, whose types form an
inheritance hierarchy that admits multiple inheritance; two feature
structures unify only when their types are compatible in that
hierarchy, and definite clause programs are a special case of phrase
structure grammars over them.
In both, definite clauses over these structures take the place of
methods.
\fi

\subsection{\bilingual{Inheritance, Objects, and Records}{继承、对象与记录}}

\paragraph{\bilingual{Object and record calculi}{对象演算与记录演算}}
\ifPaperChinese
Boudol~\cite{boudol2004-recursive-record-semantics} 证明了自引用记录需要一个不安全的不动点算子,其适定性依赖于求值策略。Cardelli~\cite{cardelli1992-extensible-records} 研究了带子类型的可扩展记录;Rémy~\cite{remy1989-row-polymorphism} 为 ML 中的记录和变体引入了行多态。这些演算全都在\emph{平坦}记录上操作:没有嵌套结构的递归合并,没有惰性观察,$\this$ 也需要一个显式的不动点组合子。Cardelli 的多重继承语义~\cite{cardelli1988-multiple-inheritance-semantics} 逐字段地对记录类型取 meet,并让构造子不同的类型的 meet 无定义;self 是类型固定的闭合值递归,没有晚绑定,而超类方法的访问“cannot be simulated in any general and direct way”。Wand~\cite{wand1991-record-concatenation-inference} 为一种右偏的记录连接做类型推导,这是优先级合并,self 是一个显式参数,晚绑定到单一目标。Self~\cite{ungar1987-self-language} 把类与对象合为一个 prototype、把字段与方法合为一个 slot,建立在可变状态之上;一个方法是一个除 slot 之外还带有代码的对象。Scala~\cite{odersky2004-scala-overview} 没有静态成员:一个类的共享成员放在其伴生 object 中,而 class、trait、object 仍是三种构造。
\else
\begin{sloppypar}
Boudol~\cite{boudol2004-recursive-record-semantics} showed that self-referential
records require an unsafe fixed-point operator whose
well-definedness depends on the evaluation strategy.
Cardelli~\cite{cardelli1992-extensible-records} studied extensible records
with subtyping;
R\'emy~\cite{remy1989-row-polymorphism} introduced row polymorphism for
records and variants in ML.
All of these calculi operate on \emph{flat} records: there is
no recursive merging of nested structure, no lazy observation,
and $\this$ requires an explicit fixed-point combinator.
Cardelli's semantics of multiple
inheritance~\cite{cardelli1988-multiple-inheritance-semantics} meets
record types field by field and leaves the meet undefined on types
of different constructors; self is closed value recursion with a
fixed type, without late binding, and superclass method access
``cannot be simulated in any general and direct way''.
Wand~\cite{wand1991-record-concatenation-inference} infers types for a
record concatenation that is right-biased, a priority merge, with self an explicit parameter bound
late to a single target.
Self~\cite{ungar1987-self-language} makes a class and an object one
prototype, and a field and a method one slot, over mutable state;
a method is an object that carries code besides its slots.
Scala~\cite{odersky2004-scala-overview} has no static members: a
class's shared members live in a companion object, while class,
trait, and object remain three constructs.
\end{sloppypar}
\fi

\paragraph{\bilingual{Virtual classes, family polymorphism, and DOT}{虚类、族多态与 DOT}}
\ifPaperChinese
BETA~\cite{kristensen1983-beta-abstraction-mechanisms} 把类、类型、过程、函数与泛型包统一在 pattern 之中,而 pattern 与对象仍然分开,pattern 属性与引用属性是两种声明形式;一个 pattern 的动作是它的 do-part,即一串命令式语句。虚 pattern 沿单一的超 pattern 链经进一步绑定而被扩展:其设计者“due to the lack of a profound theoretical understanding”而把多重继承排除在 BETA 之外,并补充说超类的顺序应当无关紧要~\cite[p.~103]{madsen1993-beta-programming-language}。Madsen 与 M{\o}ller-Pedersen~\cite{madsen1989-virtual-classes} 提出了虚类,虚类沿单一的超类链被扩展,没有二元组合算子。在 gbeta~\cite{ernst1999-gbeta, ernst2001-family-polymorphism} 中,Ernst 把一个族的扩展传播进其嵌套虚类,并引入了族多态,而方法按 C3 线性化排序,当两个操作数要求相反的顺序时组合在运行期失败。Ernst、Ostermann 与 Cook~\cite{ernst2006-virtual-class-calculus} 在虚类演算中形式化了族多态,其中类通过路径表达式 \texttt{this.out.C} 来访问,而线性化让最右的操作数胜出,因此组合不可交换。Amin 等人~\cite{amin2016-dependent-object-types} 在 DOT 演算中形式化了路径依赖类型,这是 Scala 的理论基础。在沿每条路线解析外层类这一点上,J\&~\cite{nystrom2006-nested-intersection} 最为接近:它建立在 Jx~\cite{nystrom2004-nested-inheritance} 的前缀类型之上(前缀类型命名一个被继承进新容器的嵌套类的外层类),以嵌套交集组合族,于是一个继承分属多个容器的嵌套类的类,以这些容器的交为其前缀。\footnote{其技术报告对前缀的形式定义~\cite[图~12]{nystrom2006-nested-intersection-report}只在两个容器都再绑定第三个容器的同一嵌套类时才把它们联系起来,所以对写在这些容器的共同超类中的名字,它给出的前缀只是该超类本身,与报告为之声明的性质相反~\cite[附录~A.3]{nystrom2006-nested-intersection-report};其发布的编译器则收集所有是该超类之子类型的容器。}它的组合不是深度合并:同名的嵌套类只有在其名字由共同祖先引入时才取交,而两个方法覆写同一方法的交集在子类覆写它之前是抽象的~\cite[第~3.3~节]{nystrom2006-nested-intersection}。它的实现既不接受循环继承,也不接受作为其容器之子类型的嵌套类,所以单凭嵌套类不能做一般计算;计算留在方法体中,而一个方法覆写另一个方法时执行哪个方法体,其形式语义不作规定~\cite[附录~A]{nystrom2006-nested-intersection-report};在运行时一个引用指向单一对象。在 gbeta、虚类演算与 DOT 中,每条路径都解析为\emph{单一}封闭对象中的\emph{单一}类;它们均未提供集合值解析的机制,而 Scala 在编译时拒绝多目标情形(附录~\ref{app:scala-multi-target})。
\else
BETA~\cite{kristensen1983-beta-abstraction-mechanisms} unifies
classes, types, procedures, functions, and generic packages in the
pattern, while patterns and objects stay apart and pattern and
reference attributes are two forms of declaration; the actions of a
pattern are its do-part, a sequence of imperative statements.
A virtual pattern is extended by further binding along a single
superpattern chain: its designers left multiple inheritance out of
BETA ``due to the lack of a profound theoretical understanding'',
adding that the order of the superclasses should be
insignificant~\cite[p.~103]{madsen1993-beta-programming-language}.
Madsen and M{\o}ller-Pedersen~\cite{madsen1989-virtual-classes}
present virtual classes, which are extended along a single superclass
chain, with no binary composition operator.
In gbeta~\cite{ernst1999-gbeta, ernst2001-family-polymorphism}, Ernst
propagates the extension of a family into its nested virtual classes
and introduces family polymorphism, while methods are ordered by C3
linearization and a composition fails at run time when its two
operands require opposite orders.
Ernst, Ostermann, and Cook~\cite{ernst2006-virtual-class-calculus}
formalized family polymorphism in the virtual class calculus, where
classes are accessed via path expressions \texttt{this.out.C} and the
linearization lets the rightmost operand win, so composition is not
commutative.
Amin et al.~\cite{amin2016-dependent-object-types} formalized path-dependent types
in the DOT calculus, the theoretical foundation of Scala.
On resolving an enclosing class along every route, J\&~\cite{nystrom2006-nested-intersection} comes
closest: building on the prefix types of Jx~\cite{nystrom2004-nested-inheritance}, which name the
enclosing class of a nested class inherited into a new container, it composes families by nested
intersection, so that a class extending classes nested in several containers has the intersection of
those containers as its prefix.\footnote{The formal definition of the prefix in its technical
  report~\cite[Figure~12]{nystrom2006-nested-intersection-report} relates two containers only when both
  further bind one nested class of a third, so for a name written in a common superclass of the
  containers it gives that superclass alone, contrary to the property the report states for
  it~\cite[Appendix~A.3]{nystrom2006-nested-intersection-report}; its released compiler instead
  collects every container that is a subtype of the superclass.}
Its composition is not deep merge: nested classes that share a name are intersected only when a common
ancestor introduced the name, and an intersection in which two methods override one method is abstract
until a subclass overrides it~\cite[\S~3.3]{nystrom2006-nested-intersection}.
Its implementation admits neither cyclic inheritance nor a nested class that is a subtype of its
container, so its nested classes alone cannot compute in general; computation stays in method bodies, whose dispatch its formal semantics leaves unspecified when one
method overrides another~\cite[Appendix~A]{nystrom2006-nested-intersection-report}, and at run time a
reference denotes a single object.
In gbeta, the virtual class calculus, and DOT, each path resolves to a
\emph{single} class in a \emph{single} enclosing object;
none provides set-valued resolution, and Scala rejects
the multi-target situation at compile time
(Appendix~\ref{app:scala-multi-target}).
\fi

\paragraph{\bilingual{Merge operators and nested composition}{合并算子与嵌套组合}}
\ifPaperChinese
Bi、Oliveira 与 Schrijvers~\cite{bi2018-nested-composition} 在 NeColus 中为不相交交类型加入了嵌套组合:合并在类型不相交的记录中递归,否则不良类型,因此“\texttt{1 ,{}, 1} is not expressible”,两个都包含同一个 \texttt{implLang} 的扩展无法组合,而其交换律与结合律对良类型的合并成立。组合式编程~\cite{zhang2021-compositional-programming} 建立在同样的不相交性之上,因此一个提供了任何东西的 trait 无法与自身合并,而它用于消解冲突的算子是有偏的。DeepFJig~\cite{corradi2012-deepfjig} 递归且对称地对嵌套类求和,但两个同名成员中至多一个可以是具体的,因此一个带具体成员的类无法与自身求和。按第~\ref{sec:semantic-variants}~节的分类,三者都在重叠时取冲突拒绝;NeColus 早期绑定 self,而组合式编程与 DeepFJig 把它晚绑定到单一目标。
\else
Bi, Oliveira, and Schrijvers~\cite{bi2018-nested-composition} add
nested composition to disjoint intersection types in NeColus: a merge
recurses through records whose types are disjoint and is ill-typed
otherwise, so ``\texttt{1 ,{}, 1} is not expressible'', two extensions
that both contain the same \texttt{implLang} cannot be composed, and
commutativity and associativity hold for merges that are well typed.
Compositional Programming~\cite{zhang2021-compositional-programming}
builds on the same disjointness, so a trait that provides anything
cannot be merged with itself, and its operators for resolving a
conflict are biased.
DeepFJig~\cite{corradi2012-deepfjig} sums nested classes recursively
and symmetrically, but at most one of two same-named members may be
concrete, so a class with a concrete member cannot be summed with
itself.
In the taxonomy of Section~\ref{sec:semantic-variants}, all three take conflict rejection on an
overlap; NeColus binds self early, while Compositional Programming and DeepFJig bind it late to a
single target.
\fi

\subsection{\bilingual{Module Calculi and Feature Composition}{模块演算与特征组合}}

\ifPaperChinese
Ancona 与 Zucca 的 CMS~\cite{ancona2002-module-systems-calculus} 是一个建立在任意核心演算之上的模块演算,并把 $\lambda$-演算与 Abadi--Cardelli 对象演算编码进其中;它的和拒绝两个模块都导出的名字,且并是扁平的。在 Newspeak~\cite{bracha2010-newspeak-modules-as-objects} 中,顶层类定义模块,每个类都是 mixin;slot、方法与嵌套类可以彼此覆写,但仍是三种声明形式,而覆写整体替换一个类。Apel 与 Hutchins 的 gDeep~\cite{apel2010-uniform-feature-composition} 以有序的函数复合跨制品统一地组合特征,其中后来的特征精化先前的特征,而未标注的名字冲突被拒绝。按第~\ref{sec:semantic-variants}~节的分类,CMS 的和与 gDeep 未标注的冲突取冲突拒绝,Newspeak 的覆写与 gDeep 标注的覆写取优先级合并;CMS 早期绑定其引用。
\else
Ancona and Zucca's CMS~\cite{ancona2002-module-systems-calculus} is
a module calculus over an arbitrary core calculus, into which it
encodes the $\lambda$-calculus and the Abadi--Cardelli object
calculus; its sum rejects a name that both modules export, and the
union is flat.
In Newspeak~\cite{bracha2010-newspeak-modules-as-objects}, a top-level
class defines a module and every class is a mixin; slots, methods, and
nested classes override one another but remain three forms of
declaration, and overriding replaces a class whole.
Apel and Hutchins's gDeep~\cite{apel2010-uniform-feature-composition}
composes the features of any artifact uniformly by ordered function
composition, in which a later feature refines an earlier one and an
unmarked name clash is rejected.
In the taxonomy of Section~\ref{sec:semantic-variants}, the sum of CMS and the unmarked clashes of
gDeep take conflict rejection, while the overriding of Newspeak and the marked overrides of gDeep take
priority merge; CMS binds its references early.
\fi


\section{\bilingual{Conclusion}{结论}}
\label{sec:conclusion}

\ifPaperChinese
第~\ref{sec:introduction}~节问了一门语言能否既是通用的,又由构造保证模块化。继承演算回答说能,凭借一种构造,即深可合并 mixin;它统一了面向对象语言与配置语言的构造。在通用一侧,它充分地嵌入惰性 $\lambda$-演算,并按 Felleisen 的意义比它更具宏表达力(第~\ref{sec:translation}~节),而其中的普通算术符合逻辑编程的关系语义(第~\ref{sec:case-study}~节与第~\ref{sec:emergent-phenomena}~节)。在模块化一侧,第~\ref{sec:modularity}~节证明它由构造满足 Meyer 的全部五条标准,其中可理解性读作局部性。在第~\ref{sec:prior-conflicts}~节所回顾的标准之间的冲突中,它解决了两个。一个 mixin 在广度上与深度上一直拆到它的叶,而每个部分都可以放到别处,所以可分解性与可组合性一致。先前的设计由构造满足其中一些标准,但把它们合在一起,会沿多条继承路线带入同一个嵌套作用域,而只有一个目标的晚绑定引用在那里必须拒绝该组合或选择一条路线(第~\ref{sec:prior-modularity}~节与第~\ref{sec:semantic-variants}~节);multi-target $\this$ 沿每条路线解析该引用。可分解性与信息隐藏之间的冲突仍然存在:信息隐藏可以在语言的一个变种中、以可分解性为代价加入。

继承演算保留了面向对象语言的多重继承,只替换了其方法层面的遮蔽:$\overrides$ 方程~(\ref{eq:overrides}) 把子类中重新定义的方法与被继承的方法深合并,而不是遮蔽后者,因此覆写一个方法本身就是多重继承。在这个意义上,继承演算是更纯粹的面向对象。面向对象编程传统上被归类为命令式的;去掉可变状态和函数体,便只剩成员定义与继承,这就是第~\ref{sec:introduction-approach}~节意义上的纯声明式编程。

单个 mixin 同时是数据、程序、配置、函数、关系和对象。我们提出继承演算作为通用编程的基础。
\else
Section~\ref{sec:introduction} asked whether a language can be both
general-purpose and modular by construction. Inheritance-calculus answers
that it can, with one construct, the deep-mergeable mixin, which unifies the
constructs of object-oriented and configuration languages. Toward generality, it embeds the lazy
$\lambda$-calculus adequately and is more macro-expressive than it in
Felleisen's sense (Section~\ref{sec:translation}), and ordinary arithmetic
in it agrees with the relational semantics of logic programming
(Sections~\ref{sec:case-study} and~\ref{sec:emergent-phenomena}). Toward
modularity, Section~\ref{sec:modularity} proves that it meets all five of
Meyer's criteria by construction, with understandability read as locality.
Of the conflicts among the criteria that Section~\ref{sec:prior-conflicts} recalls, it resolves two.
A mixin splits in breadth and in depth down to its leaves, and each part can be placed elsewhere, so
decomposability and composability agree.
Prior designs meet some of the criteria by construction, but combining them brings one nested scope in
along several inheritance routes, where a late-bound reference with a single target must reject the
composition or select one route (Sections~\ref{sec:prior-modularity}
and~\ref{sec:semantic-variants}); multi-target $\this$ resolves the reference along every route.
The conflict between decomposability and information hiding remains: information hiding can be added at
the cost of decomposability, in a variant of the language.

Inheritance-calculus keeps the multiple inheritance of
object-oriented languages and replaces only their method-level
shadowing: the $\overrides$ equation~(\ref{eq:overrides})
deep-merges a method redefined in a subclass with the inherited
method instead of shadowing it, so overriding a method is itself
multiple inheritance.
In this sense, inheritance-calculus is more purely object-oriented.
Object-oriented programming is traditionally classified as
imperative; removing mutable state and function bodies leaves
member definitions and inheritance, which is pure declarative
programming in the sense of Section~\ref{sec:introduction-approach}.

A single mixin is at once data, program, configuration, function, relation,
and object. We propose inheritance-calculus as a foundation for
general-purpose programming.
\fi
\section*{\bilingual{Data-Availability Statement}{数据可用性声明}}
\ifPaperChinese
MIXINv2 实现(本文每个示例都在其测试套件中可执行)与 Lean~4 开发(机械化附录~\ref{app:well-definedness},以及第~\ref{sec:modularity}~节的结合律、交换律、幂等律与分解定理和可替换性结果的语义部分)随本文一同提供\anon[,作为补充材料]{,见 MIXINv2 源码~\cite{yang2026-mixinv2}}。其中 \texttt{lean/README.md} 把每条被机械化的陈述映射到其 Lean 名称,列出为信任这些陈述所需阅读的文件,并给出构建说明;构建在任何 \texttt{sorry} 处失败,并钉死每条主定理的公理集。
\else
The MIXINv2 implementation, in whose test suite every example of this
paper is executable, and the Lean~4 development mechanizing
Appendix~\ref{app:well-definedness} and, from
Section~\ref{sec:modularity}, the associativity, commutativity,
idempotence and decomposition theorems and the semantic part of
the substitutability results accompany the
paper\anon[~as supplementary material]{~in the MIXINv2
source~\cite{yang2026-mixinv2}}. There, \texttt{lean/README.md} maps each
mechanized statement to its Lean name, lists the files to read to trust
the statements, and gives build instructions; the build fails on any
\texttt{sorry} and pins the axiom set of every headline theorem.
\fi

\bibliographystyle{ACM-Reference-Format}
\bibliography{references}

\appendix

\section{\bilingual{Well-Definedness of the Semantic Functions}{语义函数的良定义性}}
\label{app:well-definedness}

\ifPaperChinese
五个语义方程~(\ref{eq:supers})--(\ref{eq:this}) 是一阶的:其输入与输出均为路径及路径的集合,不涉及函数空间、闭包或高阶构造。
本节通过证明这五个方程定义了一个单调算子来确立良定义性;该算子有最小不动点,由 tabling 求值~\cite{tamaki1986-tabled-resolution} 计算,其可靠性由定理~\ref{thm:tabling-soundness} 确立。
我们在 Lean~4 中机械化了本附录的结果;本附录结尾的一段把每条陈述映射到其 Lean 名称。
本节中我们用 $|p|$ 表示路径 $p$ 的长度,即它所含标签的个数。
\else
The five semantic
equations~(\ref{eq:supers})--(\ref{eq:this}) are
first-order: their inputs and outputs are paths and sets of
paths, with no function spaces, closures, or higher-order
constructs.
This section establishes well-definedness by showing that
the five equations define a monotone operator with a least
fixed point, computed by tabled
evaluation~\cite{tamaki1986-tabled-resolution}; Theorem~\ref{thm:tabling-soundness}
establishes that this computation is sound.
We mechanize the results of this appendix in Lean~4; the closing paragraph
of the appendix maps each statement to its Lean name.
Throughout this section, $|p|$ denotes the length of a path $p$, the number
of labels it contains.
\fi

\begin{definition}[\bilingual{Dependency graph}{依赖图}]\label{def:dependency-graph}
\ifPaperChinese
  固定一棵 AST 及其原始函数 $\inherits$。
  \emph{依赖图}~$G$ 以由该 AST 产生的所有语义函数应用
  ($\supers(p)$、\allowbreak$\overrides(p)$、\allowbreak$\bases(p)$、\allowbreak
  $\resolve(\ldots)$、\allowbreak$\this(\ldots)$) 为顶点。
  当方程~(\ref{eq:supers})--(\ref{eq:this}) 中定义~$u$ 的右侧
  在某个查询--答案映射下读取~$v$ 的值时,
  从顶点~$u$ 到顶点~$v$ 有一条有向边。
  因此边是语法性的,不依赖任何近似阶段。
\else
  Fix an AST with its primitive function
  $\inherits$.
  The \emph{dependency graph}~$G$ has as vertices all
  semantic-function applications
  ($\supers(p)$, $\overrides(p)$,
    $\bases(p)$, $\resolve(\ldots)$,
  $\this(\ldots)$) that arise from the AST\@.
  There is a directed edge from vertex~$u$ to vertex~$v$ when the
  right-hand side defining~$u$ in
  equations~(\ref{eq:supers})--(\ref{eq:this}) reads the value
  of~$v$ under some query--answer map.
  The edges are therefore syntactic and do not depend on any stage
  of approximation.
\fi
\end{definition}

\begin{definition}[\bilingual{Immediate consequence operator $T_P$}{直接后果算子 $T_P$}]%
  \label{def:tp-operator}
\ifPaperChinese
  设 $\mathcal{Q}$ 为所有语义函数查询(即~$G$ 的顶点)的集合。
  对每个查询 $q \in \mathcal{Q}$,设 $\mathcal{V}_q$ 为~$q$ 的单个答案可取值的集合:
  $\supers$ 查询取路径对(继承位,override 位),
  $\overrides$、$\bases$、$\resolve$、$\this$ 查询取路径;
  定义方程的集合取值由下文整个答案集合 $M(q)$ 承载,
  故 $\mathcal{V}_q$ 的元素本身绝不是集合。
  \emph{查询--答案映射}是逐点乘积
  $\prod_{q \in \mathcal{Q}} \mathcal{P}(\mathcal{V}_q)$ 中的一个元素;
  等价地,它为每个查询~$q$ 指派一个候选答案集合 $M(q) \subseteq \mathcal{V}_q$。
  \emph{直接后果算子}~\cite{vanemden1976-predicate-logic-semantics}
  $T_P$ 将查询--答案映射~$M$ 映射为新的查询--答案映射~$T_P(M)$:
  $T_P(M)(q)$ 收集定义~$q$ 的那条方程右侧的集合推导式的\emph{所有实例}所产出的值,
  其中每个实例以当前近似~$M(q')$ 解释其绑定的每个子查询事实~$q'$。
  两点约定使该产出对 $M$ 逐实例正向:
  若某右侧构造自反传递闭包 $\bases^*$,则该闭包是对由子查询值
  $M(\bases(\cdot))$ 装配出的关系 $\bases$ 取闭包,而非~$M$ 的独立坐标;
  而 $\this$ 方程按其对前沿的可加性逐单点前沿展开
  (可加性
  $\this(S_1 \cup S_2, p, n) = \this(S_1, p, n) \cup \this(S_2, p, n)$
  对 $n$ 归纳直接验证:$n = 0$ 时方程~\eqref{eq:this} 原样返回前沿;
  $n > 0$ 时其推导式按 $p_{\mathrm{current}} \in S$ 逐元素生成,
  故 $S_1 \cup S_2$ 产出的新前沿是 $S_1$、$S_2$ 各自新前沿之并,
  再对同一 $\init(p)$ 与 $n-1$ 用归纳假设)——
  $T_P(M)(\this(S, p, n))$ 收集
  $\bigcup_{s \in S} M(\this(\{s\}, p, n))$,
  且在单点前沿查询处,$n > 0$ 的一步按右侧对每个匹配对
  $(p_{\mathrm{site}}, p) \in M(\supers(s))$ 读取
  $M(\this(\{p_{\mathrm{site}}\}, \init(p), n-1))$ 并取并。
  于是每个坐标的产出都是对 $M$-事实的正向存在式收集,不含以 $M$
  为索引的整体替换。
\else
  Let $\mathcal{Q}$ be the set of all
  semantic-function queries
  (vertices of~$G$).
  For each query $q \in \mathcal{Q}$, let
  $\mathcal{V}_q$ be the set of values a single answer to~$q$ may
  take: pairs of paths (inheritance site, override) for a $\supers$
  query, and paths for an $\overrides$, $\bases$, $\resolve$, or
  $\this$ query; the set value of the defining equation is carried
  by the whole answer set $M(q)$ below, so the elements of
  $\mathcal{V}_q$ are never themselves sets.
  A \emph{query--answer map} is an element of the
  pointwise product
  $\prod_{q \in \mathcal{Q}} \mathcal{P}(\mathcal{V}_q)$;
  equivalently, it assigns to each query~$q$ a set
  $M(q) \subseteq \mathcal{V}_q$ of candidate answers.
  The \emph{immediate consequence
  operator}~\cite{vanemden1976-predicate-logic-semantics}
  $T_P$ maps a query--answer map~$M$ to a new
  query--answer map~$T_P(M)$: the set $T_P(M)(q)$ collects the
  values produced by \emph{all instances} of the set comprehension
  on the right-hand side of the equation defining~$q$, each instance
  interpreting every sub-query fact it binds by the current
  approximation~$M(q')$.
  Two conventions keep this production instance-wise positive
  in~$M$: a right-hand side that forms the reflexive-transitive
  closure $\bases^*$ takes that closure of the relation $\bases$
  assembled from the sub-query values $M(\bases(\cdot))$, not a
  separate coordinate of~$M$; and the $\this$ equation is unfolded
  per singleton frontier, by its additivity in the frontier
  (additivity,
  $\this(S_1 \cup S_2, p, n) = \this(S_1, p, n) \cup \this(S_2, p, n)$,
  is verified directly by induction on~$n$: at $n = 0$
  equation~\eqref{eq:this} returns the frontier itself, and at
  $n > 0$ its comprehension generates over
  $p_{\mathrm{current}} \in S$ element by element, so the new
  frontier produced by $S_1 \cup S_2$ is the union of the new
  frontiers of $S_1$ and $S_2$, to which the induction hypothesis
  applies at the same $\init(p)$ and $n-1$) ---
  $T_P(M)(\this(S, p, n))$ collects
  $\bigcup_{s \in S} M(\this(\{s\}, p, n))$, and at a
  singleton-frontier query the $n > 0$ step reads, for each matching
  pair $(p_{\mathrm{site}}, p) \in M(\supers(s))$, the set
  $M(\this(\{p_{\mathrm{site}}\}, \init(p), n-1))$ and takes the
  union. Every coordinate's production is thus a positive
  existential collection of $M$-facts, with no wholesale substitution
  indexed by~$M$.
\fi
\end{definition}

\begin{lemma}[\bilingual{Monotonicity}{单调性}]\label{lem:monotonicity}
  \ifPaperChinese
  算子 $T_P$ 在完备格
  $\Bigl(\prod_{q \in \mathcal{Q}} \mathcal{P}(\mathcal{V}_q),\; \subseteq\Bigr)$
  上是单调的。
  \else
  \begin{sloppypar}
  The operator $T_P$ is monotone on the complete lattice
  $\Bigl(\prod_{q \in \mathcal{Q}} \mathcal{P}(\mathcal{V}_q),\; \subseteq\Bigr)$.
  \end{sloppypar}
  \fi
\end{lemma}

\begin{proof}
\ifPaperChinese
  每条方程的右侧均由 $\in$、等式测试以及正集合
  概括 $\Set{ x | x \in S, \ldots }$ 构成。
  超出这些构造子的唯一算子,即 $\supers$ 中出现的自反传递闭包
  $\bases^*$,是经 $M$ 解释的关系 $\bases$ 的闭包;闭包对其所闭合的
  关系单调,故更大的 $M$ 给出更大的 $\bases$,从而更大的 $\bases^*$。
  其中出现的不等与否定,即 $\lexical$、$\qualified$ 的分支守卫
  与 $\dom(\inherits)$ 内的 $\neq$、$\notin$ 等测试,均作用于固定的原始量 $\inherits$
  及其派生的 AST 查询,而非作用于经 $M$ 解释的任何子查询,因而是 $T_P$ 的常量。
  方程 $\this$ 按定义~\ref{def:tp-operator} 的可加展开读取:
  一般前沿处的产出是前沿各单点处 $M$-值的并;
  单点前沿处,$n > 0$ 的一步是以
  $M(\this(\{p_{\mathrm{site}}\}, \init(p), n-1))$ 为项、
  在满足 $(p_{\mathrm{site}}, p) \in M(\supers(s))$ 的对上取并的正向并,
  其指标集只随 $M$ 增大。
  没有任何右侧对子查询施以否定、集合差或补集。
  因此,若 $M \subseteq M'$(逐点),则 $T_P(M) \subseteq T_P(M')$。
\else
  Each equation's right-hand side is built from
  $\in$, equality tests, and positive set
  comprehension $\Set{ x | x \in S, \ldots }$.
  The one operator beyond these, the reflexive-transitive
  closure $\bases^*$ in $\supers$, is the closure of the
  relation $\bases$ read through $M$; a closure is monotone in
  the relation it closes, so a larger $M$ yields a larger
  $\bases$ and hence a larger $\bases^*$.
  The inequalities and negations that occur, namely
  the branch guards of $\lexical$ and $\qualified$
  and the $\neq$, $\notin$ within $\dom(\inherits)$, all act on the
  fixed primitive $\inherits$ and its derived AST queries rather than
  on any sub-query read through $M$, so they are constants of $T_P$.
  The $\this$ equation is read through the additive unfolding of
  Definition~\ref{def:tp-operator}: at a general frontier the
  production is the union of the $M$-values at the frontier's
  singletons; at a singleton frontier, the $n > 0$ step is a positive
  union of the sets
  $M(\this(\{p_{\mathrm{site}}\}, \init(p), n-1))$ over the
  pairs $(p_{\mathrm{site}}, p) \in M(\supers(s))$, an index set
  that only grows with $M$.
  No right-hand side applies negation, set difference, or
  complement to a sub-query.
  Therefore, if $M \subseteq M'$ (pointwise), then
  $T_P(M) \subseteq T_P(M')$.
\fi
\end{proof}

\begin{theorem}[\bilingual{Existence and uniqueness of
  $\lfp(T_P)$}{$\lfp(T_P)$ 的存在性与唯一性}]\label{thm:lfp-exists}
  \ifPaperChinese
  算子 $T_P$ 有唯一的最小不动点 $\lfp(T_P)$。
  \else
  The operator $T_P$ has a unique least fixed point
  $\lfp(T_P)$.
  \fi
\end{theorem}

\begin{proof}
\ifPaperChinese
  每个幂集格
  $(\mathcal{P}(\mathcal{V}_q), \subseteq)$
  都是完备的,
  因此逐点乘积格
  $\Bigl(\prod_{q \in \mathcal{Q}} \mathcal{P}(\mathcal{V}_q), \subseteq\Bigr)$
  也是完备的。
  由 Knaster--Tarski 定理~\cite{tarski1955-lattice-fixpoint-theorem},
  完备格上的每个单调函数都有唯一的最小不动点。
\else
  \begin{sloppypar}
  Each powerset lattice
  $(\mathcal{P}(\mathcal{V}_q), \subseteq)$
  is complete,
  so the pointwise product lattice
  $\Bigl(\prod_{q \in \mathcal{Q}} \mathcal{P}(\mathcal{V}_q), \subseteq\Bigr)$
  is also complete.
  By the Knaster--Tarski
  theorem~\cite{tarski1955-lattice-fixpoint-theorem},
  every monotone function on a complete lattice has a
  unique least fixed point.
  \end{sloppypar}
\fi
\end{proof}

\begin{theorem}[\bilingual{Well-definedness}{良定义性}]\label{thm:well-defined}
  \ifPaperChinese
  语义函数 $\supers$、\allowbreak$\overrides$、\allowbreak$\bases$、\allowbreak
  $\resolve$、\allowbreak$\this$ 是良定义的:
  其指称语义为 $\lfp(T_P)$,即直接后果算子的唯一最小不动点。
  \else
  The semantic functions
  $\supers$, $\overrides$, $\bases$,
  $\resolve$, $\this$ are well-defined:
  their denotational semantics is
  $\lfp(T_P)$, the unique least fixed point of
  the immediate consequence operator.
  \fi
\end{theorem}

\begin{proof}
\ifPaperChinese
  由引理~\ref{lem:monotonicity} 与定理~\ref{thm:lfp-exists},$T_P$ 有唯一的最小不动点;
  将每个函数取为 $\lfp(T_P)$ 的相应分量,即得良定义。此最小不动点对\emph{任何}
  AST 都存在,与程序是否良构无关;良构性是另一回事,即下文界定的语法辅助函数的有定义性。
\else
  By Lemma~\ref{lem:monotonicity} and
  Theorem~\ref{thm:lfp-exists}, $T_P$ has a unique least
  fixed point; taking each function to be the
  corresponding component of $\lfp(T_P)$ makes it
  well-defined. This least fixed point exists for
  \emph{every} AST, regardless of whether the program is
  well-formed; well-formedness is a separate property, the
  definedness of the syntactic helpers, which we settle
  next.
\fi
\end{proof}

\begin{lemma}[\bilingual{Kleene iteration}{Kleene 迭代}]\label{lem:kleene-iteration}
  \ifPaperChinese
  算子 $T_P$ 是 $\omega$-连续的:对每个非空有向族 $\{M_j\}_{j \in J}$,
  $T_P\bigl(\bigcup_{j} M_j\bigr) = \bigcup_{j} T_P(M_j)$。
  因此 $\lfp(T_P) = \bigcup_{i < \omega} T_P^i(\varnothing)$;
  $\lfp(T_P)$ 中的每条事实都落入某个有限阶段 $T_P^n(\varnothing)$,
  使其成立的最小下标 $n$ 称为该事实的 \emph{Kleene 秩}
  (且 $n \geq 1$,因 $T_P^0(\varnothing) = \varnothing$);
  并且 $T_P^{i+1}(\varnothing)$ 中的一条事实恰由其查询所属方程的右侧
  从 $T_P^i(\varnothing)$ 中有限多条事实产生(\emph{末步反演})。
  \else
  The operator $T_P$ is $\omega$-continuous: for every nonempty directed family
  $\{M_j\}_{j \in J}$ of query--answer maps,
  $T_P\bigl(\bigcup_{j} M_j\bigr) = \bigcup_{j} T_P(M_j)$.
  Consequently $\lfp(T_P) = \bigcup_{i < \omega} T_P^i(\varnothing)$;
  every fact of $\lfp(T_P)$ lies in some finite stage
  $T_P^n(\varnothing)$, and the least such index $n$ is the fact's
  \emph{Kleene rank} (with $n \geq 1$, since
  $T_P^0(\varnothing) = \varnothing$); moreover a fact in
  $T_P^{i+1}(\varnothing)$ is produced by the right-hand side of its
  query's equation from finitely many facts in $T_P^i(\varnothing)$
  (\emph{last-step inversion}).
  \fi
\end{lemma}

\begin{proof}
\ifPaperChinese
  单调性(引理~\ref{lem:monotonicity})给出
  $\bigcup_{j} T_P(M_j) \subseteq T_P\bigl(\bigcup_{j} M_j\bigr)$。
  反向包含:设事实 $v \in T_P\bigl(\bigcup_{j} M_j\bigr)(q)$。
  按定义~\ref{def:tp-operator},$v$ 由定义 $q$ 的方程右侧的集合推导式的
  \emph{一个实例}产生,而一个实例只绑定有限多条前提事实:每个推导式绑定子对应一条
  $\supers$、$\overrides$、$\bases$、$\resolve$ 或 $\this$ 事实,而闭包
  $\bases^*$ 中的一对由一条有限的 $\bases$ 事实链见证。
  每条前提事实都属于某个 $M_j$;由有向性,存在单个 $M_{j_0}$ 同时包含全部前提,
  故 $v \in T_P(M_{j_0})(q)$。于是 $T_P$ 是 $\omega$-连续的。
  记 $L = \bigcup_{i} T_P^i(\varnothing)$;链
  $\varnothing \subseteq T_P(\varnothing) \subseteq \cdots$ 有向,连续性给出
  $T_P(L) = \bigcup_{i} T_P^{i+1}(\varnothing) = L$,故 $L$ 是不动点;
  又对任何不动点 $F$,对 $i$ 归纳得 $T_P^i(\varnothing) \subseteq F$,
  故 $L \subseteq F$,即 $L = \lfp(T_P)$。
  有限秩与末步反演随即得出:事实属于某个 $T_P^{i+1}(\varnothing)$,
  而按定义~\ref{def:tp-operator},这一成员资格恰好就是
  "由方程右侧从 $T_P^i(\varnothing)$ 产生"。
\else
  Monotonicity (Lemma~\ref{lem:monotonicity}) gives
  $\bigcup_{j} T_P(M_j) \subseteq T_P\bigl(\bigcup_{j} M_j\bigr)$.
  For the converse, let $v \in T_P\bigl(\bigcup_{j} M_j\bigr)(q)$ be a
  fact. By Definition~\ref{def:tp-operator}, $v$ is produced by
  \emph{one instance} of the set comprehension on the right-hand side
  of the equation defining $q$, and one instance binds only finitely
  many premise facts: one $\supers$, $\overrides$, $\bases$,
  $\resolve$, or $\this$ fact per comprehension binder, and a pair of
  the closure $\bases^*$ is witnessed by a finite chain of $\bases$
  facts. Every premise lies in some $M_j$; by directedness a single
  $M_{j_0}$ contains them all, so $v \in T_P(M_{j_0})(q)$. Hence
  $T_P$ is $\omega$-continuous.
  Write $L = \bigcup_{i} T_P^i(\varnothing)$; the chain
  $\varnothing \subseteq T_P(\varnothing) \subseteq \cdots$ is
  directed, so continuity gives
  $T_P(L) = \bigcup_{i} T_P^{i+1}(\varnothing) = L$, and $L$ is a
  fixed point; for any fixed point $F$, induction on $i$ gives
  $T_P^i(\varnothing) \subseteq F$, so $L \subseteq F$, that is,
  $L = \lfp(T_P)$.
  Finite rank and last-step inversion follow: a fact of $L$ lies in
  some $T_P^{i+1}(\varnothing)$, and by
  Definition~\ref{def:tp-operator} that membership is exactly
  production by the equation's right-hand side over
  $T_P^i(\varnothing)$.
\fi
\end{proof}

\ifPaperChinese
我们对含偏函数的集合推导式采用标准约定：推导式 $\Set{ f(x) | x \in S }$
仅收集使 $f(x)$ 有定义的元素，当其源整体无定义时取空，对无定义的情况一律静默忽略。
这里有三件相邻而不同的事，我们把它们分开。

\emph{其一，语法辅助函数是偏函数}，且分两种。一种是\emph{缺失}：$\at(p \snoc \ell)$ 无条目，因为 $\ell$ 未写在 $p$ 下的任何处。另一种是\emph{畸形}：写下的节点形状不符（$\items$ 无定义）、单个 mixin 内有重复键（$\at$ 非唯一）、写下的引用其首标签爬过了根（de Bruijn 索引查找函数无定义），或某成员条目既非引用亦非定义（$\kv$ 对之无定义）。

\begin{sloppypar}
\emph{其二，$\dom(\inherits)$ 把缺失显式化}。其成员判定 $\dom(\inherits) = \Set{ p | p = () \vee \has(\init(p), \last(p)) }$ 由~\eqref{eq:helpers} 给出，是一个 AST 本地查询，因而可判定；落在其外的路径未写下任何 mixin，被上述约定静默略去，并不构成不良构。缺失因此以一个布尔值报告，绝不以无定义的形式泄入语义。
\end{sloppypar}

\emph{其三，本演算要求畸形不出现}。一个程序\emph{良构}，当且仅当下列情形均不出现于任何写下位置：
\begin{enumerate}
  \item $\items(\mathit{node})$ 无定义：$\mathit{node}$ 不是 $\seqtag$ 节点。
  \item $\at(p \snoc \ell)$ 无唯一解：$\items(\at(p))$ 中有两个或更多以标签 $\ell$ 为键的
        单键 $\maptag$ 条目（即重复键）。
  \item $\lexical(p, \ell)$ 或 $\qualified(p, \ell)$ 无定义：向外爬升到根 $()$ 仍未命中，其 $\text{otherwise}$ 支在 $p = ()$ 处需要 $\init(())$，而 $\init$ 是偏函数、对空路径 $()$ 无定义；亦即写下的引用其首标签（词法）或锚标签（限定 this）爬过了所有外围作用域。
  \item $\kv(\mathit{node})$ 对某成员条目无定义：$\items(\at(p))$ 的某元素既不是 $\seqtag$ 节点（继承引用，由~\eqref{eq:helpers} 的 $\items$ 读取），也不是键为 $\strtag$ 标量的单键 $\maptag$ 节点（定义，由 $\kv(\cdot)$ 读取），例如多键映射、标量条目或非标量键；这样的条目会被两个对偶查询同时静默略过，写下的内容被无声丢弃。
\end{enumerate}
情形 (1)、(2) 与 (4) 是节点局部的，在有限多个不同写下节点上可判定（$\items$、$\at$、$\kv$ 为 AST 查询）。情形 (3) 在每个位置 $p$ 处至多 $|p|$ 步内终止。当文档是有限树（无 YAML 锚）时，位置有限，故良构性是可判定的有限静态性质。YAML 锚使文档成为图，其展开成的路径可能无穷；此时越根的畸形若存在，可由枚举发现（不良构性是半可判定的），但其在无穷多个位置上的缺席无法由有限检查证实，故良构性是余半可判定的（$\Pi_1$）。在良构程序的每个写下位置上，每个成员条目恰被两个对偶查询之一消费（引用归 $\items$，定义归 $\kv$），语法辅助函数在各自要读的条目上全有定义，没有写下的条目被无声丢弃；本演算即定义在这样的程序之上。

\else
\begin{sloppypar}
We adopt the standard convention for set comprehensions involving partial functions:
$\Set{ f(x) | x \in S }$ collects only those elements for which $f(x)$ is defined,
is empty when its source is itself undefined, and silently omits undefined cases.
Three notions meet here, and we keep them apart.
\end{sloppypar}

\emph{First, the syntactic helpers are partial,} in two ways. One is \emph{absence}:
$\at(p \snoc \ell)$ has no entry because $\ell$ is written nowhere under $p$. The other is
\emph{malformedness}: a written node of the wrong shape ($\items$ undefined), a mixin with a
duplicate key ($\at$ non-unique), a written reference whose leading label climbs past the
root (the de Bruijn index helper undefined), or a member entry that is neither a reference
nor a definition ($\kv$ undefined on it).

\begin{sloppypar}
\emph{Second, $\dom(\inherits)$ makes absence explicit.} Its membership
$\dom(\inherits) = \Set{ p | p = () \vee \has(\init(p), \last(p)) }$, given by
\eqref{eq:helpers}, is a local AST query and hence decidable; a path outside it writes no mixin
and is silently omitted by the convention above, not ill-formed. Absence is thereby reported as
a Boolean and never enters the semantics as undefinedness.
\end{sloppypar}

\emph{Third, the calculus requires malformedness to be absent.} A program is \emph{well-formed}
when none of the following arises at a written position:
\begin{enumerate}
  \item $\items(\mathit{node})$ is undefined: $\mathit{node}$ is not a $\seqtag$ node.
  \item $\at(p \snoc \ell)$ has no unique solution: $\items(\at(p))$ contains two or more
        single-key $\maptag$ entries with label~$\ell$, a duplicate key.
  \item $\lexical(p, \ell)$ or $\qualified(p, \ell)$ is undefined: the climb
        reaches the root $()$ without a match, where the $\text{otherwise}$ branch at $p = ()$ requires
        $\init(())$ and $\init$, a partial function, is undefined on the empty
        path $()$; a written reference whose leading label (lexical) or
        anchor label (qualified this) climbs past every enclosing scope.
  \item $\kv(\mathit{node})$ is undefined on some member
        entry: an element of $\items(\at(p))$ is neither a $\seqtag$ node
        (an inheritance reference, read by $\items$ in
        equation~\eqref{eq:helpers}) nor a single-key $\maptag$ node whose
        key is a $\strtag$ scalar (a definition, read by
        $\kv(\cdot)$), such as a multi-key mapping, a scalar entry,
        or a non-scalar key; such an entry would be silently skipped by
        both dual queries, dropping written content without report.
\end{enumerate}
Cases (1), (2), and (4) are node-local and decidable on the finitely many distinct written
nodes ($\items$, $\at$, and $\kv$ are AST queries). Case (3) halts in at most $|p|$
steps at each position $p$. When the document is a finite tree (no YAML anchors), the positions are finite, so
well-formedness is a decidable, finite static property. YAML anchors make the document a graph
whose unfolding into paths may be infinite; a malformed climb, if present, is then found by
enumeration (ill-formedness is semi-decidable), but its absence across infinitely many positions
is not certifiable by a finite check, so well-formedness is co-semi-decidable ($\Pi_1$). At every
written position of a well-formed program every member entry is consumed by exactly one of the
two dual queries ($\items$ for references, $\kv$ for definitions), the syntactic
helpers are defined on the entries each is meant to read, and no written entry is silently
dropped; the calculus is defined on such programs.

\fi


\begin{proof}[\bilingual{Proof of Proposition~\ref{prop:coverage}}{命题~\ref{prop:coverage} 的证明}]
\ifPaperChinese
良构性排除畸形的四种情形(上文其三);在良构程序的每个写下位置上,语法辅助函数在各自要读的条目上全有定义。余下的无定义只有缺席,即某标签 $\ell$ 在 $p$ 之下无处写下时 $\at(p \snoc \ell)$ 没有条目。上文采用的集合推导约定略去它,而由其二,$\dom(\inherits)$ 把它作为布尔值报告,它从不作为无定义进入语义。由于 $\lfp(T_P)$ 对任何 AST 都存在(定理~\ref{thm:well-defined}),它在这里无需任何辅助函数失败即可到达。
\else
Well-formedness excludes the four cases of malformedness (the third point above), and at every
written position of a well-formed program the syntactic helpers are defined on the entries each is
meant to read. The undefinedness that remains is absence: $\at(p \snoc \ell)$ has no entry because
$\ell$ is written nowhere under $p$. The convention for set comprehensions adopted above omits it,
and by the second point $\dom(\inherits)$ reports it as a Boolean, so it never enters the semantics
as undefinedness. Since $\lfp(T_P)$ exists for every AST (Theorem~\ref{thm:well-defined}), it is
reached here without any helper failing.
\fi
\end{proof}

\ifPaperChinese
求值是否终止,则是下文另行处理的问题。
\else
Whether evaluation terminates is the separate question settled below.
\fi

\begin{lemma}[\bilingual{Path existence}{路径存在性}]\label{lem:path-existence}
  \ifPaperChinese
  称一条路径\emph{存在},若它是根 $()$,或它形如 $p \snoc \ell$ 且 $\supers(p)$ 的某个 override 分量直接写下 $\ell$,
  即存在 $(\_,\, p_{\mathrm{branch}}) \in \supers(p)$ 使 $p_{\mathrm{branch}} \snoc \ell \in \dom(\inherits)$。
  则对每条路径 $q$,$\supers(q) \neq \varnothing$ 当且仅当 $q$ 存在。因此一个标签 $\ell$ 在 $p$ 处存在
  （第~\ref{sec:mixin-trees}~节的意义下）,恰当 $\supers(p \snoc \ell) \neq \varnothing$。
  \else
  Call a path \emph{existent} if it is the root $()$, or it is $p \snoc \ell$ and some override
  component of $\supers(p)$ writes $\ell$ directly, that is
  $p_{\mathrm{branch}} \snoc \ell \in \dom(\inherits)$ for some
  $(\_,\, p_{\mathrm{branch}}) \in \supers(p)$. Then $\supers(q) \neq \varnothing$ if and only if
  $q$ is existent, for every path $q$. Consequently a label $\ell$ exists at $p$, in the sense of
  Section~\ref{sec:mixin-trees}, exactly when $\supers(p \snoc \ell) \neq \varnothing$.
  \fi
\end{lemma}
\begin{proof}
  \ifPaperChinese
  首先,$\supers(q) \neq \varnothing$ 当且仅当 $\overrides(q) \neq \varnothing$。自反项 $q \in \bases^*(q)$
  经~\eqref{eq:supers} 把 $\overrides(q)$ 注入 $\supers(q)$,故 $\overrides(q) \neq \varnothing$ 蕴含
  $\supers(q) \neq \varnothing$。反之,若 $\overrides(q) = \varnothing$,则由~\eqref{eq:bases} 有
  $\bases(q) = \varnothing$,从而 $\bases^*(q) = \{q\}$,再由~\eqref{eq:supers} 得 $\supers(q) = \varnothing$。

  对于根,$\overrides(()) = \{()\} \neq \varnothing$ 且 $()$ 存在。对于 $q = p \snoc \ell$,
  \eqref{eq:overrides} 使 $\overrides(p \snoc \ell)$ 非空,恰当存在 $(\_,\, p_{\mathrm{branch}}) \in \supers(p)$
  满足 $p_{\mathrm{branch}} \snoc \ell \in \dom(\inherits)$,这正是 $p \snoc \ell$ 存在的定义。
  \else
  First, $\supers(q) \neq \varnothing$ if and only if $\overrides(q) \neq \varnothing$. The reflexive
  term $q \in \bases^*(q)$ contributes $\overrides(q)$ to $\supers(q)$ by~\eqref{eq:supers}, so
  $\overrides(q) \neq \varnothing$ gives $\supers(q) \neq \varnothing$. Conversely, if
  $\overrides(q) = \varnothing$ then $\bases(q) = \varnothing$ by~\eqref{eq:bases}, hence
  $\bases^*(q) = \{q\}$ and $\supers(q) = \varnothing$ by~\eqref{eq:supers}.

  For the root, $\overrides(()) = \{()\} \neq \varnothing$ and $()$ is existent. For
  $q = p \snoc \ell$, \eqref{eq:overrides} makes $\overrides(p \snoc \ell)$ non-empty exactly when
  some $(\_,\, p_{\mathrm{branch}}) \in \supers(p)$ satisfies
  $p_{\mathrm{branch}} \snoc \ell \in \dom(\inherits)$, which is the definition of $p \snoc \ell$
  being existent.
  \fi
\end{proof}

\begin{lemma}[\bilingual{De Bruijn bound}{De Bruijn 界}]\label{lem:db-bound}
  \ifPaperChinese
  设程序良构，$p$ 为任意使 $\indexed(p, r) = (n, \cdot)$ 有定义的路径（故 $p \neq ()$，因 $\indexed$ 要求 $\init(p)$）。
  则 $n \leq |\init(p)| = |p| - 1$。
  \else
  Let the program be well-formed, and let $p$ be any path with $\indexed(p, r) = (n, \cdot)$ defined (so $p \neq ()$, since $\indexed$ evaluates $\init(p)$).
  Then $n \leq |\init(p)| = |p| - 1$.
  \fi
\end{lemma}
\begin{proof}
  \ifPaperChinese
  我们对~\eqref{eq:helpers} 中的两种引用形式分别讨论。

  \paragraph{词法引用 $r = (\ell_1, \ldots, \ell_k)$}
  此时 $n = \lexical(\init(p),\; \ell_1)$，它从 $\init(p)$ 起、每步经 $\init$ 向外退一层测试 $\has(\cdot, \ell_1)$。
  若爬升到根 $()$ 仍未命中，则其 otherwise 支需 $\init(())$，而偏函数 $\init$ 对 $()$ 无定义，故 $\lexical$ 无定义，
  程序不良构——已被排除。
  故递归在某路径 $p_0$（满足 $\has(p_0, \ell_1)$）处停止，此时 $n$ 为从
  $\init(p)$ 至 $p_0$ 的上爬步数；由于每步路径长度减一，$n \leq |\init(p)| = |p| - 1$。

  \paragraph{限定 this 引用 $r = (\ell, {\sim}, \ell_1, \ldots, \ell_k)$}
  此时 $n = \qualified(\init(p),\; \ell)$，其停止点处的基础守卫 $\last(p_0) = \ell$ 只在 $\last$ 有定义处成立，故 $p_0 \neq ()$（$\last$ 偏，对 $()$ 无定义）；同上，$n \leq |p| - 1$。

  \paragraph{推论：$\this$ 不越根}
  在 $\resolve$ 中，$\this$ 以
  $\this(\{p_{\mathrm{site}}\},\; \init(p_{\mathrm{override}}),\; n)$
  的形式被调用，其中 $(n, w) \in \inherits(p_{\mathrm{override}})$，
  故由上述界 $n \leq |p_{\mathrm{override}}| - 1 = |\init(p_{\mathrm{override}})|$。
  记起始路径 $q = \init(p_{\mathrm{override}})$，$n$ 步递归依次把 $\init$
  作用于 $q, \init(q), \ldots, \init^{n-1}(q)$，其长度为
  $|q|, |q| - 1, \ldots, |q| - (n - 1)$；因 $n \leq |q|$，每条这样的路径
  长度至少 $|q| - (n - 1) \geq 1$、非根，故 $\init$ 始终有定义，$\this$ 绝不越根。
  \else
  We consider the two reference forms in~\eqref{eq:helpers} separately.

  \paragraph{Lexical reference $r = (\ell_1, \ldots, \ell_k)$}
  Here $n = \lexical(\init(p),\; \ell_1)$, which tests $\has(\cdot, \ell_1)$ at each
  enclosing scope outward from $\init(p)$, stepping by $\init$.
  If the climb reaches the root $()$ without a match, its otherwise branch needs $\init(())$,
  undefined since $\init$ is partial, so $\lexical$ is undefined,
  making the program ill-formed, which we have excluded.
  Therefore the recursion stops at some $p_0$ (where $\has(p_0, \ell_1)$), after $n$ steps
  from $\init(p)$; since each step shortens the path by one label,
  $n \leq |\init(p)| = |p| - 1$.

  \paragraph{Qualified-this reference $r = (\ell, {\sim}, \ell_1, \ldots, \ell_k)$}
  Here $n = \qualified(\init(p),\; \ell)$, whose base-case guard $\last(p_0) = \ell$ at the
  stopping scope can hold only where $\last$ is defined, so $p_0 \neq ()$ ($\last$ is partial, undefined
  on $()$); by the same argument, $n \leq |p| - 1$.

  \paragraph{Corollary: $\this$ never reaches root}
  In $\resolve$, $\this$ is called as
  $\this(\{p_{\mathrm{site}}\},\; \init(p_{\mathrm{override}}),\; n)$
  with $(n, w) \in \inherits(p_{\mathrm{override}})$, so
  $n \leq |p_{\mathrm{override}}| - 1 = |\init(p_{\mathrm{override}})|$
  by the bound above. Writing $q = \init(p_{\mathrm{override}})$ for the
  starting path, the $n$ recursive steps apply $\init$ to
  $q, \init(q), \ldots, \init^{n-1}(q)$, of lengths
  $|q|, |q| - 1, \ldots, |q| - (n - 1)$; since $n \leq |q|$, each such
  path has length at least $|q| - (n - 1) \geq 1$ and is non-root, so
  $\init$ is always defined and $\this$ never reaches the root.
  \fi
\end{proof}

\begin{lemma}[\bilingual{Dependency restriction}{依赖限制}]\label{lem:dependency-restriction}
  \ifPaperChinese
  设 $R \subseteq \mathcal{Q}$ 对依赖图 $G$ 的边封闭
  ($R$ 中查询的方程所读的每个子查询也在 $R$ 中),
  记 $T_P{\restriction}R$ 为 $T_P$ 在 $R$ 上的映射限制。
  则对每个 $i$,
  $T_P^i(\varnothing)\,{\restriction}R = (T_P{\restriction}R)^i(\varnothing)$;
  因此 $\lfp(T_P)\,{\restriction}R = \lfp(T_P{\restriction}R)$,
  且 $R$ 中查询处的一条事实在全局迭代与受限迭代中于同一阶段下标
  (即同一 Kleene 秩)出现。
  \else
  Let $R \subseteq \mathcal{Q}$ be closed under the dependency edges
  of $G$ (every sub-query read by the equation of a query in $R$ is
  itself in $R$), and let $T_P{\restriction}R$ denote the restriction
  of $T_P$ to maps on $R$. Then for every $i$,
  $T_P^i(\varnothing)\,{\restriction}R = (T_P{\restriction}R)^i(\varnothing)$;
  hence $\lfp(T_P)\,{\restriction}R = \lfp(T_P{\restriction}R)$, and a
  fact at a query in $R$ enters the global and the restricted
  iteration at the same stage index (the same Kleene rank).
  \fi
\end{lemma}

\begin{proof}
\ifPaperChinese
  对 $i$ 归纳。基例 $i = 0$ 两侧均为 $\varnothing$。
  归纳步:对 $q \in R$ 求值其方程的右侧只读 $R$ 内子查询的值
  ($R$ 依赖封闭),故 $T_P(M)(q)$ 只依赖 $M{\restriction}R$;
  由归纳假设
  $T_P^i(\varnothing)\,{\restriction}R = (T_P{\restriction}R)^i(\varnothing)$,
  对两侧应用 $T_P$ 并限制到 $R$,得
  $T_P^{i+1}(\varnothing)\,{\restriction}R = (T_P{\restriction}R)^{i+1}(\varnothing)$。
  对 $i$ 取并(引理~\ref{lem:kleene-iteration})即得不动点相等;
  逐阶段相等直接给出阶段下标一致。
\else
  Induction on $i$. The base case $i = 0$ has $\varnothing$ on both
  sides. For the step, evaluating the equation of a query $q \in R$
  reads only values at sub-queries in $R$ ($R$ is dependency-closed),
  so $T_P(M)(q)$ depends only on $M{\restriction}R$; with the
  induction hypothesis
  $T_P^i(\varnothing)\,{\restriction}R = (T_P{\restriction}R)^i(\varnothing)$,
  applying $T_P$ to both sides and restricting to $R$ gives
  $T_P^{i+1}(\varnothing)\,{\restriction}R = (T_P{\restriction}R)^{i+1}(\varnothing)$.
  Taking the union over $i$ (Lemma~\ref{lem:kleene-iteration}) yields
  the fixed-point equality; the stagewise equality gives the agreement
  of stage indices directly.
\fi
\end{proof}

\ifPaperChinese
我们 recall \emph{tabling},即按需求解最小不动点系统的标准方法
~\cite{tamaki1986-tabled-resolution, chen1996-tabled-evaluation-delaying, vanemden1976-predicate-logic-semantics},
并以算法~\ref{alg:tabled} 把本文所用的 tabled evaluation 规范下来:
自空累加器起做 Kleene 迭代,每一轮按需、带备忘地对可达查询各重算一次其方程,
重入调用读取运行中的逼近,直到某一轮不再新增任何事实。
\else
We recall \emph{tabling}, the standard method for solving a
least-fixpoint system on
demand~\cite{tamaki1986-tabled-resolution, chen1996-tabled-evaluation-delaying, vanemden1976-predicate-logic-semantics},
and pin down the tabled evaluation used in this paper as
Algorithm~\ref{alg:tabled}: Kleene iteration from the empty
accumulator, each round recomputing the equation of every reachable
query once, on demand and memoized, a re-entrant call reading the
running approximation, until a round adds no fact.
\fi

\begin{algorithm}
  \DontPrintSemicolon
  \KwData{\bilingual{$\mathit{approx}$, every query's answer set so far, kept across rounds, empty until produced; $\mathit{cache}$, the queries recomputed in the current round; $\mathit{stack}$, the queries being computed}{$\mathit{approx}$,每个查询迄今的答案集,跨轮保留,在被产出前为空;$\mathit{cache}$,本轮已重算的查询;$\mathit{stack}$,正在计算中的查询}}
  \Fn{\TABLED{$q_0$}}{
    \Fn{\RESOLVE{$q$}}{
      \lIf{$q \in \mathit{cache}$}{\Return $\mathit{approx}[q]$ \tcp*[f]{\bilingual{recomputed this round}{本轮已重算}}}
      \lIf{$q \in \mathit{stack}$}{\Return $\mathit{approx}[q]$ \tcp*[f]{\bilingual{re-entry reads the running approximation}{重入读取运行中的逼近}}}
      $\mathit{stack} \gets \mathit{stack} \cup \{q\}$\;
      $\mathit{produced} \gets$ \bilingual{the values of all instances of the right-hand side of $q$'s equation, every sub-query fact read through \RESOLVE{}}{$q$ 的方程右侧全部实例产出的值,每个子查询事实经 \RESOLVE{} 读取}\;
      $\mathit{stack} \gets \mathit{stack} \setminus \{q\}$\;
      \lIf{$\mathit{produced} \not\subseteq \mathit{approx}[q]$}{$\mathit{changed} \gets \mathbf{true}$}
      $\mathit{approx}[q] \gets \mathit{approx}[q] \cup \mathit{produced}$;\; $\mathit{cache} \gets \mathit{cache} \cup \{q\}$\;
      \Return $\mathit{approx}[q]$\;
    }
    \Repeat{$\lnot\,\mathit{changed}$}{
      $\mathit{cache} \gets \emptyset$;\; $\mathit{stack} \gets \emptyset$;\; $\mathit{changed} \gets \mathbf{false}$\;
      \RESOLVE{$q_0$}\;
    }
    \Return $\mathit{approx}[q_0]$\;
  }
  \caption{\bilingual{Tabled evaluation of equations~(\ref{eq:supers})--(\ref{eq:this}). \textsf{Tabled} iterates from the empty accumulator to the least fixed point: each round, \textsf{Resolve}, the memoized recomputation, evaluates one round of every reachable query's equation, reading sub-query facts through itself, and merges what it produces into the query's answer set; a re-entrant call, a cycle in the dependency graph, returns the running approximation; the merge is set union, the join of the powerset lattice. The memo is keyed by the query itself.}{方程~(\ref{eq:supers})--(\ref{eq:this}) 的 tabled 求值。\textsf{Tabled} 自空累加器迭代至最小不动点:每一轮,\textsf{Resolve}(带备忘的重算)对每个可达查询的方程求值一轮,经由自身读取子查询事实,并把产出并入该查询的答案集;重入调用,即依赖图中的环,返回运行中的逼近;合并为集合并,即幂集格上的并。备忘以查询自身为键。}}
  \label{alg:tabled}
\end{algorithm}

\ifPaperChinese
MIXINv2 基于 tablambda 项目的 \textsf{fixpoints} tabling
库~\cite{yang2026-tablambda}实现算法~\ref{alg:tabled}。
\else
MIXINv2 implements Algorithm~\ref{alg:tabled} on the
\textsf{fixpoints} tabling library of the tablambda
project~\cite{yang2026-tablambda}.
\fi

\begin{theorem}[\bilingual{Tabling soundness}{Tabling 可靠性}]%
  \label{thm:tabling-soundness}
  \ifPaperChinese
  方程~(\ref{eq:supers})--(\ref{eq:this}) 的 tabled 求值
  (算法~\ref{alg:tabled}),
  即使依赖图~$G$ 是无穷的,也绝不会返回错误答案。具体而言:
  \begin{enumerate}
    \item \textbf{可靠性(无假正例)。}
      tabling 返回的每个事实均由某条方程的右侧在其输入上求值产生,
      且每个返回的事实都属于 $\lfp(T_P)$。

    \item \textbf{单调积累。}
      tabling 累加器在各轮迭代之间只通过集合并增长,任何事实都不会被删除。

    \item \textbf{不动点检测的正确性。}
      若求值过程中没有子查询被重入,则所有子查询均在无近似的情况下被完整计算,结果是精确的。
      若发生了重入但累加器是稳定的(即没有新事实被添加),则累加器等于
      $\lfp(T_P)$ 在可达查询上的限制,因为它是一个可靠的前不动点,而最小不动点即最小前不动点。

    \item \textbf{可靠的失败。}
      当迭代不稳定时,算法~\ref{alg:tabled} 不返回;它绝不返回错误答案。
  \end{enumerate}
  \else
  Tabled evaluation of
  equations~(\ref{eq:supers})--(\ref{eq:this})
  (Algorithm~\ref{alg:tabled})
  never returns a wrong answer, even when the
  dependency graph~$G$ is infinite.
  Specifically:
  \begin{enumerate}
    \item \textbf{Soundness (no false positives).}
      Every fact returned by tabling was produced by
      evaluating some equation's right-hand side on its
      inputs, and every returned fact belongs to
      $\lfp(T_P)$.

    \item \textbf{Monotone accumulation.}
      The tabling accumulator only grows between
      iterations via set union; no fact is ever
      removed.

    \item \textbf{Correct fixpoint detection.}
      If no sub-query was re-entered during evaluation,
      all sub-queries were fully computed with no
      approximation, and the result is exact.
      If re-entry occurred but the accumulator is
      stable, meaning no new facts were added, the
      accumulator equals
      $\lfp(T_P)$ restricted to the reachable
      queries, since it is a sound pre-fixed point and
      the least fixed point is the least pre-fixed point.

    \item \textbf{Sound failure.}
      When the iteration does not stabilize,
      Algorithm~\ref{alg:tabled} does not return; it never
      returns a wrong answer.
  \end{enumerate}
  \fi
\end{theorem}

\begin{proof}
\ifPaperChinese
  对于~(1),对算法~\ref{alg:tabled} 中向 $\mathit{approx}$ 的赋值序列作归纳。
  初始时每个答案集为空,包含于 $\lfp(T_P)$。每次赋值把 $q$ 的方程诸实例的产出
  并入 $\mathit{approx}[q]$,而这些实例读取的子查询事实来自 $\mathit{approx}$
  更早的状态,由归纳假设均属 $\lfp(T_P)$;由单调性(引理~\ref{lem:monotonicity})
  与 $T_P(\lfp(T_P)) = \lfp(T_P)$,产出亦属 $\lfp(T_P)$。
  故重入调用所读的每个事实、以及每个返回的事实,都属于 $\lfp(T_P)$。

  对于~(2),实现通过 $A \leftarrow A \cup T_P(A)$ 进行积累,这对~$A$ 是单调的。

  对于~(3),当没有重入时,求值是有限无环可达子图上的普通记忆化递归,
  因此每个子查询都被精确计算,不涉及任何近似。
  当发生重入时,设 $R \subseteq \mathcal{Q}$ 为从初始查询可达的查询集合,
  设 $T_P^R = T_P{\restriction}R$ 表示 $T_P$ 到~$R$ 上映射的限制。
  $R$ 对 $G$ 的依赖边封闭(可达查询所读的每个子查询本身也可达),
  故由依赖限制引理(引理~\ref{lem:dependency-restriction}),
  $\lfp(T_P^R)$ 等于 $\lfp(T_P)$ 在 $R$ 上的限制。
  若累加器稳定,即 $A = A \cup T_P^R(A)$,则
  $T_P^R(A) \subseteq A$,从而~$A$ 是~$T_P^R$ 的一个前不动点。

  由~(1),曾被添加到累加器中的每个事实均由某次右侧求值产生,
  因此属于 $\lfp(T_P^R)$。
  故 $A \subseteq \lfp(T_P^R)$。

  反之,由 Knaster--Tarski 将最小不动点刻画为最小前不动点,
  $\lfp(T_P^R)$ 包含于~$T_P^R$ 的每个前不动点中,特别地包含于~$A$ 中。
  从而 $\lfp(T_P^R) \subseteq A$。

  因此 $A = \lfp(T_P^R)$,即稳定的累加器等于可达查询上的精确指称。

  对于~(4),算法~\ref{alg:tabled} 的循环仅在某一轮没有任何重算产出新事实后退出,
  故不稳定的运行不返回;结合~(1),返回的答案不会是错的。实现可以另行给轮数设上限、
  并在越界时抛出错误;那属于制品的行为,不属于本定理。
\else
  For~(1), induction on the sequence of assignments to
  $\mathit{approx}$ in Algorithm~\ref{alg:tabled}. Initially every
  answer set is empty, so contained in $\lfp(T_P)$. Each assignment
  merges into $\mathit{approx}[q]$ the values produced by instances
  of $q$'s equation whose sub-query facts were read from earlier
  states of $\mathit{approx}$, all in $\lfp(T_P)$ by the induction
  hypothesis; by monotonicity (Lemma~\ref{lem:monotonicity}) and
  $T_P(\lfp(T_P)) = \lfp(T_P)$, the produced values lie in
  $\lfp(T_P)$ as well. Hence every fact a re-entrant call reads, and
  every fact returned, lies in $\lfp(T_P)$.

  For~(2), the implementation accumulates via
  $A \leftarrow A \cup T_P(A)$, which is monotone in~$A$.

  For~(3), when no reentry occurs, the evaluation is
  ordinary memoised recursion on a finite acyclic
  reachable subgraph, so every sub-query is computed
  exactly and no approximation is involved.
  When reentry occurs, let $R \subseteq \mathcal{Q}$ be
  the set of queries reachable from the initial query,
  and let $T_P^R = T_P{\restriction}R$ denote the
  restriction of $T_P$ to maps on~$R$.
  $R$ is closed under the dependency edges of $G$
  (every sub-query a reachable query reads is itself
  reachable), so by the Dependency Restriction Lemma
  (Lemma~\ref{lem:dependency-restriction}) $\lfp(T_P^R)$
  equals $\lfp(T_P)$ restricted to $R$.
  If the accumulator is stable, that is,
  $A = A \cup T_P^R(A)$, then
  $T_P^R(A) \subseteq A$, so $A$ is a pre-fixed point
  of~$T_P^R$.

  By~(1), every fact ever added to the accumulator was
  produced by some right-hand side evaluation and
  therefore belongs to $\lfp(T_P^R)$.
  Hence $A \subseteq \lfp(T_P^R)$.

  Conversely, by the Knaster--Tarski characterization
  of the least fixed point as the least pre-fixed
  point, $\lfp(T_P^R)$ is contained in every
  pre-fixed point of~$T_P^R$, in particular in~$A$.
  Thus $\lfp(T_P^R) \subseteq A$.

  Therefore $A = \lfp(T_P^R)$, that is, the stable
  accumulator equals the exact denotation on the
  reachable queries.

  For~(4), the loop of Algorithm~\ref{alg:tabled} exits only after a
  round in which no recomputation produced a new fact, so a
  non-stabilizing run does not return; with~(1), no returned answer
  is wrong. An implementation may additionally bound the number of
  rounds and raise an error at the bound; that behavior belongs to
  the artifact, not to this theorem.
\fi
\end{proof}

\paragraph{\bilingual{Termination}{终止性}}
\ifPaperChinese
求值是否终止取决于程序:
\begin{itemize}
  \item \textbf{无环、有限可达子图。}
    一趟即足;求值终止。
  \item \textbf{有环、有限可达子图。}
    当可达查询只有有限多个且每个的答案域有限时,格是有限的,
    迭代在有限多趟内到达不动点~\cite{bancilhon1986-recursive-query-strategies}。
    这是 Datalog 的情形(第~\ref{sec:recursive-rules} 节)。
  \item \textbf{无穷可达子图。}
    求值可能发散;这是图灵完备性的固有属性,而非求值策略的缺陷。
\end{itemize}
\else
Whether evaluation terminates depends on the
program:
\begin{itemize}
  \item \textbf{Acyclic, finite reachable subgraph.}
    One pass suffices; evaluation terminates.
  \item \textbf{Cyclic, finite reachable subgraph.}
    When finitely many queries are reachable and each has
    a finite answer domain, the lattice is finite and
    iteration reaches the fixed point in finitely many
    passes~\cite{bancilhon1986-recursive-query-strategies}.
    This is the Datalog case
    (Section~\ref{sec:recursive-rules}).
  \item \textbf{Infinite reachable subgraph.}
    Evaluation may diverge; this is inherent to
    Turing completeness and not a defect of the
    evaluation strategy.
\end{itemize}
\fi

\paragraph{\bilingual{Mechanization}{机械化}}
\ifPaperChinese
我们用 Claude Code 在 Lean~4 中,基于抽象原语 $\inherits$,机械化了本附录的结果;产生 $\inherits$ 的语法层未被机械化。该开发\anon[随补充材料一同提供]{是 MIXINv2 源码~\cite{yang2026-mixinv2} 的一部分},位于 \texttt{lean/} 目录下,构建说明见 \texttt{lean/README.md};表~\ref{tab:lean-correspondence} 把每条陈述映射到其 Lean 名称。同一开发还覆盖附录~\ref{app:merge-algebra} 的结合律、交换律与幂等律定理、对程序的单调性、向上封闭的观测、改指模拟,以及由它们得到的、在单个文档内的定理~\ref{thm:inheritance-substitutability}~(ii)(以 $p_T$ 非根为假设)与推论~\ref{cor:behavior-compatibility} 的包含;读取语法层的引理~\ref{lem:composite-extension} 与引理~\ref{lem:nonroot-targets}、命题~\ref{prop:upward-closed-observation} 关于组合体与文本编辑的末句、定理~\ref{thm:inheritance-substitutability} 的组合体形式,以及推论~\ref{cor:fragile-base-class},是纸面证明。该开发还包含一个发证的 tabled 解释器;该解释器发出的证书在构建时由一个已证可靠的检查器验证。每个文件都把警告提升为错误,且构建把表~\ref{tab:lean-correspondence} 中每条定理的公理钉死为 \texttt{propext}、\texttt{Classical.choice} 与 \texttt{Quot.sound} 的一个子集,因此证明中任何位置的 \texttt{sorry} 占位符都会使构建失败。附录~\ref{app:levy-longo-tree-proofs} 及其后诸附录的证明仅为纸面证明。
\else
We mechanize the results of this appendix in Lean~4 with Claude Code, over the abstract primitive $\inherits$; the syntax layer that produces $\inherits$ is not mechanized. The development ships\anon[~with the supplementary material]{~with the MIXINv2 source~\cite{yang2026-mixinv2}} in the directory \texttt{lean/}, with build instructions in \texttt{lean/README.md}; Table~\ref{tab:lean-correspondence} maps each statement to its Lean name. The same development also covers the associativity, commutativity, and idempotence theorems of Appendix~\ref{app:merge-algebra}, the monotonicity in the program, the upward-closed observation, and the redirection simulation, with the containments that follow from them: Theorem~\ref{thm:inheritance-substitutability}~(ii) within one document, assuming $p_T$ nonroot, and those of Corollary~\ref{cor:behavior-compatibility}; Lemmas~\ref{lem:composite-extension} and~\ref{lem:nonroot-targets}, which read the syntax layer, the closing sentence of Proposition~\ref{prop:upward-closed-observation} on composites and edits, the composite form of Theorem~\ref{thm:inheritance-substitutability}, and Corollary~\ref{cor:fragile-base-class} are proved on paper. The development also contains a certifying tabled interpreter whose emitted certificates are validated at build time by a checker proved sound. Every file elevates warnings to errors, and the build pins the axioms of every theorem in Table~\ref{tab:lean-correspondence} to a subset of \texttt{propext}, \texttt{Classical.choice}, and \texttt{Quot.sound}, so a \texttt{sorry} placeholder anywhere in a proof fails the build. The proofs of Appendix~\ref{app:levy-longo-tree-proofs} and the later appendices are on paper only.
\fi

\begin{table}[h]
  \centering
  \footnotesize
  \begin{tabular}{lll}
    \textbf{\bilingual{Statement}{陈述}} & \textbf{\bilingual{Lean name}{Lean 名称}} & \textbf{\bilingual{File}{文件}} \\
    \hline
    \bilingual{Equations}{方程}~(\ref{eq:supers})--(\ref{eq:this}) & \texttt{supersSet\_nil}, \ldots, \texttt{thisSet\_union} & \texttt{Inductive.lean} \\
    \bilingual{Definition}{定义}~\ref{def:tp-operator} & \texttt{TP} & \texttt{Lattice.lean} \\
    \bilingual{Lemma}{引理}~\ref{lem:monotonicity} & \texttt{TPfun\_mono} & \texttt{Lattice.lean} \\
    \bilingual{Theorem}{定理}~\ref{thm:lfp-exists} & \texttt{TP\_map\_lfp}, \texttt{TP\_isLeast\_lfp} & \texttt{Lattice.lean} \\
    \bilingual{Theorem}{定理}~\ref{thm:well-defined} & \texttt{lfp\_TP\_eq\_toQAMap} & \texttt{WellDefinedness.lean} \\
    \bilingual{Lemma}{引理}~\ref{lem:kleene-iteration} & \texttt{lfp\_TP\_eq\_iSup\_stage}, \texttt{exists\_kleene\_rank} & \texttt{Kleene.lean} \\
    \bilingual{Lemma}{引理}~\ref{lem:path-existence} & \texttt{supersSet\_nonempty\_iff\_existent} & \texttt{PathExistence.lean} \\
    \bilingual{Lemma}{引理}~\ref{lem:dependency-restriction} & \texttt{restrict\_lfp\_eq} & \texttt{Restriction.lean} \\
    \bilingual{Theorem}{定理}~\ref{thm:tabling-soundness} & \texttt{tabled\_le\_lfp}, \texttt{tabled\_stable\_eq\_lfp} & \texttt{AlgorithmSoundness.lean} \\
    \bilingual{Theorem}{定理}~\ref{thm:merge-associative} & \texttt{AssociativityDocument.merge\_associative} & \texttt{MergeAssociativity.lean} \\
    \bilingual{Theorem}{定理}~\ref{thm:merge-commutative} & \texttt{CommutativityDocument.merge\_commutative} & \texttt{MergeCommutativity.lean} \\
    \bilingual{Theorem}{定理}~\ref{thm:merge-idempotent} & \texttt{IdempotenceDocument.merge\_idempotent} & \texttt{MergeIdempotence.lean} \\
    \bilingual{Lemma}{引理}~\ref{lem:program-monotonicity} & \texttt{Extends.fact} & \texttt{ProgramMonotonicity.lean} \\
    \bilingual{Proposition}{命题}~\ref{prop:upward-closed-observation} & \texttt{fact\_upward\_closed} & \texttt{ProgramMonotonicity.lean} \\
    \bilingual{Lemma}{引理}~\ref{lem:redirection-simulation} & \texttt{Redirection.simulation} & \texttt{Substitutability.lean} \\
    \bilingual{Theorem}{定理}~\ref{thm:inheritance-substitutability}~(b2), \bilingual{one document}{单个文档} & \texttt{inheritance\_substitutable} & \texttt{Substitutability.lean} \\
    \bilingual{Corollary}{推论}~\ref{cor:behavior-compatibility} & \texttt{Redirection.supers\_written\_preserved}, \ldots & \texttt{Substitutability.lean} \\
    \bilingual{Theorem}{定理}~\ref{thm:decomposition} & \texttt{Decomposition.decomposition} & \texttt{Decomposition.lean} \\
    \bilingual{Certificate checker soundness}{证书检查器的可靠性} & \texttt{checkCert\_sound} & \texttt{Certificate.lean} \\
  \end{tabular}
  \caption{\bilingual{The Lean names of the mechanized statements; the files are
    under \texttt{lean/InheritanceCalculus/} in the supplementary
  material.}{被机械化陈述的 Lean 名称;文件位于补充材料的 \texttt{lean/InheritanceCalculus/} 之下。}}
  \label{tab:lean-correspondence}
\end{table}

\section{\bilingual{Algebraic Properties of Deep Merge}{深度合并的代数性质}}
\label{app:merge-algebra}

\ifPaperChinese
本附录证明第~\ref{sec:modularity}~节关于组合、可扩展性、可替换性与分解的结果;覆盖由附录~\ref{app:well-definedness} 证明。
我们还在 Lean~4 中证明了结合律、交换律、幂等律与分解四条定理(表~\ref{tab:lean-correspondence})。
\else
This appendix proves the results of Section~\ref{sec:modularity} on composition, extensibility,
substitutability, and decomposition; Appendix~\ref{app:well-definedness} proves coverage.
We also prove the associativity, commutativity, idempotence, and decomposition theorems in Lean~4
(Table~\ref{tab:lean-correspondence}).
\fi


\begin{proof}[\bilingual{Proof of Theorem~\ref{thm:merge-associative}}{定理~\ref{thm:merge-associative} 的证明}]
\ifPaperChinese
  全程只用五条方程直接计算,配对 Kleene 秩的强归纳与末步反演
  (引理~\ref{lem:kleene-iteration})。

  \paragraph{\bilingual{}{基链展平}}
  记保留根 $p_{\mathrm{reserved}} = (\text{``\_\_merge''})$。在顶层,
  $\supers(()) = \{((),\, ())\}$ 是方程~(\ref{eq:supers})的基例,故
  $\overrides(p_{\mathrm{reserved}}) = \{p_{\mathrm{reserved}}\}$
  (唯一分支为根);保留 mixin 只含定义、不含引用,
  $\bases(p_{\mathrm{reserved}}) = \varnothing$,于是
  $\supers(p_{\mathrm{reserved}}) = \{((),\, p_{\mathrm{reserved}})\}$。
  对每个包装标签 $X$,
  $\overrides(p_{\mathrm{reserved}} \snoc X) =
  \{p_{\mathrm{reserved}} \snoc X\}$:唯一分支是 $p_{\mathrm{reserved}}$,
  两个 override 候选重合。包装引用的解析:
  \mintinline{yaml}{[LeftPair]} 型引用的首标签命名于同一保留作用域,
  de~Bruijn 索引 $n = 0$,目标
  $p_{\mathrm{reserved}} \snoc \text{``LeftPair''}$;
  \mintinline[escapeinside=||]{yaml}{[|$p_i$|]} 型引用的首标签定义于顶层,
  索引 $n = 1$,其一步走步恰消费共享对 $((),\, p_{\mathrm{reserved}})$,
  目标 $p_i$。于是
  $\bases(p_{\mathrm{reserved}} \snoc \text{``mergedLeft''}) =
  \{p_{\mathrm{reserved}} \snoc \text{``LeftPair''},\; p_3\}$、
  $\bases(p_{\mathrm{reserved}} \snoc \text{``LeftPair''}) =
  \{p_1,\, p_2\}$,故三条合并路径的 $\bases^*$ 都含三个基
  $p_1$、$p_2$、$p_3$;三者在三个基之外的差异只在包装自身:包装无自有定义,
  其自反项对任何标签 $\ell$ 的守卫 $p \snoc \ell \in \dom(\inherits)$
  均不满足,故包装对 $\supers$ 的书写 override 分量没有贡献。

  \paragraph{\bilingual{}{结构事实}}
  称保留字下的五个包装路径 $p_{\mathrm{reserved}} \snoc X$ 为\emph{合成根},
  合成根的延伸(含合成根自身,后缀 $w$ 可空)为\emph{合成路径};保留根 $p_{\mathrm{reserved}}$ 与上段算出的
  其诸事实为各合并路径共享;其余路径(书写于三个基的子树内或文档他处)为
  \emph{固定路径}。三条结构事实,各由一次方程计算得出:
  (i) $\inherits$ 在每个合成根处有定义(其书写的引用表),而在合成根的每个
  真延伸 $\sigma \cat w$($|w| \ge 1$)处无定义:包装体只含引用、不含定义,
  故 $|w| = 1$ 时 $\has$ 查询失败,而 $|w| \ge 2$ 时 $\at$ 在
  $\sigma \cat w$ 的某个前缀处已无定义(其下未写任何节点);两种情形都给出
  $\sigma \cat w \notin \dom(\inherits)$。于是合成
  override 不向方程~(\ref{eq:bases})贡献引用,而 $\sigma \cat w \snoc \ell$
  处的自反 override 是否被产出与某书写分支是否合格同判
  (方程~\ref{eq:overrides})。
  (ii) 地面链即上段:三条合并路径的地面链都恰好抵达 $p_1$、$p_2$、$p_3$,
  配对包装至多多出一个不贡献书写 override 的中间链节点。
  (iii) 走步与字面解析的隔离,分三点。其一,每次被消费的 $\this$ 走步,
  其匹配的 override 分量 $p_{\mathrm{def}}$ 是固定路径或保留根,绝非合成
  路径:引用只书写于固定位置与合成根自身,合成根的
  \mintinline{yaml}{[LeftPair]} 型引用索引为 0、无走步,其余引用的走步以
  保留根为匹配目标。其二,走步从不输出保留根为位点:位点分量为
  $p_{\mathrm{reserved}}$ 的 $\supers$ 对只能来自链节点
  $p_{\mathrm{base}} = p_{\mathrm{reserved}} \snoc X$,其 override 分量是
  合成自反项,不等于其一中的任何匹配目标;而共享对
  $((),\, p_{\mathrm{reserved}})$ 的位点是 $()$。其三,固定目标从不落入
  保留字子树:保留根路径长度为 1,$p_{\mathrm{current}} \cat w_{\mathrm{ref}}$
  以之为前缀,要么 $p_{\mathrm{current}}$ 自身以之为前缀(合成一侧,链的
  传承),要么 $p_{\mathrm{current}} = ()$ 且 $w_{\mathrm{ref}}$ 首标签为
  保留字,而保留字不出现于任何引用;由其二,
  $p_{\mathrm{current}} = p_{\mathrm{reserved}}$ 从不发生。故合成路径在
  任何事实中的每次出现都来自链的传承,而非字面解析。

  \paragraph{\bilingual{}{对 Kleene 秩的模拟}}
  固定有序对 $(R, R')$,取值于三条合并路径;其\emph{合成根集}为其地面链上
  的合成根:mergedLeft 侧取
  $\{p_{\mathrm{reserved}} \snoc \text{``mergedLeft''},\;
  p_{\mathrm{reserved}} \snoc \text{``LeftPair''}\}$,
  mergedRight 侧取
  $\{p_{\mathrm{reserved}} \snoc \text{``mergedRight''},\;
  p_{\mathrm{reserved}} \snoc \text{``RightPair''}\}$,
  mergedFlat 侧取
  $\{p_{\mathrm{reserved}} \snoc \text{``mergedFlat''}\}$。定义映射
  $\varphi$:对合成根集中的 $\sigma$,$\varphi(\sigma \cat w) = R' \cat w$;
  固定路径与保留根不动。由 (iii),自 $R \cat w$ 出发的推导只触及合成根集的
  合成路径、保留根与固定路径,故 $\varphi$ 对事实的改写良定义;且 $\varphi$
  保持后缀拼接与 $\init$。注意 $\varphi$ 并非单射:例如 mergedLeft 侧的两个
  合成根都映到同一 $R'$,尽管配对包装的地面可达集($\{p_1, p_2\}$)小于
  合并包装的(全部三基)。这一坍缩是可靠的,因为下面四个条款都只是\emph{单向}
  蕴涵,各自只断言某事实族在 $R'$ 侧\emph{可导出};而由 (ii),$R'$ 作为顶层
  合并包装,其地面链抵达全部三个基,是每个合成根地面可达集的超集,且诸方程的
  右侧全为正存在式推导,可达集扩大只会增加可导出事实,故经较窄中间节点产出的
  每条事实在 $R'$ 侧都有同样的分支可用。固定路径处的事实其推导不触及合成路径(其链与走步
  只经固定位置与保留根,(iii)),各合并路径之间字面共享。对诸事实的最大 Kleene 秩 $m$ 作强归纳
  (降秩全部来自末步反演,引理~\ref{lem:kleene-iteration}),同时证四个条款:

  (C1) 秩 $\le m$ 的 $v \in \overrides(\sigma \cat w)$ 蕴涵
  $\varphi(v) \in \overrides(R' \cat w)$ 可导出。$w = ()$:两侧均为经
  共享保留根分支产出的自反单点集,$\varphi(\sigma) = R'$。
  $w \snoc \ell$:由 (i) 合格分支
  全为书写路径 $p_{\mathrm{branch}}$;反演暴露的对
  $(\_,\; p_{\mathrm{branch}}) \in \supers(\sigma \cat w)$ 秩更小,
  (C3) 给出 $R'$ 侧同一书写分支,产出同一书写 override
  $p_{\mathrm{branch}} \snoc \ell$ 与自反 override
  $R' \cat w \snoc \ell = \varphi(\sigma \cat w \snoc \ell)$。

  (C2) 每条自 $\sigma \cat w$ 出发、逐跳事实秩 $\le m$ 的 $\bases^*$ 链
  抵达 $q$,蕴涵 $R'$ 侧存在自 $R' \cat w$ 抵达 $\varphi(q)$ 的链。
  对链长作内层归纳,逐跳映射后拼接。地面跳($w = ()$)由 (ii) 显式核对:
  两侧地面链均恰抵三个基,像中源与目标重合的退化跳删去,一跳可展为经
  配对包装的两跳。深处跳($|w| \ge 1$):由 (i),跳来自某书写 override
  携带的引用在位点 $\sigma \cat w$ 的解析;反演得走步事实,秩更小,
  (C4) 给出同一引用在位点 $R' \cat w$ 的走步,目标
  $\varphi(p_{\mathrm{current}}) \cat w_{\mathrm{ref}}
  = \varphi(p_{\mathrm{current}} \cat w_{\mathrm{ref}})$
  ($\varphi$ 保持后缀拼接;固定的 $p_{\mathrm{current}}$ 的拼接不落入
  保留字子树,结构事实 (iii))。

  (C3) 秩 $\le m$ 的 $(s,\; v) \in \supers(\sigma \cat w)$ 蕴涵
  $(\varphi(s),\; \varphi(v)) \in \supers(R' \cat w)$ 可导出。反演得链
  $p_{\mathrm{base}} \in \bases^*(\sigma \cat w)$ 与
  $v \in \overrides(p_{\mathrm{base}})$,秩均更小;(C2) 给出 $R'$ 侧抵达
  $\varphi(p_{\mathrm{base}})$ 的链,(C1)(合成)或字面共享(固定)给出
  $\varphi(v) \in \overrides(\varphi(p_{\mathrm{base}}))$;位点分量
  $\varphi(\init(p_{\mathrm{base}})) = \init(\varphi(p_{\mathrm{base}}))$。

  (C4) 秩 $\le m$ 的走步事实 $\this(\{s\},\; p_{\mathrm{def}},\; n)$ 产出
  位点 $t$,蕴涵 $\this(\{\varphi(s)\},\; p_{\mathrm{def}},\; n)$ 产出
  $\varphi(t)$。$n = 0$ 即 $\varphi$ 逐点。$n > 0$:反演得被消费的对
  $(s',\; p_{\mathrm{def}}) \in \supers(s)$ 与余下 $n - 1$ 步,秩均更小;
  由 (iii) $p_{\mathrm{def}}$ 为固定路径或保留根、$\varphi$ 不动之;
  合成 $s$ 用 (C3)、固定或保留 $s$ 用字面共享得 $R'$ 侧的对
  $(\varphi(s'),\; p_{\mathrm{def}})$,余下走步用归纳假设。

  \paragraph{\bilingual{}{收尾}}
  对每个 $w$,把 (C3) 分别用于有序对 $(R, R')$ 与 $(R', R)$:保留字子树外
  的路径在 $\varphi$ 下不动,而合成路径的像仍在子树内,故
  $\supers(R \cat w)$ 与 $\supers(R' \cat w)$ 的 override 分量中子树外
  路径的集合互相包含,即相同——陈述第一部分。存在性部分:由路径存在性引理
  (引理~\ref{lem:path-existence}),$R \cat w \snoc \ell$ 的存在性等价于
  $\supers(R \cat w)$ 的某分支满足
  $p_{\mathrm{branch}} \snoc \ell \in \dom(\inherits)$;合格分支绝非合成
  路径(包装体不定义任何标签,合成路径的任何延伸都不在 $\dom(\inherits)$,
  结构事实 (i)),故该判据只读 $R \cat w$ 处子树外的 override 集,
  而后者三边相同。
\else
  Throughout, we compute directly with the five equations, combined
  with a strong induction on Kleene rank and last-step inversion
  (Lemma~\ref{lem:kleene-iteration}).

  \paragraph{Chain flattening}
  Write $p_{\mathrm{reserved}} = (\text{``\_\_merge''})$ for the
  reserved root. At the top level,
  $\supers(()) = \{((),\, ())\}$ is the base case of
  equation~(\ref{eq:supers}), so
  $\overrides(p_{\mathrm{reserved}}) = \{p_{\mathrm{reserved}}\}$
  (the only branch is the root); the reserved mixin carries
  definitions and no references, so
  $\bases(p_{\mathrm{reserved}}) = \varnothing$ and
  $\supers(p_{\mathrm{reserved}}) =
  \{((),\, p_{\mathrm{reserved}})\}$. For each wrapper label $X$,
  $\overrides(p_{\mathrm{reserved}} \snoc X) =
  \{p_{\mathrm{reserved}} \snoc X\}$: the only branch is
  $p_{\mathrm{reserved}}$, and both override candidates coincide.
  Resolving the wrapper references: a
  \mintinline{yaml}{[LeftPair]}-style reference names a label of the
  same reserved scope, with de~Bruijn index $n = 0$ and target
  $p_{\mathrm{reserved}} \snoc \text{``LeftPair''}$; a
  \mintinline[escapeinside=||]{yaml}{[|$p_i$|]}-style reference
  names a label defined at the top level, with index $n = 1$, and
  its one walk step consumes exactly the shared pair
  $((),\, p_{\mathrm{reserved}})$, with target $p_i$. Hence
  $\bases(p_{\mathrm{reserved}} \snoc \text{``mergedLeft''}) =
  \{p_{\mathrm{reserved}} \snoc \text{``LeftPair''},\; p_3\}$ and
  $\bases(p_{\mathrm{reserved}} \snoc \text{``LeftPair''}) =
  \{p_1,\, p_2\}$, so the $\bases^*$ of each of the three merged
  paths contains the three bases $p_1$, $p_2$, $p_3$. Beyond the
  three bases the chains differ only in the wrappers themselves; a
  wrapper has no own definitions, so for any label $\ell$ the
  reflexive term's guard $p \snoc \ell \in \dom(\inherits)$ fails,
  and wrappers contribute nothing to the written override
  components of $\supers$.

  \paragraph{Structural facts}
  Call the five wrapper paths $p_{\mathrm{reserved}} \snoc X$ under
  the reserved word \emph{synthetic roots} and their extensions
  (including the roots themselves, with an empty suffix $w$)
  \emph{synthetic paths}; the reserved root $p_{\mathrm{reserved}}$
  and its facts computed above are shared by the three merged paths;
  every other path (written inside the subtrees of the three bases
  or elsewhere in the document) is \emph{fixed}. Three structural
  facts, each one equation computation:
  (i) $\inherits$ is defined at each synthetic root (its written
  reference list) and undefined at every proper extension
  $\sigma \cat w$ with $|w| \ge 1$: a wrapper body carries references
  and no definitions, so at $|w| = 1$ the $\has$ query fails, and at
  $|w| \ge 2$ the query $\at$ is already undefined at some prefix of
  $\sigma \cat w$ (nothing is written below); either way
  $\sigma \cat w \notin \dom(\inherits)$. Hence a synthetic override
  contributes no reference to equation~(\ref{eq:bases}), and the
  reflexive override at $\sigma \cat w \snoc \ell$ is produced
  exactly when some written branch qualifies
  (equation~\ref{eq:overrides}).
  (ii) The ground chains are the previous paragraph: the ground
  chains of each of the three merged paths reach exactly $p_1$,
  $p_2$, $p_3$, with the pair wrapper contributing at most one
  intermediate chain node that carries no written override.
  (iii) Walks and literal resolution are insulated, in three parts.
  First, every $\this$ walk consumed anywhere in the document
  matches pairs whose override component $p_{\mathrm{def}}$ is a
  fixed path or the reserved root, never a synthetic path:
  references are written only at fixed positions and at the
  synthetic roots themselves, a synthetic root's
  \mintinline{yaml}{[LeftPair]}-style references have index~0 and
  take no walk step, and the walks of its remaining references
  match the reserved root. Second, a walk never outputs the
  reserved root as a site: a $\supers$ pair whose site component is
  $p_{\mathrm{reserved}}$ can only come from a chain node
  $p_{\mathrm{base}} = p_{\mathrm{reserved}} \snoc X$, whose
  override component is the synthetic reflexive term, equal to none
  of the matched targets of the first part; and the shared pair
  $((),\, p_{\mathrm{reserved}})$ has site $()$. Third, a fixed
  target never lands inside the reserved-word subtree: the reserved
  root has length one, so
  $p_{\mathrm{current}} \cat w_{\mathrm{ref}}$ has it as a prefix
  only when $p_{\mathrm{current}}$ itself does (the synthetic side,
  chain provenance) or when $p_{\mathrm{current}} = ()$ and the
  first label of $w_{\mathrm{ref}}$ is the reserved word, which is
  written in no reference; and
  $p_{\mathrm{current}} = p_{\mathrm{reserved}}$ never happens, by
  the second part. Hence every occurrence of a synthetic path in
  any fact arises through chain provenance, never through literal
  resolution.

  \paragraph{Simulation by Kleene rank}
  Fix an ordered pair $(R, R')$ ranging over the three merged paths
  and take the \emph{synthetic-root set} of $R$ to be the synthetic
  roots on its ground chains:
  $\{p_{\mathrm{reserved}} \snoc \text{``mergedLeft''},\;
  p_{\mathrm{reserved}} \snoc \text{``LeftPair''}\}$ for the
  mergedLeft side,
  $\{p_{\mathrm{reserved}} \snoc \text{``mergedRight''},\;
  p_{\mathrm{reserved}} \snoc \text{``RightPair''}\}$ for the
  mergedRight side, and
  $\{p_{\mathrm{reserved}} \snoc \text{``mergedFlat''}\}$ for the
  mergedFlat side. Define the map $\varphi$: for $\sigma$ in the
  synthetic-root set, $\varphi(\sigma \cat w) = R' \cat w$; fixed
  paths and the reserved root are unchanged. By~(iii), a derivation
  starting from $R \cat w$ touches only synthetic paths of the
  synthetic-root set, the reserved root, and fixed paths, so
  rewriting a fact along $\varphi$ is well defined; and $\varphi$
  commutes with suffix concatenation and with $\init$. Note that
  $\varphi$ is not injective: both synthetic roots of the mergedLeft
  side, for instance, map onto the same $R'$, although the pair
  wrapper's ground reach ($\{p_1, p_2\}$) is smaller than a merged
  wrapper's (all three bases). This collapse is sound because the
  four clauses below are \emph{one-directional} implications, each
  asserting only that a fact family is \emph{derivable} on the $R'$
  side; and by~(ii) $R'$, a top-level merged wrapper, has ground
  chains reaching all three bases, a superset of every synthetic
  root's ground reach, while the right-hand sides of the equations
  are positive existential comprehensions, so enlarging the reach
  only adds derivable facts: every fact produced through a narrower
  intermediate node has the same branches available on the $R'$
  side. Facts at
  fixed paths have derivations touching no synthetic path (their
  chains and walks pass only through fixed positions and the
  reserved root, by~(iii)) and are literally shared by the three
  merged paths. By strong induction on the
  maximum Kleene rank $m$ of the given facts (every rank decrease
  comes from last-step inversion, Lemma~\ref{lem:kleene-iteration}),
  we prove four clauses simultaneously:

  (C1) $v \in \overrides(\sigma \cat w)$ of rank $\le m$ implies
  $\varphi(v) \in \overrides(R' \cat w)$ is derivable. At $w = ()$
  both sides are the reflexive singletons produced through the
  shared reserved-root branch, and $\varphi(\sigma) = R'$. At
  $w \snoc \ell$: by~(i) the
  qualifying branches are all written paths $p_{\mathrm{branch}}$;
  the pair $(\_,\; p_{\mathrm{branch}}) \in \supers(\sigma \cat w)$
  exposed by inversion has smaller rank, (C3) yields the same
  written branch on the $R'$ side, producing the same written
  override $p_{\mathrm{branch}} \snoc \ell$ and the reflexive
  override $R' \cat w \snoc \ell = \varphi(\sigma \cat w \snoc \ell)$.

  (C2) Every $\bases^*$ chain from $\sigma \cat w$ reaching $q$
  whose hop facts have rank $\le m$ implies a chain from
  $R' \cat w$ reaching $\varphi(q)$ on the $R'$ side. Inner
  induction on chain length; map each hop and concatenate. Ground
  hops ($w = ()$) are checked explicitly by~(ii): both sides' ground
  chains reach exactly the three bases; a degenerate image hop
  (source and target coincide) is dropped, and one hop may expand to
  two through the pair wrapper. A deeper hop ($|w| \ge 1$): by~(i)
  it arises from resolving a reference carried by some written
  override at the site $\sigma \cat w$; inversion exposes the walk
  facts, of smaller rank, and (C4) yields the walk of the same
  reference at the site $R' \cat w$, with target
  $\varphi(p_{\mathrm{current}}) \cat w_{\mathrm{ref}}
  = \varphi(p_{\mathrm{current}} \cat w_{\mathrm{ref}})$
  ($\varphi$ commutes with suffix concatenation; the concatenation
  of a fixed $p_{\mathrm{current}}$ never lands inside the
  reserved-word subtree, by structural fact~(iii)).

  (C3) $(s,\; v) \in \supers(\sigma \cat w)$ of rank $\le m$ implies
  $(\varphi(s),\; \varphi(v)) \in \supers(R' \cat w)$ is derivable.
  Inversion yields a chain $p_{\mathrm{base}} \in
  \bases^*(\sigma \cat w)$ and $v \in \overrides(p_{\mathrm{base}})$,
  both of smaller rank; (C2) yields the $R'$-side chain reaching
  $\varphi(p_{\mathrm{base}})$, and (C1) (synthetic) or literal
  sharing (fixed) yields
  $\varphi(v) \in \overrides(\varphi(p_{\mathrm{base}}))$; the site
  component satisfies $\varphi(\init(p_{\mathrm{base}})) =
  \init(\varphi(p_{\mathrm{base}}))$.

  (C4) A walk fact $\this(\{s\},\; p_{\mathrm{def}},\; n)$ of rank
  $\le m$ producing the site $t$ implies
  $\this(\{\varphi(s)\},\; p_{\mathrm{def}},\; n)$ produces
  $\varphi(t)$. At $n = 0$ this is $\varphi$ pointwise. At $n > 0$:
  inversion exposes the consumed pair
  $(s',\; p_{\mathrm{def}}) \in \supers(s)$ and the remaining
  $n - 1$ steps, all of smaller rank; by~(iii) $p_{\mathrm{def}}$ is
  a fixed path or the reserved root, unchanged by $\varphi$; (C3)
  (synthetic $s$) or literal sharing (fixed or reserved $s$) yields
  the $R'$-side pair $(\varphi(s'),\; p_{\mathrm{def}})$, and the
  induction hypothesis covers the remaining steps.

  \paragraph{Conclusion}
  For every $w$, apply (C3) to the ordered pairs $(R, R')$ and
  $(R', R)$: paths outside the reserved-word subtree are unchanged
  by $\varphi$, while the image of a synthetic path stays inside
  the subtree, so the sets of override components outside the
  subtree of $\supers(R \cat w)$ and $\supers(R' \cat w)$ include
  each other, hence agree: the first part of the statement. For the
  existence part: by the Path Existence
  Lemma~\ref{lem:path-existence}, the existence of
  $R \cat w \snoc \ell$ is equivalent to some branch of
  $\supers(R \cat w)$ satisfying
  $p_{\mathrm{branch}} \snoc \ell \in \dom(\inherits)$; a
  qualifying branch is never a synthetic path (a wrapper body
  defines no label, so no extension of a synthetic root lies in
  $\dom(\inherits)$, structural fact~(i)), so the criterion reads
  only the override components outside the subtree at $R \cat w$,
  which the three share.
\fi
\end{proof}

\begin{corollary}[\bilingual{Regrouping preserves $\this$
    resolution}{重组保持 $\this$ 解析}]
  \label{cor:merge-this-preservation}
  \ifPaperChinese
  在定理~\ref{thm:merge-associative} 的设定下,令 $\varphi$ 为保持后缀的替换
  \[
    \begin{aligned}
      \varphi((\text{``\_\_merge''}, \text{``mergedLeft''}) \cat w)
      &= \varphi((\text{``\_\_merge''}, \text{``LeftPair''}) \cat w) \\
      &= (\text{``\_\_merge''}, \text{``mergedFlat''}) \cat w,
    \end{aligned}
  \]
  \sloppy
  其余路径不动(mergedRight 侧同理,以 RightPair 代 LeftPair;反方向的替换把
  $(\text{``\_\_merge''}, \text{``mergedFlat''}) \cat w$ 送回
  $(\text{``\_\_merge''}, \text{``mergedLeft''}) \cat w$)。则对 $\lfp(T_P)$
  的每条事实:$(s,\; v) \in \supers(q)$ 蕴涵
  $(\varphi(s),\; \varphi(v)) \in \supers(\varphi(q))$;走步
  $\this(\{s\},\; p_{\mathrm{def}},\; n)$ 产出位点 $t$ 蕴涵
  $\this(\{\varphi(s)\},\; p_{\mathrm{def}},\; n)$ 产出 $\varphi(t)$;
  反方向替换同理。特别地,对每个 $w$,同一书写引用(同一定义点与索引)锚定于
  三个位点 $(\text{``\_\_merge''}, X) \cat w$
  ($X$ 取 mergedLeft、mergedRight、mergedFlat)的走步,其落在保留字子树
  之外的产出三边字面相同:重组不改变脚手架之外的任何解析结果。
  \else
  In the setting of Theorem~\ref{thm:merge-associative}, let
  $\varphi$ be the suffix-preserving replacement
  \[
    \begin{aligned}
      \varphi((\text{``\_\_merge''}, \text{``mergedLeft''}) \cat w)
      &= \varphi((\text{``\_\_merge''}, \text{``LeftPair''}) \cat w) \\
      &= (\text{``\_\_merge''}, \text{``mergedFlat''}) \cat w,
    \end{aligned}
  \]
  fixing every other path (likewise for the mergedRight side, with
  RightPair for LeftPair; the reverse replacement sends
  $(\text{``\_\_merge''}, \text{``mergedFlat''}) \cat w$ back to
  $(\text{``\_\_merge''}, \text{``mergedLeft''}) \cat w$). Then for
  every fact of $\lfp(T_P)$: $(s,\; v) \in \supers(q)$ implies
  $(\varphi(s),\; \varphi(v)) \in \supers(\varphi(q))$; a walk
  $\this(\{s\},\; p_{\mathrm{def}},\; n)$ producing a site $t$
  implies $\this(\{\varphi(s)\},\; p_{\mathrm{def}},\; n)$ produces
  $\varphi(t)$; and likewise for the reverse replacement. In
  particular, for every $w$, the walks of one written reference (one
  definition site and index) anchored at the three sites
  $(\text{``\_\_merge''}, X) \cat w$ for $X$ among mergedLeft,
  mergedRight, mergedFlat produce literally the same sites outside
  the reserved-word subtree: regrouping changes no resolution
  outcome beyond its own scaffolding.
  \fi
\end{corollary}

\begin{proof}
\ifPaperChinese
  这是模拟的条款 (C3) 与 (C4) 在有序对 $(R, R')$ 与 $(R', R)$ 上的逐字读出:
  固定路径处的事实各侧字面共享,合成情形即诸条款本身。解析目标为
  方程~(\ref{eq:resolve})的 $p_{\mathrm{current}} \cat w_{\mathrm{ref}}$,
  而 $\varphi$ 保持后缀拼接;子树外的产出被 $\varphi$ 不动,在三个合并位点处
  两个方向的替换互相拼接,得字面重合。
\else
  This reads off clauses (C3) and (C4) of the simulation at the
  ordered pairs $(R, R')$ and $(R', R)$: facts at fixed paths are
  literally shared, and the synthetic cases are the clauses
  themselves. A resolution target is
  equation~(\ref{eq:resolve})'s
  $p_{\mathrm{current}} \cat w_{\mathrm{ref}}$, and $\varphi$
  commutes with suffix concatenation; a produced site outside the
  reserved-word subtree is fixed by $\varphi$, and at the three
  merged sites the two directions of replacement compose, giving
  literal coincidence.
\fi
\end{proof}

\begin{proof}[\bilingual{Proof of Theorem~\ref{thm:merge-commutative}}{定理~\ref{thm:merge-commutative} 的证明}]
\ifPaperChinese
  证明就是定理~\ref{thm:merge-associative} 的模拟,用于有序对
  (mergedForward, mergedReversed) 与 (mergedReversed, mergedForward),
  合成根集分别为 $\{(\text{``\_\_merge''}, \text{``mergedForward''})\}$ 与
  $\{(\text{``\_\_merge''}, \text{``mergedReversed''})\}$。基链展平对这两个包装
  如彼处成立:每条引用的首标签定义于顶层,索引 $n = 1$,其一步走步恰消费共享对
  $((),\, p_{\mathrm{reserved}})$。方程~\eqref{eq:helpers} 中 $\inherits$ 取值于
  引用之集,两条引用的书写次序在其中不留痕迹,故
  \[
    \bases((\text{``\_\_merge''}, \text{``mergedForward''})) =
    \bases((\text{``\_\_merge''}, \text{``mergedReversed''})) = \{p_1,\, p_2\}.
  \]
  结构事实 (i) 与 (iii) 逐字成立:包装只携带引用、不携带定义,且两个包装都不携带
  索引为 $0$ 的引用,故没有任何引用字面解析进另一侧的合成根集。结构事实 (ii)
  就是上面 bases 的相等:一侧的每一地面跳都是另一侧的地面跳。于是条款
  (C1)--(C4) 由对 Kleene 秩的同一强归纳得出,彼处的收尾给出陈述的两部分。
\else
  The proof is the simulation of Theorem~\ref{thm:merge-associative},
  applied to the ordered pairs (mergedForward, mergedReversed) and
  (mergedReversed, mergedForward), with the synthetic-root sets
  $\{(\text{``\_\_merge''}, \text{``mergedForward''})\}$ and
  $\{(\text{``\_\_merge''}, \text{``mergedReversed''})\}$. The chain
  flattening holds for the two wrappers as there: each reference names a
  label defined at the top level, with index $n = 1$, and its one walk
  step consumes exactly the shared pair $((),\, p_{\mathrm{reserved}})$.
  In equation~\eqref{eq:helpers}, $\inherits$ takes its values in sets of
  references, where the written order of the two references leaves no
  trace, so
  \[
    \bases((\text{``\_\_merge''}, \text{``mergedForward''})) =
    \bases((\text{``\_\_merge''}, \text{``mergedReversed''})) = \{p_1,\, p_2\}.
  \]
  Structural facts (i) and (iii) hold verbatim: the wrappers carry
  references and no definitions, and neither carries a reference of
  index $0$, so no reference resolves literally into the synthetic-root
  set of the other side. Structural fact (ii) is the equality of the
  bases above: every ground hop of one side is a ground hop of the
  other. Clauses (C1)--(C4) then follow by the same strong induction on
  Kleene rank, and the conclusion there gives both parts of the
  statement.
\fi
\end{proof}


\begin{proof}[\bilingual{Proof of Theorem~\ref{thm:merge-idempotent}}{定理~\ref{thm:merge-idempotent} 的证明}]
\ifPaperChinese
  证明把定理~\ref{thm:merge-associative} 的模拟特化到单个包装与固定的
  目标根;Kleene 秩与末步反演的用法同彼处。

  \paragraph{\bilingual{}{地面事实}}
  与定理~\ref{thm:merge-associative} 的基链展平一致:
  \[
    \overrides(p_{\mathrm{reserved}}) = \{p_{\mathrm{reserved}}\},
    \quad
    \supers(p_{\mathrm{reserved}}) = \{((),\, p_{\mathrm{reserved}})\},
    \quad
    \overrides(\sigma) = \{\sigma\}.
  \]
  包装处书写的两条引用是同一引用 $(1, p)$:方程~\eqref{eq:helpers} 中
  $\inherits$ 取值于引用之集,故 $\inherits(\sigma) = \{(1, p)\}$。该引用的
  一步走步恰消费共享对 $((),\, p_{\mathrm{reserved}})$,故
  \[
    \bases(\sigma) = \{p\},
    \quad
    \bases^*(\sigma) = \{\sigma\} \cup \bases^*(p),
    \quad
    \supers(\sigma) = \{(p_{\mathrm{reserved}},\, \sigma)\} \cup \supers(p).
  \]
  又 $p$ 存在,由路径存在性引理(引理~\ref{lem:path-existence})
  $\supers(\init(p))$ 有合格分支,方程~(\ref{eq:overrides})以之产出自反
  override,即 $p \in \overrides(p)$。$w = ()$ 情形即得:多出的一对是合成的,
  被陈述过滤。

  \paragraph{\bilingual{}{结构事实}}
  结合律证明的结构事实 (i) 与 (iii) 对单个包装逐字成立:(i) $\inherits$ 在
  $\sigma$ 处有定义、在 $\sigma$ 的每个真延伸处无定义;(iii) 走步与字面解析
  的隔离不变。本定理新增的不对称在脊柱:比较目标 $p \cat w$ 自身是固定路径,
  两条被比较的路径处于同一文档、固定路径处的事实字面共享;而
  $\sigma \cat w$ 的自反 override 是合成路径(被陈述过滤),
  $p \cat w$ 的自反 override 却是子树外的固定路径。方程~(\ref{eq:overrides})
  产出自反 override 的守卫是某书写分支合格,与自反路径自身是否书写无关;
  未书写的自反 override 因此只出现于 $p$ 侧、在 $\sigma$ 侧没有对应物——
  这正是陈述只比较\emph{书写}分量的原因。前缀封闭在同一处进入:合格分支
  $p_{\mathrm{branch}}$ 满足
  $p_{\mathrm{branch}} \snoc \ell \in \dom(\inherits)$,故
  $p_{\mathrm{branch}}$ 自身书写,于是每个合格分支都是书写分量。

  \paragraph{\bilingual{}{模拟:包装到基}}
  定义 $\varphi(\sigma \cat w) = p \cat w$,其余路径不动。由 (iii),自
  $\sigma \cat w$ 出发的推导只触及根 $\sigma$ 的合成路径、保留根与固定
  路径,故 $\varphi$ 对事实的改写良定义;结合律证明的条款 (C1)--(C4) 对
  根集 $\{\sigma\}$ 与目标根 $p$ 逐字成立,仍由对 Kleene 秩的强归纳给出,
  唯一改变的基例是 (C1) 的 $w = ()$:目标侧读 $p \in \overrides(p)$,
  上段已由 $p$ 的存在性建立。于是 $\supers(\sigma \cat w)$ 的每个书写的
  子树外 override 分量被 $\varphi$ 不动,且作为
  $\supers(p \cat w)$ 的分量可导出。

  \paragraph{\bilingual{}{基到包装}}
  反方向不存在全局改写映射:固定引用可以字面解析到脊柱路径 $p \cat u$,
  这些出现必须原样保留。改为在锚点处断言可导出性,同一强归纳同时证四条:
  (C1$'$) $\overrides(p \cat w)$ 的每个书写元素在
  $\overrides(\sigma \cat w)$ 中可导出;且当某分支在 $p \cat w$ 处合格时,
  同一分支在 $\sigma \cat w$ 处合格(顺带产出其合成自反 override)。
  (C2$'$) 自 $p \cat w$ 出发的每条 $\bases^*$ 链,在 $\sigma$ 侧有自
  $\sigma \cat w$ 出发的链,其终点或为同一路径,或当终点是经脊柱继承到达的
  $p \cat u$ 时为 $\sigma \cat u$;后一情形由 (C1$'$) 对齐两终点的书写
  override 分量。基例 $w = ()$ 即地面事实:$\bases(\sigma) = \{p\}$ 在
  $p$ 的每条链前插入一跳。(C3$'$) 于是 $p \cat w$ 处每个 override 分量
  书写的 $\supers$ 对,在 $\sigma \cat w$ 处有同分量的对(位点分量在跨越
  脊柱处相应替换)。(C4$'$) 锚定于脊柱位点的走步转移:由 (iii) 被匹配的
  $p_{\mathrm{def}}$ 绝非合成路径、$\varphi$ 不动之,消费的对由 (C3$'$)
  或字面共享给出,余下走步用归纳假设。

  \paragraph{\bilingual{}{收尾}}
  两个方向给出所比较集合的互相包含,即相同。存在性:由路径存在性引理,
  $\sigma \cat w_0 \snoc \ell$ 的存在性等价于 $\supers(\sigma \cat w_0)$
  的某分支合格;合格分支书写(前缀封闭),且绝非合成路径(合成路径的延伸
  不在 $\dom(\inherits)$,结构事实 (i)),保留根也不在所比较分量之列
  (固定推导不进入子树,(iii)),故判据只读被比较的集合,而后者两侧相同;
  $w = ()$ 时 $\sigma$ 与 $p$ 都存在。
\else
  The proof specializes the simulation of
  Theorem~\ref{thm:merge-associative} to a single wrapper and a
  fixed target root; Kleene ranks and last-step inversion are used
  as there.

  \paragraph{Ground facts}
  As in the chain flattening of
  Theorem~\ref{thm:merge-associative}:
  \[
    \overrides(p_{\mathrm{reserved}}) = \{p_{\mathrm{reserved}}\},
    \quad
    \supers(p_{\mathrm{reserved}}) = \{((),\, p_{\mathrm{reserved}})\},
    \quad
    \overrides(\sigma) = \{\sigma\}.
  \]
  The two references written at the wrapper are one reference
  $(1, p)$: in equation~\eqref{eq:helpers}, $\inherits$ takes its
  values in sets of references, so $\inherits(\sigma) = \{(1, p)\}$.
  Its one walk step consumes exactly the shared pair
  $((),\, p_{\mathrm{reserved}})$, so
  \[
    \bases(\sigma) = \{p\},
    \quad
    \bases^*(\sigma) = \{\sigma\} \cup \bases^*(p),
    \quad
    \supers(\sigma) = \{(p_{\mathrm{reserved}},\, \sigma)\} \cup \supers(p).
  \]
  Moreover
  $p$ exists, so by the Path Existence
  Lemma~\ref{lem:path-existence} some branch of
  $\supers(\init(p))$ qualifies, and
  equation~(\ref{eq:overrides}) produces the reflexive override
  from it: $p \in \overrides(p)$. The case $w = ()$ follows: the
  one extra pair is synthetic and filtered by the statement.

  \paragraph{Structural facts}
  Structural facts (i) and (iii) of the associativity proof hold
  verbatim for the single wrapper: (i) $\inherits$ is defined at
  $\sigma$ and undefined at every proper extension of $\sigma$;
  (iii) the insulation of walks and literal resolution is
  unchanged. The asymmetry new to this theorem sits on the spine:
  the comparison target $p \cat w$ is itself a fixed path, the two
  compared paths live in one document, and facts at fixed paths are
  literally shared; but the reflexive override of
  $\sigma \cat w$ is a synthetic path, filtered by the statement,
  while the reflexive override of $p \cat w$ is a fixed path
  outside the subtree. Equation~(\ref{eq:overrides}) produces a
  reflexive override under the guard that some written branch
  qualifies, whether or not the reflexive path itself is written;
  an unwritten reflexive override therefore appears on the $p$ side
  with no $\sigma$-side counterpart, which is why the statement
  compares \emph{written} components only. Prefix closure enters at
  the same point: a qualifying branch $p_{\mathrm{branch}}$
  satisfies
  $p_{\mathrm{branch}} \snoc \ell \in \dom(\inherits)$, so
  $p_{\mathrm{branch}}$ is itself written, and every qualifying
  branch is a written component.

  \paragraph{Simulation, wrapper to base}
  Define $\varphi(\sigma \cat w) = p \cat w$ and let $\varphi$ fix
  every other path. By~(iii), a derivation from $\sigma \cat w$
  touches only synthetic paths of the root $\sigma$, the reserved
  root, and fixed paths, so $\varphi$ rewrites facts
  well-definedly, and clauses (C1)--(C4) of the associativity proof
  hold verbatim for the root set $\{\sigma\}$ and the target root
  $p$, by the same strong induction on Kleene rank; the one base
  case that changes is (C1) at $w = ()$, where the target side
  reads $p \in \overrides(p)$, established above from the existence
  of $p$. The clauses give: every written override component of
  $\supers(\sigma \cat w)$ outside the subtree is fixed by
  $\varphi$ and derivable as a component of $\supers(p \cat w)$.

  \paragraph{Base to wrapper}
  The reverse direction admits no global rewriting map: a fixed
  reference may resolve literally to a spine path $p \cat u$, and
  such occurrences must stay as they are. Instead we assert
  derivability at the anchors, proving four clauses by the same
  strong induction: (C1$'$) every written element of
  $\overrides(p \cat w)$ is derivable in
  $\overrides(\sigma \cat w)$; and whenever some branch qualifies
  at $p \cat w$, the same branch qualifies at $\sigma \cat w$,
  producing along the way its synthetic reflexive override.
  (C2$'$) for every $\bases^*$ chain from $p \cat w$ there is a
  chain from $\sigma \cat w$ whose endpoint is either the same
  path, or $\sigma \cat u$ when the endpoint is a spine path
  $p \cat u$ reached by spine inheritance; in the latter case
  (C1$'$) aligns the written override components of the two
  endpoints. The base case $w = ()$ is the ground paragraph:
  $\bases(\sigma) = \{p\}$ inserts one hop in front of every chain
  from $p$. (C3$'$) hence every $\supers$ pair at $p \cat w$ with a
  written override component has a pair at $\sigma \cat w$ with the
  same component, the site component replaced accordingly where it
  crosses the spine. (C4$'$) walks anchored at spine sites
  transfer: by~(iii) the matched $p_{\mathrm{def}}$ is never a
  synthetic path and $\varphi$ fixes it; the consumed pair is given
  by (C3$'$) or by literal sharing, and the induction hypothesis
  covers the remaining steps.

  \paragraph{Conclusion}
  The two directions give mutual inclusion of the compared sets,
  hence agreement. For existence: by the Path Existence Lemma, the
  existence of $\sigma \cat w_0 \snoc \ell$ is equivalent to some
  branch of $\supers(\sigma \cat w_0)$ qualifying; a qualifying
  branch is written (prefix closure) and never a synthetic path (no
  extension of a synthetic path lies in $\dom(\inherits)$,
  structural fact~(i)), and the reserved root is not among the
  compared components (fixed derivations never enter the subtree,
  (iii)), so the criterion reads only the compared sets, which the
  two sides share; at $w = ()$ both $\sigma$ and $p$ exist.
\fi
\end{proof}


\begin{proof}[\bilingual{Proof of Theorem~\ref{thm:decomposition}}{定理~\ref{thm:decomposition} 的证明}]
\ifPaperChinese
\begin{sloppypar}
我们先证 (a),再分别证明 (b) 的两个方向;两个方向都已在 Lean~4 中机械化(表~\ref{tab:lean-correspondence})。(a):$p$ 写下,故 $\init(p)$ 写下;$\init(p) \snoc R$ 与每个 $q_i \cat v$ 的前缀都写下,因为 $W_i$ 前缀封闭;部分的引用由定义给出。

对 (b),记 $\bar r = \init(p) \snoc R$,令 $A$ 为不出现 $R$ 的路径,连同 $\bar r$ 与以某个 $q_i$ 为前缀的路径所成的集合。$\varphi$ 在 $A$ 上有定义,与追加不出现 $R$ 的后缀交换,并在 $A \setminus \{\bar r\}$ 上与 $\init$ 交换。先立三个基本事实。(K1) $\supers_{D'}(\bar r) = \{(\init(p), \bar r)\}$:以 $R$ 结尾的写下路径只有 $\bar r$,故 $\overrides_{D'}(\bar r) = \{\bar r\}$,而 $\inherits_{D'}(\bar r) = \varnothing$ 给出 $\bases_{D'}(\bar r) = \varnothing$。(K2) 对 $b \in A$,$q_i \in \overrides_{D'}(b)$ 蕴含 $b = q_i$:否则 $\bar r$ 是 $\supers_{D'}(\init(b))$ 的一个 override 分量,这要求 $\init(b)$ 以 $R$ 结尾,于是 $\init(b) = \bar r$,$b = q_i$。(K3) 写在 $q_i \cat v$ 的引用,其 $\this_{D'}$ 爬升所访问的定义站点 $\init(q_i \cat v), \init^2(q_i \cat v), \ldots$ 在第 $|v| + 1$ 个处是 $\bar r$,所以爬升经过 $\bar r$ 当且仅当其索引超过 $|v|$,由平移这恰当该引用离开 $M$ 时成立;且爬升从不停在 $\bar r$,那需要索引恰为 $|v|$,平移排除了这一点。$D$ 中的引用不指名 $R$,而 $r_i = (|\init(o)|, q_i)$ 是绝对的,故从 $A$ 出发的 $\bases^*_{D'}$ 链停留在 $A$ 中。

对 $D'$ 中的 Kleene 秩作强归纳(每次降秩都来自末步反演,引理~\ref{lem:kleene-iteration}),对 $a \in A \setminus \{\bar r\}$ 同时证明四条。
(B1) $(s, o) \in \supers_{D'}(a)$ 蕴含 $(\varphi(s), \varphi(o)) \in \supers_D(\varphi(a))$。
(B2) $o \in \overrides_{D'}(a)$ 蕴含 $\varphi(o) \in \overrides_D(\varphi(a))$。
(B3) $t \in \bases_{D'}(a)$ 蕴含 $t \in A \setminus \{\bar r\}$ 且 $\varphi(t) \in \bases_D(\varphi(a))$;从而从 $a$ 出发的 $\bases^*_{D'}$ 链映为从 $\varphi(a)$ 出发的 $\bases^*_D$ 链。
(B4) 对 $s \in A$ 与写下的 $d \in A$,$t \in \this_{D'}(\{s\}, d, n')$ 蕴含 $\varphi(t) \in \this_D(\{\varphi(s)\}, \varphi(d), n)$,其中当 $\bar r$ 在爬升访问的定义站点 $d, \init(d), \ldots, \init^{n'-1}(d)$ 之中时 $n = n' - 1$,否则 $n = n'$。

(B1):若 $a = ()$,两边都是 $\{((), ())\}$。否则反演给出一条链 $a \to^* b$ 与 $o \in \overrides_{D'}(b)$,$s = \init(b)$;由 (B3),$b \in A \setminus \{\bar r\}$ 且 $\varphi(a) \to^* \varphi(b)$;由 (B2),$\varphi(o) \in \overrides_D(\varphi(b))$;而 $\varphi(s) = \init(\varphi(b))$。

(B2):若 $a = ()$,则 $o = ()$。设 $a = a_0 \snoc \ell$。若 $a_0 = \bar r$,则 $a = q_i$,由 (K1) 唯一的分支是 $\bar r$,故 $o = q_i$,而 $\varphi(o) = p \in \overrides_D(p)$。否则反演给出分支 $(\cdot, b) \in \supers_{D'}(a_0)$,$b \snoc \ell$ 写下,$o \in \{a, b \snoc \ell\}$;$b \ne \init(p)$ 或 $\ell \ne R$,因为 $a \ne \bar r$;由 (B1),$\varphi(b)$ 是 $\supers_D(\varphi(a_0))$ 的分支,且 $\varphi(b \snoc \ell) = \varphi(b) \snoc \ell$ 在 $D$ 中写下;故 $\varphi(o) \in \overrides_D(\varphi(a))$。

(B3):反演给出 $o \in \overrides_{D'}(a)$,携带 $(n', w) \in \inherits_{D'}(o)$,以及 $c \in \this_{D'}(\{a_0\}, \init(o), n')$,$t = c \cat w$。由 (B2),$\varphi(o) \in \overrides_D(\varphi(a))$。若 $o$ 不出现 $R$,且该引用不是某个 $r_i$,它也写在 $D$ 中的 $o$ 处;$\bar r$ 不在其爬升访问的站点之中,由 (B4),$\varphi(c) \in \this_D(\{\varphi(a_0)\}, \init(o), n')$,故 $\varphi(t) = \varphi(c) \cat w \in \bases_D(\varphi(a))$。若该引用是 $r_i$,它是绝对的,如引理~\ref{lem:redirection-simulation} 的证明,爬升终于根,$c = ()$,$t = q_i$;由 (B4),$() \in \this_D(\{\varphi(a_0)\}, \init(o), n')$,故 $D$ 中 $o$ 处的 $r$ 给出 $p = \varphi(t)$。若 $o = q_i \cat v$,则该引用为 $\mathit{shift}_v(\rho)$,$\rho = (n, w) \in \inherits_D(p \cat v)$;由 (K3) 与 (B4),$\varphi(c) \in \this_D(\{\varphi(a_0)\}, \init(p \cat v), n)$,故 $\varphi(t) \in \bases_D(\varphi(a))$,由 $p \cat v = \varphi(o)$ 处的 $\rho$ 给出。$o = \bar r$ 不可能,因为 $a \ne \bar r$。$t \in A$ 已在上面说明;$t \ne \bar r$ 由 (K3) 与 $r_i$ 的形状得出。链的情形对长度归纳。

(B4):$n' = 0$ 的情形是直接的。对 $n' > 0$,反演给出 $(s_1, d) \in \supers_{D'}(s)$ 与 $t \in \this_{D'}(\{s_1\}, \init(d), n' - 1)$,两者秩更小。若 $d = \bar r$:该对来自一条链 $s \to^* b$,$\bar r \in \overrides_{D'}(b)$,这迫使 $b$ 以 $R$ 结尾,故 $b = \bar r$;而由 (B3),从 $s \ne \bar r$ 出发的链不终于 $\bar r$,所以 $s = \bar r$;由 (K1),$s_1 = \init(p) = \varphi(s)$,且 $\init(d) = \init(p) = \varphi(d)$;归纳假设给出余下的爬升,这一步在 $D$ 中没有对应,恰好扣去 $n$ 中的 $1$。若 $d \ne \bar r$,则 $s \ne \bar r$(否则由 (K1),$d = \bar r$);由 (B1),$(\varphi(s_1), \varphi(d)) \in \supers_D(\varphi(s))$,这是 $D$ 中的一步;$\varphi(\init(d)) = \init(\varphi(d))$,归纳假设覆盖余下的 $n' - 1$ 步。

对不出现 $R$ 的 $a$,$\varphi(a) = a$,由 (B1),$\supers_{D'}(a)$ 的每个书写分量的像都是 $\supers_D(a)$ 的书写分量;特别地,$a$ 在 $D'$ 中存在蕴含它在 $D$ 中存在。

由 $D$ 到 $D'$ 的包含。下面的路径都不出现 $R$。对 $D$ 中的 Kleene 秩作强归纳,同时证明四条。
(F1) 对 $(s_1, e) \in \supers_D(s)$:若 $e$ 不在 $p$ 的子树中,则 $(s_1, e) \in \supers_{D'}(s)$;若 $e = p \cat y$,则或者 $(s_1, e) \in \supers_{D'}(s)$,或者对每个满足 $y \in W_i$ 的 $i$,都有 $(s'_1, q_i \cat y) \in \supers_{D'}(s)$,$\varphi(s'_1) = s_1$。
(F2) 对 $e \in \overrides_D(b)$:若 $e$ 不在 $p$ 的子树中,或 $e = p$,则 $e \in \overrides_{D'}(b)$;若 $e = p \cat y$ 且 $y \ne ()$,则或者 $e \in \overrides_{D'}(b)$,或者对每个满足 $y \in W_i$ 的 $i$,都有 $q_i \cat y \in \overrides_{D'}(b)$。
(F3) $D$ 中一条链 $x \to^* t$ 给出 $D'$ 中一条链 $x \to^* t$,或者 $t = p$ 且对每个 $i$ 都有 $D'$ 中的链 $x \to^* q_i$。
(F4) 对 $p$ 的子树中一条写下的路径 $d = p \cat y_0$ 与 $t \in \this_D(\{s_0\}, d, n)$:或者 $t \in \this_{D'}(\{s_0\}, d, n)$,或者对每个满足 $y_0 \in W_i$ 的 $i$,都有 $t \in \this_{D'}(\{s_0\}, q_i \cat y_0, n')$,其中当爬升访问 $p$,即 $n > |y_0|$ 时 $n' = n + 1$,否则 $n' = n$。

(F1):若 $s = ()$ 则是直接的。否则反演给出一条链 $s \to^* b$ 与 $e \in \overrides_D(b)$,$s_1 = \init(b)$。若由 (F3) 链在 $D'$ 中不变,(F2) 给出结论,站点仍是 $s_1$。否则 $b = p$,而 $\overrides_D(p) = \{p\}$ 给出 $e = p$;每条链 $s \to^* q_i$ 与 $\overrides_{D'}(q_i) = \{q_i\}$(由 (K1))给出 $(\bar r, q_i) \in \supers_{D'}(s)$,而 $\varphi(\bar r) = \init(p) = s_1$。

(F2):若 $b = ()$ 则 $e = ()$。设 $b = b_0 \snoc \ell$;反演给出分支 $(\cdot, e_0) \in \supers_D(b_0)$,$e_0 \snoc \ell$ 写下,$e \in \{b, e_0 \snoc \ell\}$。对这个分支用 (F1)。若它在 $D'$ 中不变,$e$ 的两种取法都在 $D'$ 中产生。否则 $e_0 = p \cat y$,对每个满足 $y \in W_i$ 的 $i$,$q_i \cat y$ 是 $\supers_{D'}(b_0)$ 的分支。由覆盖,$p \cat y \snoc \ell$ 是某个 $p \cat v$,$v \in W_j$;于是 $y \in W_j$,$q_j \cat y \snoc \ell$ 写下,故自反的 $b$ 在 $\overrides_{D'}(b)$ 中;而对每个满足 $y \snoc \ell \in W_i$ 的 $i$,$q_i \cat y \snoc \ell \in \overrides_{D'}(b)$。在这种情形下 $e_0$ 在子树中,所以 $e_0 \snoc \ell \ne p$,产生的 override 只有自反的 $b$ 与后缀非空的 $p \cat y \snoc \ell$,两者上面都已处理。

(F3):先看一跳 $t \in \bases_D(x)$;反演给出携带 $\rho = (n, w) \in \inherits_D(o)$ 的 $o \in \overrides_D(x)$,以及一次爬升 $c \in \this_D(\{\init(x)\}, \init(o), n)$,$t = c \cat w$。若 $o$ 不在子树中,或 $o = p$,则由 (F2) $o \in \overrides_{D'}(o)$ 的对应成立,即 $o \in \overrides_{D'}(x)$;爬升访问的定义站点都不在子树中,由 (F1) 每一步在 $D'$ 中不变。若 $\rho \ne r$,$\rho \in \inherits_{D'}(o)$,给出 $t \in \bases_{D'}(x)$;若 $\rho = r$,则 $c = ()$,$t = p$,而 $(|\init(o)|, q_i)$ 给出每个 $q_i \in \bases_{D'}(x)$。若 $o = p \cat v$,$v \ne ()$:由覆盖,设 $i$ 使 $(p \cat v, \rho) \in E_i$,于是 $v \in W_i$。$n = 0$ 时用 (F2):$o$ 或 $q_i \cat v$ 在 $\overrides_{D'}(x)$ 中,分别携带 $\rho$ 或 $\mathit{shift}_v(\rho) = \rho$,给出同一个 $t = \init(x) \cat w$。$n > 0$ 时用 (F4),$y_0 = \init(v)$:第一种情形的第一步说明 $p \cat \init(v)$ 是 $\supers_{D'}(\init(x))$ 的分支,故 $o \in \overrides_{D'}(x)$(或 $x = o$),由 $\rho$ 给出 $t$;第二种情形同理给出 $q_i \cat v \in \overrides_{D'}(x)$,其引用 $\mathit{shift}_v(\rho)$ 的索引恰为 $n'$,因为 $n \ge |v|$ 当且仅当 $n > |y_0|$,给出同一个 $t$。链对长度归纳:若前一段止于 $p$ 并对应每个 $q_i$,由 $\overrides_D(p) = \{p\}$,下一跳由某个 $\rho \in \inherits_D(p)$ 给出,$(p, \rho) \in E_j$;则在 $D'$ 中从 $q_j$ 出发,它是自己唯一的 override(由 (K1)),经 $\mathit{shift}_{()}(\rho)$:其爬升先经 (K1) 从 $\bar r$ 到 $\init(p)$,其后与 $D$ 中的爬升逐步相同,给出同一个目标。

(F4):反演给出站点 $s_0, \ldots, s_n$,第 $j$ 步来自一条链 $s_j \to^* b_j$,其中 $d_j = \init^j(d) \in \overrides_D(b_j)$,$s_{j+1} = \init(b_j)$,各事实秩更小。令 $m = |y_0| + 1$,对 $j < m$,$d_j = p \cat y_j$,$d_{m-1} = p$。先定最上面一层。若 $n \ge m$:落到 $p$ 的那一步,由 (F3),链 $s_{m-1} \to^* b_{m-1}$ 或者在 $D'$ 中不变,此时由 (F2) $p \in \overrides_{D'}(b_{m-1})$,取原件版本;或者 $b_{m-1} = p$ 且有链 $s_{m-1} \to^* q_i$,于是 $(\bar r, q_i) \in \supers_{D'}(s_{m-1})$,由 (K1) 再走一步到 $\init(p) = s_m$,取部分版本;$j \ge m$ 的各步站点不在子树中,由 (F1) 不变。若 $n < m$:由最后一步的反演,$p \cat y_n$ 或者是 $\supers_D(s_n)$ 的分支(秩更小,用 (F1)),或者 $b_{n-1} = d_{n-1}$ 且 $s_n = p \cat y_n$ 本身写下(对长度归纳,$\dom(\inherits)$ 前缀封闭的文档中每条写下的路径都是自己的 override);于是 $s_n$ 在 $D'$ 中承载 $p \cat y_n$(原件版本),或对每个满足 $y_n \in W_i$ 的 $i$ 承载 $q_i \cat y_n$(部分版本)。再逐层向下:设所选版本的 $y_{j+1}$ 处的作用域是 $\supers_{D'}(s_{j+1})$ 的 override 分量。由 $\overrides_D(p) = \{p\}$,$b_j \ne p$,故由 (F3) 链 $s_j \to^* b_j$ 在 $D'$ 中不变;$b_j$ 与 $d_j$ 的末标签相同,而所选版本 $y_j$ 处的作用域写下($y_j$ 是 $y_0$ 的前缀,$W_i$ 前缀封闭),由 overrides 方程的分支条款,它属于 $\overrides_{D'}(b_j)$,故 $(s_{j+1}, \cdot) \in \supers_{D'}(s_j)$ 给出这一步。$n = 0$ 的情形是直接的。拼起来即得结论。

对不出现 $R$ 的路径 $a$,由 (F1),$\supers_D(a)$ 的每个书写分量 $e$ 或者本身是 $\supers_{D'}(a)$ 的书写分量,或者是某个书写分量 $q_i \cat y$ 的像,由覆盖,这样的 $i$ 存在;特别地,$a$ 在 $D$ 中存在蕴含它在 $D'$ 中存在。
\end{sloppypar}
\else
\begin{sloppypar}
We prove (a), and then the two inclusions of (b); both inclusions are mechanized in Lean~4
(Table~\ref{tab:lean-correspondence}). (a): $p$ is written, so $\init(p)$ is
written; every prefix of $\init(p) \snoc R$ and of each $q_i \cat v$ is written, as $W_i$ is
prefix-closed; the references of the parts are as defined.

For (b), write $\bar r = \init(p) \snoc R$, and let $A$ be the set of paths in which $R$ does not occur,
together with $\bar r$ and the paths with a prefix $q_i$. The map $\varphi$ is defined on $A$, commutes
with appending a suffix in which $R$ does not occur, and commutes with $\init$ on
$A \setminus \{\bar r\}$. Three ground facts first. (K1) $\supers_{D'}(\bar r) = \{(\init(p), \bar r)\}$:
the only written path ending in $R$ is $\bar r$, so $\overrides_{D'}(\bar r) = \{\bar r\}$, and
$\inherits_{D'}(\bar r) = \varnothing$ gives $\bases_{D'}(\bar r) = \varnothing$. (K2) For $b \in A$,
$q_i \in \overrides_{D'}(b)$ implies $b = q_i$: otherwise $\bar r$ is an override component of
$\supers_{D'}(\init(b))$, which needs $\init(b)$ to end in $R$, so $\init(b) = \bar r$ and $b = q_i$.
(K3) For a reference written at $q_i \cat v$, the definition sites
$\init(q_i \cat v), \init^2(q_i \cat v), \ldots$ that its climb in $\this_{D'}$ visits reach $\bar r$ at
the $(|v| + 1)$-st, so the climb passes $\bar r$ exactly when its index exceeds $|v|$, which by the
shift holds exactly when the reference leaves $M$; and the climb never ends at $\bar r$, which would
need the index to be $|v|$, as the shift excludes. The references of $D$ name no path through $R$, and
each $r_i = (|\init(o)|, q_i)$ is absolute, so a $\bases^*_{D'}$ chain from a path of $A$ stays in $A$.

By strong induction on the Kleene rank in $D'$ (every rank decrease from last-step inversion,
Lemma~\ref{lem:kleene-iteration}), we prove four clauses simultaneously, for $a \in A \setminus \{\bar r\}$.
(B1) $(s, o) \in \supers_{D'}(a)$ implies $(\varphi(s), \varphi(o)) \in \supers_D(\varphi(a))$.
(B2) $o \in \overrides_{D'}(a)$ implies $\varphi(o) \in \overrides_D(\varphi(a))$.
(B3) $t \in \bases_{D'}(a)$ implies $t \in A \setminus \{\bar r\}$ and
$\varphi(t) \in \bases_D(\varphi(a))$; hence a $\bases^*_{D'}$ chain from $a$ maps to a $\bases^*_D$
chain from $\varphi(a)$.
(B4) For $s \in A$ and a written $d \in A$, $t \in \this_{D'}(\{s\}, d, n')$ implies
$\varphi(t) \in \this_D(\{\varphi(s)\}, \varphi(d), n)$, where $n = n' - 1$ if $\bar r$ is among the
definition sites $d, \init(d), \ldots, \init^{n'-1}(d)$ that the climb visits, and $n = n'$ otherwise.

(B1): if $a = ()$, both sides are $\{((), ())\}$. Otherwise inversion gives a chain $a \to^* b$ and
$o \in \overrides_{D'}(b)$ with $s = \init(b)$; by (B3), $b \in A \setminus \{\bar r\}$ and
$\varphi(a) \to^* \varphi(b)$; by (B2), $\varphi(o) \in \overrides_D(\varphi(b))$; and
$\varphi(s) = \init(\varphi(b))$.

(B2): if $a = ()$ then $o = ()$. Let $a = a_0 \snoc \ell$. If $a_0 = \bar r$, then $a = q_i$, its only
branch is $\bar r$ by (K1), so $o = q_i$, and $\varphi(o) = p \in \overrides_D(p)$. Otherwise inversion
gives a branch $(\cdot, b) \in \supers_{D'}(a_0)$ with $b \snoc \ell$ written and
$o \in \{a, b \snoc \ell\}$; $b \snoc \ell \ne \bar r$, since $a \ne \bar r$; by (B1), $\varphi(b)$ is a
branch of $\supers_D(\varphi(a_0))$, and $\varphi(b \snoc \ell) = \varphi(b) \snoc \ell$ is written in
$D$; so $\varphi(o) \in \overrides_D(\varphi(a))$.

(B3): inversion gives $o \in \overrides_{D'}(a)$ carrying $(n', w) \in \inherits_{D'}(o)$ and
$c \in \this_{D'}(\{a_0\}, \init(o), n')$ with $t = c \cat w$. By (B2),
$\varphi(o) \in \overrides_D(\varphi(a))$. If $R$ does not occur in $o$ and the reference is no $r_i$,
it is written at $o$ in $D$ as well; its climb visits no $\bar r$, so by (B4)
$\varphi(c) \in \this_D(\{\varphi(a_0)\}, \init(o), n')$, and
$\varphi(t) = \varphi(c) \cat w \in \bases_D(\varphi(a))$. If the reference is $r_i$, it is absolute, so,
as in the proof of Lemma~\ref{lem:redirection-simulation}, its climb ends at the root, $c = ()$, and
$t = q_i$; by (B4), $() \in \this_D(\{\varphi(a_0)\}, \init(o), n')$, so the reference $r$ at $o$ in $D$
gives $p = \varphi(t)$. If $o = q_i \cat v$, the reference is $\mathit{shift}_v(\rho)$ with
$\rho = (n, w) \in \inherits_D(p \cat v)$; by (K3) and (B4),
$\varphi(c) \in \this_D(\{\varphi(a_0)\}, \init(p \cat v), n)$, so $\varphi(t) \in \bases_D(\varphi(a))$,
given by $\rho$ at $p \cat v = \varphi(o)$. The case $o = \bar r$ is impossible, as $a \ne \bar r$. That
$t \in A$ was shown above, and $t \ne \bar r$ follows from (K3) and the shape of the $r_i$. Chains
follow by induction on their length.

(B4): the case $n' = 0$ is immediate. For $n' > 0$, inversion gives $(s_1, d) \in \supers_{D'}(s)$ and
$t \in \this_{D'}(\{s_1\}, \init(d), n' - 1)$, both of smaller rank. If $d = \bar r$: the pair comes from a chain
$s \to^* b$ with $\bar r \in \overrides_{D'}(b)$, which forces $b$ to end in $R$, so $b = \bar r$, and since by (B3) no chain from $s \ne \bar r$ ends at
$\bar r$, $s = \bar r$; by (K1),
$s_1 = \init(p) = \varphi(s)$, and $\init(d) = \init(p) = \varphi(d)$; the induction hypothesis gives
the rest of the climb, and this step, which has no counterpart in $D$, accounts for the $1$ subtracted
from $n$. If $d \ne \bar r$, then $s \ne \bar r$, since by (K1) $d$ would be $\bar r$; by (B1),
$(\varphi(s_1), \varphi(d)) \in \supers_D(\varphi(s))$, a step in $D$; and
$\varphi(\init(d)) = \init(\varphi(d))$, so the induction hypothesis covers the remaining $n' - 1$
steps.

For a path $a$ in which $R$ does not occur, $\varphi(a) = a$, and by (B1) the image of every written
override component of $\supers_{D'}(a)$ is a written override component of $\supers_D(a)$; in
particular, if $a$ exists in $D'$, it exists in $D$.

The inclusion from $D$ to $D'$. All paths below are free of $R$. By strong induction on the Kleene rank in
$D$, we prove four clauses simultaneously.
(F1) For $(s_1, e) \in \supers_D(s)$: if $e$ is not in the subtree of $p$, then
$(s_1, e) \in \supers_{D'}(s)$; if $e = p \cat y$, then either $(s_1, e) \in \supers_{D'}(s)$, or for every
$i$ with $y \in W_i$ there is $(s'_1, q_i \cat y) \in \supers_{D'}(s)$ with $\varphi(s'_1) = s_1$.
(F2) For $e \in \overrides_D(b)$: if $e$ is not in the subtree of $p$, or $e = p$, then
$e \in \overrides_{D'}(b)$; if $e = p \cat y$ with $y \ne ()$, then either $e \in \overrides_{D'}(b)$, or
$q_i \cat y \in \overrides_{D'}(b)$ for every $i$ with $y \in W_i$.
(F3) A chain $x \to^* t$ in $D$ gives a chain $x \to^* t$ in $D'$, or else $t = p$ and there are chains
$x \to^* q_i$ in $D'$ for every $i$.
(F4) For a written path $d = p \cat y_0$ in the subtree of $p$ and $t \in \this_D(\{s_0\}, d, n)$: either
$t \in \this_{D'}(\{s_0\}, d, n)$, or $t \in \this_{D'}(\{s_0\}, q_i \cat y_0, n')$ for every $i$ with
$y_0 \in W_i$, where $n' = n + 1$ if the climb visits $p$, that is, if $n > |y_0|$, and $n' = n$
otherwise.

(F1): the case $s = ()$ is immediate. Otherwise inversion gives a chain $s \to^* b$ and
$e \in \overrides_D(b)$ with $s_1 = \init(b)$. If by (F3) the chain is unchanged in $D'$, (F2) gives the
claim, with the same site $s_1$. Otherwise $b = p$, and $\overrides_D(p) = \{p\}$ gives $e = p$; each
chain $s \to^* q_i$, with $\overrides_{D'}(q_i) = \{q_i\}$ by (K1), gives
$(\bar r, q_i) \in \supers_{D'}(s)$, and $\varphi(\bar r) = \init(p) = s_1$.

(F2): if $b = ()$ then $e = ()$. Let $b = b_0 \snoc \ell$; inversion gives a branch
$(\cdot, e_0) \in \supers_D(b_0)$ with $e_0 \snoc \ell$ written and $e \in \{b, e_0 \snoc \ell\}$. Apply
(F1) to this branch. If it is unchanged in $D'$, both choices of $e$ are produced in $D'$. Otherwise
$e_0 = p \cat y$ and $q_i \cat y$ is a branch of $\supers_{D'}(b_0)$ for every $i$ with $y \in W_i$. By
the cover, $p \cat y \snoc \ell$ is $p \cat v$ for some $v \in W_j$; so $y \in W_j$ and
$q_j \cat y \snoc \ell$ is written, which puts the reflexive $b$ in $\overrides_{D'}(b)$; and
$q_i \cat y \snoc \ell \in \overrides_{D'}(b)$ for every $i$ with $y \snoc \ell \in W_i$. In this case $e_0$ is in the subtree, so
$e_0 \snoc \ell \ne p$, and the only overrides produced are the reflexive $b$ and $p \cat y \snoc \ell$
with a nonempty suffix, both treated above.

(F3): consider first a hop $t \in \bases_D(x)$; inversion gives $o \in \overrides_D(x)$ carrying
$\rho = (n, w) \in \inherits_D(o)$, and a climb $c \in \this_D(\{\init(x)\}, \init(o), n)$ with
$t = c \cat w$. If $o$ is not in the subtree, or $o = p$, then $o \in \overrides_{D'}(x)$ by (F2), and the
definition sites that the climb visits are outside the subtree, so by (F1) each step is unchanged in
$D'$. If $\rho \ne r$, then $\rho \in \inherits_{D'}(o)$ gives $t \in \bases_{D'}(x)$; if $\rho = r$, then
$c = ()$ and $t = p$, and the references $(|\init(o)|, q_i)$ give $q_i \in \bases_{D'}(x)$ for every $i$.
If $o = p \cat v$ with $v \ne ()$, let $i$ be such that $(p \cat v, \rho) \in E_i$, by the cover, so
that $v \in W_i$. If $n = 0$, (F2) puts $o$ or $q_i \cat v$ in $\overrides_{D'}(x)$, carrying
$\rho$ or $\mathit{shift}_v(\rho) = \rho$ respectively, and either gives the same $t = \init(x) \cat w$. If
$n > 0$, apply (F4) with $y_0 = \init(v)$: in its first case, the first step shows that $p \cat \init(v)$ is
a branch of $\supers_{D'}(\init(x))$, so $o \in \overrides_{D'}(x)$ (or $x = o$), and $\rho$ gives $t$; in
its second case, the same argument gives $q_i \cat v \in \overrides_{D'}(x)$, whose reference
$\mathit{shift}_v(\rho)$ has index exactly $n'$, since $n \ge |v|$ exactly when $n > |y_0|$, and it gives
the same $t$. Chains follow by induction on their length: when a prefix ends at $p$ and corresponds to
every $q_i$, the next hop is given, as $\overrides_D(p) = \{p\}$, by some $\rho \in \inherits_D(p)$
with $(p, \rho) \in E_j$; continue in $D'$ from $q_j$, its own only override by (K1), with
$\mathit{shift}_{()}(\rho)$, whose climb first steps from $\bar r$ to $\init(p)$ by (K1) and then follows
the climb in $D$ step by step, giving the same target.

(F4): inversion gives sites $s_0, \ldots, s_n$ whose $j$-th step comes from a chain $s_j \to^* b_j$ with
$d_j = \init^j(d) \in \overrides_D(b_j)$ and $s_{j+1} = \init(b_j)$, all of smaller rank. Let
$m = |y_0| + 1$, so that $d_j = p \cat y_j$ for $j < m$ and $d_{m-1} = p$. First fix the top. If
$n \ge m$, consider the step that lands on $p$: by (F3), either the chain $s_{m-1} \to^* b_{m-1}$ is
unchanged in $D'$, and then $p \in \overrides_{D'}(b_{m-1})$ by (F2), and we take the version of the
original; or $b_{m-1} = p$ with chains $s_{m-1} \to^* q_i$, and then $(\bar r, q_i) \in \supers_{D'}(s_{m-1})$
and one more step by (K1) reaches $\init(p) = s_m$, and we take the version of the parts; the steps with
$j \ge m$ visit sites outside the subtree and are unchanged by (F1). If $n < m$: by inversion of the last
step, $p \cat y_n$ is either a branch of $\supers_D(s_n)$, of smaller rank, to which (F1) applies, or
$b_{n-1} = d_{n-1}$ and $s_n = p \cat y_n$ is written, and so, by induction on its length as every written path of a document with
prefix-closed $\dom(\inherits)$, is its own override in $D'$; so $s_n$ hosts in $D'$ either $p \cat y_n$ (the version of the original) or $q_i \cat y_n$ for every
$i$ with $y_n \in W_i$ (the version of the parts). Then descend level by level: suppose the scope of the
chosen version at $y_{j+1}$ is an override component of $\supers_{D'}(s_{j+1})$. As
$\overrides_D(p) = \{p\}$, $b_j \ne p$, so by (F3) the chain $s_j \to^* b_j$ is unchanged in $D'$; $b_j$
and $d_j$ have the same last label, and the scope of the chosen version at $y_j$ is written, since $y_j$
is a prefix of $y_0$ and $W_i$ is prefix-closed; by the branch clause of the equation for $\overrides$,
it belongs to $\overrides_{D'}(b_j)$, which gives the step $(s_{j+1}, \cdot) \in \supers_{D'}(s_j)$.
The case $n = 0$ is immediate. Putting the steps together gives the claim.

For a path $a$ free of $R$, by (F1) every written override component $e$ of $\supers_D(a)$ is either a
written override component of $\supers_{D'}(a)$ or the image of one, $q_i \cat y$, where such an $i$ exists by the cover; in particular, if $a$
exists in $D$, it exists in $D'$.
\end{sloppypar}
\fi
\end{proof}


\begin{proof}[\bilingual{Proof of Proposition~\ref{prop:extensibility-by-construction}}{命题~\ref{prop:extensibility-by-construction} 的证明}]
\ifPaperChinese
  存在性:记 $w_j = (\text{``\_\_extend''}) \cat (\ell_1, \ldots, \ell_j)$
  ($0 \le j \le k$),每个 $w_j$ 都是写下位置(嵌套定义链)。对 $j$ 作
  平凡归纳,不变量为 $w_j \in \overrides(w_j)$。基例:分支
  $((),\, ()) \in \supers(())$(方程~\ref{eq:supers} 基例)满足守卫
  $() \snoc \text{``\_\_extend''} \in \dom(\inherits)$,方程~(\ref{eq:overrides})
  产出 $w_0 \in \overrides(w_0)$。归纳步:由 $w_j \in \overrides(w_j)$
  与 $\bases^*$ 的自反性,$(\init(w_j),\, w_j) \in \supers(w_j)$;
  守卫 $w_j \snoc \ell_{j+1} \in \dom(\inherits)$ 成立(写下),
  产出 $w_{j+1} \in \overrides(w_{j+1})$。末了,分支
  $(\_,\, w_k) \in \supers(w_k)$ 满足
  $w_k \snoc \text{``\_\_new''} \in \dom(\inherits)$,由路径存在性引理
  (引理~\ref{lem:path-existence}),
  $(\text{``\_\_extend''}) \cat p \cat (\text{``\_\_new''})$ 存在。
  归纳只自底向上构造事实、不消费事实,故无需秩。

  隔离与 (ii):$\ell_P$ 之下的文本按构造原样。保留根
  $(\text{``\_\_extend''})$ 路径长度为 1,且 \mintinline{yaml}{__extend}
  与 \mintinline{yaml}{__new} 不出现于任何引用,故固定目标
  $p_{\mathrm{current}} \cat w_{\mathrm{ref}}$ 落入保留子树只可能经
  $p_{\mathrm{current}} = ()$ 且首投影标签为保留字,被排除;保留子树外
  路径的推导从不触及该子树,其事实字面不变,(ii) 得证。

  两点收尾。其一,(i) 中后缀 $(\text{``\_\_new''})$ 是新的:\mintinline{yaml}{__new}
  不书写于 $P$ 中任何位置,而按引理~\ref{lem:path-existence},
  $(\ell_P) \cat p \snoc \text{``\_\_new''}$ 在 $P$ 中存在须有某书写分支
  定义 \mintinline{yaml}{__new}。其二,合成体良构:包装的引用在合成体中
  解析到顶层 $\ell_P$,不再爬空,其余良构性情形
  (附录~\ref{app:well-definedness})均节点局部,在写下的包装处逐一成立。
\else
  Existence: write
  $w_j = (\text{``\_\_extend''}) \cat (\ell_1, \ldots, \ell_j)$ for
  $0 \le j \le k$; every $w_j$ is a written position (the nested
  definition chain). By a plain induction on $j$ with the invariant
  $w_j \in \overrides(w_j)$. Base: the branch
  $((),\, ()) \in \supers(())$ (the base case of
  equation~\ref{eq:supers}) satisfies the guard
  $() \snoc \text{``\_\_extend''} \in \dom(\inherits)$, so
  equation~(\ref{eq:overrides}) produces $w_0 \in \overrides(w_0)$.
  Step: from $w_j \in \overrides(w_j)$ and the reflexivity of
  $\bases^*$, $(\init(w_j),\, w_j) \in \supers(w_j)$; the guard
  $w_j \snoc \ell_{j+1} \in \dom(\inherits)$ holds (written), and
  the comprehension produces $w_{j+1} \in \overrides(w_{j+1})$.
  Finally the branch $(\_,\, w_k) \in \supers(w_k)$ satisfies
  $w_k \snoc \text{``\_\_new''} \in \dom(\inherits)$, and the Path
  Existence Lemma~\ref{lem:path-existence} yields the existence of
  $(\text{``\_\_extend''}) \cat p \cat (\text{``\_\_new''})$. The
  induction only constructs facts bottom-up and consumes none, so no
  rank is needed.

  Insulation and (ii): the text under $\ell_P$ is untouched by
  construction. The reserved root $(\text{``\_\_extend''})$ has
  length one, and \mintinline{yaml}{__extend} and
  \mintinline{yaml}{__new} appear in no reference, so a fixed target
  $p_{\mathrm{current}} \cat w_{\mathrm{ref}}$ lands inside the
  reserved subtree only through $p_{\mathrm{current}} = ()$ with the
  reserved word as first projection label, which is excluded; the
  derivations of paths outside the reserved subtree never touch it,
  their facts are literally unchanged, and (ii) holds.

  Two closing points. First, the suffix $(\text{``\_\_new''})$ of (i)
  is new: \mintinline{yaml}{__new} is written nowhere in $P$, and by
  Lemma~\ref{lem:path-existence} the existence of
  $(\ell_P) \cat p \snoc \text{``\_\_new''}$ in $P$ would require
  some written branch defining \mintinline{yaml}{__new}. Second, the
  composite is well-formed: the wrapper's reference resolves to the
  top-level $\ell_P$ in the composite, no longer climbing past the
  root, and the remaining well-formedness cases
  (Appendix~\ref{app:well-definedness}) are node-local and hold at
  the written wrapper.
\fi
\end{proof}


\begin{lemma}[\bilingual{Monotonicity in the program}{对程序的单调性}]\label{lem:program-monotonicity}
\ifPaperChinese
设 $\inherits$ 与 $\inherits'$ 是两个原始查询,满足 $\dom(\inherits) \subseteq \dom(\inherits')$,且对每个 $p \in \dom(\inherits)$ 有 $\inherits(p) \subseteq \inherits'(p)$;我们说 $\inherits'$ \emph{扩展}了 $\inherits$。则 $\inherits$ 之上 $\lfp(T_P)$ 的每个事实都是 $\inherits'$ 之上 $\lfp(T_P)$ 的事实;特别地,对每条路径 $q$,$\inherits$ 之上的 $\supers(q)$ 包含于 $\inherits'$ 之上的 $\supers(q)$,且在 $\inherits$ 之上存在的每条路径在 $\inherits'$ 之上也存在。
\else
Let $\inherits$ and $\inherits'$ be two primitive queries with
$\dom(\inherits) \subseteq \dom(\inherits')$ and $\inherits(p) \subseteq \inherits'(p)$ for every
$p \in \dom(\inherits)$; we say that $\inherits'$ \emph{extends} $\inherits$. Then every fact of
$\lfp(T_P)$ over $\inherits$ is a fact of $\lfp(T_P)$ over $\inherits'$; in particular, for every
path $q$, $\supers(q)$ over $\inherits$ is contained in $\supers(q)$ over $\inherits'$, and every path
that exists over $\inherits$ exists over $\inherits'$.
\fi
\end{lemma}

\begin{proof}
\ifPaperChinese
我们在一切路径上的五个函数的全部查询之上读取两个最小不动点;这一查询集包含由任一 AST 产生的查询,且由引理~\ref{lem:dependency-restriction},最小不动点在由一个 AST 产生的、对依赖封闭的查询集上的限制不因加入其他查询而改变。对 $\inherits$ 之上的 Kleene 秩作强归纳(引理~\ref{lem:kleene-iteration})。秩为 $m + 1$ 的事实由其方程右侧的一个实例,从秩至多为 $m$ 的有限多个事实产出;由归纳假设,这些事实都是 $\inherits'$ 之上的事实。五个方程恰在两处读取原始查询,且两处都是正向的:方程~\eqref{eq:overrides} 的守卫 $p_{\mathrm{branch}} \snoc \ell \in \dom(\inherits)$,以及方程~\eqref{eq:bases} 的绑定子 $(n, w) \in \inherits(p_{\mathrm{override}})$;二者在 $\inherits$ 之上成立时在 $\inherits'$ 之上也成立。因此同一个实例在 $\inherits'$ 之上产出该事实。存在性的断言由 $\supers$ 的包含得到。
\else
We read both least fixed points over the queries of the five functions at all paths; this query set
contains the queries arising from either AST, and by Lemma~\ref{lem:dependency-restriction} the
restriction of a least fixed point to the dependency-closed queries arising from one AST is unchanged
by adding the others. By strong induction on the Kleene rank over $\inherits$
(Lemma~\ref{lem:kleene-iteration}). A fact of rank $m + 1$ is produced by one instance of the right-hand side of its equation from finitely many
facts of rank at most $m$, which are facts over $\inherits'$ by the induction hypothesis. The five
equations read the primitive query in exactly two places, both positive: the guard
$p_{\mathrm{branch}} \snoc \ell \in \dom(\inherits)$ of equation~\eqref{eq:overrides} and the binder
$(n, w) \in \inherits(p_{\mathrm{override}})$ of equation~\eqref{eq:bases}; both hold over
$\inherits'$ when they hold over $\inherits$. Hence the same instance produces the fact over
$\inherits'$. The existence claim follows from the containment of $\supers$.
\fi
\end{proof}

\begin{lemma}[\bilingual{Composites}{组合体}]\label{lem:composite-extension}
\ifPaperChinese
设 $P$ 是良构程序,$Q$ 是一列标签两两不同的顶层定义,且这些标签都不是 $P$ 的顶层标签;设 $C$ 是它们的组合体(第~\ref{sec:semantic-variants}~节),其根 mixin 依次列出 $P$ 的根的成员与 $Q$ 的成员。则:
(i) $P$ 中书写的每个位置都是 $C$ 的书写位置,其节点不变,且在 $C$ 中该处不出现任何良构性情形;
(ii) $\inherits_C$ 扩展 $\inherits_P$,且对每个 $p \in \dom(\inherits_P)$ 有 $\inherits_C(p) = \inherits_P(p)$;
(iii) 在任意文档中,一个书写位置 $p$ 是否良构只取决于 $p$ 处的节点、$p$ 的各前缀的末尾标签,以及 $p$ 的各外围作用域写下的标签;对 $C$ 中 $Q$ 的位置,根处写下的标签是 $P$ 与 $Q$ 的顶层标签。
若 $Q$ 重用了 $P$ 的一个顶层标签,或把一个标签写了两次,则根带有一个重复键(情形 2),$C$ 不良构。
\else
Let $P$ be a well-formed program and $Q$ a list of top-level definitions with pairwise distinct labels,
none of which is a top-level label of $P$, and let $C$ be their composite (Section~\ref{sec:semantic-variants}), whose
root mixin lists the members of $P$'s root followed by those of $Q$. Then:
(i) every position written in $P$ is a written position of $C$ with the same node, and no case of
well-formedness arises at it in $C$;
(ii) $\inherits_C$ extends $\inherits_P$, with $\inherits_C(p) = \inherits_P(p)$ for every
$p \in \dom(\inherits_P)$;
(iii) in any document, whether a written position $p$ is well formed depends only on the node at $p$,
on the final labels of the prefixes of $p$, and on the labels written at the enclosing scopes of $p$;
for a position of $Q$ in $C$, the labels written at the root are the top-level labels of $P$ and $Q$.
If $Q$ reuses a top-level label of $P$, or writes one label twice, the root carries a duplicate key
(case~2) and $C$ is ill formed.
\fi
\end{lemma}

\begin{proof}
\ifPaperChinese
(iii):四种情形逐一读取。情形 (1)、(4) 只读 $p$ 处的节点。情形 (2) 在 $\at$ 沿 $p$ 的各前缀下降时只读各前缀处成员的键。情形 (3) 的 $\lexical$ 在外围作用域 $q$ 处读 $\has(q, \cdot)$,即 $q$ 处写下的标签,$\qualified$ 读 $\last(q)$;二者都不读任何语义函数。
(i):$C$ 的根列出 $P$ 的全部根成员,且由于 $Q$ 的标签是新鲜的,$\at_C$ 在每个顶层标签处的解是唯一的,并与 $\at_P$ 相同;此后 $\at$ 沿 $P$ 的路径下降时经过的是 $P$ 的节点,故 $P$ 的每个书写位置在 $C$ 中有相同的节点。情形 (1)、(2)、(4) 是节点局部的,在非根节点处不变;在根处,$C$ 的根是一个序列,其成员是 $P$ 的根成员与 $Q$ 的定义,故情形 (4) 成立,且其键两两不同。对情形 (3):$P$ 的一个引用的爬升所经过的外围作用域,除根以外都是 $P$ 的非根节点,它们写下的标签不变;根处写下的标签只是增多了。$P$ 的爬升在 $P$ 中停在某个作用域 $q_0$;在 $C$ 中,它在到达 $q_0$ 之前测试的是同样的非根作用域,因而停在同一个 $q_0$,而 $\qualified$ 只读末尾标签,不受影响。
(ii):由 (i),$\dom(\inherits_P) \subseteq \dom(\inherits_C)$,且对 $p \in \dom(\inherits_P)$,$p$ 处的引用成员相同,每个引用的索引由 (i) 中的爬升论证保持不变,故 $\inherits_C(p) = \inherits_P(p)$。
最后一句:根成员中出现同一个键两次,即情形 (2)。
\else
(iii): we read the four cases. Cases (1) and (4) read only the node at $p$. Case (2) reads, as $\at$
descends along the prefixes of $p$, only the keys of the members at each prefix. In case (3),
$\lexical$ reads $\has(q, \cdot)$ at each enclosing scope $q$, that is, the labels written at $q$, and
$\qualified$ reads $\last(q)$; neither reads any semantic function.
(i): the root of $C$ lists every root member of $P$, and since the labels of $Q$ are fresh, $\at_C$ has
a unique solution at every top-level label, equal to that of $\at_P$; from there $\at$ descends along a
path of $P$ through the nodes of $P$, so every written position of $P$ has the same node in $C$. Cases
(1), (2), and (4) are node-local and unchanged at nonroot nodes; at the root, the root of $C$ is a
sequence whose members are the root members of $P$ and the definitions of $Q$, so case~(4) holds there,
and its keys are pairwise distinct. For
case (3): the enclosing scopes that the climb of a reference of $P$ passes are, apart from the root,
nonroot nodes of $P$, whose written labels are unchanged; at the root the written labels only grow. In
$P$ the climb stops at some scope $q_0$; in $C$ it tests the same nonroot scopes before reaching $q_0$,
so it stops at the same $q_0$, and $\qualified$, which reads final labels only, is unaffected.
(ii): by (i), $\dom(\inherits_P) \subseteq \dom(\inherits_C)$, and for $p \in \dom(\inherits_P)$ the
reference members at $p$ are the same and the index of each is preserved by the climb argument of (i),
so $\inherits_C(p) = \inherits_P(p)$.
The last sentence: the same key occurs twice among the root members, which is case~(2).
\fi
\end{proof}

\begin{proposition}[\bilingual{Upward-closed observation}{向上封闭的观测}]\label{prop:upward-closed-observation}
\ifPaperChinese
称一个程序的一条路径 $q$ 为\emph{读点},称 $\supers(q)$ 的各个对为 $q$ 处的\emph{事实};$q$ 存在当且仅当 $q$ 处有一个事实,而观测由这些事实决定。则对每个文档以及其原始查询的每个扩展,每个读点处的每个事实都保持。因此在原始查询的扩展之下,没有事实依赖于另一个事实的缺失:若 $q$ 存在而 $r$ 不存在,则在每个使 $r$ 存在的扩展之上 $q$ 仍然存在。一个读点所表达的每个性质,即使 $q$ 存在、或使某个给定的对属于 $\supers(q)$ 的那类原始查询,在扩展之下都是向上封闭的。一个良构程序与新鲜顶层定义的组合体是其原始查询的扩展(引理~\ref{lem:composite-extension}~(ii));文本的一次编辑是原始查询的扩展,只要它保持每个引用的索引并只增加书写位置与引用,而在一个词法引用所爬经的作用域写下其首标签的定义,会给该引用重新索引。
\else
We call a path $q$ of a program a \emph{read point}, and the pairs of $\supers(q)$ the \emph{facts} at
$q$; $q$ exists exactly when there is a fact at $q$, and the facts determine the observation. Then for
every document and every extension of its primitive query, every fact at every read point persists.
Consequently, under extensions of the primitive query, no fact depends on the absence of another: if
$q$ exists and $r$ does not, then $q$ still exists over every extension in which $r$ exists. Every
property that a read point expresses, the class of primitive queries over which $q$ exists or over
which a given pair belongs to $\supers(q)$, is upward closed under extension. The composite of a
well-formed program with fresh top-level definitions extends its primitive query
(Lemma~\ref{lem:composite-extension}~(ii)); an edit of the text extends the primitive query
when it keeps the index of every reference and only adds written positions and references, while a
definition of the leading label of a lexical reference, written at a scope its climb passes, re-indexes
that reference.
\fi
\end{proposition}

\begin{proof}
\ifPaperChinese
程序的意义恰是 $\lfp(T_P)$(定理~\ref{thm:well-defined}),而读点处的事实是它的事实;由引理~\ref{lem:program-monotonicity},这些事实在扩展之下保持。关于 $q$ 与 $r$ 的那句是其特例:$q$ 的存在性是一个事实,与 $r$ 无关地保持。
\else
The meaning of a program is $\lfp(T_P)$ (Theorem~\ref{thm:well-defined}), and the facts at read points
are its facts; by Lemma~\ref{lem:program-monotonicity} these facts persist under extension. The sentence about
$q$ and $r$ is a special case: the existence of $q$ is a fact, and it persists regardless of $r$.
\fi
\end{proof}

\begin{lemma}[\bilingual{Nonroot targets}{非根目标}]\label{lem:nonroot-targets}
\ifPaperChinese
在一个良构文档中,对每条路径 $q$,$\bases(q)$ 的元素都不是根。
\else
In a well-formed document, no element of $\bases(q)$ is the root, for every path $q$.
\fi
\end{lemma}

\begin{proof}
\ifPaperChinese
先看两个事实。其一,一个顶层路径 $(\ell)$ 的 override 只有它自己:方程~\eqref{eq:overrides} 的分支来自 $\supers(()) = \{((), ())\}$。其二,因此一个 $\supers$ 对 $(p_{\mathrm{site}}, p_{\mathrm{override}})$ 若有 $|p_{\mathrm{override}}| \ge 2$,则 $p_{\mathrm{site}} \neq ()$:位点为根意味着该对来自一个顶层的链节点,其 override 长度为一。
现在一个目标是 $p_{\mathrm{current}} \cat w$,其中 $(n, w) \in \inherits(p_{\mathrm{override}})$ 是写在 $q$ 的某个 override $p_{\mathrm{override}}$ 处的一个引用 $r$ 的索引化,而 $p_{\mathrm{current}} \in \this(\{\init(q)\}, \init(p_{\mathrm{override}}), n)$。词法引用的 $w$ 以其首标签开头,非空。限定 $\this$ 引用 $(\ell, {\sim}, \ell_1, \ldots, \ell_k)$ 在 $k \ge 1$ 时 $w$ 非空;在 $k = 0$ 时,$n = \qualified(\init(p_{\mathrm{override}}), \ell)$ 停在一个末尾标签为 $\ell$ 的作用域,故 $n < m$,其中 $m = |\init(p_{\mathrm{override}})|$。若 $n \ge 1$,最后一步消费的对的 override 分量长度为 $m - n + 1 \ge 2$,由其二,输出的位点非根。若 $n = 0$,目标是 $\init(q)$;若它是根,则 $q$ 是顶层路径,由其一 $p_{\mathrm{override}} = q$,于是 $\init(p_{\mathrm{override}}) = ()$,而 $\qualified((), \ell)$ 无定义,这与良构性矛盾。
\else
Two facts first. First, the overrides of a top-level path $(\ell)$ are that path alone: the branches of
equation~\eqref{eq:overrides} come from $\supers(()) = \{((), ())\}$. Second, therefore a $\supers$
pair $(p_{\mathrm{site}}, p_{\mathrm{override}})$ with $|p_{\mathrm{override}}| \ge 2$ has
$p_{\mathrm{site}} \neq ()$: a root site means that the pair comes from a top-level chain node, whose
override has length one.
Now a target is $p_{\mathrm{current}} \cat w$, where $(n, w) \in \inherits(p_{\mathrm{override}})$ indexes a
reference $r$ written at some override $p_{\mathrm{override}}$ of $q$ and
$p_{\mathrm{current}} \in \this(\{\init(q)\}, \init(p_{\mathrm{override}}), n)$. The $w$ of a lexical
reference begins with its leading label and is nonempty. A qualified-this reference
$(\ell, {\sim}, \ell_1, \ldots, \ell_k)$ has a nonempty $w$ when $k \ge 1$; when $k = 0$,
$n = \qualified(\init(p_{\mathrm{override}}), \ell)$ stops at a scope whose final label is $\ell$, so
$n < m$ where $m = |\init(p_{\mathrm{override}})|$. If $n \ge 1$, the pair consumed by the last step has
an override component of length $m - n + 1 \ge 2$, so by the second fact the output site is nonroot. If
$n = 0$, the target is $\init(q)$; were it the root, $q$ would be top-level, so
$p_{\mathrm{override}} = q$ by the first fact and $\init(p_{\mathrm{override}}) = ()$, where
$\qualified((), \ell)$ is undefined, contradicting well-formedness.
\fi
\end{proof}

\begin{lemma}[\bilingual{Redirection simulation}{改指模拟}]\label{lem:redirection-simulation}
\ifPaperChinese
设 $D$ 是一个 $\dom(\inherits)$ 前缀封闭的文档,$p_S$ 是一条路径,$p_T$ 是一条非根路径。称写在路径 $o$ 处的一个引用 $(n, w)$ 为\emph{绝对的},当且仅当 $n = |\init(o)|$,即它爬升到根并从根命名一条路径。设 $D'$ 是一个文档,满足 $\dom(\inherits_D) \subseteq \dom(\inherits_{D'})$ 与 $p_T \in \bases^*_{D'}(p_S)$,并且对每个 $(n, w) \in \inherits_D(o)$,要么 $(n, w) \in \inherits_{D'}(o)$,要么该引用是绝对的、$w = p_T \cat u$ 且 $(n, p_S \cat u) \in \inherits_{D'}(o)$。当 $a = b$,或对某个后缀 $v$ 有 $a = p_S \cat v$ 且 $b = p_T \cat v$ 时,我们说路径 $a$ \emph{对应}路径 $b$。则对每个 $(s, o) \in \supers_D(q)$ 和每个对应 $q$ 的 $a$,存在 $(s', o') \in \supers_{D'}(a)$,使 $s'$ 对应 $s$、$o'$ 对应 $o$,且当 $o \in \dom(\inherits_D)$ 时 $o' = o$。特别地,对每条路径 $q$ 与每个后缀 $w$,$D$ 在 $q$ 处的书写观测包含于 $D'$ 在 $q$ 处的书写观测,$D$ 在 $p_T \cat w$ 处的书写观测包含于 $D'$ 在 $p_S \cat w$ 处的书写观测,且 $D$ 中 $p_T \cat w$ 的存在蕴涵 $D'$ 中 $p_S \cat w$ 的存在。
\else
Let $D$ be a document whose $\dom(\inherits)$ is prefix-closed, $p_S$ a path, and $p_T$ a nonroot path.
Call a reference $(n, w)$ written at a path $o$ \emph{absolute} when $n = |\init(o)|$, so that it climbs
to the root and names a path from there. Let $D'$ be a document with
$\dom(\inherits_D) \subseteq \dom(\inherits_{D'})$ and $p_T \in \bases^*_{D'}(p_S)$, such that for every
$(n, w) \in \inherits_D(o)$ either $(n, w) \in \inherits_{D'}(o)$, or the reference is absolute,
$w = p_T \cat u$, and $(n, p_S \cat u) \in \inherits_{D'}(o)$. We say that a path $a$
\emph{corresponds} to a path $b$ when $a = b$, or $a = p_S \cat v$ and $b = p_T \cat v$ for some suffix
$v$. Then for every $(s, o) \in \supers_D(q)$ and every $a$ corresponding to $q$ there is
$(s', o') \in \supers_{D'}(a)$ with $s'$ corresponding to $s$, $o'$ corresponding to $o$, and
$o' = o$ when $o \in \dom(\inherits_D)$. In particular, for every path $q$ and every suffix $w$, the
written observation of $D$ at $q$ is contained in that of $D'$ at $q$, the written observation of $D$ at
$p_T \cat w$ is contained in that of $D'$ at $p_S \cat w$, and the existence of $p_T \cat w$ in $D$ implies that of $p_S \cat w$ in
$D'$.
\fi
\end{lemma}

\begin{proof}
\ifPaperChinese
我们说 $a$ \emph{严格对应} $b$,当且仅当 $a$ 对应 $b$ 且 $(a, b) \neq (p_S, p_T)$,即 $a = b$ 或后缀 $v$ 非空;此时若 $b = b_0 \snoc \ell$,则 $a = a_0 \snoc \ell$ 且 $a_0$ 对应 $b_0$。先看两个基础事实。(G1) $p_S \neq ()$:$\bases(()) = \varnothing$,故从根出发的链停在根,而 $p_T$ 非根。(G2) 对每个使 $p_T \cat v \in \dom(\inherits_D)$ 的 $v$,$\supers_{D'}(p_S \cat v)$ 中有一个形如 $(\cdot,\, p_T \cat v)$ 的对。对 $v$ 归纳:$v = ()$ 时 $p_T$ 是书写的,沿用命题~\ref{prop:extensibility-by-construction} 的存在性归纳,由前缀封闭,每条书写路径都是它自己的 override,故 $(\init(p_T), p_T) \in \supers_{D'}(p_T)$;沿链 $p_S \to^* p_T$ 把这一对带到 $\supers_{D'}(p_S)$($\bases^*$ 是传递的,且由 (G1) $p_S$ 非根)。在 $v \snoc \ell$ 处:$p_T \cat v$ 是书写的(前缀封闭),归纳假设给出分支 $p_T \cat v$ 位于 $\supers_{D'}(p_S \cat v)$ 中,其守卫 $p_T \cat v \snoc \ell \in \dom(\inherits_{D'})$ 成立,故方程~\eqref{eq:overrides} 产出 $p_T \cat v \snoc \ell \in \overrides_{D'}(p_S \cat v \snoc \ell)$,再由自反链得到 $(p_S \cat v,\, p_T \cat v \snoc \ell) \in \supers_{D'}(p_S \cat v \snoc \ell)$。
\else
We say that $a$ \emph{strictly corresponds} to $b$ when $a$ corresponds to $b$ and
$(a, b) \neq (p_S, p_T)$, that is, $a = b$ or the suffix $v$ is nonempty; then $b = b_0 \snoc \ell$
gives $a = a_0 \snoc \ell$ with $a_0$ corresponding to $b_0$. Two ground facts first. (G1)
$p_S \neq ()$: $\bases(()) = \varnothing$, so a chain from the root stays at the root, while $p_T$ is
nonroot. (G2) For every $v$ with $p_T \cat v \in \dom(\inherits_D)$, $\supers_{D'}(p_S \cat v)$ has a
pair of the form $(\cdot,\, p_T \cat v)$. By induction on $v$: at $v = ()$, $p_T$ is written, and as in
the existence induction of Proposition~\ref{prop:extensibility-by-construction}, by prefix closure every written path is its
own override, so $(\init(p_T), p_T) \in \supers_{D'}(p_T)$; the chain $p_S \to^* p_T$ carries this pair
to $\supers_{D'}(p_S)$, since $\bases^*$ is transitive and $p_S$ is nonroot by (G1). At $v \snoc \ell$:
$p_T \cat v$ is written (prefix closure), the induction hypothesis gives the branch $p_T \cat v$ in
$\supers_{D'}(p_S \cat v)$, its guard $p_T \cat v \snoc \ell \in \dom(\inherits_{D'})$ holds, so
equation~\eqref{eq:overrides} produces $p_T \cat v \snoc \ell \in \overrides_{D'}(p_S \cat v \snoc \ell)$,
and the reflexive chain gives
$(p_S \cat v,\, p_T \cat v \snoc \ell) \in \supers_{D'}(p_S \cat v \snoc \ell)$.
\fi

\ifPaperChinese
对给定事实在 $D$ 中的 Kleene 秩作强归纳(每次降秩都来自末步反演,引理~\ref{lem:kleene-iteration}),我们同时证明四个条款。
(E1) $(s, o) \in \supers_D(q)$ 且 $a$ 对应 $q$,蕴涵存在 $(s', o') \in \supers_{D'}(a)$,$s'$ 对应 $s$,$o'$ 对应 $o$,且 $o$ 书写时 $o' = o$。
(E2) $o \in \overrides_D(q)$ 且 $a$ 严格对应 $q$,蕴涵存在 $o' \in \overrides_{D'}(a)$,$o'$ 对应 $o$,且 $o$ 书写时 $o' = o$。
(E3) $t \in \bases_D(q)$ 且 $a$ 严格对应 $q$,蕴涵存在 $t' \in \bases_{D'}(a)$ 对应 $t$;以及从 $q$ 到 $t$ 的一条 $\bases^*_D$ 链且 $a$ 严格对应 $q$,蕴涵存在一条从 $a$ 到某个严格对应 $t$ 的 $t'$ 的 $\bases^*_{D'}$ 链。
(E4) $t \in \this_D(\{s\}, d, n)$,$d$ 书写,且 $s'$ 对应 $s$,蕴涵存在对应 $t$ 的 $t' \in \this_{D'}(\{s'\}, d, n)$。
\else
By strong induction on the Kleene rank in $D$ of the given fact (every rank decrease from last-step
inversion, Lemma~\ref{lem:kleene-iteration}), we prove four clauses simultaneously.
(E1) $(s, o) \in \supers_D(q)$ and $a$ corresponding to $q$ imply some
$(s', o') \in \supers_{D'}(a)$ with $s'$ corresponding to $s$, $o'$ corresponding to $o$, and $o' = o$
when $o$ is written.
(E2) $o \in \overrides_D(q)$ and $a$ strictly corresponding to $q$ imply some
$o' \in \overrides_{D'}(a)$ corresponding to $o$, with $o' = o$ when $o$ is written.
(E3) $t \in \bases_D(q)$ and $a$ strictly corresponding to $q$ imply some $t' \in \bases_{D'}(a)$
corresponding to $t$; and a $\bases^*_D$ chain from $q$ to $t$, with $a$ strictly corresponding to $q$,
implies a $\bases^*_{D'}$ chain from $a$ to some $t'$ strictly corresponding to $t$.
(E4) $t \in \this_D(\{s\}, d, n)$ with $d$ written and $s'$ corresponding to $s$ imply some
$t' \in \this_{D'}(\{s'\}, d, n)$ corresponding to $t$.
\fi

\ifPaperChinese
(E1):$q = ()$ 时 $a = ()$($p_T$ 非根),取根的对。否则反演给出一条链 $q \to^* p_{\mathrm{base}}$ 与 $o \in \overrides_D(p_{\mathrm{base}})$,秩都更小,且 $s = \init(p_{\mathrm{base}})$。若 $a$ 严格对应 $q$,(E3) 给出一条链 $a \to^* p'$,$p'$ 严格对应 $p_{\mathrm{base}}$,(E2) 给出 $o' \in \overrides_{D'}(p')$;于是 $(\init(p'), o') \in \supers_{D'}(a)$,且 $\init(p')$ 对应 $\init(p_{\mathrm{base}})$。若 $(a, q) = (p_S, p_T)$,先对 $(p_T, p_T)$ 照此处理,再在前面接上链 $p_S \to^* p_T$。
\else
(E1): if $q = ()$ then $a = ()$, as $p_T$ is nonroot, and the root pair serves. Otherwise inversion gives
a chain $q \to^* p_{\mathrm{base}}$ and $o \in \overrides_D(p_{\mathrm{base}})$, both of smaller rank,
with $s = \init(p_{\mathrm{base}})$. If $a$ strictly corresponds to $q$, (E3) gives a chain $a \to^* p'$
with $p'$ strictly corresponding to $p_{\mathrm{base}}$, and (E2) gives $o' \in \overrides_{D'}(p')$; so
$(\init(p'), o') \in \supers_{D'}(a)$, and $\init(p')$ corresponds to $\init(p_{\mathrm{base}})$. If
$(a, q) = (p_S, p_T)$, treat the pair $(p_T, p_T)$ in this way and prefix the chain $p_S \to^* p_T$.
\fi

\ifPaperChinese
(E2):$q = ()$ 时 $o = ()$。对 $q = q_0 \snoc \ell$,反演给出一个分支 $(\cdot, p_{\mathrm{branch}}) \in \supers_D(q_0)$,且 $p_{\mathrm{branch}} \snoc \ell$ 是书写的;由前缀封闭 $p_{\mathrm{branch}}$ 是书写的,于是 (E1) 在 $a_0$ 处给出同一个 $p_{\mathrm{branch}}$ 位于 $\supers_{D'}(a_0)$ 中,而其守卫在 $D'$ 中成立。分支 override $p_{\mathrm{branch}} \snoc \ell$ 因此在两边都被产出。对自反 override $o = q$:若 $a = q$,它在 $a$ 处被产出;否则 $q = p_T \cat v$,$a = p_S \cat v$,$v$ 非空,此时自反地产出 $o' = a$,它对应 $o$;若 $o$ 是书写的,则由 (G2) 与守卫 $o \in \dom(\inherits_{D'})$,$o$ 本身属于 $\overrides_{D'}(a)$。
\else
(E2): if $q = ()$ then $o = ()$. For $q = q_0 \snoc \ell$, inversion gives a branch
$(\cdot, p_{\mathrm{branch}}) \in \supers_D(q_0)$ with $p_{\mathrm{branch}} \snoc \ell$ written; by prefix
closure $p_{\mathrm{branch}}$ is written, so (E1) at $a_0$ yields the same $p_{\mathrm{branch}}$ in
$\supers_{D'}(a_0)$, and its guard holds in $D'$. The branch override $p_{\mathrm{branch}} \snoc \ell$ is
therefore produced on both sides. For the reflexive override $o = q$: if $a = q$ it is produced at $a$;
otherwise $q = p_T \cat v$ and $a = p_S \cat v$ with $v$ nonempty, and the reflexive $o' = a$, which
corresponds to $o$, is produced; if $o$ is written, then by (G2) and the guard
$o \in \dom(\inherits_{D'})$, $o$ itself belongs to $\overrides_{D'}(a)$.
\fi

\ifPaperChinese
(E3):反演给出一个 override $o \in \overrides_D(q)$,它携带 $(n, w) \in \inherits_D(o)$,因而是书写的,以及 $p_{\mathrm{current}} \in \this_D(\{q_0\}, \init(o), n)$,$t = p_{\mathrm{current}} \cat w$。(E2) 给出 $o \in \overrides_{D'}(a)$ 本身;$\init(o)$ 是书写的,(E4) 给出 $p'_{\mathrm{current}} \in \this_{D'}(\{a_0\}, \init(o), n)$,它对应 $p_{\mathrm{current}}$。若 $(n, w) \in \inherits_{D'}(o)$,则 $t' = p'_{\mathrm{current}} \cat w$ 对应 $t$,因为对应关系在接上同一个后缀时保持。若该引用被改指,则它是绝对的,$n = |\init(o)|$,$w = p_T \cat u$:从一个非根定义点出发、步数等于其长度的走步停在根(每一步匹配一个短一个标签的 override 分量,而长度为一的 override 分量所在的对以根为位点),而当 $|\init(o)| = 0$ 时 $o$ 是顶层路径,由 $q$ 的 override 是顶层路径知 $q_0$ 是根。故 $p_{\mathrm{current}} = ()$,而对应它的 $p'_{\mathrm{current}}$ 也是 $()$;于是 $t = p_T \cat u$,$t' = p_S \cat u$ 对应 $t$。对链,对其长度归纳:每一跳由前一部分映射;一跳若落在 $t' = p_S$ 而 $t = p_T$,则接上链 $p_S \to^* p_T$,使它落在 $p_T$ 本身,于是终点仍严格对应。
\else
(E3): inversion gives an override $o \in \overrides_D(q)$ carrying $(n, w) \in \inherits_D(o)$, hence
written, and $p_{\mathrm{current}} \in \this_D(\{q_0\}, \init(o), n)$ with
$t = p_{\mathrm{current}} \cat w$. (E2) yields $o$ itself in $\overrides_{D'}(a)$; $\init(o)$ is written,
and (E4) yields $p'_{\mathrm{current}} \in \this_{D'}(\{a_0\}, \init(o), n)$ corresponding to
$p_{\mathrm{current}}$. If $(n, w) \in \inherits_{D'}(o)$, then $t' = p'_{\mathrm{current}} \cat w$
corresponds to $t$, since correspondence is kept when the same suffix is appended. If the reference is
redirected, it is absolute, $n = |\init(o)|$, and $w = p_T \cat u$: a walk from a nonroot definition site
with as many steps as its length ends at the root, since each step matches an override component one
label shorter and a pair whose override component has length one has the root as its site; and when
$|\init(o)| = 0$, $o$ is a top-level path, and since $o$ is an override of $q$, the scope $q_0$ is the
root. So $p_{\mathrm{current}} = ()$, and $p'_{\mathrm{current}}$, which corresponds to it, is $()$ as
well; hence $t = p_T \cat u$, and $t' = p_S \cat u$ corresponds to $t$. For chains, by induction on their
length: each hop is mapped by the first part; a hop landing at $t' = p_S$ with $t = p_T$ is extended by
the chain $p_S \to^* p_T$ so as to land at $p_T$ itself, and the endpoint strictly corresponds.
\fi

\ifPaperChinese
(E4):$n = 0$ 时是平凡的。$n > 0$ 时,反演给出 $(s_1, d) \in \supers_D(s)$ 与 $t \in \this_D(\{s_1\}, \init(d), n - 1)$,秩都更小;$d$ 是书写的,(E1) 在 $s'$ 处给出 $(s'_1, d) \in \supers_{D'}(s')$,$s'_1$ 对应 $s_1$;$\init(d)$ 是书写的,归纳假设完成余下的 $n - 1$ 步。
结论:在 $a = q$ 处与在 $a = p_S \cat w$、$q = p_T \cat w$ 处应用 (E1);存在性由一个对的存在得到。
\else
(E4): the case $n = 0$ is immediate. For $n > 0$, inversion gives $(s_1, d) \in \supers_D(s)$ and
$t \in \this_D(\{s_1\}, \init(d), n - 1)$, both of smaller rank; $d$ is written, so (E1) at $s'$ gives
$(s'_1, d) \in \supers_{D'}(s')$ with $s'_1$ corresponding to $s_1$; $\init(d)$ is written, and the
induction hypothesis covers the remaining $n - 1$ steps.
The conclusions apply (E1) at $a = q$, and at $a = p_S \cat w$ with $q = p_T \cat w$; existence
follows from the existence of a pair.
\fi
\end{proof}

\ifPaperChinese
在一个 $\dom(\inherits)$ 前缀封闭的文档里,对一切 $w$ 的 $p \cat w$ 处的书写观测决定 $p$ 处的观测:由引理~\ref{lem:path-existence},$p \cat w \snoc \ell$ 存在当且仅当 $\supers(p \cat w)$ 的某个分支 $p_{\mathrm{branch}}$ 满足 $p_{\mathrm{branch}} \snoc \ell \in \dom(\inherits)$,而由前缀封闭这样的分支是书写的;且一条路径存在当且仅当它的书写观测非空。
\else
In
a document whose $\dom(\inherits)$ is prefix-closed, the written observations at $p \cat w$ for all $w$
determine the observation at $p$: by Lemma~\ref{lem:path-existence}, $p \cat w \snoc \ell$ exists exactly
when some branch $p_{\mathrm{branch}}$ of $\supers(p \cat w)$ satisfies
$p_{\mathrm{branch}} \snoc \ell \in \dom(\inherits)$, and such a branch is written by prefix closure; and
a path exists exactly when its written observation is nonempty.
\fi


\begin{proof}[\bilingual{Proof of Theorem~\ref{thm:inheritance-substitutability}}{定理~\ref{thm:inheritance-substitutability} 的证明}]
\ifPaperChinese
(a) 是引理~\ref{lem:composite-extension} 的 (i) 与 (iii);一个 $Q$ 的位置所读取的 $P$ 的内容只是 $P$ 的顶层标签,不含 $P$ 的任何体。(b1) 由引理~\ref{lem:composite-extension}~(ii) 与引理~\ref{lem:program-monotonicity} 得到。对 (b2),若 $p_S = p_T$,结论是平凡的;否则 $p_T$ 是一跳 $\bases$ 的目标,由引理~\ref{lem:nonroot-targets} 非根。语法层产生的文档的 $\dom(\inherits)$ 是前缀封闭的($\at$ 沿前缀递归)。取 $D' = D = C$、不改指任何引用,应用引理~\ref{lem:redirection-simulation}:$C$ 中 $p_T \cat w$ 处的书写观测包含于 $C$ 中 $p_S \cat w$ 处的书写观测;由 (b1),$P$ 中 $p_T \cat w$ 处的书写观测包含于前者,而 $P$ 的书写路径在 $C$ 中仍是书写的。观测的包含由存在性的转移得到。一般情形是同一论证,不经过 (b1)。
\else
(a) is Lemma~\ref{lem:composite-extension}~(i) and (iii); what a position of $Q$ reads of $P$ is the
top-level labels of $P$ only, never a body of $P$. (b1) follows from
Lemma~\ref{lem:composite-extension}~(ii) and Lemma~\ref{lem:program-monotonicity}. For (b2), the claim
is immediate when $p_S = p_T$; otherwise $p_T$ is the target of a $\bases$ hop and is nonroot by
Lemma~\ref{lem:nonroot-targets}. A document the syntax layer produces has a prefix-closed
$\dom(\inherits)$, since $\at$ recurses through prefixes. Lemma~\ref{lem:redirection-simulation} with
$D' = D = C$ and no reference redirected gives that the written observation of $C$ at $p_T \cat w$ is
contained in that of $C$ at $p_S \cat w$; by (b1), the written observation of $P$ at $p_T \cat w$ is
contained in the former, as the written paths of $P$ remain written in $C$. The containment of
observations follows from the transfer of existence. The general case is the same argument without
(b1).
\fi
\end{proof}

\ifPaperChinese
取 $p_S = (\text{``\_\_extend''})$、$p_T = (\ell_P)$,(b2) 表明命题~\ref{prop:extensibility-by-construction} 的扩展在每条对应路径处都保留 $P$ 的观测:可扩展性给出新的成员,可替换性保证旧的成员都还在。
\else
Taking $p_S = (\text{``\_\_extend''})$ and $p_T = (\ell_P)$, (b2) shows that the extension of
Proposition~\ref{prop:extensibility-by-construction} keeps the observation of $P$ at every
corresponding path: extensibility provides the new member, and substitutability keeps every old one.
\fi

\begin{corollary}[\bilingual{Behavior compatibility}{行为兼容}]\label{cor:behavior-compatibility}
\ifPaperChinese
设 $D$ 是一个良构文档,带有位于 $p_S$ 的 mixin $S$ 与位于非根路径 $p_T$ 的 mixin $T$。本推论关乎一个上下文中命名 $T$ 或其子树中一条路径的绝对引用,以及该子树中的读点;一个上下文还通过指向 $T$ 的外围 mixin 的引用触及 $T$,这类引用把 $T$ 作为成员继承,此处不改指。设 $D'$ 是一个良构文档,其原始查询与 $D$ 的相同,只是一组绝对引用 $(n, p_T \cat u)$ 被改指为 $(n, p_S \cat u)$,并且在 $D'$ 中 $p_T \in \bases^*(p_S)$,例如当 $S$ 以一个不在改指之列的引用继承 $T$、并写下另外的成员时。在表层,一个绝对引用是一个词法引用,其首标签写在根处,且不写在其爬升经过的其他任何作用域处(限定 $\this$ 引用停在一个非根作用域,从不是绝对的);把其投影列表开头的 $p_T$ 的标签改写为 $p_S$ 的标签,当爬升经过的作用域中除根以外没有一个写下新的首标签时,就得到这样的 $D'$。则对每条路径 $q$,$D$ 在 $q$ 处的书写观测包含于 $D'$ 在 $q$ 处的书写观测,且对每个后缀 $w$,$D$ 在 $p_T \cat w$ 处的书写观测包含于 $D'$ 在 $p_S \cat w$ 处的书写观测。因此在一个上下文的绝对使用处以 $S$ 替换 $T$,并在上下文读取 $T$ 之处读取 $S$,不丢失任何观测。
\else
Let $D$ be a well-formed document with a mixin $S$ at $p_S$ and a mixin $T$ at a nonroot path $p_T$. The
corollary concerns the absolute references of a context that name $T$ or a path in its subtree, and
the read points of the context in that subtree; a context also reaches $T$ through references to
enclosing mixins of $T$, which inherit $T$ as a member and are not redirected here. Let $D'$ be a
well-formed document whose primitive query is that of $D$ except that a set of absolute references
$(n, p_T \cat u)$ is redirected to $(n, p_S \cat u)$, such that $p_T \in \bases^*(p_S)$ in $D'$, for
instance when $S$ inherits $T$ through a reference not among those redirected and writes further
members. At the surface, an absolute reference is a lexical reference whose leading label is written at
the root and at no other scope its climb passes (a qualified-this reference stops at a nonroot scope and
is never absolute); rewriting the labels of $p_T$ at the head of its projection list to those of $p_S$
yields such a $D'$ when no scope the climb passes, other than the root, writes the new leading label.
Then for every path $q$ the written observation of $D$ at $q$ is contained in that of $D'$ at $q$, and
for every suffix $w$ the written observation of $D$ at $p_T \cat w$ is contained in that of $D'$ at
$p_S \cat w$. Hence replacing $T$ by $S$ at the absolute uses of a context, and reading $S$ where the
context read $T$, loses no observation.
\fi
\end{corollary}

\begin{proof}
\ifPaperChinese
引理~\ref{lem:redirection-simulation};$D$ 的书写位置在 $D'$ 中仍是书写的,因为改指只改引用。
\else
Lemma~\ref{lem:redirection-simulation}; the written positions of $D$ remain written in $D'$, since
redirection changes references only.
\fi
\end{proof}

\begin{corollary}[\bilingual{Fragile base class}{脆弱基类}]\label{cor:fragile-base-class}
\ifPaperChinese
Mikhajlov 与 Sekerinski~\cite{mikhajlov1998-fragile-base-class} 的 flexibility 性质要求:若一个基类 $C$ 被其修订 $C'$ 精化,且被扩展 $M \mathbin{\mathrm{mod}} C$ 精化,则 $C$ 被 $M \mathbin{\mathrm{mod}} C'$ 精化。我们把``$X$ 被 $Y$ 精化''读作在对应路径处 $X$ 的观测包含于 $Y$ 的观测。取一个良构文档,其中位于 $p_C$ 的 $C$ 与位于 $p_\Delta$ 的一个 mixin 在一个共同的外围作用域中命名;加法修订 $C'$ 是深度合并 \mintinline[escapeinside=||]{yaml}{[[|$p_C$|], [|$p_\Delta$|]]}(定义~\ref{def:merge});$M \mathbin{\mathrm{mod}} C$ 是一个 mixin,其引用中有一个解析到 $p_C$ 的绝对引用,并写下 $M$ 的成员;并设 $M \mathbin{\mathrm{mod}} C'$ 是同一个 mixin,其中该引用如推论~\ref{cor:behavior-compatibility} 那样改指到 $C'$,且该改指在表层可实现。则 $C$ 被 $C'$ 与 $M \mathbin{\mathrm{mod}} C$ 精化,因为二者都继承 $C$;$C$ 被 $M \mathbin{\mathrm{mod}} C'$ 精化,因为 $M \mathbin{\mathrm{mod}} C'$ 经 $C'$ 传递地继承 $C$;并且 $M \mathbin{\mathrm{mod}} C$ 被 $M \mathbin{\mathrm{mod}} C'$ 精化。在这一读法下,flexibility 性质的结论不使用其前提:它对每个继承 $C$ 的修订 $C'$ 都成立。包含允许在 $M$ 与 $\Delta$ 都写下的标签处发生深度合并:在 $M \mathbin{\mathrm{mod}} C'$ 中,这样的标签同时带有两处定义。
\else
The flexibility property of Mikhajlov and Sekerinski~\cite{mikhajlov1998-fragile-base-class} requires
that whenever a base class $C$ is refined by its revision $C'$ and by an extension $M \mathbin{\mathrm{mod}} C$, the base
class $C$ is refined by $M \mathbin{\mathrm{mod}} C'$. We read ``$X$ is refined by $Y$'' as containment of the
observation of $X$ in that of $Y$ at corresponding paths. Take a well-formed document in which $C$ at
$p_C$ and a mixin at $p_\Delta$ are named in a common enclosing scope; the additive revision $C'$ is the
deep merge \mintinline[escapeinside=||]{yaml}{[[|$p_C$|], [|$p_\Delta$|]]} (Definition~\ref{def:merge});
$M \mathbin{\mathrm{mod}} C$ is a mixin whose references include an absolute one that resolves to $p_C$, and
which writes the members of $M$; and let $M \mathbin{\mathrm{mod}} C'$ be the same mixin with that reference
redirected to $C'$ as in Corollary~\ref{cor:behavior-compatibility}, the redirection being realizable at
the surface. Then $C$ is refined by $C'$ and by $M \mathbin{\mathrm{mod}} C$, since both inherit $C$; $C$ is
refined by $M \mathbin{\mathrm{mod}} C'$, since $M \mathbin{\mathrm{mod}} C'$ inherits $C$ transitively through $C'$; and
$M \mathbin{\mathrm{mod}} C$ is refined by $M \mathbin{\mathrm{mod}} C'$. In this reading the conclusion of the flexibility
property does not use its premises: it holds for every revision $C'$ that inherits $C$. Containment
admits deep merges at labels that $M$ and $\Delta$ both write: in $M \mathbin{\mathrm{mod}} C'$ such a label carries
both definitions.
\fi
\end{corollary}

\begin{proof}
\ifPaperChinese
前三条精化是定理~\ref{thm:inheritance-substitutability} 的一般情形,因为 $\bases^*$ 是传递的。最后一条是推论~\ref{cor:behavior-compatibility},把 $M \mathbin{\mathrm{mod}} C$ 指向 $C$ 的绝对引用改指为指向 $C'$,而 $C'$ 指向 $C$ 的引用不在改指之列。关于前提的那句:上述链只用到 $C'$ 继承 $C$。关于深度合并的那句:若 $M$ 与 $\Delta$ 都写下标签 $\ell$,则二者的定义都是 $M \mathbin{\mathrm{mod}} C'$ 处 $\ell$ 的分支 override(方程~\eqref{eq:overrides})。
\else
The first three refinements are the general case of Theorem~\ref{thm:inheritance-substitutability},
since $\bases^*$ is transitive. The last is Corollary~\ref{cor:behavior-compatibility}, redirecting the
absolute reference of $M \mathbin{\mathrm{mod}} C$ from $C$ to $C'$, while the reference of $C'$ to $C$ is not
redirected. The sentence on the premises: the chains above use only that $C'$ inherits $C$. The
sentence on deep merges: if $M$ and $\Delta$ both write a label $\ell$, both definitions are branch
overrides of $\ell$ at $M \mathbin{\mathrm{mod}} C'$ (equation~\eqref{eq:overrides}).
\fi
\end{proof}

\section{\bilingual{Adequacy for the Lazy \texorpdfstring{$\lambda$}{λ}-Calculus: Proofs}{惰性 $\lambda$-演算的充分性：证明}}
\label{app:levy-longo-tree-proofs}

\ifPaperChinese
本附录证明与 $\lambda$-演算对应的继承演算子语言关于惰性
$\lambda$-演算是充分的（adequate，第~\ref{sec:levy-longo-tree}~节）：
一个闭合 $\lambda$-项有弱头范式当且仅当其翻译收敛（定理~\ref{thm:adequacy}）。
\else
This appendix proves that the sublanguage of
inheritance-calculus corresponding to the $\lambda$-calculus
is adequate for the lazy $\lambda$-calculus
(Section~\ref{sec:levy-longo-tree}):
a closed $\lambda$-term has a weak head normal form if and only
if its translation converges (Theorem~\ref{thm:adequacy}).
\fi

\subsection{\bilingual{Weak Head Normal Form}{弱头范式}}

\ifPaperChinese
我们简要回顾弱头归约与 Lévy--Longo 树~\cite{levy1978-reductions-lambda-calcul, longo1983-set-theoretical-models}。
\emph{弱头归约} $M \to_h M'$ 沿函数脊收缩最左的 $\beta$-redex，但不进入抽象体：
若 $M$ 的形式为 $(\lambda x.\, e)\; v\; M_1 \cdots M_k$，则
$M \to_h e[v/x]\; M_1 \cdots M_k$。
若项 $M$ 没有弱头 redex，则称其处于\emph{弱头范式}（WHNF），
即一个抽象 $\lambda x.\, M'$，或一个以变量为首的应用 $y\; M_1 \cdots M_k$。
\else
We briefly recall weak-head reduction and the
L\'evy--Longo tree~\cite{levy1978-reductions-lambda-calcul, longo1983-set-theoretical-models}.
A \emph{weak-head reduction} $M \to_h M'$ contracts the
leftmost $\beta$-redex along the function spine, without
reducing under an abstraction: if $M$ has the form
$(\lambda x.\, e)\; v\; M_1 \cdots M_k$, then
$M \to_h e[v/x]\; M_1 \cdots M_k$.
A term $M$ is in \emph{weak head normal form} (WHNF) if it has
no weak-head redex, that is, $M$ is an abstraction
$\lambda x.\, M'$ or an application $y\; M_1 \cdots M_k$ headed
by a variable.
\fi

\ifPaperChinese
$\lambda$-项 $M$ 的 \emph{Lévy--Longo 树} $\mathrm{LLT}(M)$ 是其遗传式弱头范式：暴露弱头范式的顶层构造子，递归进入直接子项，并在到不了弱头范式之处放上 $\bot$。其根为 $\bot$ 当且仅当 $M$ 没有弱头范式；Lévy--Longo 树是惰性 $\lambda$-演算的标准树语义~\cite{abramsky1993-full-abstraction-lazy-lambda}。
\else
The \emph{L\'evy--Longo tree} $\mathrm{LLT}(M)$ of a
$\lambda$-term $M$ is its hereditary weak head normal form:
expose the top constructor of the WHNF, recurse into the
immediate sub-terms, and place $\bot$ where no WHNF is reached.
Its root is $\bot$ exactly when $M$ has no WHNF; the
L\'evy--Longo tree is the standard tree semantics of the lazy
$\lambda$-calculus~\cite{abramsky1993-full-abstraction-lazy-lambda}.
\fi

\subsection{\bilingual{Path Encoding}{路径编码}}

\ifPaperChinese
我们定义 Lévy--Longo 树中的位置与 mixin 树中的路径之间的对应关系。
当翻译 $\mathcal{T}$ 作用于一个纯 $\lambda$-项时，所得 mixin 树具有特定形状：
\else
We define a correspondence between positions in the L\'evy--Longo tree
and paths in the mixin tree.
When the translation $\mathcal{T}$ is applied to a pure
$\lambda$-term, the resulting mixin tree has a specific shape:
\fi
\begin{itemize}
  \item \ifPaperChinese
    抽象 $\lambda x.\, M$ 翻译为一条仅有单一自有定义 $\olbl{__whnf}$ 的 mixin；$\olbl{__whnf}$ 之下的抽象形状具有 label $\{x, \text{``\_\_parameter''}, \text{``\_\_result''}\}$，其中 $x$ 是别名，$\olbl{__parameter}$ 扮演形参声明，$\olbl{__result}$ 持有体。其下 $\olbl{__parameter}$ 有定义的节点称之为\emph{抽象形状}。
    \else
    An abstraction $\lambda x.\, M$ translates to a mixin with a
    single own definition $\olbl{__whnf}$, under which the abstraction
    shape has labels $\{x, \text{``\_\_parameter''}, \text{``\_\_result''}\}$,
    where $x$ is an alias, $\olbl{__parameter}$ plays the role of a
    parameter declaration, and $\olbl{__result}$ holds the body.
    A node under which $\olbl{__parameter}$ is defined
    is referred to as the \emph{abstraction shape}.
    \fi
  \item \ifPaperChinese
    应用 $M_1\; M_2$ 翻译为一条具有自有定义
    $\{\text{``\_\_callee''}, \text{``\_\_call''}, \text{``\_\_whnf''}\}$ 的 mixin，其中
    $\olbl{__callee}$ 命名被调用者 $\mathcal{T}(M_1)$，$\olbl{__call}$
    继承 \mintinline{yaml}{[__callee, __whnf]} 并以 $\mathcal{T}(M_2)$ 覆写 $\olbl{__parameter}$，
    $\olbl{__whnf}$ 投影 \mintinline{yaml}{[__call, __result, __whnf]} 以沿规约链行走。
    \else
    \begin{sloppypar}
    An application $M_1\; M_2$ translates to a mixin with own
    definitions $\{\text{``\_\_callee''},\allowbreak \text{``\_\_call''},\allowbreak \text{``\_\_whnf''}\}$, where
    $\olbl{__callee}$ names the operator $\mathcal{T}(M_1)$,
    $\olbl{__call}$ inherits \mintinline{yaml}{[__callee, __whnf]} and
    overrides $\olbl{__parameter}$ with $\mathcal{T}(M_2)$,
    and $\olbl{__whnf}$ projects
    \mintinline{yaml}{[__call, __result, __whnf]} to walk the reduction chain.
    \end{sloppypar}
    \fi
  \item \ifPaperChinese
    变量 $x$ 翻译为单成员 mixin \mintinline[escapeinside=||]{yaml}{[[|$x$|]]}，其成员是词法引用 \mintinline[escapeinside=||]{yaml}{[|$x$|]}，在翻译后的树中解析为绑定 $x$ 的最近抽象的别名。
    \else
    A variable $x$ translates to the one-member mixin \mintinline[escapeinside=||]{yaml}{[[|$x$|]]},
    whose member is the lexical reference \mintinline[escapeinside=||]{yaml}{[|$x$|]} that resolves in
    the translated tree to the alias of the nearest
    enclosing abstraction that binds $x$.
    \fi
\end{itemize}

\noindent
\ifPaperChinese
Lévy--Longo 树中的每个位置都是从根出发的一系列导航步骤。
在节点 $\lambda x.\, M'$（一个抽象）或 $y\; M_1 \cdots M_k$（以变量为首）处，
可能的步骤为：进入抽象的函数体，以及进入首变量的第 $i$ 个参数 $M_i$。
在 mixin 树中，这些步骤分别对应于跟随 $\olbl{__result}$ 标签（进入函数体）
和在别名解析后跟随 $\olbl{__parameter}$ 标签（进入参数位置）。
\else
Each L\'evy--Longo tree position is a sequence of navigation
steps from the root.
At a node $\lambda x.\, M'$ (an abstraction) or
$y\; M_1 \cdots M_k$ (headed by a variable), the possible steps
are: entering the body of the abstraction, and
entering the $i$-th argument $M_i$ of the head variable.
In the mixin tree, these correspond to following the
$\olbl{__result}$ label (for entering the body) and the
$\olbl{__parameter}$ label after alias resolution (for
entering an argument position).
\fi

\ifPaperChinese
与其在不同 mixin 树之间比较绝对路径（这对内部结构敏感），
我们只观测通过 $\olbl{__whnf}$ 投影可访问的\emph{收敛}行为
（定义~\ref{def:convergence}）。
$\olbl{__whnf}$ 投影继承规约链：沿 $\bases^*$ 的每一跳即一次弱头归约步骤，
因为应用的 $\olbl{__whnf}$ 投影 $(\text{``\_\_call''}, \text{``\_\_result''}, \text{``\_\_whnf''})$，
透过 $\olbl{__call}$ 节点的继承收缩 head redex。
链末尾的抽象形状表明已到达弱头范式。
\else
Rather than comparing absolute paths across different mixin
trees, which would be sensitive to internal structure, we
observe only the \emph{convergence} behavior accessible via
the $\olbl{__whnf}$ projection
(Definition~\ref{def:convergence}).
The $\olbl{__whnf}$ projection inherits the reduction chain:
each hop along $\bases^*$ is one weak-head reduction step,
since an application's $\olbl{__whnf}$ projects
$(\text{``\_\_call''}, \text{``\_\_result''}, \text{``\_\_whnf''})$,
contracting the head redex through the inheritance at the
$\olbl{__call}$ node.
The abstraction shape at the end of the chain signals that a
weak head normal form has been reached.
\fi

\begin{lemma}[\bilingual{Node classification}{节点分类}]\label{lem:node-classification}
  \ifPaperChinese
  设 $M_0$ 为闭合 $\lambda$-项,$P = \mathcal{T}(M_0)$。
  则 $P$ 的每个书写节点(即 $\dom(\inherits)$ 中的位置及其成员)恰属于下列八类之一,
  且每个书写节点至多携带一条引用:
  \begin{enumerate}
    \item \emph{应用节点}:自有 label $\{\text{``\_\_callee''}, \text{``\_\_call''}, \text{``\_\_whnf''}\}$,无引用;
    \item \emph{redex 节点}(应用之下的 $\olbl{__call}$):一条引用 \mintinline{yaml}{[__callee, __whnf]},
      $\lexical$ 为其构建的 de~Bruijn 索引引用为 $(0,\; (\text{``\_\_callee''}, \text{``\_\_whnf''}))$,另有自有 label $\olbl{__parameter}$;
    \item \emph{链节点}(应用之下的 $\olbl{__whnf}$):一条引用 \mintinline{yaml}{[__call, __result, __whnf]},
      $\lexical$ 为其构建的 de~Bruijn 索引引用为 $(0,\; (\text{``\_\_call''}, \text{``\_\_result''}, \text{``\_\_whnf''}))$,无自有 label;
    \item \emph{抽象节点}:自有 label $\{\text{``\_\_whnf''}\}$,无引用;
    \item \emph{抽象形状}(抽象之下的 $\olbl{__whnf}$):自有 label $\{x, \text{``\_\_parameter''}, \text{``\_\_result''}\}$,无引用;
    \item \emph{别名节点}(形状之下的 $x$):一条引用 \mintinline[escapeinside=||]{yaml}{[parameter]},
      $\lexical$ 为其构建的 de~Bruijn 索引引用为 $(0,\; (\text{``\_\_parameter''}))$,因为形状自身即定义 $\olbl{__parameter}$ 的最近外围作用域;
    \item \emph{形参声明}(形状之下的 $\olbl{__parameter}$):无成员;
    \item \emph{变量节点}($\mathcal{T}(x)$ 的像):一条引用 \mintinline[escapeinside=||]{yaml}{[|$x$|]},
      $\lexical$ 为其构建的 de~Bruijn 索引引用为 $(n,\; (x))$,其中 $n$ 为从该节点的外围作用域 $\init$-上行、
      直到绑定 $x$ 的抽象形状所经过的作用域层数。
  \end{enumerate}
  子项的像(体、被调用者、实参)各居于 $\olbl{__result}$、$\olbl{__callee}$、
  redex 节点的 $\olbl{__parameter}$ 覆写之下,递归地属于第 1、4、8 类。
  特别地,书写引用的\emph{起点}(origin)只会是 redex 节点、链节点、别名节点或变量节点,
  而变量引用的 $\init$-上行只经过第 1、2、4、5 类节点与各子项槽位,
  绝不以形参声明或别名节点为其定义点作用域。
  \else
  Let $M_0$ be a closed $\lambda$-term and $P = \mathcal{T}(M_0)$.
  Every written node of $P$ (the positions in $\dom(\inherits)$ and
  their members) belongs to exactly one of the following eight
  classes, and every written node carries at most one reference:
  \begin{enumerate}
    \item an \emph{application node}: own labels
      $\{\text{``\_\_callee''}, \text{``\_\_call''}, \text{``\_\_whnf''}\}$, no reference;
    \item a \emph{redex node} (the $\olbl{__call}$ of an application):
      one reference \mintinline{yaml}{[__callee, __whnf]}, for which $\lexical$ builds the
      de-Bruijn-indexed reference $(0,\; (\text{``\_\_callee''}, \text{``\_\_whnf''}))$, plus the own label
      $\olbl{__parameter}$;
    \item a \emph{chain node} (the $\olbl{__whnf}$ of an application):
      one reference \mintinline{yaml}{[__call, __result, __whnf]}, for which $\lexical$
      builds the de-Bruijn-indexed reference $(0,\; (\text{``\_\_call''}, \text{``\_\_result''}, \text{``\_\_whnf''}))$, no own labels;
    \item an \emph{abstraction node}: the single own label
      $\{\text{``\_\_whnf''}\}$, no reference;
    \item \begin{sloppypar}an \emph{abstraction shape} (the $\olbl{__whnf}$ of an
      abstraction): own labels $\{x, \text{``\_\_parameter''}, \text{``\_\_result''}\}$,
      no reference;\end{sloppypar}
    \item an \emph{alias node} (the $x$ of a shape): one reference
      \mintinline[escapeinside=||]{yaml}{[parameter]}, for which $\lexical$ builds the
      de-Bruijn-indexed reference $(0,\; (\text{``\_\_parameter''}))$, because the shape itself is the
      nearest enclosing scope defining $\olbl{__parameter}$;
    \item a \emph{parameter declaration} (the $\olbl{__parameter}$ of a
      shape): no members;
    \item a \emph{variable node} (the image of $\mathcal{T}(x)$): one
      reference \mintinline[escapeinside=||]{yaml}{[|$x$|]}, for which $\lexical$ builds the
      de-Bruijn-indexed reference $(n,\; (x))$, where $n$ is the number of scope levels climbed by
      $\init$ from the node's enclosing scope up to the abstraction
      shape binding $x$.
  \end{enumerate}
  The images of sub-terms (body, operator, argument) live under
  $\olbl{__result}$, under $\olbl{__callee}$, and under the redex node's
  $\olbl{__parameter}$ override, and belong recursively to
  classes~1, 4, and~8.
  In particular, the \emph{origin} of a written reference is always a
  redex node, a chain node, an alias node, or a variable node, and the
  $\init$-climb of a variable reference passes only through nodes of
  classes~1, 2, 4, 5 and the sub-term slots; it never has a parameter
  declaration or an alias node as its definition-site scope.
  \fi
\end{lemma}

\begin{proof}
\ifPaperChinese
  对 $\mathcal{T}$ 的三条规则作结构归纳。三条规则各自书写的节点恰为上列类别
  (应用规则产出第 1、2、3 类,抽象规则产出第 4、5、6、7 类,变量规则产出第 8 类),
  且每条规则在每个节点处至多写下一条引用。
  别名引用的 de~Bruijn 索引:词法引用 \mintinline[escapeinside=||]{yaml}{[parameter]} 自
  $\init(\text{别名节点}) = \text{形状}$ 起测试 $\has(\cdot, \text{``\_\_parameter''})$,
  形状自有 $\olbl{__parameter}$,故 $n = 0$。
  变量引用的 de~Bruijn 索引同理:自 $\init(\text{变量节点})$ 起逐层上行,首个定义 label~$x$
  的作用域是绑定 $x$ 的抽象形状(合成 label 与用户变量名不相交),
  故 $n$ 即所经作用域层数。
  最后一句:引用只出现在第 2、3、6、8 类节点(逐规则检查);
  变量节点的 $\init$-祖先是包含它的槽位与第 1、2、4、5 类节点,
  而形参声明(第 7 类,无成员)与别名节点(第 6 类,其成员只是引用)之下不含任何变量节点,
  故它们不会作为任何变量引用的定义点作用域出现。
\else
  By structural induction on the three rules of $\mathcal{T}$. The
  nodes written by each rule are exactly the classes listed (the
  application rule produces classes~1, 2, 3, the abstraction rule
  classes~4, 5, 6, 7, and the variable rule class~8), and each rule
  writes at most one reference per node.
  For the alias reference, $\lexical$ builds the index as follows: the lexical reference
  \mintinline[escapeinside=||]{yaml}{[parameter]} tests
  $\has(\cdot, \text{``\_\_parameter''})$ outward from
  $\init(\text{alias}) = \text{the shape}$, and the shape owns
  $\olbl{__parameter}$, so $n = 0$.
  For the variable reference likewise: the climb from
  $\init(\text{variable node})$ stops at the first scope defining
  label~$x$, which is the abstraction shape binding $x$ (synthetic
  labels are disjoint from user variable names), so $n$ is the number of scope levels
  passed.
  For the final sentence: references occur only at nodes of
  classes~2, 3, 6, 8 (check per rule); the $\init$-ancestors of a
  variable node are the slots containing it and nodes of
  classes~1, 2, 4, 5, while a parameter declaration (class~7, no
  members) and an alias node (class~6, whose only member is a
  reference) contain no variable node beneath them, so neither ever
  occurs as the definition-site scope of a variable reference.
\fi
\end{proof}


\subsection{\bilingual{The Call-by-Name Machine}{Call-by-Name 机}}

\ifPaperChinese
收敛的充分性的证明把 $P = \mathcal{T}(M_0)$ 的收敛观测与一台标准的
call-by-name 求值机联系起来。这台机器只是叙述词汇:它的每个配置都将被
解码为\emph{同一棵} mixin 树 $P$ 中已经存在的位置,每条迁移都对应对
方程~(\ref{eq:supers})--(\ref{eq:this}) 的一次直接计算;
全部证明自始至终只涉及这一个程序的一个最小不动点。
\else
The proof of adequacy for convergence relates the convergence observation of
$P = \mathcal{T}(M_0)$ to a standard call-by-name evaluation machine.
The machine is presentation vocabulary only: each of its
configurations will be decoded to a position already present in the
\emph{one} mixin tree $P$, and each of its transitions corresponds to
one direct calculation with
equations~(\ref{eq:supers})--(\ref{eq:this}); throughout, the proofs
concern a single least fixed point of this single program.
\fi

\begin{definition}[\bilingual{Krivine machine~\cite{krivine2007-callbyname-machine}}{Krivine 机~\cite{krivine2007-callbyname-machine}}]\label{def:krivine}
  \ifPaperChinese
  固定闭合 $\lambda$-项 $M_0$。\emph{配置}为三元组
  $\langle u,\; \rho,\; S \rangle$,其中 $u$ 是 $M_0$ 的一个子项出现,
  \emph{环境} $\rho$ 把 $u$ 的每个自由变量映射到一个\emph{闭包}
  (一个子项出现与其自身环境的对),\emph{栈} $S$ 是闭包的有限列表。
  迁移规则:
  \begin{align*}
    \langle u_1\, u_2,\; \rho,\; S \rangle
      &\longrightarrow \langle u_1,\; \rho,\; (u_2, \rho) \cdot S \rangle
      && \text{(push)} \\
    \langle \lambda x.\, u',\; \rho,\; c \cdot S \rangle
      &\longrightarrow \langle u',\; \rho[x \mapsto c],\; S \rangle
      && \text{($\beta$)} \\
    \langle x,\; \rho,\; S \rangle
      &\longrightarrow \langle \rho(x).\mathrm{term},\; \rho(x).\mathrm{env},\; S \rangle
      && \text{(lookup)}
  \end{align*}
  初始配置为 $\langle M_0,\; \varnothing,\; \varnothing \rangle$;
  \emph{停机}配置为栈空的抽象 $\langle \lambda x.\, u',\; \rho,\; \varnothing \rangle$。
  因 $M_0$ 闭合,lookup 恒有定义,且停机配置是唯一无后继的配置。
  \else
  Fix a closed $\lambda$-term $M_0$. A \emph{configuration} is a
  triple $\langle u,\; \rho,\; S \rangle$ where $u$ is an occurrence
  of a sub-term of $M_0$, the \emph{environment} $\rho$ maps each free
  variable of $u$ to a \emph{closure} (a pair of a sub-term occurrence
  and its own environment), and the \emph{stack} $S$ is a finite list
  of closures. The transitions are:
  \begin{align*}
    \langle u_1\, u_2,\; \rho,\; S \rangle
      &\longrightarrow \langle u_1,\; \rho,\; (u_2, \rho) \cdot S \rangle
      && \text{(push)} \\
    \langle \lambda x.\, u',\; \rho,\; c \cdot S \rangle
      &\longrightarrow \langle u',\; \rho[x \mapsto c],\; S \rangle
      && \text{($\beta$)} \\
    \langle x,\; \rho,\; S \rangle
      &\longrightarrow \langle \rho(x).\mathrm{term},\; \rho(x).\mathrm{env},\; S \rangle
      && \text{(lookup)}
  \end{align*}
  The initial configuration is
  $\langle M_0,\; \varnothing,\; \varnothing \rangle$; a \emph{halting}
  configuration is an abstraction under the empty stack,
  $\langle \lambda x.\, u',\; \rho,\; \varnothing \rangle$. Since
  $M_0$ is closed, lookup is always defined, and the halting
  configurations are exactly the configurations without a successor.
  \fi
\end{definition}

\begin{proposition}[\bilingual{Machine correctness}{机器正确性}]\label{prop:krivine-whnf}
  \ifPaperChinese
  闭合 $\lambda$-项 $M_0$ 有弱头范式,当且仅当 Krivine 机自
  $\langle M_0, \varnothing, \varnothing \rangle$ 起的运行停机;
  运行中的每次 $\beta$-迁移对应一步弱头归约 $\to_h$,
  而 push 与 lookup 是管理步。
  \else
  A closed $\lambda$-term $M_0$ has a weak head normal form if and
  only if the Krivine machine run from
  $\langle M_0, \varnothing, \varnothing \rangle$ halts; each
  $\beta$-transition of the run corresponds to one weak-head reduction
  step $\to_h$, while push and lookup are administrative.
  \fi
\end{proposition}

\begin{proof}
\ifPaperChinese
  这是 Krivine 机对 call-by-name 弱头归约的标准正确性
  ~\cite{krivine2007-callbyname-machine}:把配置读回为
  $\lambda$-项(对控制项施行环境替换,再依次应用栈中各闭包的读回),
  则 push 与 lookup 保持读回不变,$\beta$-迁移把读回按 $\to_h$ 收缩一步;
  闭合项的弱头范式必为抽象,恰对应栈空的抽象配置。
\else
  This is the standard correctness of the Krivine machine for
  call-by-name weak-head
  reduction~\cite{krivine2007-callbyname-machine}: reading back a
  configuration as a $\lambda$-term (apply the environment
  substitution to the control, then apply the readbacks of the stacked
  closures in order), push and lookup preserve the readback, and a
  $\beta$-transition contracts it by one $\to_h$ step; the weak head
  normal form of a closed term is an abstraction, which corresponds
  exactly to a halting configuration.
\fi
\end{proof}

\subsection{\bilingual{Machine Simulation}{机器模拟}}

\ifPaperChinese
本小节把 Krivine 机的每个可达配置解码为 $P = \mathcal{T}(M_0)$ 中的路径。
关键观察:翻译从不执行 $\beta$-替换。redex 就是 $\olbl{__call}$ 节点
(它经 \mintinline{yaml}{[__callee, __whnf]} 继承被调用者的抽象形状并覆写
$\olbl{__parameter}$),reduct 就是其下的 $(\text{``\_\_call''}, \text{``\_\_result''})$
子树:收缩一次 head redex 即向下投影一层,整条归约序列在一棵树中
空间化地铺开。机器的环境与栈也不是额外数据:一个实参 thunk 就是一条
形如 $\alpha \cat (\text{``\_\_call''}, \text{``\_\_parameter''})$ 的路径,
其中 $\alpha$ 是绑定它的那次应用的解码位置。
\else
This subsection decodes every reachable configuration of the Krivine
machine to a path of $P = \mathcal{T}(M_0)$. The key observation is
that the translation never performs $\beta$-substitution. The redex
\emph{is} the $\olbl{__call}$ node (it inherits the callee's abstraction
shape through \mintinline{yaml}{[__callee, __whnf]} and overrides
$\olbl{__parameter}$), and the reduct \emph{is} the
$(\text{``\_\_call''}, \text{``\_\_result''})$ subtree beneath it: contracting
one head redex is projecting one level deeper, and the whole reduction
sequence is laid out spatially in one tree. Environments and stacks
are not extra data either: an argument thunk is a path of the form
$\alpha \cat (\text{``\_\_call''}, \text{``\_\_parameter''})$, where $\alpha$
is the decoded position of the application that bound it.
\fi

\begin{definition}[\bilingual{Configuration decoding}{配置解码}]\label{def:config-paths}
  \ifPaperChinese
  沿机器运行以递归定义:每个可达配置
  $K = \langle u, \rho, S \rangle$ 获得一条\emph{焦点路径} $\pi(K)$,
  每个闭包 $c$ 获得一条\emph{thunk 路径} $\theta(c)$ 与一条
  \emph{绑定点路径};环境条目与栈条目沿用其闭包的指派。
  \begin{itemize}
    \item 初始配置的焦点为根路径 $()$。
    \item \textbf{push}(焦点 $\pi$,控制项为应用 $u_1\, u_2$):
      新闭包 $(u_2, \rho)$ 得
      $\theta = \pi \cat (\text{``\_\_call''}, \text{``\_\_parameter''})$;
      新焦点为 $\pi \cat (\text{``\_\_callee''})$。
    \item \begin{sloppypar}\textbf{$\beta$}(焦点 $\pi$,控制项为抽象 $\lambda x.\, u'$,
      栈顶闭包的 thunk 路径为
      $\theta = \alpha \cat (\text{``\_\_call''},\allowbreak \text{``\_\_parameter''})$):
      新焦点为 $\alpha \cat (\text{``\_\_call''}, \text{``\_\_result''})$;
      环境条目 $x \mapsto$ 该闭包,其\emph{绑定点路径}为 redex 节点
      $\alpha \cat (\text{``\_\_call''})$。
      \end{sloppypar}
    \item \textbf{lookup}(控制项为变量 $x$):新焦点为
      $\theta(\rho(x))$,即所记录的 thunk 路径。
  \end{itemize}
  记 $\mathrm{wpos}(K)$ 为控制项 $u$ 的像在 $P$ 中的书写槽位
  (引理~\ref{lem:node-classification} 的第 1、4、8 类节点所在的位置)。
  \else
  By recursion along the machine run, every reachable configuration
  $K = \langle u, \rho, S \rangle$ receives a \emph{focus path}
  $\pi(K)$, and every closure $c$ receives a \emph{thunk path}
  $\theta(c)$ and a \emph{binder path}; environment and stack entries
  carry the assignments of their closures.
  \begin{itemize}
    \item The initial configuration has focus $()$, the root.
    \item \textbf{push} (focus $\pi$, control an application
      $u_1\, u_2$): the new closure $(u_2, \rho)$ receives
      $\theta = \pi \cat (\text{``\_\_call''}, \text{``\_\_parameter''})$;
      the new focus is $\pi \cat (\text{``\_\_callee''})$.
    \item \textbf{$\beta$} (focus $\pi$, control an abstraction
      $\lambda x.\, u'$, top closure with thunk path
      $\theta = \alpha \cat (\text{``\_\_call''},\allowbreak \text{``\_\_parameter''})$):
      the new focus is
      $\alpha \cat (\text{``\_\_call''}, \text{``\_\_result''})$; the
      environment entry $x \mapsto$ that closure has \emph{binder
      path} the redex node $\alpha \cat (\text{``\_\_call''})$.
    \item \textbf{lookup} (control a variable $x$): the new focus is
      $\theta(\rho(x))$, the recorded thunk path.
  \end{itemize}
  Write $\mathrm{wpos}(K)$ for the written slot of the control's image
  in $P$ (the position of the class-1, 4, or 8 node of
  Lemma~\ref{lem:node-classification}).
  \fi
\end{definition}

\begin{lemma}[\bilingual{Run invariant and forward simulation}{运行不变量与正向模拟}]
\label{lem:decode-invariant}
  \ifPaperChinese
  对每个可达配置 $K = \langle u, \rho, S \rangle$(焦点 $\pi$,环境条目自内而外为
  $e_1, \ldots, e_m$),下列事实属于 $\lfp(T_P)$:
  \begin{enumerate}
    \item[(i)] $\mathrm{wpos}(K) \in \overrides(\pi)$;
    \item[(ii)] 存在\emph{梯子} $\sigma_0, \sigma_1, \ldots, \sigma_m$ 与书写作用域
      $D_1, \ldots, D_m$,其中 $\sigma_0 = \init(\pi)$,$D_j$ 为绑定第 $j$ 层
      (自内而外)的书写作用域,$\sigma_j$ 为其被并入的位置:对 thunk 条目
      $\sigma_j = \alpha_j \cat (\text{``\_\_call''})$(其 $\olbl{__parameter}$ 覆写
      即该条目的 thunk 路径),对声明条目 $\sigma_j$ 为相应形状实例;
      且对每个 $1 \le j \le m$,对
      $(\sigma_j,\; D_j) \in \supers(\sigma_{j-1}$ 的相应查询$)$
      的匹配事实成立——精确地,
      $(\init(\sigma_j),\, D_j) \in \supers(\sigma_{j-1})$,
      并且 $\mathrm{wpos}$ 沿 $u$ 内部各自有 label 的延伸满足
      $\mathrm{wpos}(K) \cat w \in \overrides(\pi \cat w)$
      (对 $u$ 的像内每条自有路径 $w$,由 (i) 与方程~\ref{eq:overrides}
      沿 $w$ 逐标签得出)。
  \end{enumerate}
  进一步(\emph{正向模拟}):若自 $K$ 的运行\emph{停机}——即在有限步内到达栈深回到
  $|S|$、控制项为抽象的配置 $K'$——则 $K'$ 亦满足 (i)、(ii),且存在
  $\lfp(T_P)$ 中的链事实族:一条自 $\chi(K) = \pi \cat (\text{``\_\_whnf''})$ 到某
  $c_j$ 的逐跳 $\bases$ 事实链,以及 override 事实
  $q \in \overrides(c_j)$,其中 $q$ 为 $K'$ 控制项之像的抽象形状
  (末标签为 $\olbl{__whnf}$,其下 $\olbl{__parameter}$、$\olbl{__result}$ 为自有 label)。
  \else
  For every reachable configuration $K = \langle u, \rho, S \rangle$
  (focus $\pi$; environment entries $e_1, \ldots, e_m$ from innermost
  outward), the following facts belong to $\lfp(T_P)$:
  \begin{enumerate}
    \item[(i)] $\mathrm{wpos}(K) \in \overrides(\pi)$;
    \item[(ii)] there are a \emph{ladder} $\sigma_0, \sigma_1,
      \ldots, \sigma_m$ and written scopes $D_1, \ldots, D_m$, where
      $\sigma_0 = \init(\pi)$, $D_j$ is the written scope binding
      level $j$ (innermost first), and $\sigma_j$ is the position at
      which it is incorporated: for a thunk entry
      $\sigma_j = \alpha_j \cat (\text{``\_\_call''})$ (whose
      $\olbl{__parameter}$ override is that entry's thunk path), and
      for a declaration entry $\sigma_j$ is the corresponding shape
      instance; and for each $1 \le j \le m$ the matching fact
      $(\init(\sigma_j),\, D_j) \in \supers(\sigma_{j-1})$ holds.
      Moreover $\mathrm{wpos}$ extends along own labels inside $u$:
      $\mathrm{wpos}(K) \cat w \in \overrides(\pi \cat w)$ for every
      own-label path $w$ inside $u$'s image (from~(i) and
      equation~\ref{eq:overrides}, label by label along $w$).
  \end{enumerate}
  Furthermore (\emph{forward simulation}): if the run from $K$
  \emph{halts} --- reaches, in finitely many steps, a configuration
  $K'$ with stack depth back to $|S|$ and an abstraction in control
  --- then $K'$ also satisfies (i) and (ii), and $\lfp(T_P)$
  contains a family of chain facts: a hop-by-hop $\bases$ chain from
  $\chi(K) = \pi \cat (\text{``\_\_whnf''})$ to some $c_j$, together
  with an override fact $q \in \overrides(c_j)$ where $q$ is the
  abstraction shape of $K'$'s control image (its last label is
  $\olbl{__whnf}$, owning $\olbl{__parameter}$ and $\olbl{__result}$).
  \fi
\end{lemma}

\begin{proof}
\ifPaperChinese
  对停机段的长度作强归纳,同时按控制项的类别
  (引理~\ref{lem:node-classification})分情形;(i)、(ii) 的维持与链事实的
  构造在同一归纳内完成,故不再有跨引理的相互借用。
  初始配置:$\mathrm{wpos} = () = \pi$,自反 override
  (方程~\ref{eq:overrides} 的根基例);$m = 0$,梯子平凡。

  \paragraph{\bilingual{}{控制项为抽象}}
  段长为零。链事实:由 (i) 与方程~(\ref{eq:supers})的自反项,
  $(\_,\; \mathrm{wpos}(K) \cat (\text{``\_\_whnf''})) \in \supers(\chi(K))$,
  即零跳链加 override 事实,而
  $\mathrm{wpos}(K) \cat (\text{``\_\_whnf''})$ 正是该抽象之像的形状。

  \paragraph{\bilingual{}{控制项为应用 $u_1\, u_2$}}
  push 步:新焦点 $\pi \cat (\text{``\_\_callee''})$ 的 (i) 由 (i) 与
  方程~(\ref{eq:overrides})一步得出($\olbl{__callee}$ 为应用节点的自有 label);
  新闭包记录当前梯子与
  $\mathrm{wpos}(K) \cat (\text{``\_\_call''}, \text{``\_\_parameter''})
  \in \overrides(\pi \cat (\text{``\_\_call''}, \text{``\_\_parameter''}))$
  ——后者由 (i) 沿自有路径 $(\text{``\_\_call''}, \text{``\_\_parameter''})$
  延伸得出,这就是实参槽位的 (i)。
  停机的运行必含 callee 段的停机前缀;对其应用归纳假设(段更短),
  得 $\beta$ 点配置满足 (i)、(ii),并得 callee 链事实族:自
  $\pi \cat (\text{``\_\_callee''}, \text{``\_\_whnf''})$ 到 callee 形状 $q_c$ 的链。
  $\beta$ 步:新焦点 $\alpha \cat (\text{``\_\_call''}, \text{``\_\_result''})$
  (此处 $\alpha = \pi$)的 (i) 计算如下:redex 节点的引用
  \mintinline{yaml}{[__callee, __whnf]}($n = 0$,与 override 无关)给出
  $\bases(\pi \cat (\text{``\_\_call''})) =
  \{\pi \cat (\text{``\_\_callee''}, \text{``\_\_whnf''})\}$;
  拼上 callee 链事实($\bases^*$ 在一次 $T_P$ 应用内闭包),得
  $(\_,\; q_c) \in \supers(\pi \cat (\text{``\_\_call''}))$;
  于是 $q_c \cat (\text{``\_\_result''}) \in
  \overrides(\pi \cat (\text{``\_\_call''}, \text{``\_\_result''}))$
  (方程~\ref{eq:overrides},分支项),此即新焦点的 (i)。
  新梯子:$\sigma_1 = \pi \cat (\text{``\_\_call''})$、$D_1 = q_c$ 的外围形状,
  匹配事实即上句的 $\supers$ 对;旧梯子逐级上移一位,其事实原样保留。
  体段停机由归纳假设(段更短)给出其链事实族;根侧链事实族由
  $\chi(K)$ 的一跳
  $\bases(\chi(K)) = \{\pi \cat (\text{``\_\_call''}, \text{``\_\_result''},
  \text{``\_\_whnf''})\}$(class 3 引用,$n = 0$,与 override 无关地为单点)
  接体段链而成。

  \paragraph{\bilingual{}{控制项为变量 $x$,$\rho(x)$ 为 thunk 条目}}
  设 $x$ 的书写引用为 $(n, (x))$。由 (ii) 的梯子,自
  $\sigma_0 = \init(\pi)$ 起,逐级消费匹配事实
  $(\init(\sigma_j), D_j) \in \supers(\sigma_{j-1})$,$n$ 步后前沿为
  $\{\sigma_n\}$,$\sigma_n$ 为绑定 $x$ 的形状的并入位置
  $\alpha_x \cat (\text{``\_\_call''})$;按方程~(\ref{eq:resolve})拼接 $(x)$
  得别名拷贝 $\sigma_n' = $ 形状实例 $\cat (x)$——其 (i) 由 $\sigma_n$ 处
  callee 链的形状对沿自有 label $x$ 延伸得出(与 $\beta$ 情形同一计算)。
  别名引用 $(0, (\text{``\_\_parameter''}))$ 的定点目标为
  $\alpha_x \cat (\text{``\_\_call''}, \text{``\_\_parameter''}) = \theta(\rho(x))$。
  lookup 步取回闭包所存 (i)、(ii)(其 push 时刻已证);
  链事实族 = 上述转发事实接 thunk 段(归纳假设)的链事实族。

  \paragraph{\bilingual{}{声明条目与 enter/block}}
  见引理~\ref{lem:open-simulation};停机段从不查找声明条目。
\else
  By strong induction on the length of the halting segment, with
  cases on the class of the control
  (Lemma~\ref{lem:node-classification}); maintaining (i) and (ii)
  and constructing the chain facts happen within the same induction,
  so there is no mutual borrowing across lemmas.
  For the initial configuration, $\mathrm{wpos} = () = \pi$ is a
  reflexive override (the root base case of
  equation~\ref{eq:overrides}); $m = 0$ and the ladder is trivial.

  \paragraph{Control is an abstraction}
  The segment has length zero. Chain facts: by (i) and the reflexive
  term of equation~(\ref{eq:supers}),
  $(\_,\; \mathrm{wpos}(K) \cat (\text{``\_\_whnf''})) \in
  \supers(\chi(K))$: a zero-hop chain plus the override fact, and
  $\mathrm{wpos}(K) \cat (\text{``\_\_whnf''})$ is exactly the shape of
  the abstraction's image.

  \paragraph{Control is an application $u_1\, u_2$}
  The push step: (i) for the new focus $\pi \cat (\text{``\_\_callee''})$
  follows from (i) in one application of
  equation~(\ref{eq:overrides}) ($\olbl{__callee}$ is an own label of
  the application node); the new closure records the current ladder
  together with
  $\mathrm{wpos}(K) \cat (\text{``\_\_call''}, \text{``\_\_parameter''})
  \in \overrides(\pi \cat (\text{``\_\_call''}, \text{``\_\_parameter''}))$
  --- obtained by extending (i) along the own-label path
  $(\text{``\_\_call''}, \text{``\_\_parameter''})$ --- which is (i) for the
  argument slot.
  A halting run contains a halting prefix for the callee segment;
  the induction hypothesis (shorter segment) gives (i) and (ii) at
  the $\beta$-point configuration and the callee chain-fact family,
  from $\pi \cat (\text{``\_\_callee''}, \text{``\_\_whnf''})$ to the callee
  shape $q_c$.
  The $\beta$ step: (i) for the new focus
  $\alpha \cat (\text{``\_\_call''}, \text{``\_\_result''})$ (here
  $\alpha = \pi$) is computed as follows: the redex node's reference
  \mintinline{yaml}{[__callee, __whnf]} ($n = 0$, independent of the
  override) gives
  $\bases(\pi \cat (\text{``\_\_call''})) =
  \{\pi \cat (\text{``\_\_callee''}, \text{``\_\_whnf''})\}$; splicing the
  callee chain facts ($\bases^*$ closes within one application of
  $T_P$) yields
  $(\_,\; q_c) \in \supers(\pi \cat (\text{``\_\_call''}))$; hence
  $q_c \cat (\text{``\_\_result''}) \in
  \overrides(\pi \cat (\text{``\_\_call''}, \text{``\_\_result''}))$
  (a branch term of equation~\ref{eq:overrides}), which is (i) for
  the new focus. The new ladder: $\sigma_1 = \pi \cat
  (\text{``\_\_call''})$ and $D_1 = $ the enclosing shape of $q_c$, with
  the matching fact being the $\supers$ pair of the previous
  sentence; the old ladder shifts up by one, its facts unchanged.
  The body segment halts by the induction hypothesis (shorter
  segment), yielding its chain-fact family; the root-side family is
  the one hop
  $\bases(\chi(K)) = \{\pi \cat (\text{``\_\_call''}, \text{``\_\_result''},
  \text{``\_\_whnf''})\}$ (the class-3 reference, $n = 0$, a singleton
  independently of the override) spliced onto the body family.

  \paragraph{Control is a variable $x$ with a thunk entry $\rho(x)$}
  Let $x$'s written reference be $(n, (x))$. By the ladder of (ii),
  starting from $\sigma_0 = \init(\pi)$ and consuming the matching
  facts $(\init(\sigma_j), D_j) \in \supers(\sigma_{j-1})$ level by
  level, after $n$ steps the frontier is $\{\sigma_n\}$ with
  $\sigma_n = \alpha_x \cat (\text{``\_\_call''})$, the incorporation
  position of the shape binding $x$; equation~(\ref{eq:resolve})
  appends $(x)$, giving the alias copy --- whose (i) follows by
  extending the shape pair at $\sigma_n$ along the own label $x$
  (the same calculation as the $\beta$ case). The alias reference
  $(0, (\text{``\_\_parameter''}))$ has the site-relative target
  $\alpha_x \cat (\text{``\_\_call''}, \text{``\_\_parameter''}) =
  \theta(\rho(x))$. The lookup step takes back the closure's stored
  (i) and (ii) (established at its push time); the chain-fact family
  is the forwarding facts above spliced onto the thunk segment's
  family (induction hypothesis).

  \paragraph{Declaration entries, enter, and block}
  See Lemma~\ref{lem:open-simulation}; a halting segment never looks
  up a declaration entry.
\fi
\end{proof}

\begin{lemma}[\bilingual{Forward simulation}{正向模拟}]\label{lem:sim-forward}
  \ifPaperChinese
  若机器自初始配置起的运行停机,则 $P$ 收敛(定义~\ref{def:convergence})。
  \else
  If the machine run from the initial configuration halts, then $P$
  converges (Definition~\ref{def:convergence}).
  \fi
\end{lemma}

\begin{proof}
\ifPaperChinese
  \begin{sloppypar}
  引理~\ref{lem:decode-invariant} 的正向模拟部分在初始配置的实例:
  链事实族给出 $(\_,\, q) \in \supers((\text{``\_\_whnf''}))$ 且
  $q \snoc \text{``\_\_parameter''} \in \dom(\inherits)$。
  \end{sloppypar}
\else
  The forward-simulation part of Lemma~\ref{lem:decode-invariant},
  instantiated at the initial configuration: the chain-fact family
  yields $(\_,\, q) \in \supers((\text{``\_\_whnf''}))$ with
  $q \snoc \text{``\_\_parameter''} \in \dom(\inherits)$.
\fi
\end{proof}

\begin{lemma}[\bilingual{Lookup}{查找}]\label{lem:lookup}
  \ifPaperChinese
  设可达配置 $K$ 的控制项为变量 $x$,其书写引用为 $(n, (x))$,
  $\rho(x)$ 为 thunk 条目。则引理~\ref{lem:decode-invariant}(ii) 的梯子事实
  逐级可导出 $x$ 的解析:走步经 $n$ 步抵达 $\{\sigma_n\}$、经别名一步转发到
  $\theta(\rho(x))$,如该引理变量情形所计算。特别地,当 $y \neq x$ 时,
  $(n_y, (y))$ 的爬升止于绑定 $y$ 的形状,不经过 $x$ 的绑定形状:
  在 $x$ 的 redex 处以 $\olbl{__parameter}$ 覆写实参,对 $y$ 的解析没有影响。
  \else
  Let a reachable configuration $K$ have a variable $x$ as control,
  with written reference $(n, (x))$ and a thunk entry $\rho(x)$.
  Then the ladder facts of Lemma~\ref{lem:decode-invariant}(ii)
  derive $x$'s resolution level by level: the walk reaches
  $\{\sigma_n\}$ in $n$ steps and forwards through the alias to
  $\theta(\rho(x))$, as computed in that lemma's variable case. In
  particular, for $y \neq x$ the climb of $(n_y, (y))$ stops at the
  shape binding $y$ and does not pass through $x$'s binding shape:
  overriding $\olbl{__parameter}$ with an argument at $x$'s redex has
  no effect on the resolution of $y$; the climb stops at the first
  scope defining label $y$, synthetic labels being disjoint from
  user variable names.
  \fi
\end{lemma}

\begin{lemma}[\bilingual{Backward simulation}{反向模拟}]\label{lem:sim-backward}
  \ifPaperChinese
  设 $K$ 为可达配置,焦点 $\pi$,观测节点 $\chi = \pi \cat (\text{``\_\_whnf''})$。
  设存在链 $\chi = c_0 \to c_1 \to \cdots \to c_k$
  (每步 $c_{j+1} \in \bases(c_j)$ 为 $\lfp(T_P)$ 中的事实)
  与 override 事实 $q \in \overrides(c_k)$,其中
  $q \snoc \text{``\_\_parameter''} \in \dom(\inherits)$,
  且这 $k{+}1$ 条事实的 Kleene 秩最大为 $m$。则:
  \begin{enumerate}
    \item[(1)] 机器自 $K$ 起,经有限步后到达控制项为抽象、栈与 $K$ 相同的配置;
    \item[(2)] (\emph{一致性})所给链与解码运行逐跳一致:$c_1, \ldots, c_k$
      依次为该运行沿途的观测节点,$q$ 为停机时控制项之像的形状;
      给定诸事实的推导中,每个被消费的 $\this$ 走步事实与
      引理~\ref{lem:lookup} 的机器走步一致。
  \end{enumerate}
  \else
  Let $K$ be a reachable configuration with focus $\pi$ and
  observation node $\chi = \pi \cat (\text{``\_\_whnf''})$. Suppose there
  are a chain $\chi = c_0 \to c_1 \to \cdots \to c_k$ (each step
  $c_{j+1} \in \bases(c_j)$ a fact of $\lfp(T_P)$) and an override
  fact $q \in \overrides(c_k)$ with
  $q \snoc \text{``\_\_parameter''} \in \dom(\inherits)$, these $k{+}1$
  facts having maximum Kleene rank $m$. Then:
  \begin{enumerate}
    \item[(1)] the machine, run from $K$, reaches in finitely many
      steps a configuration whose control is an abstraction and
      whose stack equals that of $K$;
    \item[(2)] (\emph{coincidence}) the given chain coincides hop by
      hop with the decoded run: $c_1, \ldots, c_k$ are, in order,
      the observation nodes visited by that run, and $q$ is the
      shape of the halting control's image; within the given facts'
      derivations, every consumed $\this$-walk fact coincides with
      the machine walk of Lemma~\ref{lem:lookup}.
  \end{enumerate}
  \fi
\end{lemma}

\begin{proof}
\ifPaperChinese
  对 $m$ 作强归纳,同一 $m$ 内对 $k$ 作副归纳;按控制项的类别分情形。
  所有降秩均来自末步反演(引理~\ref{lem:kleene-iteration}):
  被反演事实的前提秩严格更小。

  \paragraph{\bilingual{}{控制项为抽象}}
  零步停机,(1) 即毕。(2):$k = 0$——若 $k \ge 1$,则 $c_1 \in \bases(c_0)$,
  而抽象控制下 $c_0 = \chi$ 的 override 均为末标签 $\olbl{__whnf}$ 的形状类
  节点(class 5),不携带引用,方程~(\ref{eq:bases})的推导式为空,矛盾;
  故链即自反对,$q$ 为该形状。

  \paragraph{\bilingual{}{控制项为应用}}
  先排除 $k = 0$:见证 override $q$ 须满足
  $q \snoc \text{``\_\_parameter''} \in \dom(\inherits)$ 且末标签为
  $\olbl{__whnf}$($\overrides$ 保持末标签,而 $\chi$ 以 $\olbl{__whnf}$ 结尾;
  redex 节点以 $\olbl{__call}$ 结尾,故被排除——见证必为抽象形状),
  而应用控制下 $c_0$ 的 override 为 class-3 链节点,自有 label 为空,
  自反项守卫不满足;故 $k \ge 1$,且由 class-3 引用的 $n = 0$ 定点解析,
  $c_1 = \pi \cat (\text{``\_\_call''}, \text{``\_\_result''}, \text{``\_\_whnf''})$
  被迫唯一,与解码运行的第一跳一致。
  callee 见证的提取:对事实 $c_2 \in \bases(c_1)$(若 $k \ge 2$;
  否则对 $q \in \overrides(c_1)$)作末步反演,其前提含
  $\overrides(c_1)$ 的某分支事实;再反演,其前提含
  $(\_,\; q_c \cat (\text{``\_\_result''}))$ 形的
  $\supers(\pi \cat (\text{``\_\_call''}, \text{``\_\_result''}))$ 对;
  又反演,其前提为装配
  $\bases^*(\pi \cat (\text{``\_\_call''}))$ 的诸 $\bases$ 事实——
  即 redex 的一跳
  $\pi \cat (\text{``\_\_callee''}, \text{``\_\_whnf''})$
  与 callee 链的全部逐跳事实、以及 callee 形状的 override 事实,
  秩均 $\le m - 3$。以之为链,对焦点 $\pi \cat (\text{``\_\_callee''})$ 的
  push 后配置应用主归纳假设($m$ 严格减小):callee 段停机且
  (2) 给出该链与 callee 段解码一致,故停机控制之像的形状恰为 $q_c$,
  $\beta$ 步的解码焦点恰为 $\pi \cat (\text{``\_\_call''}, \text{``\_\_result''})$,
  其观测节点即 $c_1$。剩余链 $c_1 \to \cdots \to c_k$ 与 override 事实
  秩 $\le m$、链长 $k - 1$:副归纳假设给出体段停机与一致性。
  两段拼接得 (1);(2) 由两段的一致性拼接。

  \paragraph{\bilingual{}{控制项为变量 $x$}}
  此时 $c_0 = \chi$ 的 override 分支经变量节点的引用产生:
  对 $c_1 \in \bases(c_0)$ 作末步反演,前提含 $\bases(\pi)$ 侧的
  $\resolve$/$\this$ 事实:一条 $n$ 步走步的逐步匹配对,每步在单点前沿
  $\{\sigma\}$ 处消费 $\supers(\sigma)$ 中 override 分量为某书写作用域
  $D$ 的对(引理~\ref{lem:node-classification} 末句:$D$ 为作用域类节点)。
  \emph{一致性}:对每个被消费的匹配对,再作末步反演,其前提为装配
  $\bases^*(\sigma)$ 的诸 $\bases$ 事实,每条秩严格更小。沿该链逐跳检查
  (对跳数再作内层归纳):class 2、3、6 节点的跳由 $n = 0$ 定点解析强制,
  与解码一致;class 8(变量)节点的跳,其事实本身构成一条更小秩的
  链事实族——秩的记账:$c_1 \in \bases(c_0)$ 在所给诸事实中秩 $\le m$,
  反演暴露的 $\resolve$/$\this$ 走步事实秩 $\le m - 1$,
  每个被消费的匹配对是某走步事实的前提、秩 $\le m - 2$,
  再反演得装配 $\bases^*(\sigma)$ 的诸 $\bases$ 事实、秩 $\le m - 3$,
  而 class-8 跳的链事实族恰由这类事实组成,故其最大秩 $\le m - 3 < m$,
  对 $m$ 的强归纳假设适用——其焦点配置为该变量出现的解码位置,
  主归纳假设的 (2) 恰断言该跳与机器一致。于是 $\bases^*(\sigma)$ 的被消费前缀与解码运行
  一致,故匹配对的 site 分量为梯子事实(引理~\ref{lem:decode-invariant}(ii))
  所指的机器位置,走步与引理~\ref{lem:lookup} 的机器走步逐级重合,
  链经别名转发抵达 $\theta(\rho(x))$:$c_1$ 落在 thunk 的观测链上,
  与解码一致。剩余链秩 $\le m$、长 $k - 1$,对 lookup 后配置用副归纳假设,
  得 (1)、(2)。
\else
  By strong induction on $m$, with a secondary induction on $k$
  within the same $m$; cases on the class of the control. Every rank
  decrease comes from last-step inversion
  (Lemma~\ref{lem:kleene-iteration}): the premises of an inverted
  fact have strictly smaller rank.

  \paragraph{Control is an abstraction}
  Halting in zero steps gives (1). For (2), $k = 0$: if $k \ge 1$
  then $c_1 \in \bases(c_0)$, but under abstraction control the
  overrides of $c_0 = \chi$ are shape-class nodes with last label
  $\olbl{__whnf}$ (class 5) carrying no reference, so the
  comprehension of equation~(\ref{eq:bases}) is empty --- a
  contradiction; the chain is the reflexive pair and $q$ is the
  shape.

  \paragraph{Control is an application}
  \begin{sloppypar}
  First exclude $k = 0$: a witnessing override $q$ must satisfy
  $q \snoc \text{``\_\_parameter''} \in \dom(\inherits)$ and end in
  $\olbl{__whnf}$ ($\overrides$ preserves the last label, and $\chi$
  ends in $\olbl{__whnf}$; a redex node ends in $\olbl{__call}$ and is
  thereby excluded --- the witness is necessarily an abstraction
  shape), while under application control the overrides of $c_0$ are
  class-3 chain nodes owning no label, so the reflexive term's guard
  fails; hence $k \ge 1$, and by the $n = 0$ site-relative
  resolution of the class-3 reference,
  $c_1 = \pi \cat (\text{``\_\_call''}, \text{``\_\_result''},
  \text{``\_\_whnf''})$ is forced, coinciding with the decoded run's
  first hop.
  Extraction of the callee witness: invert the fact
  $c_2 \in \bases(c_1)$ (if $k \ge 2$; otherwise
  $q \in \overrides(c_1)$); its premises include a branch fact of
  $\overrides(c_1)$; invert again: the premises include a
  $\supers(\pi \cat (\text{``\_\_call''}, \text{``\_\_result''}))$ pair of
  the form $(\_,\; q_c \cat (\text{``\_\_result''}))$; invert once more:
  the premises are the $\bases$ facts assembling
  $\bases^*(\pi \cat (\text{``\_\_call''}))$ --- the redex's one hop to
  $\pi \cat (\text{``\_\_callee''}, \text{``\_\_whnf''})$ together with all
  hop facts of the callee chain and the callee shape's override
  fact, all of rank $\le m - 3$. Using these as the chain, apply the
  main induction hypothesis (strictly smaller $m$) to the
  configuration after the push, at focus
  $\pi \cat (\text{``\_\_callee''})$: the callee segment halts, and
  by~(2) the chain coincides with the callee segment's decoding, so
  the halting control's shape is exactly $q_c$, the $\beta$ step's
  decoded focus is exactly
  $\pi \cat (\text{``\_\_call''}, \text{``\_\_result''})$, and its
  observation node is $c_1$. The remaining chain
  $c_1 \to \cdots \to c_k$ with the override fact has rank $\le m$
  and length $k - 1$: the secondary induction hypothesis gives the
  body segment's halting and coincidence. Concatenating the two
  segments gives (1); (2) follows by splicing the two coincidences.
  \end{sloppypar}

  \paragraph{Control is a variable $x$}
  Here the override branches of $c_0 = \chi$ arise through the
  variable node's reference: inverting $c_1 \in \bases(c_0)$ exposes
  the $\resolve$/$\this$ facts on the $\bases(\pi)$ side: an
  $n$-step walk of matching pairs, each consuming, at a singleton
  frontier $\{\sigma\}$, a pair of $\supers(\sigma)$ whose override
  component is a written scope $D$ (a scope-class node, by the final
  clause of Lemma~\ref{lem:node-classification}).
  \emph{Coincidence}: invert each consumed matching pair; its
  premises are the $\bases$ facts assembling $\bases^*(\sigma)$,
  each of strictly smaller rank. Check the chain hop by hop (an
  inner induction on the hop count): hops at class-2, 3, 6 nodes are
  forced by $n = 0$ site-relative resolution and coincide with the
  decoding; a hop at a class-8 (variable) node is itself a
  chain-fact family of strictly smaller rank (the rank bookkeeping:
  $c_1 \in \bases(c_0)$ has rank $\le m$ among the given facts, the
  $\resolve$/$\this$ walk facts exposed by its inversion have rank
  $\le m - 1$, each consumed matching pair is a premise of some walk
  fact and has rank $\le m - 2$, one more inversion yields the
  $\bases$ facts assembling $\bases^*(\sigma)$ at rank $\le m - 3$,
  and the class-8 hop's chain-fact family consists of exactly such
  facts, so its maximum rank is $\le m - 3 < m$ and the strong
  induction hypothesis on~$m$ applies) whose focus configuration is
  the decoded position of that variable occurrence, and clause~(2)
  of the main induction hypothesis asserts exactly that this hop
  coincides with the machine. The consumed prefix of
  $\bases^*(\sigma)$ therefore coincides with the decoded run, so
  the matching pair's site component is the machine position named
  by the ladder facts (Lemma~\ref{lem:decode-invariant}(ii)), the
  walk coincides level by level with the machine walk of
  Lemma~\ref{lem:lookup}, and the chain forwards through the alias
  to $\theta(\rho(x))$: $c_1$ lies on the thunk's observation chain,
  coinciding with the decoding. The remaining chain has rank
  $\le m$ and length $k - 1$; the secondary induction hypothesis at
  the post-lookup configuration gives (1) and (2).
\fi
\end{proof}

\subsection{\bilingual{Adequacy for Convergence}{收敛的充分性}}


\begin{proof}[\bilingual{Proof of Theorem~\ref{thm:adequacy}}{定理~\ref{thm:adequacy} 的证明}]
  \ifPaperChinese
  （$\Rightarrow$）设 $M$ 有弱头范式。由机器正确性
  （命题~\ref{prop:krivine-whnf}），Krivine 机自初始配置起停机；
  由正向模拟（引理~\ref{lem:sim-forward}），$\mathcal{T}(M)$ 收敛。

  （$\Leftarrow$）设 $\mathcal{T}(M)$ 收敛。由定义~\ref{def:convergence}，
  存在见证对 $(\_,\, q) \in \supers((\text{``\_\_whnf''}))$ 且
  $q \snoc \text{``\_\_parameter''} \in \dom(\inherits)$；该事实的
  Kleene 秩有限（引理~\ref{lem:kleene-iteration}）。对它作一次末步反演，
  得到装配它的链事实族：一条自 $(\text{``\_\_whnf''})$ 到 $q$ 的
  $\bases$-链与 override 事实，秩均严格更小。以初始配置
  （焦点为根，观测节点即 $(\text{``\_\_whnf''})$）与这条链应用反向模拟
  （引理~\ref{lem:sim-backward}），机器停机；再由机器正确性
  （命题~\ref{prop:krivine-whnf}），$M$ 有弱头范式。
  \else
  ($\Rightarrow$) Suppose $M$ has a weak head normal form. By machine
  correctness (Proposition~\ref{prop:krivine-whnf}) the Krivine
  machine halts from the initial configuration; by forward simulation
  (Lemma~\ref{lem:sim-forward}), $\mathcal{T}(M)$ converges.

  \begin{sloppypar}
  ($\Leftarrow$) Suppose $\mathcal{T}(M)$ converges. By
  Definition~\ref{def:convergence} there is a witnessing pair
  $(\_,\, q) \in \supers((\text{``\_\_whnf''}))$ with
  $q \snoc \text{``\_\_parameter''} \in \dom(\inherits)$, and this fact
  has a finite Kleene rank (Lemma~\ref{lem:kleene-iteration}). One
  last-step inversion yields the family of facts assembling it: a
  $\bases$ chain from $(\text{``\_\_whnf''})$ to $q$ together with the
  override fact, all of strictly smaller rank. Applying backward
  simulation (Lemma~\ref{lem:sim-backward}) to the initial
  configuration (whose focus is the root, so the observation node is
  $(\text{``\_\_whnf''})$) with this chain, the machine halts; by machine
  correctness (Proposition~\ref{prop:krivine-whnf}), $M$ has a weak
  head normal form.
  \end{sloppypar}
  \fi
\end{proof}

\subsection{\bilingual{Lévy--Longo Readback}{Lévy--Longo 读回}}

\ifPaperChinese
收敛的充分性定理只观测根层收敛。要兑现与 Lévy--Longo 树的对应,须把观测推进到每个树位置:
进入抽象的体(即使它未被应用)并在头变量处停下。机器与解码相应扩展到\emph{开配置}。
\else
Adequacy for convergence observes only root-level convergence. To earn
the correspondence with the Lévy--Longo tree, the observation must
advance to every tree position: entering the body of an abstraction
even when it is unapplied, and stopping at head variables. The
machine and its decoding extend accordingly to \emph{open
configurations}.
\fi

\begin{definition}[\bilingual{Open configurations}{开配置}]\label{def:open-krivine}
  \ifPaperChinese
  在定义~\ref{def:krivine} 与~\ref{def:config-paths} 之上,环境条目还可为
  \emph{声明条目}:一条形参声明实例的路径
  $\pi \cat (\text{``\_\_whnf''}, \text{``\_\_parameter''})$,不携带项与环境。
  新增两条迁移:
  \begin{itemize}
    \item \begin{sloppypar}\textbf{enter}(栈空的抽象 $\langle \lambda x.\, u', \rho,
      \varnothing \rangle$,焦点 $\pi$):转入
      $\langle u',\; \rho[x \mapsto
      \pi \cat (\text{``\_\_whnf''}, \text{``\_\_parameter''})],\;
      \varnothing \rangle$,新焦点
      $\pi \cat (\text{``\_\_whnf''}, \text{``\_\_result''})$。
      这不是求值步,而是 Lévy--Longo 树的一步导航:进入未应用抽象的体。
      \end{sloppypar}
    \item \textbf{block}(控制项为变量 $x$ 且 $\rho(x)$ 为声明条目):
      终止。该配置称为\emph{头变量配置};其栈中的闭包依次为头变量的实参。
  \end{itemize}
  自某配置起的运行现在恰有三种结局:停机(栈空的抽象)、阻塞(头变量配置)、
  或不终止。
  \else
  On top of Definitions~\ref{def:krivine} and~\ref{def:config-paths},
  an environment entry may also be a \emph{declaration entry}: the
  path $\pi \cat (\text{``\_\_whnf''}, \text{``\_\_parameter''})$ of a
  parameter-declaration instance, carrying no term and no
  environment. Two transitions are added:
  \begin{itemize}
    \item \textbf{enter} (an abstraction under the empty stack,
      $\langle \lambda x.\, u', \rho, \varnothing \rangle$, focus
      $\pi$): pass to
      $\langle u',\; \rho[x \mapsto
      \pi \cat (\text{``\_\_whnf''}, \text{``\_\_parameter''})],\;
      \varnothing \rangle$ with new focus
      $\pi \cat (\text{``\_\_whnf''}, \text{``\_\_result''})$. This is not an
      evaluation step but one navigation step of the Lévy--Longo
      tree: entering the body of an unapplied abstraction.
    \item \textbf{block} (control a variable $x$ with $\rho(x)$ a
      declaration entry): terminal. Such a configuration is a
      \emph{head-variable configuration}; the closures on its stack
      are, in order, the arguments of the head variable.
  \end{itemize}
  A run from a configuration now has exactly three outcomes: halting
  (an abstraction under the empty stack), blocking (a head-variable
  configuration), or non-termination.
  \fi
\end{definition}

\begin{proposition}[\bilingual{Open machine correctness}{开机器正确性}]
  \label{prop:open-krivine-whnf}
  \ifPaperChinese
  对开项(自由变量绑定于声明条目)的运行实现开的弱头归约:运行停机当且仅当
  该项的弱头范式是抽象;阻塞当且仅当其弱头范式以变量为首,且阻塞配置的
  头变量与栈中实参恰为该弱头范式的头与脊上实参;不终止当且仅当无弱头范式。
  \else
  A run on an open term (free variables bound to declaration entries)
  implements open weak-head reduction: it halts iff the term's weak
  head normal form is an abstraction; it blocks iff the weak head
  normal form is headed by a variable, and the blocked
  configuration's head variable and stacked arguments are exactly the
  head and the spine arguments of that weak head normal form; and it
  fails to terminate iff there is no weak head normal form.
  \fi
\end{proposition}

\begin{proof}
\ifPaperChinese
  与命题~\ref{prop:krivine-whnf} 的读回论证相同:push 与 lookup(thunk 条目)
  保持读回,$\beta$ 收缩一步 $\to_h$;lookup 到声明条目的配置读回为
  以该自由变量为首的应用脊,即以变量为首的弱头范式。
\else
  The same readback argument as
  Proposition~\ref{prop:krivine-whnf}: push and lookup at thunk
  entries preserve the readback and $\beta$ contracts one $\to_h$
  step; a configuration looking up a declaration entry reads back to
  an application spine headed by that free variable, a weak head
  normal form headed by a variable.
\fi
\end{proof}

\begin{lemma}[\bilingual{Open simulation}{开模拟}]\label{lem:open-simulation}
  \ifPaperChinese
  设 $K$ 为可达开配置,焦点 $\pi$,观测节点 $\chi = \pi \cat (\text{``\_\_whnf''})$。
  \begin{enumerate}
    \item[(a)] 若自 $K$ 的运行停机,则 $\supers(\chi)$ 含有抵达某抽象形状的对;
      反之,若该事实以某链事实族见证,则运行停机。
    \item[(b)] 若自 $K$ 的运行阻塞于头变量 $x$($\rho(x)$ 为声明 $d$),
      则 $\chi$ 的链事实有限地终止于 $d$:自 $\chi$ 起的逐跳 $\bases$
      事实经查找引理~\ref{lem:lookup} 的转发抵达 $d$,而 $d$ 无成员,
      链在此之后不产生任何事实;特别地 $\supers(\chi)$ 不含抵达抽象形状的对。
      反之,若 $\chi$ 的链事实有限地终止于某声明条目 $d$
      (引理~\ref{lem:node-classification} 第 7 类,无成员),
      则自 $K$ 的运行阻塞于绑定 $d$ 的头变量。
    \item[(c)] 若自 $K$ 的运行不终止,则 $\supers(\chi)$ 不含抵达抽象形状的对,
      且不存在 (b) 那样的有限终止链。
  \end{enumerate}
  \else
  Let $K$ be a reachable open configuration with focus $\pi$ and
  observation node $\chi = \pi \cat (\text{``\_\_whnf''})$.
  \begin{enumerate}
    \item[(a)] If the run from $K$ halts, then $\supers(\chi)$
      contains a pair reaching an abstraction shape; conversely, if
      that fact is witnessed by a family of chain facts, the run
      halts.
    \item[(b)] If the run from $K$ blocks at a head variable $x$
      (with $\rho(x)$ a declaration $d$), then the chain facts of
      $\chi$ terminate finitely at $d$: the hop-by-hop $\bases$ facts
      from $\chi$ reach $d$ through the forwarding of the Lookup
      Lemma~\ref{lem:lookup}, $d$ has no members, and the chain
      produces no fact beyond it; in particular $\supers(\chi)$
      contains no pair reaching an abstraction shape. Conversely, if
      the chain facts of $\chi$ terminate finitely at a declaration
      entry $d$ (class~7 of Lemma~\ref{lem:node-classification}, no
      members), the run from $K$ blocks at the head variable bound
      to~$d$.
    \item[(c)] If the run from $K$ does not terminate, then
      $\supers(\chi)$ contains no pair reaching an abstraction shape,
      and there is no finitely terminating chain as in~(b).
  \end{enumerate}
  \fi
\end{lemma}

\begin{proof}
\ifPaperChinese
  push、$\beta$ 与 thunk-lookup 情形与引理~\ref{lem:decode-invariant}、
  \ref{lem:sim-forward}、\ref{lem:sim-backward} 逐字相同:停机段从不查找
  声明条目(查找即阻塞),故 (a) 的两个方向分别就是引理~\ref{lem:sim-forward}
  与引理~\ref{lem:sim-backward},其中解码不变量对 enter 步的维持是一次
  直接计算——形状实例 $\pi \cat (\text{``\_\_whnf''})$ 的 $\olbl{__result}$
  与 $\olbl{__parameter}$ 均为其自有 label(引理~\ref{lem:node-classification}
  第 5 类),故新焦点满足不变量 (i),新声明条目满足 (ii) 的末级条件。

  对 (b):运行阻塞时,变量 $x$ 的引用按查找引理~\ref{lem:lookup} 逐级转发,
  末级抵达 $\rho(x) = d$;每级转发是一条可导出的 $\bases$ 事实,故链事实
  有限地抵达 $d$。$d$ 是形参声明(引理~\ref{lem:node-classification} 第 7 类,
  无成员):$\inherits$ 在 $d$ 的任何延伸处无定义,方程~(\ref{eq:bases})的
  推导式为空,链在 $d$ 之后不产生事实。$\supers(\chi)$ 不含形状对:
  设有,取秩最小的见证并作末步反演(引理~\ref{lem:kleene-iteration});
  由反向模拟(引理~\ref{lem:sim-backward})的一致性从句 (2),其链前缀被迫与上述转发链逐跳一致,
  故见证须含 $d$ 之后的 $\bases$ 事实,而上一句已排除——矛盾。

  对 (b) 的逆向:设 $\chi$ 的链事实有限地终止于声明 $d$,取秩最小的见证族,
  对其逐字重跑引理~\ref{lem:sim-backward} 的 $(m,k)$ 归纳,唯一改动是终止情形:
  原"链抵达抽象形状"换为"链抵达无成员声明 $d$"。诸非终止情形
  (push、$\beta$、thunk 查找、第 8 类跳步)照旧各消费一步机器步,度量同样下降,
  一致性从句 (2) 同样把每一跳钉到解码运行的转发链上;终止情形处,运行的下一步
  是查找声明条目 $d$,按定义即阻塞配置。故自 $K$ 的运行(机器决定性保证其唯一)
  在有限步内阻塞,且由一致性从句其转发链与见证链逐跳一致、末端同为 $d$,
  阻塞的头变量即绑定 $d$ 的变量。

  对 (c):不终止的运行每一步消费一条可导出事实的反演前提;若 $\supers(\chi)$
  含形状对,取秩最小见证,(a) 的反向即引理~\ref{lem:sim-backward} 给出停机,
  与不终止矛盾;若存在 (b) 型有限终止链,则由 (b) 的逆向,运行阻塞,
  亦矛盾。
\else
  The push, $\beta$, and thunk-lookup cases are verbatim those of
  Lemmas~\ref{lem:decode-invariant}, \ref{lem:sim-forward},
  and~\ref{lem:sim-backward}: a halting segment never looks up a
  declaration entry (that would block), so the two directions of~(a)
  are exactly Lemma~\ref{lem:sim-forward} and
  Lemma~\ref{lem:sim-backward}, where maintaining the decode
  invariant across an enter step is one direct calculation: the shape
  instance $\pi \cat (\text{``\_\_whnf''})$ owns both $\olbl{__result}$ and
  $\olbl{__parameter}$ (class~5 of
  Lemma~\ref{lem:node-classification}), so the new focus satisfies
  invariant~(i) and the new declaration entry satisfies the final
  level of~(ii).

  For~(b): when the run blocks, $x$'s reference forwards level by
  level per the Lookup Lemma~\ref{lem:lookup}, arriving at
  $\rho(x) = d$; each forwarding level is a derivable $\bases$ fact,
  so the chain facts reach $d$ finitely. $d$ is a parameter
  declaration (class~7 of Lemma~\ref{lem:node-classification}, no
  members): $\inherits$ is undefined at every extension of $d$, the
  comprehension of equation~(\ref{eq:bases}) is empty, and the chain
  produces no fact beyond $d$. And $\supers(\chi)$ contains no shape
  pair: suppose it did, take a witness of least rank and invert its
  last step (Lemma~\ref{lem:kleene-iteration}); by the coincidence clause~(2) of the backward simulation
  (Lemma~\ref{lem:sim-backward}), its chain prefix is forced to
  coincide hop by hop with the forwarding chain above, so the witness would need a $\bases$ fact beyond $d$,
  excluded by the previous sentence --- a contradiction.

  For the converse direction of~(b): let the chain facts of $\chi$
  terminate finitely at the declaration $d$, and take a witness
  family of least rank. Re-run the $(m,k)$ induction of
  Lemma~\ref{lem:sim-backward} verbatim on this family; the only
  change is the terminal case, where ``the chain reaches an
  abstraction shape'' becomes ``the chain reaches the member-less
  declaration $d$''. Every non-terminal case (push, $\beta$, thunk
  lookup, class-8 hop) consumes one machine step exactly as before,
  the measure decreases identically, and the coincidence clause~(2)
  pins each hop to the forwarding chain of the decoded run just as
  before; at the terminal case the run's next step is a lookup of
  the declaration entry $d$, by definition a blocked configuration.
  Hence the run from $K$ (unique, by determinism of the machine)
  blocks within finitely many steps, and by the coincidence clause
  its forwarding chain agrees hop by hop with the witness chain,
  ending at the same $d$: the blocked head variable is the one bound
  to~$d$.

  For~(c): if $\supers(\chi)$ contained a shape pair, a least-rank
  witness plus the backward direction of~(a)
  (Lemma~\ref{lem:sim-backward}) would make the run halt,
  contradicting non-termination; if a finitely terminating chain as
  in~(b) existed, the converse direction of~(b) would make the run
  block, again a contradiction.
\fi
\end{proof}

\begin{definition}[\bilingual{Observation and readback}{观测与读回}]
  \label{def:llt-readback}
  \ifPaperChinese
  可达开配置 $K$ 的\emph{观测}为:形状(运行停机)、头变量(运行阻塞)或
  $\bot$(不终止);由引理~\ref{lem:open-simulation},三者由 $\lfp(T_P)$
  的事实刻画。\emph{读回}递归定义:
  \begin{itemize}
    \item 观测为形状:读回为 $\lambda x$ 下挂 enter 步之后配置的读回,
      其中 $x$ 为该形状的别名 label;
    \item 观测为头变量:读回为以阻塞变量 $y$ 为首、诸实参读回为脊的应用,
      其中第 $i$ 个实参的读回取自以栈中第 $i$ 个闭包为初始控制与焦点、
      栈空的配置;
    \item 观测为 $\bot$:读回为 $\bot$。
  \end{itemize}
  \else
  The \emph{observation} of a reachable open configuration $K$ is:
  shape (the run halts), head variable (the run blocks), or $\bot$
  (non-termination); by Lemma~\ref{lem:open-simulation}, the three
  are characterized by facts of $\lfp(T_P)$. The \emph{readback} is defined
  recursively:
  \begin{itemize}
    \item observation shape: the readback is $\lambda x$ over the
      readback of the configuration after the enter step, where $x$
      is the alias label of the shape;
    \item observation head variable: the readback is the
      application spine headed by the blocked variable $y$ whose
      arguments are the readbacks of the configurations taking, in
      order, each stacked closure as control and focus under the
      empty stack;
    \item observation $\bot$: the readback is $\bot$.
  \end{itemize}
  \fi
\end{definition}

\begin{theorem}[\bilingual{Adequacy for Lévy--Longo trees}{Lévy--Longo 树的充分性}]
  \label{thm:llt-correspondence}
  \ifPaperChinese
  对每个闭合 $\lambda$-项 $M_0$,初始配置
  $\langle M_0, \varnothing, \varnothing \rangle$ 的读回等于其
  Lévy--Longo 树 $\mathrm{LLT}(M_0)$。
  \else
  For every closed $\lambda$-term $M_0$, the readback of the initial
  configuration $\langle M_0, \varnothing, \varnothing \rangle$
  equals its Lévy--Longo tree $\mathrm{LLT}(M_0)$.
  \fi
\end{theorem}

\begin{proof}
\ifPaperChinese
  对树深度归纳,两棵树逐位置比较。在每个位置,开机器正确性
  (命题~\ref{prop:open-krivine-whnf})把配置的三种结局与该位置子项的开弱头
  归约结局一一对应:停机对应抽象弱头范式,即 $\mathrm{LLT}$ 的 $\lambda$-节点,
  且形状的别名 label 即 binder 名;阻塞对应以变量为首的弱头范式,即
  $\mathrm{LLT}$ 的变量头节点,头与实参逐个对应;不终止对应无弱头范式,
  即 $\bot$。读回的递归恰在这三种情形分别进入体、诸实参或停止,
  与 $\mathrm{LLT}$ 的递归(第~\ref{sec:levy-longo-tree}~节前的回顾)逐位置一致。
\else
  By induction on tree depth, comparing the two trees position by
  position. At each position, open machine correctness
  (Proposition~\ref{prop:open-krivine-whnf}) matches the three
  outcomes of the configuration with the outcomes of open weak-head
  reduction of the sub-term at that position: halting corresponds to
  an abstraction weak head normal form, the $\lambda$-node of
  $\mathrm{LLT}$, with the shape's alias label being the binder name;
  blocking corresponds to a variable-headed weak head normal form,
  the head-variable node of $\mathrm{LLT}$, with head and arguments
  matched one by one; and non-termination corresponds to the absence
  of a weak head normal form, that is, $\bot$. The recursion of the
  readback descends, in the three cases, into the body, into the
  arguments, or stops --- exactly as the recursion of
  $\mathrm{LLT}$ (recalled before
  Section~\ref{sec:levy-longo-tree}) does at each position.
\fi
\end{proof}

\subsection{\bilingual{Corollaries}{推论}}

\begin{lemma}[\bilingual{Single-Path}{单路径}]\label{lem:single-path}
  \ifPaperChinese
  设 $M$ 为闭合 $\lambda$-项且 $\mathcal{T}(M)$ 收敛。则在收敛见证的
  提取（引理~\ref{lem:sim-backward} 的证明）所消费的每次
  $\this(S,\; p_{\mathrm{def}},\; n)$ 解析步中，前沿集合 $S$ 至多
  包含一条路径：$\lambda$-演算嵌入所依赖的每次解析都不分叉为多目标。
  \else
  Let $M$ be a closed $\lambda$-term such that $\mathcal{T}(M)$
  converges. Then at every resolution step
  $\this(S,\; p_{\mathrm{def}},\; n)$ consumed by the extraction of
  the convergence witness (the proof of
  Lemma~\ref{lem:sim-backward}), the frontier set $S$ contains at
  most one path: every resolution the $\lambda$-calculus embedding
  depends on is single-target.
  \fi
\end{lemma}

\begin{proof}
\ifPaperChinese
  即引理~\ref{lem:sim-backward} 的一致性从句 (2)：每个被消费的走步
  与引理~\ref{lem:lookup} 的机器走步一致，而后者的前沿逐步为单点。
\else
  This is the coincidence clause~(2) of
  Lemma~\ref{lem:sim-backward}: every consumed walk step coincides
  with the machine walk of Lemma~\ref{lem:lookup}, whose frontier is
  a singleton at every step.
\fi
\end{proof}

\begin{theorem}[\bilingual{Convergence preservation}{收敛保持}]%
  \label{thm:convergence-preservation}
  \ifPaperChinese
  若 $M \to_h N$（弱头归约），则 $\mathcal{T}(M)$ 收敛当且仅当
  $\mathcal{T}(N)$ 收敛。
  \else
  If $M \to_h N$ (weak-head reduction), then
  $\mathcal{T}(M)$ converges if and only if
  $\mathcal{T}(N)$ converges.
  \fi
\end{theorem}

\begin{proof}
  \ifPaperChinese
  弱头归约是确定性的，故 $M$ 有弱头范式当且仅当 $N$ 有
  （$M$ 的弱头归约序列必经 $N$）。两次应用收敛的充分性定理
  （定理~\ref{thm:adequacy}）即得。
  \else
  Weak-head reduction is deterministic, so $M$ has a weak head normal
  form if and only if $N$ does (the weak-head reduction sequence of
  $M$ passes through $N$). Applying adequacy for convergence
  (Theorem~\ref{thm:adequacy}) twice completes the proof.
  \fi
\end{proof}
\section{\bilingual{Expressive Asymmetry: Proofs}{表达力的不对称:证明}}
\label{app:expressiveness}

\begin{definition}[\bilingual{Macro-expressibility, after Felleisen~\cite{felleisen1991-expressive-power}}{宏可表达性,据 Felleisen~\cite{felleisen1991-expressive-power}}]
  \label{def:macro-expressibility}
  \ifPaperChinese
  语言 $\mathscr{L}_1$ 的一个构造 $C$ 在语言 $\mathscr{L}_0$ 中是\emph{宏可表达的},若存在一个翻译 $\mathcal{E}$,使得对每一个含有 $C$ 的出现的程序 $P$,翻译 $\mathcal{E}(P)$ 是通过将 $C$ 的每次出现替换为一个仅依赖于 $C$ 的子表达式的 $\mathscr{L}_0$ 表达式而得到的,并保持 $P$ 的所有其他构造不变。
  \else
  A construct $C$ of language $\mathscr{L}_1$ is \emph{macro-expressible}
  in language $\mathscr{L}_0$ if there exists a translation $\mathcal{E}$
  such that for every program $P$ containing occurrences of $C$,
  the translation $\mathcal{E}(P)$ is obtained by replacing each
  occurrence of $C$ with an $\mathscr{L}_0$ expression that depends only
  on $C$'s subexpressions, leaving all other constructs of $P$
  unchanged.
  \fi
\end{definition}

\begin{lemma}[\bilingual{Composition of macro-expressible translations}{宏可表达翻译的复合}]
  \label{lem:macro-composition}
  \ifPaperChinese
  若 $\mathcal{E}_1$ 在 $\mathscr{L}_0$ 中宏可表达 $\mathscr{L}_1$ 的构造,且 $\mathcal{E}_2$ 在 $\mathscr{L}_1$ 中宏可表达 $\mathscr{L}_2$ 的构造,则复合 $\mathcal{E}_1 \circ \mathcal{E}_2$ 在 $\mathscr{L}_0$ 中宏可表达 $\mathscr{L}_2$ 的构造。
  \else
  If $\mathcal{E}_1$ macro-expresses the constructs of
  $\mathscr{L}_1$ in $\mathscr{L}_0$ and
  $\mathcal{E}_2$ macro-expresses the constructs of
  $\mathscr{L}_2$ in $\mathscr{L}_1$, then the composition
  $\mathcal{E}_1 \circ \mathcal{E}_2$ macro-expresses the
  constructs of $\mathscr{L}_2$ in $\mathscr{L}_0$.
  \fi
\end{lemma}

\begin{proof}
  \ifPaperChinese
  $\mathscr{L}_2$ 的每个构造 $F$ 在 $\mathscr{L}_1$ 上有一个固定的语法抽象 $A_F^{(2)}$。
  出现在 $A_F^{(2)}$ 中的每个 $\mathscr{L}_1$-构造 $G$ 在 $\mathscr{L}_0$ 上有一个固定的语法抽象 $A_G^{(1)}$。
  将 $A_F^{(2)}$ 中的每个 $G$ 替换为 $A_G^{(1)}$,得到 $F$ 在 $\mathscr{L}_0$ 上的一个固定语法抽象,从而满足 E4。
  \else
  \begin{sloppypar}
  Each construct~$F$ of~$\mathscr{L}_2$ has a fixed
  syntactic abstraction~$A_F^{(2)}$
  over~$\mathscr{L}_1$.
  Each $\mathscr{L}_1$-construct~$G$ appearing
  in~$A_F^{(2)}$ has a fixed
  syntactic abstraction~$A_G^{(1)}$
  over~$\mathscr{L}_0$.
  Substituting~$A_G^{(1)}$ for every~$G$
  in~$A_F^{(2)}$ yields a fixed syntactic
  abstraction over~$\mathscr{L}_0$ for~$F$,
  satisfying~E4.
  \end{sloppypar}
  \fi
\end{proof}

\begin{definition}[\bilingual{Language universe}{语言全集}]
  \label{def:language-universe}
  \ifPaperChinese
  设 $\mathscr{U}$ 为一个语言,其构造子是继承演算的构造子与惰性 $\lambda$-演算的构造子(包括 $\lambda$-抽象、应用与变量)的并集。
  $\mathscr{U}$ 的语义通过将每个 $\mathscr{U}$-程序翻译为继承演算来给出,分两步:
  \begin{enumerate}
    \item 将翻译 $\mathcal{T}$(第~\ref{sec:forward-translation} 节)应用于每个 $\lambda$-演算子表达式。
    \item 按第~\ref{sec:mixin-trees} 节的规则对所得继承演算程序求值。
  \end{enumerate}
  当一个 $\mathscr{U}$-表达式交错使用 $\lambda$-演算与继承演算的构造子时(例如,一个命名继承演算字段,其值是一个 $\lambda$-演算表达式),翻译从叶到根组合地应用:每个 $\lambda$-演算构造子通过 $\mathcal{T}$ 替换为其继承演算等价物,同时将继承演算子表达式视为不透明的值。
  由于 $\mathcal{T}$ 的每条规则(第~\ref{sec:forward-translation} 节)都是一个固定的语法抽象,仅依赖于被翻译的子表达式,因此这种自底向上的应用是良定义的,并为步骤~(2) 产生一个纯继承演算程序。

  若一个 $\mathscr{U}$-程序 $P$ 的翻译收敛,则 $P$ 收敛,记为 $\eval_{\mathscr{U}}(P)$。
  \else
  Let $\mathscr{U}$ be the language whose constructors are
  the union of the constructors of inheritance-calculus and
  the constructors of the lazy $\lambda$-calculus,
  including $\lambda$-abstraction, application, and variable.
  The semantics of $\mathscr{U}$ is given by translating
  every $\mathscr{U}$-program to inheritance-calculus in
  two steps:
  \begin{enumerate}
    \item Apply the translation $\mathcal{T}$
      (Section~\ref{sec:forward-translation}) to every
      $\lambda$-calculus subexpression.
    \item Evaluate the resulting inheritance-calculus
      program by the rules of
      Section~\ref{sec:mixin-trees}.
  \end{enumerate}
  When a $\mathscr{U}$-expression interleaves
  $\lambda$-calculus and inheritance-calculus constructors
  (e.g.\ a named inheritance-calculus field whose value is a
  $\lambda$-calculus expression), the translation is
  applied compositionally from leaves to root: each
  $\lambda$-calculus construct is replaced by its
  inheritance-calculus equivalent via~$\mathcal{T}$,
  treating inheritance-calculus subexpressions as opaque
  values.
  Since each rule of~$\mathcal{T}$
  (Section~\ref{sec:forward-translation}) is a fixed
  syntactic abstraction depending only on the translated
  subexpressions, this bottom-up application is
  well-defined and yields a pure inheritance-calculus
  program for step~(2).

  A $\mathscr{U}$-program $P$ converges,
  $\eval_{\mathscr{U}}(P)$, iff its translation
  converges.
  \fi
\end{definition}

\begin{remark}[\bilingual{Relativity to the universe}{对 universe 的相对性}]
  \label{rem:universe-relativity}
  \ifPaperChinese
  定义~\ref{def:language-universe} 把 $\lambda$-项在 $\mathscr{U}$ 中的语义
  \emph{规定}为其 $\mathcal{T}$-像的语义;于是翻译后的项在 $\mathscr{U}$ 中携带
  可观测的内部结构($\olbl{__callee}$、$\olbl{__call}$、$\olbl{__parameter}$ 等路径),
  而这正是下文分离构造所利用的。本节的分离因此是\emph{相对于该 universe} 的:
  这是 Felleisen 框架~\cite{felleisen1991-expressive-power}中比较两种语言时的
  标准做法(表达力总是相对一个共同的语言全集来陈述),但选取不同的 universe
  可能得到不同的分离结论,本文不做更强的宣称。
  \else
  Definition~\ref{def:language-universe} \emph{stipulates} the
  semantics of a $\lambda$-term in $\mathscr{U}$ to be that of its
  $\mathcal{T}$-image; a translated term therefore carries observable
  internal structure in $\mathscr{U}$ (the paths $\olbl{__callee}$,
  $\olbl{__call}$, $\olbl{__parameter}$, and so on), and this is
  exactly what the separating construction below exploits. The
  separation of this section is accordingly \emph{relative to this
  universe}: stating expressiveness against a common language
  universe is the standard setup of Felleisen's
  framework~\cite{felleisen1991-expressive-power}, but a different
  choice of universe could yield a different separation verdict, and
  we claim nothing stronger.
  \fi
\end{remark}

\begin{remark}\label{rem:universe-properties}
  \ifPaperChinese
  惰性 $\lambda$-演算与继承演算是 $\mathscr{U}$ 的保守性限制:各自通过去掉另一方的构造子来得到。
  对于纯继承演算程序,步骤~(1) 是平凡的,因此 $\eval_{\mathscr{U}}$ 与继承演算的语义一致。
  对于纯 $\lambda$-演算程序,步骤~(1)--(2) 的复合由收敛的充分性(定理~\ref{thm:adequacy})与惰性求值一致。
  \else
  The lazy $\lambda$-calculus and inheritance-calculus are
  conservative restrictions of $\mathscr{U}$: each is
  obtained by removing the other's constructors.
  For a pure inheritance-calculus program, step~(1) is vacuous, so
  $\eval_{\mathscr{U}}$ agrees with the
  inheritance-calculus semantics.
  For a pure $\lambda$-calculus program, the composition
  of steps~(1)--(2) agrees with lazy evaluation by
  adequacy for convergence (Theorem~\ref{thm:adequacy}).
  \fi
\end{remark}

\ifPaperChinese
两步翻译是有效的($\mathcal{T}$ 是语法变换),而继承演算的递归求值(第~\ref{sec:mixin-trees} 节)是一个半判定过程,因此 $\eval_{\mathscr{U}}$ 是递归可枚举的,满足 Felleisen 框架~\cite{felleisen1991-expressive-power} 的要求。
\else
The two-step translation is effective ($\mathcal{T}$
is a syntactic transformation) and the recursive evaluation
of inheritance-calculus (Section~\ref{sec:mixin-trees})
is a semi-decision procedure, so $\eval_{\mathscr{U}}$
is recursively enumerable, as required by
Felleisen's framework~\cite{felleisen1991-expressive-power}.
\fi

\begin{theorem}[\bilingual{Forward macro-expressibility}{正向宏可表达性}]
  \label{thm:forward-macro}
  \ifPaperChinese
  翻译 $\mathcal{T}$(第~\ref{sec:forward-translation} 节)是 $\lambda$-演算嵌入继承演算的一个宏可表达的嵌入。
  \else
  The translation $\mathcal{T}$
  (Section~\ref{sec:forward-translation}) is a macro-expressible
  embedding of the $\lambda$-calculus into inheritance-calculus.
  \fi
\end{theorem}

\begin{proof}
  \ifPaperChinese
  我们验证 Felleisen~\cite[Definition~3.11]{felleisen1991-expressive-power} 的条件 E4。
  $\lambda$-演算与继承演算共享同一变量字母表;一个变量出现 $x$ 由固定的单成员抽象 \mintinline[escapeinside=||]{yaml}{[[|$x$|]]} 宏表达(一个只携带变量名、不含子表达式的引用)。
  两个复合构造各有固定的语法抽象;其中 $\lambda$ 的抽象是一个绑定构造,其被绑定变量 $x$ 作为 label 出现,这为条件 E4 所允许：
  \begin{enumerate}
    \item \emph{\bilingual{$\lambda$-abstraction}{$\lambda$-抽象}} $\lambda x.\, M$（关于体 $M$ 一元），$A_\lambda([\cdot]_1)$ 为：
      \begin{minted}[escapeinside=||]{yaml}
- __whnf:
  - |$x$|: [[__parameter]]
  - __parameter: []
  - __result: |$[\cdot]_1$|
      \end{minted}
    \item \emph{\bilingual{Application}{应用}} $M_1\; M_2$（关于 $M_1, M_2$ 二元），$A_{\text{app}}([\cdot]_1, [\cdot]_2)$ 为：
      \begin{minted}[escapeinside=||]{yaml}
- __callee: |$[\cdot]_1$|
- __call:
  - [__callee, __whnf]
  - __parameter: |$[\cdot]_2$|
- __whnf: [[__call, __result, __whnf]]
      \end{minted}
  \end{enumerate}
  在每种情况下，语法抽象是一个固定的继承演算上下文，其空位由子表达式的递归翻译填充，从而满足 E4。
  \else
  We verify condition~E4 of
  Felleisen~\cite[Definition~3.11]{felleisen1991-expressive-power}.
  The $\lambda$-calculus and inheritance-calculus share a single variable
  alphabet, and a variable occurrence $x$ is macro-expressed by the fixed
  one-member abstraction \mintinline[escapeinside=||]{yaml}{[[|$x$|]]}, a reference
  carrying only the variable name and no subexpressions.
  The two compound constructs each have a fixed syntactic abstraction; the one
  for $\lambda$ is a binding construct whose bound variable $x$ appears as a
  label, which condition~E4 admits:
  \begin{enumerate}
    \item \emph{$\lambda$-abstraction} $\lambda x.\, M$ (unary in body~$M$), $A_\lambda([\cdot]_1)$ is:
      \begin{minted}[escapeinside=||]{yaml}
- __whnf:
  - |$x$|: [[__parameter]]
  - __parameter: []
  - __result: |$[\cdot]_1$|
      \end{minted}
    \item \emph{Application} $M_1\; M_2$ (binary in~$M_1, M_2$), $A_{\text{app}}([\cdot]_1, [\cdot]_2)$ is:
      \begin{minted}[escapeinside=||]{yaml}
- __callee: |$[\cdot]_1$|
- __call:
  - [__callee, __whnf]
  - __parameter: |$[\cdot]_2$|
- __whnf: [[__call, __result, __whnf]]
      \end{minted}
  \end{enumerate}
  In each case the syntactic abstraction is a fixed
  inheritance-calculus context whose holes are filled
  by the recursive translations of the subexpressions,
  satisfying~E4.
  \fi
\end{proof}

\begin{theorem}[\bilingual{Nonexpressibility of inheritance}{继承的不可宏表达性}]
  \label{thm:cartesian-nonexpressibility}
  \ifPaperChinese
  惰性 $\lambda$-演算不能宏可表达 $\mathscr{U}$(定义~\ref{def:language-universe})的继承构造子。
  \else
  The lazy $\lambda$-calculus cannot macro-express the
  inheritance constructors of $\mathscr{U}$
  (Definition~\ref{def:language-universe}).
  \fi
\end{theorem}

\begin{proof}
  \ifPaperChinese
  我们对 $L_0 = \text{惰性 $\lambda$-演算}$,$L_1 = \mathscr{U}$ 应用 Felleisen~\cite{felleisen1991-expressive-power} 的定理 3.14(i)。
  由定义~\ref{def:language-universe},$\mathscr{U} = L_0 + \{$mixin-树构造子$\}$ 是 $L_0$ 的一个保守性扩展(Felleisen~\cite{felleisen1991-expressive-power} 的定义 3.2)。
  我们验证四个条件:
  \begin{enumerate}
    \item[(i)] \emph{$L_0 \subseteq L_1$:} $\mathscr{U}$ 的语法包含所有 $L_0$-短语。
    \item[(ii)] \emph{不引入新的 $L_0$-短语:} mixin-树构造子不在 $L_0$ 中引入新短语。
    \item[(iii)] \emph{构造子不相交:} mixin-树构造子与 $L_0$-构造子不相交。
    \item[(iv)] \emph{语义保持:} 对纯 $\lambda$-项,定义~\ref{def:language-universe} 的步骤~(1) 应用 $\mathcal{T}$、步骤~(2) 求值,其复合由收敛的充分性(定理~\ref{thm:adequacy})与惰性求值一致,因此 $\eval_{\mathscr{U}}$ 在 $L_0$-短语上与 $\eval_{L_0}$ 一致。
  \end{enumerate}
  该定理要求我们给出两个 $L_0$-项 $M$ 与 $N$,满足 $M \cong_0 N$ 但 $M \not\cong_1 N$,其中 $\cong_0$ 是 $L_0$ 中的操作等价,$\cong_1$ 是 $\mathscr{U}$ 中的操作等价。

  \paragraph{\bilingual{Operational equivalence in $L_0$}{$L_0$ 中的操作等价}}
  在惰性 $\lambda$-演算中,归约保持观测等价:若两个闭项归约到同一个值,则它们可观测等价(即 $\cong_0$)。我们用这一事实给出分离对,无需诉诸完整的上下文等价刻画。

  \paragraph{\bilingual{Separating pair}{分离对}}
  定义 Church 布尔值与 Church 布尔相等:$\mathrm{true} = \lambda t.\,\lambda f.\, t$,$\mathrm{false} = \lambda t.\,\lambda f.\, f$,以及 $\mathrm{eq} = \lambda a.\,\lambda b.\, (a\;b)\;(b\;\mathrm{false}\;\mathrm{true})$。
  令
  \[
    M = \mathrm{eq}\;\mathrm{false}\;\mathrm{false}
    \qquad
    N = \mathrm{eq}\;\mathrm{true}\;\mathrm{true}.
  \]
  两者在惰性 $\lambda$-演算中均归约为 Church~$\mathrm{true}$,因此可观测等价,即 $M \cong_0 N$。

  \paragraph{\bilingual{Distinguishing $\mathscr{U}$-context}{区分 $\mathscr{U}$-上下文}}
  我们现在构造一个区分 $M$ 与 $N$ 的 $\mathscr{U}$-上下文。
  $M$ 与 $N$ 是语法封闭的 $\lambda$-项;由定义~\ref{def:language-universe},它们在 $\mathscr{U}$ 中的语义通过直接应用 $\mathcal{T}$ 将其翻译为继承演算来给出。
  记 $\mathrm{eq} = \lambda a.\,\lambda b.\, (a\;b)\;(b\;\mathrm{false}\;\mathrm{true})$。
  则：
  \else
  We apply Theorem~3.14(i) of
  Felleisen~\cite{felleisen1991-expressive-power} to the pair
  $L_0 = \text{lazy $\lambda$-calculus}$,
  $L_1 = \mathscr{U}$.
  By Definition~\ref{def:language-universe},
  $\mathscr{U} = L_0 + \{$mixin-tree constructors$\}$
  is a conservative extension of~$L_0$
  (Definition~3.2 of
  Felleisen~\cite{felleisen1991-expressive-power}).
  We verify the four conditions:
  \begin{enumerate}
    \item[(i)] \emph{$L_0 \subseteq L_1$:}
      The syntax of~$\mathscr{U}$ includes all
      $L_0$-phrases.
    \item[(ii)] \emph{No new $L_0$-phrases:}
      The mixin-tree constructors do not introduce new
      phrases in~$L_0$.
    \item[(iii)] \emph{Disjoint constructors:}
      The mixin-tree constructors are disjoint from the
      $L_0$-constructors.
    \item[(iv)] \emph{Semantics preserved:}
      For pure $\lambda$-terms, step~(1) of
      Definition~\ref{def:language-universe} applies
      $\mathcal{T}$ and step~(2) evaluates; their
      composite agrees with lazy evaluation by adequacy for convergence
      (Theorem~\ref{thm:adequacy}), so
      $\eval_{\mathscr{U}}$ agrees with
      $\eval_{L_0}$ on $L_0$-phrases.
  \end{enumerate}
  The theorem requires us to exhibit two $L_0$-terms
  $M$ and $N$ with $M \cong_0 N$ but
  $M \not\cong_1 N$, where $\cong_0$ is operational
  equivalence in $L_0$ and $\cong_1$ is operational
  equivalence in $\mathscr{U}$.

  \paragraph{Operational equivalence in $L_0$}
  In the lazy $\lambda$-calculus, reduction preserves
  observational equivalence: if two closed terms reduce to the
  same value, they are observationally equivalent (that is,
  $\cong_0$). We use this fact to supply the separating pair,
  without appealing to a full contextual characterisation.

  \paragraph{Separating pair}
  Define Church booleans and Church boolean equality:
  $\mathrm{true} = \lambda t.\,\lambda f.\, t$,
  $\mathrm{false} = \lambda t.\,\lambda f.\, f$, and
  $\mathrm{eq} = \lambda a.\,\lambda b.\,
  (a\;b)\;(b\;\mathrm{false}\;\mathrm{true})$.
  Let
  \[
    M = \mathrm{eq}\;\mathrm{false}\;\mathrm{false}
    \qquad
    N = \mathrm{eq}\;\mathrm{true}\;\mathrm{true}.
  \]
  Both reduce to Church~$\mathrm{true}$ in the
  lazy $\lambda$-calculus, so they are observationally
  equivalent, that is, $M \cong_0 N$.

  \paragraph{Distinguishing $\mathscr{U}$-context}
  We now construct a $\mathscr{U}$-context that
  separates $M$ and $N$.
  $M$ and $N$ are syntactically closed $\lambda$-terms;
  by Definition~\ref{def:language-universe}, their
  semantics in $\mathscr{U}$ is given by directly applying
  the translation $\mathcal{T}$ to translate them to inheritance-calculus.
  \fi

  \ifPaperChinese
  $T(M)$ 是应用 $T(M_1\;M_2)$ 的翻译,其顶层为 $\{\text{``\_\_callee''}, \text{``\_\_call''}, \text{``\_\_whnf''}\}$。
  我们把空位的翻译命名为 $\olbl{hole}$,并添加一个新的可观测投影 $\olbl{firstOperand}$,读取 $\mathrm{eq}$ 的第一个操作数,即 $\lambda a$ 所绑定的实参。
  该实参位于 $(\text{``hole''}, \text{``\_\_callee''}, \text{``\_\_call''}, \text{``\_\_parameter''})$:内层应用 $\mathrm{eq}\;\mathrm{false}$ 在其 $\olbl{__call}$ 处把 $\lambda a$ 的 $\olbl{__parameter}$ 覆写为该操作数。
  定义上下文 $C[\cdot]$ 为:
  \else
  \begin{sloppypar}
  $T(M)$ is the translation of the application $T(M_1\;M_2)$, whose
  top level is $\{\text{``\_\_callee''}, \text{``\_\_call''}, \text{``\_\_whnf''}\}$.
  We name the translation of the hole $\olbl{hole}$ and add a new
  observable projection $\olbl{firstOperand}$ that reads
  $\mathrm{eq}$'s first operand, the argument bound by $\lambda a$.
  That argument sits at
  $(\text{``hole''},\allowbreak \text{``\_\_callee''},\allowbreak \text{``\_\_call''},\allowbreak \text{``\_\_parameter''})$:
  the inner application $\mathrm{eq}\;\mathrm{false}$ overrides the
  $\olbl{__parameter}$ of $\lambda a$ with that operand at its
  $\olbl{__call}$.
  Define the context $C[\cdot]$ as:
  \end{sloppypar}
  \fi
  \begin{minted}[escapeinside=||]{yaml}
- hole: |$[\cdot]$|
- firstOperand: [[hole, __callee, __call, __parameter]]
- __whnf: [[firstOperand, __whnf]]
  \end{minted}
  \ifPaperChinese
  这里 $[\cdot]$ 是空位;整个 $C[\cdot]$ 仅使用继承演算的 mixin-树构造子,不含 $\lambda$-演算构造。它不修改 $M$ 的翻译(置于 $\olbl{hole}$ 之下),只添加一个新投影 $\olbl{firstOperand}$。

  逐方程计算 $\olbl{firstOperand}$:其引用 $(0,\allowbreak (\text{``hole''}, \text{``\_\_callee''}, \text{``\_\_call''}, \text{``\_\_parameter''}))$ 在 $n = 0$ 时按方程~(\ref{eq:resolve})给出定点目标,即空位之下 $\mathrm{eq}$ 内层应用的 $\olbl{__parameter}$ 覆写——$\mathrm{eq}$ 第一个操作数的绑定值。在 $C[M]$ 中该操作数为 $T(\mathrm{false})$,在 $C[N]$ 中为 $T(\mathrm{true})$;两者均为抽象之像,其 $\olbl{__whnf}$ 即其形状(自反,方程~\ref{eq:supers})。

  定义 $D[\cdot] = (\lambda x.\, x\;\Omega\;I)\; C[\cdot]$;这是混合项:应用规则的 $\olbl{__parameter}$ 覆写为纯 mixin $C[\cdot]$,故以下不引用附录~\ref{app:levy-longo-tree-proofs} 的机器引理(它们只对纯 $\lambda$-像陈述),而对方程直接计算。$D[M]$ 的收敛链逐跳如下,每跳为一次 $n = 0$ 引用的定点解析(方程~\ref{eq:resolve})或一次别名转发(方程~\ref{eq:this},走步一步):
  根 $\olbl{__whnf}$ 继承 $(\text{``\_\_call''}, \text{``\_\_result''}, \text{``\_\_whnf''})$;$\olbl{__call}$ 继承判别子 $\lambda x$ 的形状,体为 $x\;\Omega\;I$ 之像,其脊上外层 $\olbl{__call}$ 绑定 $I$-thunk、内层 $\olbl{__call}$ 绑定 $\Omega$-thunk;$x$ 的引用经别名转发到判别子 $\olbl{__call}$ 的 $\olbl{__parameter} = C[M]$;$C[M]$ 的 $\olbl{__whnf}$ 经 $\olbl{firstOperand}$ 达 $T(\mathrm{false})$ 的形状;$\mathrm{false} = \lambda t.\,\lambda f.\, f$ 的体 $f$ 经别名转发到\emph{第二个}实参 thunk,即 $I$;$I$ 为抽象,其形状终结该链。故 $\supers((\text{``\_\_whnf''}, \text{``\_\_parameter''})) \neq \varnothing$,$\eval_{\mathscr{U}}(D[M])$ 成立。

  对 $D[N]$,同一条链改经 $T(\mathrm{true})$,而 $\mathrm{true} = \lambda t.\,\lambda f.\, t$ 的体 $t$ 转发到\emph{第一个}实参 thunk,即 $T(\Omega)$。设 $D[N]$ 的收敛见证存在,取其 Kleene 秩最小者(引理~\ref{lem:kleene-iteration});其末步反演的前提含 $T(\Omega)$ 链的收敛见证,而 $\Omega = (\lambda x.\, x\, x)(\lambda x.\, x\, x)$ 的链经一跳回到与自身同构的配置,其见证的末步反演又含同一形状的见证,秩严格更小——与最小性矛盾。故无见证,$\eval_{\mathscr{U}}(D[N])$ 不成立。
  因此 $M \not\cong_1 N$。

  \paragraph{\bilingual{Conclusion}{结论}}
  我们有 $M \cong_0 N$ 但 $M \not\cong_1 N$,故 ${\cong_0} \neq {(\cong_1|_{L_0})}$。
  由 Felleisen~\cite{felleisen1991-expressive-power} 的定理 3.14(i),惰性 $\lambda$-演算不能宏可表达 $\mathscr{U}$ 的继承构造子。
  \else
  Here $[\cdot]$ is the hole; the entire $C[\cdot]$ uses only the
  mixin-tree constructors of inheritance-calculus, with no
  $\lambda$-calculus constructs.
  It does not modify the translation of $M$ (placed under
  $\olbl{hole}$); it only adds a new projection
  $\olbl{firstOperand}$.

  \begin{sloppypar}
  We compute $\olbl{firstOperand}$ with the equations: its reference
  $(0,\allowbreak (\text{``hole''},\allowbreak \text{``\_\_callee''},\allowbreak \text{``\_\_call''},\allowbreak
  \text{``\_\_parameter''}))$ resolves at $n = 0$ to the site-relative
  target (equation~\ref{eq:resolve}), the $\olbl{__parameter}$
  override of $\mathrm{eq}$'s inner application beneath the hole ---
  the bound value of $\mathrm{eq}$'s first operand. In $C[M]$ that
  operand is $T(\mathrm{false})$; in $C[N]$ it is
  $T(\mathrm{true})$; both are images of abstractions, whose
  $\olbl{__whnf}$ is their shape (reflexively,
  equation~\ref{eq:supers}).
  \end{sloppypar}

  Define $D[\cdot] = (\lambda x.\, x\;\Omega\;I)\; C[\cdot]$.
  This is a mixed term: the application rule's $\olbl{__parameter}$ is
  overridden by the pure mixin $C[\cdot]$, so we do not invoke the
  machine lemmas of Appendix~\ref{app:levy-longo-tree-proofs} (they
  are stated for pure $\lambda$-images only) and instead compute
  directly with the equations. The convergence chain of $D[M]$, hop
  by hop, each hop a site-relative resolution of an $n = 0$ reference
  (equation~\ref{eq:resolve}) or one alias forwarding step
  (equation~\ref{eq:this}): the root $\olbl{__whnf}$ inherits
  $(\text{``\_\_call''}, \text{``\_\_result''}, \text{``\_\_whnf''})$; the
  $\olbl{__call}$ inherits the discriminator $\lambda x$'s shape,
  whose body is the image of $x\;\Omega\;I$, with the outer
  $\olbl{__call}$ of the spine binding the $I$-thunk and the inner one
  the $\Omega$-thunk; $x$'s reference forwards through the alias to
  the discriminator $\olbl{__call}$'s $\olbl{__parameter} = C[M]$;
  $C[M]$'s $\olbl{__whnf}$ reaches $T(\mathrm{false})$'s shape via
  $\olbl{firstOperand}$; the body $f$ of
  $\mathrm{false} = \lambda t.\,\lambda f.\, f$ forwards through
  its alias to the \emph{second} argument thunk, $I$; and $I$ is an
  abstraction whose shape ends the chain. Hence
  $\supers((\text{``\_\_whnf''}, \text{``\_\_parameter''})) \neq
  \varnothing$ and $\eval_{\mathscr{U}}$ holds for $D[M]$.

  For $D[N]$ the same chain passes through $T(\mathrm{true})$
  instead, and the body $t$ of
  $\mathrm{true} = \lambda t.\,\lambda f.\, t$ forwards to the
  \emph{first} argument thunk, $T(\Omega)$. Suppose a convergence
  witness for $D[N]$ existed, and take one of least Kleene rank
  (Lemma~\ref{lem:kleene-iteration}); the premises of its last-step
  inversion contain a convergence witness for $T(\Omega)$'s chain,
  and the chain of
  $\Omega = (\lambda x.\, x\, x)(\lambda x.\, x\, x)$ returns
  after one hop to a configuration isomorphic to itself, so the
  inversion of that witness again contains a witness of the same
  shape at strictly smaller rank --- contradicting minimality. Hence
  no witness exists and $\eval_{\mathscr{U}}$ fails for $D[N]$.
  Therefore $M \not\cong_1 N$.

  \paragraph{Conclusion}
  \begin{sloppypar}
  We have $M \cong_0 N$ but $M \not\cong_1 N$, so
  ${\cong_0} \neq {(\cong_1|_{L_0})}$.
  By Theorem~3.14(i) of
  Felleisen~\cite{felleisen1991-expressive-power},
  the lazy $\lambda$-calculus cannot macro-express the
  inheritance constructors of~$\mathscr{U}$.
  \end{sloppypar}
  \fi
\end{proof}


\begin{proof}[\bilingual{Proof of Theorem~\ref{thm:expressive-asymmetry}}{定理~\ref{thm:expressive-asymmetry} 的证明}]
  \ifPaperChinese
  我们须证明:(a)~继承演算能宏可表达惰性 $\lambda$-演算能宏可表达的每个 $\mathscr{U}$-构造,以及 (b)~反之不成立。

  \paragraph{\bilingual{Part (b): the lazy $\lambda$-calculus $\not\geq$
  inheritance-calculus}{第 (b) 部分:惰性 $\lambda$-演算 $\not\geq$ 继承演算}}
  $\lambda$-演算包含其自身的构造子($\lambda$、应用、变量),但不能宏可表达 $\mathscr{U}$ 的继承构造子(定理~\ref{thm:cartesian-nonexpressibility})。
  继承演算包含那些构造子。
  因此,惰性 $\lambda$-演算对于 $\mathscr{U}$ 而言表达力不及继承演算。

  \paragraph{\bilingual{Part (a): inheritance-calculus $\geq$ the lazy
  $\lambda$-calculus}{第 (a) 部分:继承演算 $\geq$ 惰性 $\lambda$-演算}}
  我们证明继承演算能宏可表达 $\lambda$-演算的每个构造。
  $\lambda$-演算包含 $\lambda$-抽象、应用与变量。
  翻译 $\mathcal{T}$(定理~\ref{thm:forward-macro})直接将每个 $\lambda$-演算构造子映射到继承演算上的一个固定语法抽象,满足 Felleisen~\cite[Definition~3.11]{felleisen1991-expressive-power} 的条件 E4。
  由于继承演算平凡地包含其自身的构造子,它能宏可表达 $\lambda$-演算的每个构造。
  \else
  We must show:
  (a)~inheritance-calculus can macro-express every
  $\mathscr{U}$-construct that the lazy $\lambda$-calculus
  can macro-express, and
  (b)~the converse fails.

  \paragraph{Part (b): the lazy $\lambda$-calculus $\not\geq$
  inheritance-calculus}
  The $\lambda$-calculus contains its own constructors
  ($\lambda$, application, variable) but
  cannot macro-express the inheritance constructors of
  $\mathscr{U}$
  (Theorem~\ref{thm:cartesian-nonexpressibility}).
  Inheritance-calculus contains those constructors.
  Therefore the lazy $\lambda$-calculus is not at least
  as expressive as inheritance-calculus with respect
  to~$\mathscr{U}$.

  \paragraph{Part (a): inheritance-calculus $\geq$ the lazy
  $\lambda$-calculus}
  We show that inheritance-calculus can macro-express
  every construct of the $\lambda$-calculus.
  The $\lambda$-calculus contains $\lambda$-abstraction,
  application, and variable.
  The translation $\mathcal{T}$ (Theorem~\ref{thm:forward-macro})
  directly maps each $\lambda$-calculus construct to a
  fixed syntactic abstraction over inheritance-calculus,
  satisfying condition~E4 of
  Felleisen~\cite[Definition~3.11]{felleisen1991-expressive-power}.
  Since inheritance-calculus trivially contains its own
  constructors, it can macro-express every construct of
  the $\lambda$-calculus.
  \fi
\end{proof}

\section{\bilingual{Multi-Target \texorpdfstring{$\this$}{this} Resolution in Other Systems}{其他系统中的多目标 \texorpdfstring{$\this$}{this} 解析}}
\label{app:scala-multi-target}

\ifPaperChinese
本附录展示两个具有代表性的系统——Scala~3 和 NixOS 模块系统——都以静态错误拒绝了继承演算能自然处理的多目标 $\this$ 模式。
\else
This appendix shows that two representative systems, Scala~3 and the NixOS
module system, both reject with a static error the multi-target $\this$ pattern
that inheritance-calculus handles naturally.
\fi

\subsection*{\bilingual{Scala~3}{Scala~3}}

\ifPaperChinese
考虑两个独立扩展同一外部类的对象，每个对象各自提供内部 trait 的一份副本：
\else
Consider two objects that independently extend an outer class,
each providing its own copy of an inner trait:
\fi

\begin{listing}[htbp]
  \caption{\bilingual{Scala 3 rejects two objects sharing an outer class through a common inner trait}{Scala 3 拒绝两个对象经由共同内部 trait 共享外层类}}
  \label{lst:scala-conflict}
\begin{minted}{scala}
class MyOuter:
  trait MyInner:
    def outer = MyOuter.this

object Object1 extends MyOuter
object Object2 extends MyOuter
object HasMultipleOuters extends Object1.MyInner
                     with Object2.MyInner
\end{minted}
\end{listing}

\noindent
\ifPaperChinese
Scala~3 以如下错误拒绝了 \mintinline{scala}{HasMultipleOuters}：
\else
Scala~3 rejects \mintinline{scala}{HasMultipleOuters} with the error:
\fi

\begin{minted}{text}
trait MyInner is extended twice
object HasMultipleOuters cannot be instantiated since
  it has conflicting base types
  Object1.MyInner and Object2.MyInner
\end{minted}

\ifPaperChinese
等价的继承演算定义如下：
\else
The equivalent inheritance-calculus definition is:
\fi
\begin{listing}[htbp]
  \caption{\bilingual{The same pattern is well-defined in inheritance-calculus}{同一模式在继承演算中是良定义的}}
  \label{lst:ic-multi-outer}
\begin{minted}{yaml}
- MyOuter:
  - MyInner:
    - outer:
      - [MyOuter, ~]
- Object1:
  - [MyOuter]
- Object2:
  - [MyOuter]
- HasMultipleOuters:
  - [Object1, MyInner]
  - [Object2, MyInner]
\end{minted}
\end{listing}
\ifPaperChinese
此定义在继承演算中是良定义的。\mintinline{yaml}{outer} 内部的引用 \mintinline{yaml}{[MyOuter, ~]} 具有 de~Bruijn 索引 $n = 1$：从 \mintinline{yaml}{outer} 的封闭作用域 \mintinline{yaml}{MyInner} 出发，经过一步 $\this$ 便到达 \mintinline{yaml}{MyOuter}。当 \mintinline{yaml}{HasMultipleOuters} 同时继承 \mintinline{yaml}{[Object1, MyInner]} 和 \mintinline{yaml}{[Object2, MyInner]} 时，$\this$ 函数（方程~\ref{eq:this}）通过在 $\supers(\text{``HasMultipleOuters''})$ 中搜索与 \mintinline{yaml}{MyInner} 定义位置匹配的覆写路径来解析 \mintinline{yaml}{[MyOuter, ~]}。它找到两条继承位置路径，一条经由 \mintinline{yaml}{Object1}，另一条经由 \mintinline{yaml}{Object2}，并返回两者。两条路径都通向继承自同一 \mintinline{yaml}{MyOuter} 的 mixin，因此查询某标签在 \mintinline{yaml}{[HasMultipleOuters, outer]} 处是否存在，无论走哪条路由，都与在 \mintinline{yaml}{MyOuter} 处给出相同答案。由于这样的路径只要某条路由到达一个定义该标签的 override 即存在,两条路由便无冲突地一致,该演算也无需任何人为的路由优先级或线性化规则。
\else
This is well-defined in inheritance-calculus.
The reference
\mintinline{yaml}{[MyOuter, ~]} inside \mintinline{yaml}{outer}
has de~Bruijn index $n = 1$:
starting from \mintinline{yaml}{outer}'s enclosing scope
\mintinline{yaml}{MyInner}, one $\this$ step reaches
\mintinline{yaml}{MyOuter}.
When \mintinline{yaml}{HasMultipleOuters} inherits from both
\mintinline{yaml}{[Object1, MyInner]} and
\mintinline{yaml}{[Object2, MyInner]},
the $\this$ function (equation~\ref{eq:this}) resolves
\mintinline{yaml}{[MyOuter, ~]} by searching through
$\supers(\text{``HasMultipleOuters''})$ for override paths matching
\mintinline{yaml}{MyInner}'s definition site.
It finds two inheritance-site paths, one through
\mintinline{yaml}{Object1} and one through \mintinline{yaml}{Object2}, and
returns both.
Both paths lead to mixins that inherit from the same
\mintinline{yaml}{MyOuter}, so whether a label exists at
\mintinline{yaml}{[HasMultipleOuters, outer]} has the same answer as at
\mintinline{yaml}{MyOuter} regardless of which
route is taken.
Because the path exists as soon as some route reaches an override that defines the label, the two routes agree without conflict, and the calculus needs no artificial route-priority or linearization rules.
\fi

\subsection*{\bilingual{NixOS Module System}{NixOS 模块系统}}

\ifPaperChinese
\begin{sloppypar}
NixOS 模块系统以静态错误拒绝了相同的模式。该翻译使用了模块系统自身的抽象：\mintinline{nix}{deferredModule} 对应 trait（即未求值的模块值，被导入到其他不动点中），\mintinline{nix}{submoduleWith} 对应对象（在其自身的不动点中求值）。
\end{sloppypar}
\else
The NixOS module system rejects the same pattern with a static error.
The translation uses the module system's own abstractions:
\mintinline{nix}{deferredModule} for traits, which are unevaluated module values imported
into other fixpoints, and \mintinline{nix}{submoduleWith} for objects, which are evaluated
in their own fixpoints.
\fi

\begin{listing}[htbp]
  \caption{\bilingual{The NixOS module system rejects the same pattern}{NixOS 模块系统拒绝同一模式}}
  \label{lst:nixos-conflict}
\begin{minted}{nix}
let
  lib = (import <nixpkgs> { }).lib;
  result = lib.evalModules {
    modules = [
      (toplevel@{ config, ... }: {
        # class MyOuter { trait MyInner {
        #   def outer = MyOuter.this } }
        options.MyOuter = lib.mkOption {
          default = { };
          type = lib.types.deferredModuleWith {
            staticModules = [
              (MyOuter: {
                options.MyInner = lib.mkOption {
                  default = { };
                  type = lib.types.deferredModuleWith {
                    staticModules = [
                      (MyInner: {
                        options.outer = lib.mkOption {
                          default = { };
                          type =
                            lib.types.deferredModuleWith {
                              staticModules =
                                [ toplevel.config.MyOuter ];
                            };
                        };
                      })
                    ];
                  };
                };
              })
            ];
          };
        };
        # object Object1 extends MyOuter
        options.Object1 = lib.mkOption {
          default = { };
          type = lib.types.submoduleWith {
            modules = [ toplevel.config.MyOuter ];
          };
        };
        # object Object2 extends MyOuter
        options.Object2 = lib.mkOption {
          default = { };
          type = lib.types.submoduleWith {
            modules = [ toplevel.config.MyOuter ];
          };
        };
        # object HasMultipleOuters extends
        #   Object1.MyInner with Object2.MyInner
        options.HasMultipleOuters = lib.mkOption {
          default = { };
          type = lib.types.submoduleWith {
            modules = [
            toplevel.config.Object1.MyInner
            toplevel.config.Object2.MyInner
          ];
          };
        };
      })
    ];
  };
in
  builtins.attrNames result.config.HasMultipleOuters.outer
\end{minted}
\end{listing}

\noindent
\ifPaperChinese
NixOS 模块系统以如下错误拒绝了此代码：
\else
The NixOS module system rejects this with:
\fi

\begin{minted}{text}
error: The option `HasMultipleOuters.outer'
  in `<unknown-file>'
  is already declared
  in `<unknown-file>'.
\end{minted}

\noindent
\ifPaperChinese
给嵌套模块设一个 \mintinline{nix}{key}(模块系统自己用来只导入一次模块的手段)能消除这个错误,但只保留一条路线:在清单~\ref{lst:nixos-conflict} 的一个变体中,\mintinline{nix}{Object1} 与 \mintinline{nix}{Object2} 各加一个选项,而 \mintinline{nix}{MyInner} 读出其外层配置的选项名;此时 \mintinline{nix}{HasMultipleOuters} 只看到 \mintinline{nix}{Object1} 与 \mintinline{nix}{Object2} 中在其模块列表里排在前面的那一个的选项,看不到另一个的。
\else
Giving the nested module a \mintinline{nix}{key}, the module system's own device for importing a
module once, removes the error but keeps one route: in a variant of Listing~\ref{lst:nixos-conflict}
where \mintinline{nix}{Object1} and \mintinline{nix}{Object2} each add an option and
\mintinline{nix}{MyInner} reads the option names of its enclosing configuration,
\mintinline{nix}{HasMultipleOuters} sees the option of whichever of \mintinline{nix}{Object1} and
\mintinline{nix}{Object2} comes first in its list of modules, and not the other's.
\fi

\noindent
\ifPaperChinese
Scala 与 NixOS 对同一模式的编码都无法通过编译或求值（上文的两处错误）；而该模式在继承演算中是良定义的：$\this$ 把 \mintinline{yaml}{[MyOuter, ~]} 解析为 \mintinline{yaml}{Object1} 与 \mintinline{yaml}{Object2} 两个继承位置（附录~\ref{app:evaluation-trace}），两条路由无冲突地一致，无需任何路由优先级或线性化规则。
\else
Both the Scala and the NixOS encodings of this pattern fail to compile or
evaluate (the errors shown above); the same pattern is well-defined in
inheritance-calculus, where $\this$ resolves \mintinline{yaml}{[MyOuter, ~]} to
the two inheritance sites \mintinline{yaml}{Object1} and
\mintinline{yaml}{Object2} (Appendix~\ref{app:evaluation-trace}), the two routes
agreeing without conflict and needing no route-priority or linearization rule.
\fi

\section{\bilingual{Encoding Datalog in Inheritance-Calculus}{在继承演算中编码 Datalog}}
\label{app:datalog-encoding}

\ifPaperChinese
第~\ref{sec:emergent-phenomena}~节观察到继承演算原生符合逻辑编程的关系语义。
本附录把该观察报告为 Datalog 到继承演算的系统性编码,以代表性规则示例呈现,
而非证明一条翻译定理。
该编码不使用演算的任何扩展;它可用第~\ref{sec:syntax}~节的三个基本原语表达。
本附录的所有示例均有可执行的 MIXINv2 实现\anon[~(补充材料)]{~\cite{yang2026-mixinv2}}。
需要区分三个要素。
首先,关系以 trie 形式存储,因此元组查找通过前缀索引,而非扫描平面事实集。
其次,控制骨架采用第~\ref{sec:case-study}~节已使用的访问者模式风格编写;从函数式编程的角度看,这与 Scott 编码的情况分析所扮演的角色相同。
第三,继承并和 tabling 求值赋予最终程序关系式的最小不动点语义。
编码自身引入的每个标签都带双下划线保留前缀:trie 与访问者机制
(\mintinline{yaml}{__Relation}、\mintinline{yaml}{__Nil}、
\mintinline{yaml}{__Cons}、\mintinline{yaml}{__TailMap}、
\mintinline{yaml}{__Acceptance} 及其余),以及派生自规则变量的
逐规则阶段作用域(如下文的 \mintinline{yaml}{__XAcceptance})。
保留字不在 Datalog 程序的词汇表中:关系名与常量从不带该前缀,
规则变量自身($X$、$Y$、$Z$)是用户词汇、从不成为标签,
故编码机制与任何源名字绝不冲突。
本附录按此顺序展开:先是数据表示,然后是控制流,最后是递归。
\else
Section~\ref{sec:emergent-phenomena} observed that
inheritance-calculus natively agrees with the relational semantics
of logic programming.
This appendix reports that observation as a systematic
encoding of Datalog into inheritance-calculus, exhibited on
representative rules rather than proved as a translation theorem.
The encoding uses no extensions to the calculus; it is
expressible in the three primitives of Section~\ref{sec:syntax}.
All examples in this appendix have executable MIXINv2
implementations\anon[~(supplementary material)]{~\cite{yang2026-mixinv2}}.
There are three ingredients to keep separate.
First, relations are stored as tries, so tuple lookup is indexed
by prefixes rather than by scanning a flat set of facts.
Second, the control skeleton is written in the visitor-pattern
style already used in Section~\ref{sec:case-study}; from a
functional-programming viewpoint, this is the same role played by
Scott-encoded case analysis.
Third, inheritance union and tabled evaluation give the resulting
program its relational, least-fixed-point semantics.
Every label the encoding itself introduces carries the
double-underscore reserved prefix: the trie and visitor machinery
(\mintinline{yaml}{__Relation}, \mintinline{yaml}{__Nil},
\mintinline{yaml}{__Cons}, \mintinline{yaml}{__TailMap},
\mintinline{yaml}{__Acceptance}, and the rest) together with the
per-rule stage scopes derived from rule variables (such as
\mintinline{yaml}{__XAcceptance} below). Reserved words are not in
the vocabulary of the Datalog program: relation names and constants
never carry the prefix, and the rule variables themselves ($X$,
$Y$, $Z$) are user vocabulary that never becomes a label, so the
encoding's machinery collides with no source name.
The appendix proceeds in that order: data representation first,
then control flow, then recursion.
\fi

\subsection{\bilingual{Relations as Tries}{关系作为 Trie}}

\ifPaperChinese
有限域 $D = \{a, b, c, \ldots\}$ 上的二元关系表示为 \emph{trie}:一棵路径编码元组的 mixin 树。
该 trie 由三个基础作用域构成:
\mintinline{yaml}{__Relation}（抽象基础）、
\mintinline{yaml}{__Cons}（非空节点）和 \mintinline{yaml}{__Nil}（终止符）。
每个常量 $d \in D$ 继承自 \mintinline{yaml}{__Cons},并携带一个 \mintinline{yaml}{__TailMap},该映射提供逐常量的子 trie:
\begin{minted}{yaml}
- __Relation:
- __Cons:
  - [__Relation]
  - __Tail:
    - [__Relation]
- __Nil:
  - [__Relation]
- d:
  - [__Cons]
  - __Tail:
  - __TailMap:
    - d:
      - [__Tail]
\end{minted}
事实 $R(d_1, d_2)$ 编码为深度为 2 的 trie 路径,带有 \mintinline{yaml}{__Nil} 终止符。
逐常量的子 trie 通过 $\olbl{__TailMap}.d$ 访问。$R(a, b)$ 的编码为:
\begin{minted}{yaml}
- R:
  - [a]
  - __TailMap:
    - a:
      - [b]
      - __TailMap:
        - b:
          - [__Nil]
\end{minted}
同一关系的多个事实通过继承（trie 并）累积。
添加 $R(b, c)$:
\begin{minted}{yaml}
- R:
  - [a]
  - [b]
  - __TailMap:
    - a:
      - [b]
      - __TailMap:
        - b:
          - [__Nil]
    - b:
      - [c]
      - __TailMap:
        - c:
          - [__Nil]
\end{minted}
trie 表示的意义不仅在于它是树形编码。
它还提供后续连接所使用的索引结构:一旦选定常量 $d$,编码就直接移至 $\olbl{__TailMap}.d$ 并从匹配的子 trie 继续。
这就是为什么本附录使用 trie 这一说法,而非有限元组集的任意编码。
\else
A binary relation over a finite domain $D = \{a, b, c, \ldots\}$
is represented as a \emph{trie}: a mixin tree whose paths
encode tuples.
The trie is built from three base scopes:
\mintinline{yaml}{__Relation} (abstract base),
\mintinline{yaml}{__Cons} (non-empty node), and \mintinline{yaml}{__Nil} (terminator).
Each constant $d \in D$ inherits from \mintinline{yaml}{__Cons} and carries
a \mintinline{yaml}{__TailMap} that provides per-constant sub-tries:
\begin{minted}{yaml}
- __Relation:
- __Cons:
  - [__Relation]
  - __Tail:
    - [__Relation]
- __Nil:
  - [__Relation]
- d:
  - [__Cons]
  - __Tail:
  - __TailMap:
    - d:
      - [__Tail]
\end{minted}
A fact $R(d_1, d_2)$ is encoded as a trie path of depth~2,
with a \mintinline{yaml}{__Nil} terminator.
The per-constant sub-tries are accessed via
$\olbl{__TailMap}.d$.
The encoding of $R(a, b)$ is:
\begin{minted}{yaml}
- R:
  - [a]
  - __TailMap:
    - a:
      - [b]
      - __TailMap:
        - b:
          - [__Nil]
\end{minted}
Multiple facts for the same relation are accumulated by
inheritance (trie union).
Adding $R(b, c)$:
\begin{minted}{yaml}
- R:
  - [a]
  - [b]
  - __TailMap:
    - a:
      - [b]
      - __TailMap:
        - b:
          - [__Nil]
    - b:
      - [c]
      - __TailMap:
        - c:
          - [__Nil]
\end{minted}
The point of the trie representation is not merely that it is a
tree-shaped encoding.
It also provides the indexing structure used by the subsequent
joins: once a constant $d$ has been selected, the encoding moves
directly to $\olbl{__TailMap}.d$ and continues from the matching
sub-trie.
This is why the appendix speaks about tries rather than about an
arbitrary encoding of finite sets of tuples.
\fi

\subsection{\bilingual{Infrastructure}{基础设施}}

\ifPaperChinese
每个常量 $d$ 携带编码所使用的四组样板定义。
这些定义的作用类似于 Haskell 中派生的类型类实例:每个常量必须提供编码的通用模式所依赖的一组固定定义。
继承演算没有内置的派生机制,因此样板需手动写出;但每个常量的样板量为 $O(1)$,不随域大小 $|D|$ 增长,因此添加新常量不需要修改任何现有作用域。
\else
Each constant $d$ carries four groups of boilerplate
definitions used by the encoding.
These play a role analogous to derived type-class instances
in Haskell: each constant must provide a fixed set of
definitions that the encoding's generic patterns rely on.
Inheritance-calculus has no built-in derivation mechanism,
so the boilerplate is written out manually; however, the
amount of boilerplate per constant is $O(1)$
and does not grow with the domain size $|D|$,
so adding a new constant
does not require modifying any existing scope.
\fi


\begin{description}
  \item[\texttt{\_\_Acceptance}]
    \ifPaperChinese
    第~\ref{sec:case-study}~节的访问者模式基础设施。
    操作上,这是对当前常量进行情况分析的面向对象写法。
    每个常量 $d$ 定义一个 \mintinline{yaml}{__VisitorMap},其中包含以 $d$ 为键的分支、一个带有用于逐常量结果的 abstract \mintinline{yaml}{__Visited} 的 \mintinline{yaml}{__Visitor},以及一个将匹配分支的 \mintinline{yaml}{__Visited} 值委托给外部的 \mintinline{yaml}{__Accepted} mixin。
    这是分派机制:给定一个可能包含多个常量的 mixin,\mintinline{yaml}{__Acceptance} 选择每个存在常量对应的分支,并收集逐常量的 \mintinline{yaml}{__Visited} 结果。
    \else
    The visitor-pattern infrastructure from
    Section~\ref{sec:case-study}.
    Operationally, this is the object-oriented spelling of a
    case split on the current constant.
    Each constant $d$ defines a \mintinline{yaml}{__VisitorMap}
    containing a branch keyed by $d$, a \mintinline{yaml}{__Visitor}
    with an abstract \mintinline{yaml}{__Visited} for per-constant results,
    and an \mintinline{yaml}{__Accepted} mixin that delegates to the
    matching branch's \mintinline{yaml}{__Visited} value.
    This is the dispatch mechanism: given a mixin that may
    contain multiple constants, \mintinline{yaml}{__Acceptance} selects the
    branch corresponding to each constant present and
    collects the per-constant \mintinline{yaml}{__Visited} results.
    \fi

  \item[\texttt{\_\_WithVisitorTail}]
    \ifPaperChinese
    \begin{sloppypar}
    将每个常量 $d$ 映射到其子 trie 条目 $\olbl{__TailMap}.d$,标记为 \mintinline{yaml}{__VisitorTail}。
    这提供了连接的右端\footnote{
      在访问者模式的语境中,这充当访问者。
    }:当两个关系的 \mintinline{yaml}{__VisitorMap} 条目组合时,\mintinline{yaml}{__VisitorTail} 携带已连接关系中匹配常量的行数据。
    \end{sloppypar}
    \else
    \begin{sloppypar}
    Maps each constant $d$ to its sub-trie entry
    $\olbl{__TailMap}.d$, labeled as \mintinline{yaml}{__VisitorTail}.
    This provides the right-hand side of a join\footnote{
      In the context of the visitor pattern, this acts as the visitor.
    }: when two relations' \mintinline{yaml}{__VisitorMap} entries are
    composed, \mintinline{yaml}{__VisitorTail} carries the joined
    relation's row data for the matched constant.
    \end{sloppypar}
    \fi

  \item[\texttt{\_\_WithVisiteeTail}]
    \ifPaperChinese
    \begin{sloppypar}
    将每个常量 $d$ 映射到其子 trie 条目 $\olbl{__TailMap}.d$,标记为 \mintinline{yaml}{__VisiteeTail}。
    这提供了左端\footnote{
      这充当被访问者。
    }:即被迭代的关系。\mintinline{yaml}{__WithVisitorTail} 和 \mintinline{yaml}{__WithVisiteeTail} 共同支持值级连接,其中只有\emph{两个}关系中均存在的常量才产生输出。
    \end{sloppypar}
    \else
    \begin{sloppypar}
    Maps each constant $d$ to its sub-trie entry
    $\olbl{__TailMap}.d$, labeled as \mintinline{yaml}{__VisiteeTail}.
    This provides the left-hand side\footnote{
      This acts as the visitee.
    }: the
    relation being iterated. Together,
    \mintinline{yaml}{__WithVisitorTail} and \mintinline{yaml}{__WithVisiteeTail}
    enable value-level joins where only constants present
    in \emph{both} relations produce output.
    \end{sloppypar}
    \fi

  \item[\texttt{\_\_Replacement}]
    \ifPaperChinese
    给定一个 \mintinline{yaml}{__NewTail} 值,产生一个 \mintinline{yaml}{__Replaced} mixin,该 mixin 继承标签 $d$,但将 $\olbl{__TailMap}.d$ 替换为 \mintinline{yaml}{__NewTail}。
    这是输出构造机制:它通过将新子 trie 替换到常量包装器中来构建输出元组。
    关键在于,\mintinline{yaml}{__Replacement} 定义在逐常量的 \mintinline{yaml}{__TailMap} 条目上（而非合并后的作用域上）,确保输出构造作用于单个常量而非其并集。
    \else
    Given a \mintinline{yaml}{__NewTail} value, produces a
    \mintinline{yaml}{__Replaced} mixin that inherits the label $d$
    but substitutes $\olbl{__TailMap}.d$ with
    \mintinline{yaml}{__NewTail}.
    This is the output construction mechanism: it
    builds output tuples by substituting new sub-tries into
    constant wrappers.
    Crucially, \mintinline{yaml}{__Replacement} is defined on the
    per-constant \mintinline{yaml}{__TailMap} entry (not on the merged
    scope), ensuring that output construction operates on
    individual constants rather than their union.
    \fi
\end{description}

\ifPaperChinese
我们以可执行的两跳编码
$\mathrm{twoHop}(X, Z) \mathrel{{:}\!{-}} \mathrm{edge}(X, Y),\; \mathrm{edge}(Y, Z)$
为贯穿示例,取自
\anon[\nolinkurl{RelationalTwoHopRepeatedVisitor.mixin.yaml}~(补充材料)]{\nolinkurl{RelationalTwoHopRepeatedVisitor.mixin.yaml}~\cite{yang2026-mixinv2}}。
其顶层作用域按编码各阶段命名:\mintinline{yaml}{__XAcceptance} 迭代第一列,\mintinline{yaml}{__YAcceptance} 执行共享变量连接,\mintinline{yaml}{__NilAcceptance0} 确认第一个事实的末尾,\mintinline{yaml}{__ZAcceptance} 迭代输出列,\mintinline{yaml}{__NilAcceptance1} 触发最终的 \mintinline{yaml}{__Replacement} 链。
下面的示意编码应理解为该可执行 MIXINv2 文件的抽象。
若暂时忽略关系累积并将每次分派读作普通的情况分析,则控制流为:从第一个关系中选取 $X$,选取匹配的 $Y$,确认第一个元组已完整,从第二个关系匹配行中选取 $Z$,确认第二个元组已完整,并输出 $(X, Z)$。
本附录的其余部分解释这一熟悉的单值控制骨架如何通过对 trie 的访问者分派实现,并随后通过继承并提升为关系语义。
\else
\begin{sloppypar}
As a running example, we use the executable two-hop encoding
$\mathrm{twoHop}(X, Z) \mathrel{{:}\!{-}} \mathrm{edge}(X, Y),\allowbreak\; \mathrm{edge}(Y, Z)$
from
\anon[\nolinkurl{RelationalTwoHopRepeatedVisitor.mixin.yaml}~(supplementary material)]{\nolinkurl{RelationalTwoHopRepeatedVisitor.mixin.yaml}~\cite{yang2026-mixinv2}}.
Its top-level scopes are named after the stages of the
encoding: \mintinline{yaml}{__XAcceptance} iterates the first column,
\mintinline{yaml}{__YAcceptance} performs the shared-variable join,
\mintinline{yaml}{__NilAcceptance0} confirms the end of the first fact,
\mintinline{yaml}{__ZAcceptance} iterates the output column, and
\mintinline{yaml}{__NilAcceptance1} triggers the final
\mintinline{yaml}{__Replacement} chain.
The schematic encodings below should be read as abstractions
of that executable MIXINv2 file.
If one temporarily ignores relational accumulation and reads
each dispatch as an ordinary case split, the control flow is:
choose $X$ from the first relation, choose a matching $Y$,
confirm that the first tuple is complete, choose $Z$ from the
second relation's matching row, confirm that the second tuple
is complete, and emit $(X, Z)$.
The rest of the appendix explains how that familiar
single-valued control skeleton is realised by visitor dispatch
over tries and then lifted to relational semantics by
inheritance union.
\end{sloppypar}
\fi

\subsection{\bilingual{Encoding Rules}{编码规则}}

\ifPaperChinese
Datalog 规则 $H(\bar{X}) \mathrel{{:}\!{-}} B_1(\bar{Y}_1), \ldots, B_n(\bar{Y}_n)$ 通过在该控制骨架内嵌套三种机制来编码:

\begin{enumerate}
  \item \textbf{列迭代。}
    对每个列位置,一次 \mintinline{yaml}{__Acceptance} 分派迭代存在的常量。
    \mintinline{yaml}{__VisitorMap} 从 \mintinline{yaml}{__WithVisiteeTail} 填充,因此每个匹配常量的 \mintinline{yaml}{__Visitor} 接收 \mintinline{yaml}{__VisiteeTail}:该位置处逐常量的子 trie。
    处理每个常量的结果放入 \mintinline{yaml}{__Visited},\mintinline{yaml}{__Accepted} 收集所有逐常量的结果。
    （详见第~\ref{sec:encoding-variables}、\ref{sec:encoding-constants}~和~\ref{sec:encoding-multi}~节。）

  \item \textbf{连接。}
    对于体原子之间的每个共享变量,内层 \mintinline{yaml}{__Acceptance} 的 \mintinline{yaml}{__VisitorMap} 同时继承自 \mintinline{yaml}{__WithVisitorTail}（来自被连接关系）和 \mintinline{yaml}{__WithVisiteeTail}（来自被迭代关系）。
    常量 $d$ 仅当出现在两个关系中时才产生 \mintinline{yaml}{__Visitor} 条目,从而实现值级连接。
    随后 \mintinline{yaml}{__Visitor} 接收 \mintinline{yaml}{__VisiteeTail} 和 \mintinline{yaml}{__VisitorTail} 两者,从而可访问匹配常量在两个关系中的数据。

  \item \textbf{输出构造。}
    在最内层的 \mintinline{yaml}{__Nil} 确认（验证每个体原子的元组已完整）内部,从最内层输出列向外链式执行 \mintinline{yaml}{__Replacement} 操作,每步将 \mintinline{yaml}{__NewTail} 设为前一个 \mintinline{yaml}{__Replaced} 值,最终由 \mintinline{yaml}{__Nil} 终止。
    最终的 \mintinline{yaml}{__Replaced} 值赋给 \mintinline{yaml}{__Visited}。
\end{enumerate}
\else
A Datalog rule $H(\bar{X}) \mathrel{{:}\!{-}} B_1(\bar{Y}_1), \ldots, B_n(\bar{Y}_n)$
is encoded by nesting three mechanisms inside that control
skeleton:

\begin{enumerate}
  \item \textbf{Column iteration.}
    For each column position, an \mintinline{yaml}{__Acceptance} dispatch
    iterates over the constants present.
    The \mintinline{yaml}{__VisitorMap} is populated from
    \mintinline{yaml}{__WithVisiteeTail}, so each matched constant's
    \mintinline{yaml}{__Visitor} receives \mintinline{yaml}{__VisiteeTail}: the
    per-constant sub-trie at that position.
    The result of processing each constant is placed in
    \mintinline{yaml}{__Visited}, and \mintinline{yaml}{__Accepted} collects all
    per-constant results.
    (Details in Sections~\ref{sec:encoding-variables},
    \ref{sec:encoding-constants}, and~\ref{sec:encoding-multi}.)

  \item \textbf{Join.}
    \begin{sloppypar}
    For each shared variable between body atoms, the inner
    \mintinline{yaml}{__Acceptance}'s \mintinline{yaml}{__VisitorMap} inherits from
    \emph{both} \mintinline{yaml}{__WithVisitorTail} (from the joined
    relation) and \mintinline{yaml}{__WithVisiteeTail} (from the iterated
    relation).
    A constant $d$ produces a \mintinline{yaml}{__Visitor} entry only if
    $d$ appears in both relations, achieving a value-level
    join.
    The \mintinline{yaml}{__Visitor} then receives both
    \mintinline{yaml}{__VisiteeTail} and \mintinline{yaml}{__VisitorTail},
    providing access to both relations' data for the
    matched constant.
    \end{sloppypar}

  \item \textbf{Output construction.}
    Inside the innermost \mintinline{yaml}{__Nil} confirmation (verifying
    that each body atom's tuple is complete), chain
    \mintinline{yaml}{__Replacement} operations from the innermost
    output column outward, each setting \mintinline{yaml}{__NewTail} to
    the previous \mintinline{yaml}{__Replaced} value, terminated by
    \mintinline{yaml}{__Nil}.
    The final \mintinline{yaml}{__Replaced} value is assigned to
    \mintinline{yaml}{__Visited}.
\end{enumerate}
\fi

\subsection{\bilingual{Variable Join}{变量连接}}
\label{sec:encoding-variables}

\ifPaperChinese
考虑最简单的连接:
$\mathrm{composed}(X, Z) \mathrel{{:}\!{-}} \alpha(X, Y),\; \beta(Y, Z)$。
该模式的可执行自连接实例为
\nolinkurl{RelationalTwoHopRepeatedVisitor.mixin.yaml},
其作用域名称在下文括号中给出。
这里可以清楚地看到节引言中的分工:访问者分派提供控制流,trie 提供对匹配尾部的直接访问,而 $\supers$ 累积每次成功匹配产生的输出。

共享变量 $Y$ 要求在 $\alpha$ 的第二列和 $\beta$ 的第一列之间匹配常量。
编码嵌套五次 \mintinline{yaml}{__Acceptance} 分派:从 $\alpha$ 迭代 $X$、在 $\alpha$ 与 $\beta$ 之间对 $Y$ 连接、确认 $\alpha$ 中 $Y$ 之后的 \mintinline{yaml}{__Nil}、从 $\beta$ 匹配行迭代 $Z$、确认 $Z$ 之后的 \mintinline{yaml}{__Nil}。
我们将每一层作为一个命名嵌套作用域呈现。
\else
Consider the simplest join:
$\mathrm{composed}(X, Z) \mathrel{{:}\!{-}} \alpha(X, Y),\; \beta(Y, Z)$.
The executable self-join instance of this pattern is
\nolinkurl{RelationalTwoHopRepeatedVisitor.mixin.yaml},
whose scope names are shown in parentheses below.
This is the right place to see the division of labour from the
section introduction: visitor dispatch provides the control flow,
the trie provides direct access to the matching tails, and
$\supers$ accumulates the outputs produced by each successful
match.

The shared variable $Y$ requires matching constants
between $\alpha$'s second column and $\beta$'s first column.
The encoding nests five \mintinline{yaml}{__Acceptance} dispatches:
iterate $X$ from $\alpha$, join on $Y$ between $\alpha$
and $\beta$, confirm \mintinline{yaml}{__Nil} after $Y$ in $\alpha$,
iterate $Z$ from $\beta$'s matched row, and confirm
\mintinline{yaml}{__Nil} after $Z$.
We present each level as a named nested scope.
\fi

\ifPaperChinese
\paragraph{第~1~层: 从 \texorpdfstring{$\alpha$}{α} 迭代 \texorpdfstring{$X$}{X}}
(可执行作用域: \mintinline{yaml}{__XAcceptance}。)
\else
\paragraph{Level~1: iterate $X$ from $\alpha$}
(Executable scope: \mintinline{yaml}{__XAcceptance}.)
\fi
{\small
\begin{minted}[escapeinside=||]{yaml}
- __xAcc:
  - [|$\alpha$|, __Acceptance]
  - __VisitorMap:
    - [|$\alpha$|, __WithVisiteeTail]
  - __Visitor:
    - __X:
      - [__VisiteeTail]
    - |$\ldots$|
    - __Visited:
      - [__yAcc, __Accepted]
\end{minted}
}

\ifPaperChinese
\paragraph{第~2~层: 将 \texorpdfstring{$\alpha$}{α} 的 \texorpdfstring{$Y$}{Y} 列与 \texorpdfstring{$\beta$}{β} 的第一列连接}
(可执行作用域: \mintinline{yaml}{__YAcceptance}。)
\else
\paragraph{Level~2: join $\alpha$'s $Y$-column with $\beta$'s first column}
(Executable scope: \mintinline{yaml}{__YAcceptance}.)
\fi

\ifPaperChinese
\mintinline{yaml}{__VisitorMap} 同时继承自 $\beta.\olbl{__WithVisitorTail}$ 和 $\olbl{__X}.\olbl{__WithVisiteeTail}$,仅当常量 $d$ 同时出现在 $\alpha$ 的第二列（通过 \mintinline{yaml}{__X}）和 $\beta$ 的第一列时才产生条目。
\else
The \mintinline{yaml}{__VisitorMap} inherits from \emph{both}
$\beta.\olbl{__WithVisitorTail}$ and
$\olbl{__X}.\olbl{__WithVisiteeTail}$, producing an entry for
constant $d$ only if $d$ appears in both $\alpha$'s
second column (via \mintinline{yaml}{__X}) and $\beta$'s first column.
\fi
{\small
\begin{minted}[escapeinside=||]{yaml}
- __yAcc:
  - [__X, __Acceptance]
  - __VisitorMap:
    - [|$\beta$|, __WithVisitorTail]
    - [__X, __WithVisiteeTail]
  - __Visitor:
    - __Y0:
      - [__VisiteeTail]
    - __VisitorTail:
      - [|$\beta$|, __Tail]
    - |$\ldots$|
    - __Visited:
      - [__nilAcc0, __Accepted]
\end{minted}
}

\ifPaperChinese
\paragraph{第~3~层: 确认 \texorpdfstring{$\alpha$}{α} 中 \texorpdfstring{$Y$}{Y} 之后的 \texorpdfstring{$\olbl{__Nil}$}{__Nil}}
(可执行作用域: \mintinline{yaml}{__NilAcceptance0}。)
\else
\paragraph{Level~3: confirm $\olbl{__Nil}$ after $Y$ in $\alpha$}
(Executable scope: \mintinline{yaml}{__NilAcceptance0}.)
\fi

\ifPaperChinese
此步骤检查 $\alpha(X, Y)$ 是一个完整的二列事实（其 trie 路径在 $Y$ 之后终止于 \mintinline{yaml}{__Nil}）。
在 \mintinline{yaml}{__Nil} 分支内部,编码继续通过 \mintinline{yaml}{__VisitorTail} 迭代 $\beta$ 的第二列:
\else
This checks that $\alpha(X, Y)$ is a complete two-column
fact (its trie path terminates at \mintinline{yaml}{__Nil} after $Y$).
Inside the \mintinline{yaml}{__Nil} branch, the encoding proceeds to
iterate $\beta$'s second column via \mintinline{yaml}{__VisitorTail}:
\fi
{\small
\begin{minted}[escapeinside=||]{yaml}
- __nilAcc0:
  - [|$\YsubZero$|, __Acceptance]
  - __VisitorMap:
    - __Nil:
      - |$\ldots$|
      - __Visited:
        - [__zAcc, __Accepted]
\end{minted}
}

\ifPaperChinese
\paragraph{第~4~层: 从 \texorpdfstring{$\beta$}{β} 的匹配行迭代 \texorpdfstring{$Z$}{Z}}
(可执行作用域: \mintinline{yaml}{__ZAcceptance}。)
\else
\paragraph{Level~4: iterate $Z$ from $\beta$'s matched row}
(Executable scope: \mintinline{yaml}{__ZAcceptance}.)
\fi
{\small
\begin{minted}[escapeinside=||]{yaml}
- __zAcc:
  - [__VisitorTail, __Acceptance]
  - __VisitorMap:
    - [__VisitorTail, __WithVisiteeTail]
  - __Visitor:
    - __Z:
      - [__VisiteeTail]
    - |$\ldots$|
    - __Visited:
      - [__nilAcc1, __Accepted]
\end{minted}
}

\ifPaperChinese
\paragraph{第~5~层: 确认 \texorpdfstring{$Z$}{Z} 之后的 \texorpdfstring{$\olbl{__Nil}$}{__Nil} 并构造输出}
(可执行作用域: \mintinline{yaml}{__NilAcceptance1},包含 \mintinline{yaml}{__ZReplacement} 和 \mintinline{yaml}{__XReplacement}。)
\else
\paragraph{Level~5: confirm $\olbl{__Nil}$ after $Z$ and construct output}
(Executable scope: \mintinline{yaml}{__NilAcceptance1}, with
\mintinline{yaml}{__ZReplacement} and \mintinline{yaml}{__XReplacement}.)
\fi
{\small
\begin{minted}[escapeinside=||]{yaml}
- __nilAcc1:
  - [__VisiteeTail, __Acceptance]
  - __VisitorMap:
    - __Nil:
      - zR:
        - [__Z, __Replacement]
        - __NewTail:
          - [__Nil]
      - xR:
        - [__X, __Replacement]
        - __NewTail:
          - [zR, __Replaced]
      - __Visited:
        - [xR, __Replaced]
\end{minted}
}

\paragraph{\bilingual{Final result}{最终结果}}
\begin{minted}{yaml}
- composed:
  - [__xAcc, __Accepted]
\end{minted}

\paragraph{\bilingual{Semantics}{语义}}
\ifPaperChinese
在第~2~层,组合的 \mintinline{yaml}{__VisitorMap} 确保常量 $d$ 仅当同时出现在 $\alpha$ 的第二列\emph{和} $\beta$ 的第一列时才产生 \mintinline{yaml}{__Visitor}。
每个匹配常量 $d$ 的 \mintinline{yaml}{__Visitor} 独立接收 \mintinline{yaml}{__VisiteeTail}（$\alpha$ 中 $d$ 的数据）和 \mintinline{yaml}{__VisitorTail}（$\beta$ 中 $d$ 的行）。
输出构造（第~5~层）在这个逐常量的 \mintinline{yaml}{__Visitor} 内部执行,因此 \mintinline{yaml}{__Replacement} 作用于\emph{单个}常量的数据,而非合并后的作用域。
每个匹配常量 $d$ 在 \mintinline{yaml}{__Visited} 中产生其自身的输出元组,这些元组由 \mintinline{yaml}{__Accepted} 通过 trie 并（$\supers$）收集。
这产生了标准 Datalog 元组级语义:只有实际在两个关系中均匹配的常量才对输出有贡献。
对于具体实例
$\mathrm{edge} = \{(a,b),\; (b,c),\; (c,a)\}$,如编码于
\nolinkurl{RelationalTwoHopRepeatedVisitor.mixin.yaml},
第~1~层选取 $a$,第~2~层匹配 $b$,第~3~层确认 $\mathrm{edge}(a,b)$ 终止于 \mintinline{yaml}{__Nil},第~4~层迭代 $\olbl{edge}.\olbl{__TailMap}.b$ 并选取 $c$,第~5~层再次确认 \mintinline{yaml}{__Nil},随后两次 \mintinline{yaml}{__Replacement} 操作重建 $(a,c)$。
\else
\begin{sloppypar}
At Level~2, the combined \mintinline{yaml}{__VisitorMap} ensures that a
constant $d$ produces a \mintinline{yaml}{__Visitor} only when $d$
appears in both $\alpha$'s second column \emph{and}
$\beta$'s first column.
The \mintinline{yaml}{__Visitor} for each matched constant $d$
independently receives \mintinline{yaml}{__VisiteeTail} ($\alpha$'s
data for~$d$) and \mintinline{yaml}{__VisitorTail} ($\beta$'s row
for~$d$).
Output construction (Level~5) operates inside this
per-constant \mintinline{yaml}{__Visitor}, so \mintinline{yaml}{__Replacement}
acts on the \emph{individual} constant's data, not on the
merged scope.
Each matched constant $d$ produces its own output tuples
in \mintinline{yaml}{__Visited}, and these are collected by
\mintinline{yaml}{__Accepted} via trie union ($\supers$).
This yields standard Datalog tuple-level semantics: only
constants that actually match between both relations
contribute to the output.
For the concrete instance
$\mathrm{edge} = \{(a,b),\; (b,c),\; (c,a)\}$, as encoded in
\nolinkurl{RelationalTwoHopRepeatedVisitor.mixin.yaml},
Level~1 selects $a$, Level~2 matches $b$, Level~3 confirms
that $\mathrm{edge}(a,b)$ ends in \mintinline{yaml}{__Nil}, Level~4
iterates $\olbl{edge}.\olbl{__TailMap}.b$ and selects $c$,
and Level~5 confirms \mintinline{yaml}{__Nil} again before the two
\mintinline{yaml}{__Replacement} operations reconstruct $(a,c)$.
\end{sloppypar}
\fi

\subsection{\bilingual{Constant Match and Navigation}{常量匹配与导航}}
\label{sec:encoding-constants}

\ifPaperChinese
Datalog 原子中的常量使用两种机制,均非连接。
它们比变量连接更简单,因为不需要将第二个关系的值与当前选定的常量进行匹配:
\else
Constants in Datalog atoms use two mechanisms, neither of
which is a join.
They are simpler than variable joins because no value from a
second relation needs to be matched against the currently chosen
constant:
\fi

\paragraph{\bilingual{Type existence match}{类型存在匹配}}
\ifPaperChinese
原子 $R(X, b)$（第二位置为常量）要求检查常量 $b$ 是否存在于 $R$ 在匹配行处的第二列。
这编码为一次 \mintinline{yaml}{__Acceptance},其手工构造的 \mintinline{yaml}{__VisitorMap} \emph{仅}包含 $b$ 分支（而非从 \mintinline{yaml}{__WithVisitorTail} 组合,后者会为所有常量创建分支）:
{\small
\begin{minted}[escapeinside=||]{yaml}
- __VisitorMap:
  - |$b$|:
    - |$\ldots$|
\end{minted}
}
若 $b$ 存在于分派作用域中,则 $b$ 分支触发；否则 \mintinline{yaml}{__Accepted} 不接收任何匹配分支。
\else
An atom $R(X, b)$ with a constant in the second position
requires checking that constant $b$ exists in $R$'s
second column at the matched row.
This is encoded as an \mintinline{yaml}{__Acceptance} with a
hand-crafted \mintinline{yaml}{__VisitorMap} containing \emph{only}
the $b$ branch (instead of composing from
  \mintinline{yaml}{__WithVisitorTail}, which would create branches for
all constants):
{\small
\begin{minted}[escapeinside=||]{yaml}
- __VisitorMap:
  - |$b$|:
    - |$\ldots$|
\end{minted}
}
If $b$ exists in the dispatching scope, the
$b$ branch fires; otherwise, \mintinline{yaml}{__Accepted}
receives no matching branch.
\fi

\paragraph{\bilingual{Direct navigation}{直接导航}}
\ifPaperChinese
原子 $R(b, Z)$（第一位置为常量）直接导航到 $b$ 的子 trie:
\begin{minted}[escapeinside=||]{yaml}
- |$\edgeBZ$|:
  - [|$R$|, __TailMap, |$b$|]
\end{minted}
这完全绕过了访问者机制,通过访问逐常量的 \mintinline{yaml}{__TailMap} 条目将作用域限制为仅 $b$ 的数据。
\else
An atom $R(b, Z)$ with a constant in the first position
navigates directly to $b$'s sub-trie:
\begin{minted}[escapeinside=||]{yaml}
- |$\edgeBZ$|:
  - [|$R$|, __TailMap, |$b$|]
\end{minted}
This bypasses the visitor mechanism entirely, accessing
the per-constant \mintinline{yaml}{__TailMap} entry to restrict the
scope to only $b$'s data.
\fi

\paragraph{\bilingual{Example}{示例}}
\ifPaperChinese
$\mathrm{viaB}(X, Z) \mathrel{{:}\!{-}} \mathrm{edge}(X, b),\; \mathrm{edge}(b, Z)$
结合了两种机制。
可执行的 MIXINv2 编码为
\anon[\nolinkurl{RelationalConstant.mixin.yaml}~(补充材料)]{\nolinkurl{RelationalConstant.mixin.yaml}~\cite{yang2026-mixinv2}}。
编码迭代 \mintinline{yaml}{edge} 以得到 $X$（第~1~层）,然后用手工构造的 \mintinline{yaml}{__VisitorMap} 检查 $X$ 的行中是否存在常量 $b$（第~2~层）,导航 $\olbl{edge}.\olbl{__TailMap}.b$ 以得到 $Z$（第~3~层）,确认 \mintinline{yaml}{__Nil}（第~4~层）,并通过 \mintinline{yaml}{__Replacement} 构造 $(X, Z)$。
\else
$\mathrm{viaB}(X, Z) \mathrel{{:}\!{-}} \mathrm{edge}(X, b),\; \mathrm{edge}(b, Z)$
combines both mechanisms.
The executable MIXINv2 encoding is
\anon[\nolinkurl{RelationalConstant.mixin.yaml}~(supplementary material)]{\nolinkurl{RelationalConstant.mixin.yaml}~\cite{yang2026-mixinv2}}.
The encoding iterates \mintinline{yaml}{edge} for $X$ (Level~1),
then checks for constant $b$ in $X$'s row with a
hand-crafted \mintinline{yaml}{__VisitorMap} (Level~2), navigates
$\olbl{edge}.\olbl{__TailMap}.b$ for $Z$ (Level~3),
confirms \mintinline{yaml}{__Nil} (Level~4), and constructs $(X, Z)$
via \mintinline{yaml}{__Replacement}.
\fi

\subsection{\bilingual{Multi-Body Rules}{多体规则}}
\label{sec:encoding-multi}

\ifPaperChinese
含 $n > 2$ 个体原子的规则链接额外的连接层。
每个共享变量在前一层的 \mintinline{yaml}{__Visitor} 内引入一个组合的 \mintinline{yaml}{__WithVisitorTail}/\mintinline{yaml}{__WithVisiteeTail} 连接,输出构造放置在最内层的 \mintinline{yaml}{__Nil} 确认内部。
因此,二元规则的编码在结构上已是完整故事:更大的规则体只是增加同一模式的更多嵌套副本。
\else
Rules with $n > 2$ body atoms chain additional join levels.
Each shared variable introduces one combined
\mintinline{yaml}{__WithVisitorTail}/\mintinline{yaml}{__WithVisiteeTail} join
inside the previous level's \mintinline{yaml}{__Visitor},
and the output construction is placed inside the innermost
\mintinline{yaml}{__Nil} confirmation.
So the binary-rule encoding is the whole story structurally:
larger bodies only add more nested copies of the same pattern.
\fi

\paragraph{\bilingual{Three-body rule}{三体规则}}
\ifPaperChinese
$\mathrm{chain}(X, W) \mathrel{{:}\!{-}} \alpha(X, Y),\; \beta(Y, Z),\; \gamma(Z, W)$
需要两次连接:对 $Y$（在 $\alpha$ 和 $\beta$ 之间）以及对 $Z$（在 $\beta$ 和 $\gamma$ 之间）。
结构为:从 $\alpha$ 迭代 $X$、对 $Y$ 连接（获取 $\alpha$ 和 $\beta$ 两者的数据）、确认 $\alpha$ 中 $Y$ 之后的 \mintinline{yaml}{__Nil},然后在 \mintinline{yaml}{__Nil} 分支内将 $\beta$ 的 $Z$ 列与 $\gamma$ 的第一列连接、从 $\gamma$ 匹配行迭代 $W$、确认 \mintinline{yaml}{__Nil},并构造 $(X, W)$。
$Z$ 连接嵌套在 $Y$ 连接的 \mintinline{yaml}{__Nil} 确认内部,确保 $\beta$ 中匹配的 $Y$ 常量的数据（通过 \mintinline{yaml}{__VisitorTail} 可得）被用于迭代 $\beta$ 的第二列。
\else
$\mathrm{chain}(X, W) \mathrel{{:}\!{-}} \alpha(X, Y),\; \beta(Y, Z),\; \gamma(Z, W)$
requires two joins: on $Y$ (between $\alpha$ and $\beta$)
and on $Z$ (between $\beta$ and $\gamma$).
The structure is: iterate $X$ from $\alpha$, join on $Y$
(getting both $\alpha$'s and $\beta$'s data), confirm
\mintinline{yaml}{__Nil} after $Y$ in $\alpha$, then inside the
\mintinline{yaml}{__Nil} branch join $\beta$'s $Z$-column with
$\gamma$'s first column, iterate $W$ from $\gamma$'s
matched row, confirm \mintinline{yaml}{__Nil}, and construct
$(X, W)$.
The $Z$-join is nested inside the $Y$-join's
\mintinline{yaml}{__Nil} confirmation, ensuring that $\beta$'s data
for the matched $Y$ constant (available via
\mintinline{yaml}{__VisitorTail}) is used to iterate $\beta$'s
second column.
\fi

\paragraph{\bilingual{Mixed constants and variables}{混合常量与变量}}
\ifPaperChinese
$\mathrm{constJoin}(X, W) \mathrel{{:}\!{-}} \mathrm{edge}(X, Y),\; \mathrm{edge}(Y, b),\; \mathrm{edge}(b, W)$
结合了对 $Y$ 的变量连接、对 $b$ 的常量匹配（手工构造的仅含 $b$ 条目的 \mintinline{yaml}{__VisitorMap}）以及通过 $\olbl{edge}.\olbl{__TailMap}.b$ 对 $W$ 的直接导航。
对应的可执行 MIXINv2 文件为
\anon[\nolinkurl{RelationalConstantJoin.mixin.yaml}~(补充材料)]{\nolinkurl{RelationalConstantJoin.mixin.yaml}~\cite{yang2026-mixinv2}}。
\else
$\mathrm{constJoin}(X, W) \mathrel{{:}\!{-}} \mathrm{edge}(X, Y),\; \mathrm{edge}(Y, b),\; \mathrm{edge}(b, W)$
combines a variable join on $Y$, a constant match on $b$
(hand-crafted \mintinline{yaml}{__VisitorMap} with only the $b$ entry),
and direct navigation via $\olbl{edge}.\olbl{__TailMap}.b$
for $W$.
The corresponding executable MIXINv2 file is
\anon[\nolinkurl{RelationalConstantJoin.mixin.yaml}~(supplementary material)]{\nolinkurl{RelationalConstantJoin.mixin.yaml}~\cite{yang2026-mixinv2}}.
\fi

\subsection{\bilingual{Multi-Variable Join}{多变量连接}}

\ifPaperChinese
\begin{sloppypar}
当一条规则在两个体原子之间共享\emph{多个}变量时,每个共享变量需要其自身的 \mintinline{yaml}{__Acceptance} 检查。
这些检查链式排列:外层 \mintinline{yaml}{__Acceptance} 使用组合的 \mintinline{yaml}{__WithVisitorTail}/\mintinline{yaml}{__WithVisiteeTail} 对一对列连接,其 \mintinline{yaml}{__Visitor} 内含对另一对列进行连接的内层 \mintinline{yaml}{__Acceptance}。
\end{sloppypar}
\else
When a rule shares \emph{multiple} variables between two
body atoms, each shared variable requires its own
\mintinline{yaml}{__Acceptance} check.
The checks are chained: the outer \mintinline{yaml}{__Acceptance}
joins on one column pair using combined
\mintinline{yaml}{__WithVisitorTail}/\mintinline{yaml}{__WithVisiteeTail}, and
its \mintinline{yaml}{__Visitor} contains an inner \mintinline{yaml}{__Acceptance}
joining on another column pair.
\fi

\paragraph{\bilingual{Example}{示例}}
\ifPaperChinese
$\mathrm{multiJoin}(X, Y) \mathrel{{:}\!{-}} r(X, Y),\; s(X, Y)$
共享 $X$ 和 $Y$ 两者。
可执行的 MIXINv2 编码为
\anon[\nolinkurl{RelationalMultiVariableJoin.mixin.yaml}~(补充材料)]{\nolinkurl{RelationalMultiVariableJoin.mixin.yaml}~\cite{yang2026-mixinv2}}。
外层 \mintinline{yaml}{__Acceptance} 使用组合的 \mintinline{yaml}{__WithVisitorTail}/\mintinline{yaml}{__WithVisiteeTail} 连接 $r$ 和 $s$ 的第一列（变量 $X$）。
在外层 \mintinline{yaml}{__Visitor} 内部,内层 \mintinline{yaml}{__Acceptance} 类似地连接 $r$ 和 $s$ 的第二列（变量 $Y$）。
两次连接均须匹配才能产生输出,逐常量的 \mintinline{yaml}{__Visited} 模式确保只有两个关系中均存在的常量才有贡献。
\else
\begin{sloppypar}
$\mathrm{multiJoin}(X, Y) \mathrel{{:}\!{-}} r(X, Y),\; s(X, Y)$
shares both $X$ and $Y$.
The executable MIXINv2 encoding is
\anon[\nolinkurl{RelationalMultiVariableJoin.mixin.yaml}~(supplementary material)]{\nolinkurl{RelationalMultiVariableJoin.mixin.yaml}~\cite{yang2026-mixinv2}}.
The outer \mintinline{yaml}{__Acceptance} joins $r$'s and $s$'s first
columns (variable $X$) using combined
\mintinline{yaml}{__WithVisitorTail}/\mintinline{yaml}{__WithVisiteeTail}.
Inside the outer \mintinline{yaml}{__Visitor}, an inner
\mintinline{yaml}{__Acceptance} joins $r$'s and $s$'s second columns
(variable $Y$) similarly.
Both joins must match for the output to be produced, and
the per-constant \mintinline{yaml}{__Visited} pattern ensures that
only constants present in both relations contribute.
\end{sloppypar}
\fi

\subsection{\bilingual{Recursive Rules}{递归规则}}
\label{sec:recursive-rules}

\ifPaperChinese
递归 Datalog 规则（例如传递闭包
$\mathrm{path}(X, Y) \mathrel{{:}\!{-}} \mathrm{edge}(X, Z),\; \mathrm{path}(Z, Y)$）的编码与非递归规则完全相同:头部关系 \mintinline{yaml}{path} 既作为已定义的作用域出现,也作为规则体中的引用出现（为 $Z$ 连接提供 \mintinline{yaml}{__WithVisitorTail}）。
这不应被理解为单值惰性程序中普通的 $\this$。
在操作上,tabling 求值（第~\ref{sec:mixin-trees}~节,附录~\ref{app:well-definedness}）计算 $\lfp(T_P)$:对 \mintinline{yaml}{path} 的每次查询通过递归规则体触发进一步求值,trie 并（$\supers$）累积所有可达元组。
这与 $T_P$ 算子~\cite{vanemden1976-predicate-logic-semantics}的\emph{朴素求值}~\cite{bancilhon1986-recursive-query-strategies}对应:每次迭代将所有规则应用于当前事实集并添加新导出的事实。
由于 Datalog 编码的关系在每次查询中只有有限多个常量,查询的答案域是有限的,收敛性由有限格上的升链条件保证。
\emph{半朴素求值}~\cite{bancilhon1986-recursive-query-strategies}（限制每次迭代仅使用至少一个新导出事实的规则）是一种自然的优化,留待未来工作。
\else
Recursive Datalog rules, such as transitive closure
$\mathrm{path}(X, Y) \mathrel{{:}\!{-}} \mathrm{edge}(X, Z),\; \mathrm{path}(Z, Y)$,
are encoded identically to non-recursive rules: the
head relation \mintinline{yaml}{path} appears both as a defined scope
and as a reference in the body (providing
\mintinline{yaml}{__WithVisitorTail} for the $Z$-join).
This should not be read as ordinary $\this$ of a
single-valued lazy program.
Operationally, tabled evaluation
(Section~\ref{sec:mixin-trees},
Appendix~\ref{app:well-definedness})
computes $\lfp(T_P)$: each query into
\mintinline{yaml}{path} triggers further evaluation through
the recursive body, and the trie union ($\supers$)
accumulates all reachable tuples.
This corresponds to \emph{Naive
evaluation}~\cite{bancilhon1986-recursive-query-strategies} of the
$T_P$
operator~\cite{vanemden1976-predicate-logic-semantics}:
each iteration applies every rule to the current
fact set and adds newly derived facts.
Because Datalog-encoded relations have finitely many
constants per query, the per-query answer domain is
finite, and convergence is guaranteed by the ascending
chain condition on a finite lattice.
\emph{Semi-Naive} evaluation~\cite{bancilhon1986-recursive-query-strategies},
which restricts each iteration to rules that use at
least one newly derived fact, is a natural
optimisation left for future work.
\fi

\paragraph{\bilingual{The recursive rule as monadic pseudocode}{递归规则作为单子伪代码}}
\ifPaperChinese
上述递归规则大致对应以下伪代码（类 Haskell 语法,带有 \mintinline{haskell}{RebindableSyntax} 和 \mintinline{haskell}{OverloadedLists} 语义,其中 do 记法脱糖为基于 \mintinline{haskell}{Set} 的绑定,列表字面量表示集合）:
\else
The recursive rule above roughly corresponds to the
following pseudocode (Haskell-like syntax with
  \mintinline{haskell}{RebindableSyntax} and \mintinline{haskell}{OverloadedLists}
  semantics, where do-notation desugars to a
\mintinline{haskell}{Set}-based bind and list literals denote sets):
\fi
\begin{listing}[htbp]
  \caption{\bilingual{The recursive Datalog rule as monadic pseudocode}{递归 Datalog 规则的单子伪代码}}
  \label{lst:datalog-monadic}
\begin{minted}{haskell}
{-# LANGUAGE RebindableSyntax, OverloadedLists #-}
s >>= f = unions [f x | x <- s]
return x = [x]
guard True  = [()]
guard False = []
edge = [("a","b"), ("b","c"), ("c","a")]
path = edge <> do
  (x, y0) <- edge
  (y1, z) <- path
  guard (y0 == y1)
  return (x, z)
\end{minted}
\end{listing}
\noindent
\ifPaperChinese
伪代码的每一行在继承演算编码（第~\ref{sec:encoding-variables}~节）中都有直接对应:
\else
Each line of the pseudocode has a direct counterpart in
the inheritance-calculus encoding
(Section~\ref{sec:encoding-variables}):
\fi
\begin{description}
  \item[\mintinline{haskell}{edge}]
    \ifPaperChinese 对应基础事实 trie;\else corresponds to the base-fact trie;\fi
  \item[\mintinline{haskell}{path = edge <> do ...}]
    \ifPaperChinese
    对应 \mintinline{yaml}{path} 继承 \mintinline{yaml}{edge}（基础事实）和递归规则体（通过 $\supers$ 得到的导出事实）;
    \else
    corresponds to \mintinline{yaml}{path} inheriting \mintinline{yaml}{edge}
    (base facts) and the recursive body (derived facts via
    $\supers$);
    \fi
  \item[\mintinline{haskell}{(x, y0) <- edge}]
    \ifPaperChinese
    对应第~1~层（\mintinline{yaml}{__xAcc}）,从 \mintinline{yaml}{edge} 迭代 $X$;
    \else
    corresponds to Level~1 (\mintinline{yaml}{__xAcc}), iterating $X$
    from \mintinline{yaml}{edge};
    \fi
  \item[\mintinline{haskell}{(y1, z) <- path}]
    \ifPaperChinese
    \begin{sloppypar}
    对应第~2~层（\mintinline{yaml}{__yAcc}）,与 \mintinline{yaml}{path} 对 $Y$ 进行连接,以及第~4~层（\mintinline{yaml}{__zAcc}）,迭代 $Z$;
    \end{sloppypar}
    \else
    corresponds to Level~2 (\mintinline{yaml}{__yAcc}), joining on $Y$
    with \mintinline{yaml}{path}, and Level~4 (\mintinline{yaml}{__zAcc}),
    iterating $Z$;
    \fi
  \item[\mintinline{haskell}{guard (y0 == y1)}]
    \ifPaperChinese
    \begin{sloppypar}
    对应第~2~层中的继承 \mintinline{yaml}{[[__WithVisitorTail], [__WithVisiteeTail]]}:常量仅当同时出现在两个关系中时才产生 \mintinline{yaml}{__Visitor} 条目;
    \end{sloppypar}
    \else
    \begin{sloppypar}
    corresponds to the inheritance
    \mintinline{yaml}{[[__WithVisitorTail], [__WithVisiteeTail]]} in
    Level~2: a constant produces a \mintinline{yaml}{__Visitor} entry
    only if it appears in both relations;
    \end{sloppypar}
    \fi
  \item[\mintinline{haskell}{return (x, z)}]
    \ifPaperChinese
    对应第~5~层的 \mintinline{yaml}{__Replacement} 链,构造输出元组 $(X, Z)$。
    \else
    corresponds to the \mintinline{yaml}{__Replacement} chain in
    Level~5, constructing the output tuple $(X, Z)$.
    \fi
\end{description}
\ifPaperChinese
唯一的区别是求值策略:Haskell 惰性求值 \mintinline{haskell}{path}（发散）,而继承演算通过 tabling $T_P$ 迭代求值（收敛到 $\lfp$）。

然而,该单子程序并不收敛,因为 \mintinline{haskell}{path} 是递归定义的。%
\footnote{
  启用 RecursiveDo 扩展并无帮助:\mintinline{haskell}{mfix} 通过惰性在\emph{元素}级别打结,而传递闭包方程需要在\emph{容器}级别求不动点,即找到满足 $S = F(S)$ 的集合 $S$。
}
继承演算通过 tabling 求值（第~\ref{sec:mixin-trees}~节）解决了这一问题,它通过在 mixin 树的格上迭代来计算 $T_P$ 算子的最小不动点,恰好是上述伪代码无法表达的容器级不动点。
\else
The only difference is the evaluation strategy:
Haskell evaluates \mintinline{haskell}{path} lazily (diverges),
whereas inheritance-calculus evaluates it via tabled
$T_P$ iteration (converges to $\lfp$).

\begin{sloppypar}
However, this monadic program does not converge because
\mintinline{haskell}{path} is defined recursively.%
\footnote{
  Enabling the RecursiveDo extension does not help:
  \mintinline{haskell}{mfix} ties a knot at the \emph{element} level via laziness,
  whereas the transitive closure equation requires a
  fixed point at the \emph{container} level, that is, finding a
  set~$S$ such that $S = F(S)$.
}
Inheritance-calculus solves this through tabled evaluation
(Section~\ref{sec:mixin-trees}), which computes the
least fixed point of the $T_P$ operator by iterating
over the lattice of mixin trees, precisely the
container-level fixed point that the pseudocode above cannot express.
\end{sloppypar}
\fi

\subsection{\bilingual{Semantic Correspondence with Standard Datalog}{与标准 Datalog 的语义对应}}
\label{sec:datalog-semantics}

\ifPaperChinese
本编码针对\emph{正向}(无否定)、\emph{范围受限}(安全)、论域有限、以基本事实给出的 Datalog；在路径 interning 下可达答案域有限。该编码产生标准 Datalog 的元组级语义。
综合前几节:trie 提供带索引的元组存储,访问者分派提供控制骨架,逐常量的 \mintinline{yaml}{__Replacement} 将输出构造局限于每个见证。
关键机制是逐常量的 \mintinline{yaml}{__Visited} 模式:输出构造（\mintinline{yaml}{__Replacement}）发生在每个匹配常量的 \mintinline{yaml}{__Visitor} 作用域内部,作用于\emph{单个}常量的数据,而非合并后的作用域。
随后 \mintinline{yaml}{__Accepted} mixin 通过 trie 并（$\supers$）收集所有逐常量的 \mintinline{yaml}{__Visited} 结果。

例如,给定 $\mathrm{edge} = \{(a,b),\; (b,c),\; (c,a)\}$（一个 3-环），两跳自连接
$\mathrm{twoHop}(X, Z) \mathrel{{:}\!{-}} \mathrm{edge}(X, Y),\; \mathrm{edge}(Y, Z)$
恰好产生 $3$ 对:$\{(a,c),\; (b,a),\; (c,b)\}$,与标准 Datalog 一致。
这是
\nolinkurl{RelationalTwoHopRepeatedVisitor.mixin.yaml}
所呈现的行为。
当 $Y$ 连接匹配到常量 $b$（同时出现在 $(a,b)$ 产生的 \mintinline{yaml}{edge} 第二列和 \mintinline{yaml}{edge} 的第一列中）时,$b$ 的 \mintinline{yaml}{__Visitor} 接收 $\olbl{__VisitorTail} = \olbl{edge}.\olbl{__TailMap}.b$（仅含 $c$），\mintinline{yaml}{__Replacement} 构造单个输出元组 $(a, c)$,而非 $\{a,b,c\}$ 的全部。

这种值级精确性源于每个常量的 \mintinline{yaml}{__Replacement} 作用于 $\olbl{__TailMap}.d$（逐常量的子 trie）,而非合并后的 trie。
同样的观点在重复访问者自连接回归示例中也可见:重复的访问者分派阻止不同常量的见证被错误合并。
尽管 $\bases$（方程~\ref{eq:bases}）在基本层仍计算叉积,该编码在任何合并发生之前将计算路由通过逐常量的 \mintinline{yaml}{__TailMap} 条目,从而达到与标准 Datalog 元组级求值相同的效果。
\else
This encoding targets \emph{positive} (negation-free),
\emph{range-restricted} (safe) Datalog over a finite domain with ground facts;
under path-interning the reachable answer domain is finite.
The encoding produces standard Datalog tuple-level semantics.
Putting the previous subsections together: tries provide indexed
tuple storage, visitor dispatch supplies the control skeleton,
and per-constant \mintinline{yaml}{__Replacement} keeps output
construction local to each witness.
The key mechanism is the per-constant \mintinline{yaml}{__Visited}
pattern: output construction (\mintinline{yaml}{__Replacement})
occurs inside each matched constant's \mintinline{yaml}{__Visitor}
scope, operating on the \emph{individual} constant's data
rather than on the merged scope.
The \mintinline{yaml}{__Accepted} mixin then collects all per-constant
\mintinline{yaml}{__Visited} results via trie union ($\supers$).

\begin{sloppypar}
For example, given $\mathrm{edge} = \{(a,b),\; (b,c),\; (c,a)\}$
(a~3-cycle), the two-hop self-join
$\mathrm{twoHop}(X, Z) \mathrel{{:}\!{-}} \mathrm{edge}(X, Y),\; \mathrm{edge}(Y, Z)$
produces exactly $3$ pairs:
$\{(a,c),\; (b,a),\; (c,b)\}$, matching standard Datalog.
This is the behavior exhibited by
\nolinkurl{RelationalTwoHopRepeatedVisitor.mixin.yaml}.
When the $Y$-join matches constant $b$ (present in both
  \mintinline{yaml}{edge}'s second column from $(a,b)$ and
\mintinline{yaml}{edge}'s first column), the \mintinline{yaml}{__Visitor} for
$b$ receives $\olbl{__VisitorTail} = \olbl{edge}.\olbl{__TailMap}.b$
(containing only~$c$), and \mintinline{yaml}{__Replacement} constructs
a single output tuple $(a, c)$, rather than all of $\{a,b,c\}$.
\end{sloppypar}

\begin{sloppypar}
This value-level precision arises because each constant's
\mintinline{yaml}{__Replacement} operates on $\olbl{__TailMap}.d$
(the per-constant sub-trie) rather than on the merged
trie.
The same point is visible in
the repeated-visitor self-join regression example,
where repeated visitor dispatch prevents witnesses from
different constants from being merged spuriously.
Although $\bases$ (equation~\ref{eq:bases}) still computes
cross-joins at the primitive level, the encoding routes
computation through per-constant \mintinline{yaml}{__TailMap} entries
before any merging occurs, achieving the same effect as
standard Datalog's tuple-level evaluation.
\end{sloppypar}
\fi

\begin{proposition}[\bilingual{Datalog correspondence}{Datalog 对应}]%
  \label{prop:datalog-correspondence}
  \ifPaperChinese
  对于正向、范围受限、论域有限且以基本事实给出的 Datalog 程序,编码的继承演算求值恰好计算出最小 Herbrand 模型:一个基本元组在程序中可导出,当且仅当对应路径存在于编码的 $\lfp(T_P)$ 中。
  \else
  For a positive, range-restricted Datalog program over a finite domain with
  ground facts, the inheritance-calculus evaluation of the encoding computes
  exactly the least Herbrand model: a ground tuple is derivable in the program
  if and only if the corresponding path exists in $\lfp(T_P)$ of the encoding.
  \fi
\end{proposition}
\ifPaperChinese
我们不给出完整的形式化证明。严格建立它需要把上面的示意编码与实现逐常量分派的可执行 RepeatedVisitor 模块对齐。我们已在编码的示例上验证了这一对应(\nolinkurl{mixinv2-examples} 测试套件,含 3-环两跳连接的无虚假元组检查);完整的形式化超出本文范围。
\else
We do not give a full formal proof. Establishing it rigorously would require
reconciling the schematic encoding above with the executable RepeatedVisitor
modules that realize the per-constant dispatch. We have verified the
correspondence on the encoded examples (the \nolinkurl{mixinv2-examples} test
suite, including the no-spurious-tuple checks for the 3-cycle two-hop join); a
complete formalization is beyond the scope of this paper.
\fi

\section{\bilingual{Evaluation Trace of Multi-Target \texorpdfstring{$\this$}{this} Resolution}{多目标 \texorpdfstring{$\this$}{this} 解析的求值轨迹}}
\label{app:evaluation-trace}

\ifPaperChinese
本附录给出附录~\ref{app:scala-multi-target}~程序的逐步求值轨迹,通过第~\ref{sec:definitions}~节的五个函数（方程~\ref{eq:supers}~至~\ref{eq:this}）判定路径 $(\text{``HasMultipleOuters''},\, \text{``outer''},\, \text{``MyInner''})$ 是否存在。
$\overrides$、$\bases$、$\resolve$、$\this$ 和 $\supers$ 的每次调用均连同其参数和返回集合一起展示。
求值轨迹存在于元语言中:每个参数和结果都是一条\emph{路径},即标识树中某个位置的标签序列 $(\ell_1, \ldots, \ell_n)$（第~\ref{sec:definitions}~节），而非继承演算表达式。
为使标签与第~\ref{sec:definitions}~节的元语言变量（斜体 $p$、$\ell$、$n$、$S$）区分,我们对每个具体标签加引号,因此路径写作带括号的逗号分隔引号标签序列。
从而 $(\text{``HasMultipleOuters''}, \text{``outer''})$ 是继承演算投影 $\olbl{HasMultipleOuters}.\olbl{outer}$ 所标识的路径,根路径为空序列 $()$。
\else
This appendix gives a step-by-step evaluation trace of the program of
Appendix~\ref{app:scala-multi-target}, deciding whether the path
$(\text{``HasMultipleOuters''},\, \text{``outer''},\, \text{``MyInner''})$
exists by the five
functions of Section~\ref{sec:definitions}
(equations~\ref{eq:supers} to~\ref{eq:this}).
Every call to $\overrides$, $\bases$, $\resolve$, $\this$, and $\supers$
is shown with its arguments and returned set.
The trace lives in the metalanguage: every argument and result is a
\emph{path}, that is, a sequence of labels
$(\ell_1, \ldots, \ell_n)$ identifying a position in the tree
(Section~\ref{sec:definitions}), not an inheritance-calculus expression.
To keep labels distinct from the metalanguage variables of
Section~\ref{sec:definitions} (the italic $p$, $\ell$, $n$, $S$), we
quote each concrete label, so a path is written as a parenthesized,
comma-separated sequence of quoted labels.
Thus $(\text{``HasMultipleOuters''}, \text{``outer''})$ is the path that
the inheritance-calculus projection
$\olbl{HasMultipleOuters}.\olbl{outer}$ identifies, and the root path
is the empty sequence $()$.
\fi

\paragraph{\bilingual{Program}{程序}}
\begin{minted}{yaml}
- MyOuter:
  - MyInner:
    - outer:
      - [MyOuter, ~]
- Object1:
  - [MyOuter]
- Object2:
  - [MyOuter]
- HasMultipleOuters:
  - [Object1, MyInner]
  - [Object2, MyInner]
\end{minted}

\paragraph{\bilingual{The primitive \texorpdfstring{$\inherits$}{inherits} and the helper \texorpdfstring{$\has$}{has}}{原始量 \texorpdfstring{$\inherits$}{inherits} 与辅助函数 \texorpdfstring{$\has$}{has}}}
\ifPaperChinese
解析填充 $\has$ 和 $\inherits$（第~\ref{sec:definitions}~节）:
\else
Parsing populates $\has$ and $\inherits$
(Section~\ref{sec:definitions}):
\fi
\begin{align*}
  \has(\tuple{}, \ell) &\iff \ell \in \{
    \text{``MyOuter''},\;
    \text{``Object1''}, \\
    &\hphantom{{}\iff{}}
    \text{``Object2''},\;
  \text{``HasMultipleOuters''}\}, \\
  \has(\tuple{\text{``MyOuter''}}, \ell) &\iff \ell = \text{``MyInner''}, \\
  \has(\tuple{\text{``MyOuter''}, \text{``MyInner''}}, \ell) &\iff \ell = \text{``outer''}, \\
  \inherits(\tuple{\text{``MyOuter''}, \text{``MyInner''}, \text{``outer''}}) &=
  \{(1,\, \tuple{})\}, \\
  \inherits(\tuple{\text{``Object1''}}) &= \{(0,\, \tuple{\text{``MyOuter''}})\}, \\
  \inherits(\tuple{\text{``Object2''}}) &= \{(0,\, \tuple{\text{``MyOuter''}})\}, \\
  \inherits(\tuple{\text{``HasMultipleOuters''}}) &=
  \{(0,\, \tuple{\text{``Object1''}, \text{``MyInner''}}), \\
  &\hphantom{{}={}}
  (0,\, \tuple{\text{``Object2''}, \text{``MyInner''}})\}.
\end{align*}
\ifPaperChinese
在所有其他路径上 $\has$ 恒为假,$\inherits$ 在其上无定义(不落在 $\dom(\inherits)$ 内)。
每个 $\inherits$ 条目是引用对 $(n,\, w)$,其中 $n$ 是向上步数,$w$ 是向下投影列表（第~\ref{sec:definitions}~节）。
\else
On all other paths $\has$ is false, and they lie outside $\dom(\inherits)$.
Each $\inherits$ entry is a reference pair $(n,\, w)$,
where $n$ is the number of upward steps and $w$ the
downward projection list (Section~\ref{sec:definitions}).
\fi

\paragraph{\bilingual{Trace}{轨迹}}
\ifPaperChinese
我们以依赖顺序将求值轨迹呈现为语义函数调用的编号列表:某次调用出现在其所依赖的调用之后,每个条目说明它使用了哪些较早的条目、如何计算以及产生什么结果。
每个不同的调用列出一次,因此单个条目可被多个后续条目使用,这正是实现可以通过 tabling~\cite{tamaki1986-tabled-resolution}实现的复用。
参数和结果是元语言值:标签是引号字符串（如 $\text{``outer''}$），路径是标签的元组（如 $(\text{``HasMultipleOuters''}, \text{``outer''})$），根路径写作 $()$。
查询为：路径 $(\text{``HasMultipleOuters''}, \text{``outer''}, \text{``MyInner''})$ 是否存在。
\else
We present the trace as a numbered list of semantic-function calls in
dependency order: a call appears after the calls it depends on, and each
item states, in English, which earlier items it uses, how it computes, and
what it yields. Each distinct call is listed once, so a single item may be
used by several later items, the reuse that an implementation could realize
by tabling~\cite{tamaki1986-tabled-resolution}.
The arguments and results are metalanguage values: a label is a quoted
string such as $\text{``outer''}$, and a path is a tuple of labels such as
$(\text{``HasMultipleOuters''}, \text{``outer''})$, with the root path
written $()$.
The query is whether the path
$(\text{``HasMultipleOuters''}, \text{``outer''}, \text{``MyInner''})$ exists.
\fi

\begin{breakablepathcommas}
\raggedright
\begin{enumerate}
  \item \label{trace:overrides-MyOuter-MyInner}Compute $\overrides(\tuple{\text{``MyOuter''}, \text{``MyInner''}})$, the paths sharing the identity of $\tuple{\text{``MyOuter''}, \text{``MyInner''}}$. It keeps $\tuple{\text{``MyOuter''}, \text{``MyInner''}}$ and adds any same-label definition reached through the supers of its enclosing scope:
  \[
    \overrides(\tuple{\text{``MyOuter''}, \text{``MyInner''}}) = \{\tuple{\text{``MyOuter''}, \text{``MyInner''}}\}
  \]
  \item \label{trace:supers-MyOuter-MyInner}Compute $\supers(\tuple{\text{``MyOuter''}, \text{``MyInner''}})$ using item~\ref{trace:overrides-MyOuter-MyInner}. It takes the reflexive-transitive closure of $\bases$ from $\tuple{\text{``MyOuter''}, \text{``MyInner''}}$, pairing each reachable override with the inheritance site through which it is reached:
  \[
    \supers(\tuple{\text{``MyOuter''}, \text{``MyInner''}}) = \{(\tuple{\text{``MyOuter''}}, \tuple{\text{``MyOuter''}, \text{``MyInner''}})\}
  \]
  \item \label{trace:overrides-Has-outer}Compute $\overrides(\tuple{\text{``HasMultipleOuters''}, \text{``outer''}})$, the paths sharing the identity of $\tuple{\text{``HasMultipleOuters''}, \text{``outer''}}$ using item~\ref{trace:supers-MyOuter-MyInner}. It keeps $\tuple{\text{``HasMultipleOuters''}, \text{``outer''}}$ and adds any same-label definition reached through the supers of its enclosing scope:
  \[
    \begin{aligned}
      &\overrides(\tuple{\text{``HasMultipleOuters''}, \text{``outer''}}) ={} \\
      &\quad \left\{\begin{aligned}
          &\tuple{\text{``HasMultipleOuters''}, \text{``outer''}}, \\
          &\tuple{\text{``MyOuter''}, \text{``MyInner''}, \text{``outer''}}
        \end{aligned}\right\}
    \end{aligned}
  \]
  \item \label{trace:bases-Has-outer}Compute $\bases(\tuple{\text{``HasMultipleOuters''}, \text{``outer''}})$. It resolves the references carried by the overrides of $\tuple{\text{``HasMultipleOuters''}, \text{``outer''}}$ at the enclosing scope, and collects their targets:
  \[
    \begin{aligned}
      &\bases(\tuple{\text{``HasMultipleOuters''}, \text{``outer''}}) ={} \\
      &\quad \left\{\begin{aligned}
          &\tuple{\text{``Object1''}}, \\
          &\tuple{\text{``Object2''}}
        \end{aligned}\right\}
    \end{aligned}
  \]
  \item \label{trace:bases-MyOuter-MyInner-outer}Compute $\bases(\tuple{\text{``MyOuter''}, \text{``MyInner''}, \text{``outer''}})$. It resolves the references carried by the overrides of $\tuple{\text{``MyOuter''}, \text{``MyInner''}, \text{``outer''}}$ at the enclosing scope, and collects their targets:
  \[
    \bases(\tuple{\text{``MyOuter''}, \text{``MyInner''}, \text{``outer''}}) = \{\tuple{\text{``MyOuter''}}\}
  \]
  \item \label{trace:this-Has-MyOuter-MyInner-1}Compute one $\this$ step: from the frontier $\{\tuple{\text{``HasMultipleOuters''}}\}$, take the supers of each frontier path and keep those whose definition site is $\tuple{\text{``MyOuter''}, \text{``MyInner''}}$, collecting their inheritance sites as the new frontier:
  \[
    \begin{aligned}
      &\this(\{\tuple{\text{``HasMultipleOuters''}}\}, \tuple{\text{``MyOuter''}, \text{``MyInner''}}, 1) ={} \\
      &\quad \left\{\begin{aligned}
          &\tuple{\text{``Object1''}}, \\
          &\tuple{\text{``Object2''}}
        \end{aligned}\right\}
    \end{aligned}
  \]
  \item \label{trace:resolve-Has-MyOuter-MyInner-1}Compute $\resolve$ for the reference pair $(1, ())$ defined at $\tuple{\text{``MyOuter''}, \text{``MyInner''}}$, reached from inheritance site $\tuple{\text{``HasMultipleOuters''}}$ using item~\ref{trace:this-Has-MyOuter-MyInner-1}. It takes one step upward through $\this$ and then appends the projection:
  \[
    \begin{aligned}
      &\resolve(\tuple{\text{``HasMultipleOuters''}}, \tuple{\text{``MyOuter''}, \text{``MyInner''}}, 1, ()) ={} \\
      &\quad \left\{\begin{aligned}
          &\tuple{\text{``Object1''}}, \\
          &\tuple{\text{``Object2''}}
        \end{aligned}\right\}
    \end{aligned}
  \]
  \item \label{trace:resolve-0-MyOuter}Compute $\resolve$ for the reference pair $(0, (\text{``MyOuter''}))$ defined at $\tuple{}$, reached from inheritance site $\tuple{}$. It takes no step upward through $\this$ and then appends the projection:
  \[
    \resolve(\tuple{}, \tuple{}, 0, (\text{``MyOuter''})) = \{\tuple{\text{``MyOuter''}}\}
  \]
  \item \label{trace:supers-Has-outer}Compute $\supers(\tuple{\text{``HasMultipleOuters''}, \text{``outer''}})$ using items~\ref{trace:overrides-Has-outer}, \ref{trace:bases-Has-outer}, \ref{trace:bases-MyOuter-MyInner-outer}, \ref{trace:resolve-Has-MyOuter-MyInner-1} and~\ref{trace:resolve-0-MyOuter}. It takes the reflexive-transitive closure of $\bases$ from $\tuple{\text{``HasMultipleOuters''}, \text{``outer''}}$, pairing each reachable override with the inheritance site through which it is reached:
  \[
    \begin{aligned}
      &\supers(\tuple{\text{``HasMultipleOuters''}, \text{``outer''}}) ={} \\
      &\quad \left\{\begin{aligned}
          &(\tuple{}, \tuple{\text{``MyOuter''}}), \\
          &(\tuple{}, \tuple{\text{``Object1''}}), \\
          &(\tuple{}, \tuple{\text{``Object2''}}), \\
          &(\tuple{\text{``HasMultipleOuters''}}, \tuple{\text{``HasMultipleOuters''}, \text{``outer''}}), \\
          &(\tuple{\text{``HasMultipleOuters''}}, \tuple{\text{``MyOuter''}, \text{``MyInner''}, \text{``outer''}})
        \end{aligned}\right\}
    \end{aligned}
  \]
  \item \label{trace:property}Decide whether $\text{``MyInner''}$ is a property at $\tuple{\text{``HasMultipleOuters''}, \text{``outer''}}$, that is, whether the path $\tuple{\text{``HasMultipleOuters''}, \text{``outer''}, \text{``MyInner''}}$ exists. Item~\ref{trace:supers-Has-outer} exhibits a witnessing override:
  \[
    \begin{aligned}
      &(\_,\; \tuple{\text{``MyOuter''}}) \in \supers(\tuple{\text{``HasMultipleOuters''}, \text{``outer''}}), \\
      &\tuple{\text{``MyOuter''}, \text{``MyInner''}} \in \dom(\inherits), \\
      &\quad \text{reached through } \tuple{\text{``Object1''}} \text{ and } \tuple{\text{``Object2''}}, \\
      &\supers(\tuple{\text{``HasMultipleOuters''}, \text{``outer''}, \text{``MyInner''}}) \neq \varnothing.
    \end{aligned}
  \]
  so $\text{``MyInner''}$ is a property at $\tuple{\text{``HasMultipleOuters''}, \text{``outer''}}$.
\end{enumerate}

\end{breakablepathcommas}

\subsection{\bilingual{Recorded Outcomes of the Nat Case
    Study}{Nat 案例研究的运行记录}}
\label{app:nat-answers}

\ifPaperChinese
第~\ref{sec:case-study}~节的行为陈述是测试套件执行该节逐字展示的源文件所
记录的结果。下列条目由套件的生成器直接产出:读回诸项经测试套件的 ToPython
入口 mixin(算术入口即该节脚注所述者,笛卡尔积诸项用相应的笛卡尔积测试
入口),末项是不经 Python 模块的元语言见证。
\else
The behavioral statements of Section~\ref{sec:case-study} are
outcomes recorded by the test suite executing the very source files
that section displays verbatim. The following items are emitted
directly by the suite's generator: the readback items go through the
suite's ToPython entry mixins (the arithmetic entry is the one
described in that section's footnote; the Cartesian-product items
use the corresponding test entry), and the last item is a metalanguage witness
without Python modules.
\fi

\begin{breakablepathcommas}
\raggedright
\begin{enumerate}
  \item The sum of the case study's addition: the ToPython readback of the path $(\text{``ArithmeticTest''}, \text{``threePlusFour''}, \text{``sum''})$ is $\{\text{7}\}$.
  \item The equality check: the readback of $(\text{``ArithmeticTest''}, \text{``threePlusFourEqualsSeven''}, \text{``equal''})$ is $\{\text{True}\}$, so $(3+4) = 7$ holds in the executed encoding.
  \item The Cartesian-product semantics: the readback of $(\text{``CartesianProductTest''}, \text{``NatOneOrTwoPlusThreeOrFour''}, \text{``sum''})$ is $\{\text{4}, \text{5}, \text{6}\}$, the pointwise sums of $\{1, 2\}$ and $\{3, 4\}$.
  \item The relational equality: the readback of $(\text{``CartesianProductTest''}, \text{``NatResultEqualitySelf''}, \text{``equal''})$ is $\{\text{False}, \text{True}\}$: comparing the set-valued sum with itself observes both outcomes, one per pair of drawn values.
  \item The FFI-free witness: from the core path $(\text{``NatArithmeticTest''}, \text{``threePlusFourEqualsSeven''}, \text{``equal''})$ shown in the paper, iterating the $\bases$ step of equation~(\ref{eq:bases}) reaches the written definition $(\text{``NatArithmeticTest''}, \text{``BooleanFactory''}, \text{``True''})$ in 16 hops: the mixin that the reference resolves to inherits \mintinline{yaml}{True}.
\end{enumerate}

\end{breakablepathcommas}

\begin{remark}
\sloppy
\ifPaperChinese
  关键在于 $\this$ 步骤（条目~\ref{trace:this-Has-MyOuter-MyInner-1}）。
  $(\text{``MyOuter''}, \text{``MyInner''}, \text{``outer''})$ 处的继承引用向上走一步,因此 $\this$ 从 $(\text{``MyOuter''}, \text{``MyInner''})$ 向上走一步,产生前沿 $\{(\text{``Object1''}), (\text{``Object2''})\}$,即合并了 $\text{``MyInner''}$ 的两个继承站点。
  这正是 Scala~3 与 NixOS 模块系统都以静态错误拒绝的多目标情况（附录~\ref{app:scala-multi-target}）:在那里 $\this$ 必须选择单个站点。
  之所以 $(\text{``MyOuter''})$ 之下的路径可被命中,是因为 $(\text{``Object1''})$ 和 $(\text{``Object2''})$ 都以 $(\text{``MyOuter''})$ 为基础,因此最终的 $\supers$ 通过两条路径都包含 $(\text{``MyOuter''})$；由于一条路径只要某条路由上的 override 定义了其末标签即存在,两条路由折叠为同一结论:路径 $(\text{``HasMultipleOuters''}, \text{``outer''}, \text{``MyInner''})$ 与 $(\text{``MyOuter''}, \text{``MyInner''})$ 都存在。
\else
  The crux is the $\this$ step (item~\ref{trace:this-Has-MyOuter-MyInner-1}).
  The inheritance reference at
  $(\text{``MyOuter''}, \text{``MyInner''}, \text{``outer''})$ takes one step
  upward, so $\this$ moves one step up from
  $(\text{``MyOuter''}, \text{``MyInner''})$ and yields the frontier
  $\{(\text{``Object1''}), (\text{``Object2''})\}$, the two inheritance
  sites that incorporated $\text{``MyInner''}$. This is precisely the
  multi-target situation that Scala~3 and the NixOS module system both reject
  with a static error (Appendix~\ref{app:scala-multi-target}): there $\this$
  would have to select a single site. The paths under $(\text{``MyOuter''})$
  are reached because both $(\text{``Object1''})$ and
  $(\text{``Object2''})$ have $(\text{``MyOuter''})$ as a base, so the final
  $\supers$ contains $(\text{``MyOuter''})$ through both routes; since
  a path exists as soon as one override on some route defines its last label,
  the two routes collapse to the same conclusion: both
  $(\text{``HasMultipleOuters''}, \text{``outer''}, \text{``MyInner''})$ and
  $(\text{``MyOuter''}, \text{``MyInner''})$ exist.
\fi
\end{remark}

\ifPaperChinese\clearpage\end{CJK*}\fi
\end{document}